\documentclass[11pt,a4paper]{article}
\usepackage{jheppub}
\usepackage{amsthm}
\usepackage{microtype}
\usepackage[section]{placeins}
\hypersetup{pdftitle={A gauge-invariant measure for lattice chiral fermions},
  pdfauthor={},pdfsubject={}}
\newtheorem{theorem}{Theorem}[section]
\newtheorem{proposition}[theorem]{Proposition}
\newtheorem{lemma}[theorem]{Lemma}

\theoremstyle{remark}

\DeclareMathOperator{\Tr}{Tr}
\DeclareMathOperator{\tr}{tr}
\DeclareMathOperator{\SU}{SU}
\DeclareMathOperator{\U}{U}
\DeclareMathOperator{\Ad}{Ad}
\DeclareMathOperator{\im}{im}
\newcommand{\Id}{\mathbf 1}
\newcommand{\dd}{\mathrm d}
\newcommand{\ii}{\mathrm i}
\newcommand{\eps}{\varepsilon}
\allowdisplaybreaks[1]

\title{A gauge-invariant measure for lattice chiral fermions}

\author[a,b,c]{Nathaniel Craig}
\affiliation[a]{Kavli Institute for Theoretical Physics,
  University of California,\\
  Santa Barbara, CA 93106, USA}

\affiliation[b]{Department of Physics,
  University of California,\\
  Santa Barbara, CA 93106, USA}

\affiliation[c]{OpenAI,\\
  San Francisco, CA 94158, USA}

\emailAdd{ncraig@ucsb.edu}

\abstract{We nonperturbatively construct a gauge-invariant measure for overlap fermions  on four-dimensional Euclidean lattices in a class of nonabelian chiral gauge theories that includes the Standard Model. On admissible gauge fields, comparing chiral projectors at successive lattice spacings yields an exponentially local current and establishes the global properties needed for L\"{u}scher’s reconstruction of a smooth, gauge-invariant fermion measure. This provides a candidate nonperturbative lattice definition of the Standard Model.
}

\begin{document}
\maketitle
\flushbottom

\section{Introduction}
\label{sec:introduction}

What is the Standard Model? On one hand the answer is obvious: it is the astonishingly successful theory of the known fundamental particles and three of the four forces that bind them \cite{GlashowElectroweak,WeinbergLeptons,SalamElectroweak,FritzschGellMannLeutwyler}, validated by experiments across eighteen decades of energy yet simple enough to put on a coffee cup. On the other hand, it's fair to say that we don't even know what the theory {\it is}. Without a nonperturbative regulator respecting both gauge invariance and chiral symmetry, we have lacked a complete definition of the theory since its inception.\footnote{We also don't know the global structure of the Standard Model gauge group, which is invisible to perturbation theory and could take one of four possible forms \cite{HucksGlobal,TongGlobal}; here we treat all four possibilities.}

Lattice gauge theory provides a promising route to such a definition, via either a path integral formulation on a Euclidean spacetime lattice or a Hamiltonian formulation with continuous time on a spatial lattice. Either approach furnishes a nonperturbative, gauge invariant regulator that allows gauge theories to be defined at finite lattice spacing. Lattice QCD has turned this idea into a powerful quantitative framework for the strong interactions, but the weak interactions pose a challenge absent in QCD: left- and right-handed fermions carry different gauge charges. Accommodating this chiral structure without sacrificing gauge invariance remains a longstanding open challenge in extending lattice methods to the full Standard Model. More broadly, a lattice formulation of nonabelian chiral gauge theories in four dimensions remains a major open problem in theoretical physics.

In this paper we present a Euclidean lattice definition of a class of nonabelian gauge theories with chiral fermions in appropriate representations, including the Standard Model. Our construction uses overlap fermions with an admissibility constraint on the gauge fields. Overlap fermions shift the burden from chiral symmetry to gauge invariance, where we construct the necessary gauge-invariant measure satisfying L\"{u}scher's reconstruction criteria.

\subsection{Overlap fermions}

The challenge of realizing chiral fermions on the lattice is encoded in the Nielsen-Ninomiya theorem \cite{NielsenNinomiya}: under modest locality and regularity assumptions, a translation-invariant lattice Dirac operator cannot simultaneously have the desired continuum behavior, avoid extra species, and satisfy the conventional chirality condition $\{\gamma_5,D\}=0$. Wilson showed how to eliminate the extra species at the expense of chiral symmetry by adding a momentum-dependent mass \cite{WilsonFermions}. This can be achieved by replacing the massless Dirac operator in the fermion action with the Wilson operator
\begin{equation}
D_W(p)
= \frac{i}{a}\sum_{\mu=1}^{4}\gamma_\mu \sin(ap_\mu)
+ \frac{r}{a}\sum_{\mu=1}^{4}\left[1-\cos(ap_\mu)\right]
\end{equation}
where $a$ is the lattice spacing, $p$ is the lattice momentum, and $r$ is the Wilson parameter, which we henceforth set to one. The second term vanishes near the physical zero but gives the extra species masses of order $1/a$. The price to be paid is that $D_W$ no longer anticommutes with the chirality matrix $\gamma_5$, and conventional chiral symmetry is explicitly broken.

However, Narayanan and Neuberger \cite{NarayananNeuberger,NarayananNeubergerConstruction} identified a route to restoring chirality. Starting from the Wilson operator $D_W$, consider the Hermitian Wilson operator
\begin{equation}
H = \gamma_5(a D_W - 1)
\end{equation}
whose sign ${\rm sign} H \equiv  H/ \sqrt{H^2}$ returns $+1$ for every positive eigenvalue of $H$ and $-1$ for every negative eigenvalue while preserving the corresponding eigenspaces. We can use the Hermitian Wilson operator to construct the {\it overlap operator}
\begin{equation}
D = \frac{1}{a} (1 + \gamma_5 \, {\rm sign} H) .
\end{equation}
(Here {\it overlap} refers to Narayanan and Neuberger's original construction using overlaps of auxiliary many-fermion ground states, while Neuberger later introduced the above formulation in terms of the Hermitian Wilson operator \cite{Neuberger}.) In the free theory, $D = 0$ at $p = 0$ as desired, while $D = 2/a$ at each of the extra zeros of the naive lattice Dirac operator. The physical massless mode survives while the extra species are lifted, as with the Wilson operator.

The fact that $({\rm sign} H)^2 =1$ implies that the overlap operator satisfies the Ginsparg-Wilson relation \cite{GinspargWilson,NeubergerGW}
\begin{equation}
\{ \gamma_5,D \} = a D \gamma_5 D
\end{equation}
which restores the usual chirality condition in the continuum limit, exploiting a clear loophole in the Nielsen-Ninomiya theorem. However, as L\"{u}scher observed \cite{LuscherExactChiral}, overlap fermions do not merely enjoy chiral symmetry as $a \rightarrow 0$, but in fact possess an {\it exact} chiral symmetry at finite lattice spacing. To see this, define the modified chirality operator
\begin{equation}
\widehat{\gamma}_5 = \gamma_5 (1 - a D)
\end{equation}
in terms of which we can implement an infinitesimal lattice axial rotation by acting separately on the independent Euclidean Grassmann fields $\psi, \bar \psi$ as 
\begin{equation}
\delta \psi = i \alpha \widehat{\gamma}_5 \psi \hspace{1cm} \delta \bar{\psi} = i \alpha \bar{\psi} \gamma_5 \, .
\end{equation}
This transformation clearly leaves the free action invariant since $\gamma_5 D + D \widehat{\gamma}_5 = \gamma_5 D + D \gamma_5 - a D \gamma_5 D = 0,$ while recovering the familiar axial rotation in the continuum limit. Since $\widehat{\gamma}_5$ is Hermitian and squares to one, we can use it (along with the ordinary $\gamma_5$) to define genuine chiral projectors
\begin{equation}
\widehat{P}_- = \frac{1 - \widehat{\gamma}_5}{2} \hspace{1cm} P_+ = \frac{1 + \gamma_5}{2} .
\end{equation}
These projectors can be used to define a chiral theory by imposing e.g. $\widehat P_- \psi = \psi$ and $\bar \psi P_+ = \bar \psi$. We thus obtain lattice chiral fermions without extra species by constructing a fermion action from $D$ and imposing chirality conditions on the fermion and antifermion fields.

The discussion carries over to the interacting theory upon including gauge interactions, which also reveals the catch. The Wilson operator depends on gauge links $U$ via
\begin{equation}
D_W[U]=\frac{1}{2a}\sum_{\mu=1}^4
\left[
\gamma_\mu\bigl(T_\mu[U]-T_\mu[U]^\dagger\bigr)
+r\bigl(2-T_\mu[U]-T_\mu[U]^\dagger\bigr)
\right]
\end{equation}
where $(T_\mu[U]\psi)(x) = R(U(x,\mu))\,\psi(x+a\hat\mu)$; here $\hat \mu$ is the unit vector in the direction $\mu$, $U(x,\mu)$ carries out parallel transport from $x+a \hat \mu$ to $x$, and $R$ is the fermion representation. This introduces $U$-dependence into $H$, $D$, $\widehat{\gamma}_5$, and $\widehat P_-$. (Of course, the action and chirality constraints remain invariant when the gauge and fermion fields are transformed together.) Some care is required to keep the sign of $H[U]$ well-defined as $U$ varies; this can be accomplished by restricting to {\it admissible gauge fields} in which plaquettes are sufficiently close to the identity, which keeps the spectrum of $H[U]$ away from zero and thus makes $D[U]$ a smooth function of $U$~\cite{HJL}. 

But even with this restriction, there is a seemingly-insuperable problem: while we can again construct a chiral theory from $D[U]$, the subspace of fermion fields we integrate over now changes with $U$. The chirality condition $\widehat{P}_-[U] \psi = \psi$ specifies that subspace, but does not fix the phase of its integration measure. The difficulty is choosing this phase consistently as $U$ varies, while preserving locality, smoothness, and gauge invariance. Overlap fermions have apparently transferred the lattice challenge from chirality to gauge invariance.

To see the phase ambiguity explicitly, choose a smooth orthonormal basis $v_k(U)$ of the allowed fermion subspace near a given gauge field $U$, and a fixed orthonormal basis $\bar v_l$ for the antifermions. We can then write the fermion fields as
\begin{equation}
 \psi(x)=\sum_k v_k(x;U)c_k,
 \qquad
 \bar\psi(x)=\sum_l\bar c_l\bar v_l(x),
\end{equation}
where $c_k$ and $\bar c_l$ are independent Grassmann integration variables, and the sums run over the respective bases. The fermion path integral is then an integral over these coefficients, using the product of their Grassmann differentials. However, changing the orthonormal fermion basis while keeping $\psi$ fixed transforms the coefficients $c_k$ by a unitary matrix. Their integration differentials acquire the inverse determinant of this matrix, which is a phase; simply taking the product of differentials therefore leaves the measure dependent on an arbitrary basis choice. We can make the phase choice explicit by writing
\begin{equation}
 \mathcal D[\psi]\,\mathcal D[\bar\psi]
 =e^{i\varphi_v(U)}
   \left(\prod_k dc_k\right)
   \left(\prod_l d\bar c_l\right).
\end{equation}
Here $\varphi_v(U)$ is a real angle, defined modulo $2\pi$. If we change the basis, the factor $e^{i\varphi_v(U)}$ must compensate the phase acquired by the differentials, so that we still describe the same measure. This tells us how to compare different bases at the same gauge field. The remaining problem is to choose the phase as $U$ varies, so that transforming the gauge and fermion fields together leaves the full measure invariant.

\subsection{The measure current}

L\"uscher showed that this could be accomplished nonperturbatively for abelian gauge theories \cite{LuscherAbelian}. Rather than choosing $\varphi_v(U)$ directly, he first determined how the measure should respond to infinitesimal changes in the links $U$. Both the changing fermion basis $v_k(U)$ and the derivative of $\varphi_v(U)$ contribute to this response; together they define the \emph{measure term} $\mathfrak L$, which is linear in the link variation. The coefficient of each link variation is a component of the \emph{measure current}. This current must be smooth and gauge covariant, and its value at a given link must depend only exponentially weakly on distant links.

Gauge invariance constrains this response. Varying the fermion effective action gives one contribution from the change in the overlap operator $D[U]$ and another from the measure term $\mathfrak L$. Under a gauge transformation, these contributions must cancel, although neither need vanish separately. The operator contribution is expressed in terms of the lattice gauge anomaly; its cancellation requires the covariant divergence of the measure current to equal this anomaly.

This is not quite enough to define a consistent measure, however. As $U$ varies, the chiral projector $\widehat P_-[U]$ supplies a rule for transporting a fermion basis between the allowed subspaces. After a closed loop of gauge fields, the transported basis can differ from the initial basis by a unitary rotation, whose determinant gives the relative phase of the products $\prod_k dc_k$. (The antifermion basis $\bar v_l$ remains fixed.) The phase obtained by integrating $\mathfrak L$ around the loop and exponentiating must match this relative phase. This is the \emph{integrability condition}: it allows the compensating factor $e^{i\varphi_v(U)}$ to make the full measure return to itself. The condition must hold for every closed loop, not just infinitesimal ones; a global gauge anomaly can obstruct it even when the local anomaly cancels~\cite{WittenSU2,BarCampos}.

L\"uscher's abelian construction begins on the infinite lattice, with an exact classification of the local terms that can appear in the lattice anomaly~\cite{LuscherAbelianTopology}. For $U(1)$, the term that cannot be written as the divergence of a gauge-invariant local current has a coefficient proportional to the sum of the fermion charges cubed. When this sum vanishes, the entire lattice anomaly is such a divergence, at nonzero lattice spacing and not just in the continuum limit.

To impose integrability as well, he interpolates the abelian gauge potential from zero to its given value. The abelian field strength is linear in the potential, so the whole interpolation remains admissible. Varying the endpoint gauge field $U$ gives two nearby interpolation paths, enclosing a narrow strip in the space of gauge fields. The chiral projector $\widehat P_-[U]$ determines the phase around infinitesimal loops, described by its curvature. Integrating this projector curvature over the strip gives a quantity linear in the endpoint variation, and hence a candidate measure term. This construction supplies the required infinitesimal-loop phases. A correction built from the local current supplied by the anomaly classification then gives the measure current the required gauge invariance and divergence, without changing the phases around infinitesimal loops. To reconstruct the measure on a finite periodic lattice, one must also account for magnetic flux and the phases acquired by parallel transport around the periodic directions. On sufficiently large periodic lattices, L\"uscher treats these separately, adds finite-volume corrections, and verifies the remaining closed-loop conditions, with an additional restriction on the charges in nontrivial flux sectors~\cite{LuscherAbelian}.

For anomaly-free nonabelian theories, L\"uscher constructed measures to all orders in perturbation theory~\cite{LuscherPerturbative}. Nonperturbative extensions to nonabelian theories include Suzuki's measure for real representations~\cite{SuzukiReal} and Kadoh and Kikukawa's electroweak measure for vanishing hypercharge magnetic flux on sufficiently large lattices~\cite{KadohKikukawaEW,KadohKikukawaErratum}. Kikukawa also proposed a gauge-invariant measure for the sixteen-dimensional spinor representation of $\mathrm{Spin}(10)$; its smoothness and locality on general admissible fields remain open~\cite{KikukawaSpinTen}. Many other approaches to four-dimensional lattice chiral gauge theories have also been proposed and investigated~\cite{Swift1984,Smit1986,EichtenPreskill1986,Montvay1987,CreutzTytgatRebbiXue1997,Xue1997,BorrelliEtAl1990,Vink1994,Shamir1998Phase,GoltermanShamir1997,BockGoltermanShamir1998,GoltermanShamir2004,Smit1992Staggered,BockSmitVink1994,AlonsoEtAl1992,DrellWeinsteinYankielowicz1976,QuinnWeinstein1986,Rebbi1987,BietenholzWiese1996,FrolovSlavnov1994,Slavnov1995,GockelerSchierholz1992,tHooft1995,Hsu1995,Kronfeld1995,HernandezSundrumChiral,Bodwin1996,Kaplan1992,Shamir1993,GoltermanJansenPetcherVink1994,Kikukawa2002,GrabowskaKaplan2016Regulator,GrabowskaKaplan2016Chiral,FukayaEtAl2017,Kaplan2024,KaplanSen2026,BhattacharyaMartinPoppitz2006,PoppitzShang2007,Wen2013,YouBenTovXu2014,YouXu2015,WangWen2020,Catterall2021,Catterall2024,ThorngrenPreskillFidkowski2026,Zenkin1993,AoyamaKikukawa2019,PedersenKikukawa2023,GoltermanPetcher1989,GoltermanPetcherRivas1993,BockGoltermanShamir1998Perturbation,BockGoltermanLeungShamir2000,GoltermanShamir2003,FurmanShamir1995,GockelerKronfeldSchierholzWiese,Kronfeld1996,GoltermanJansenKaplan1993,KaplanSchmaltz1996,KaplanSen2024,ClancyKaplan2025}. The general problem for nonabelian gauge theories remains unsolved.

\subsection{This work}

L\"uscher's abelian construction does not extend directly to nonabelian gauge fields. One clear obstacle is that shrinking an abelian gauge potential also shrinks its field strength, whereas for a nonabelian potential the derivative and commutator terms scale differently, so shrinking the potential can spoil cancellations that kept the field strength small. The interpolation toward zero can therefore leave the admissible domain. Moreover, the abelian classification of the lattice anomaly does not supply the corresponding nonperturbative result for nonabelian theories~\cite{LuscherNonabelian}.

\begin{figure}[p]
\centering
\includegraphics[width=0.94\linewidth]{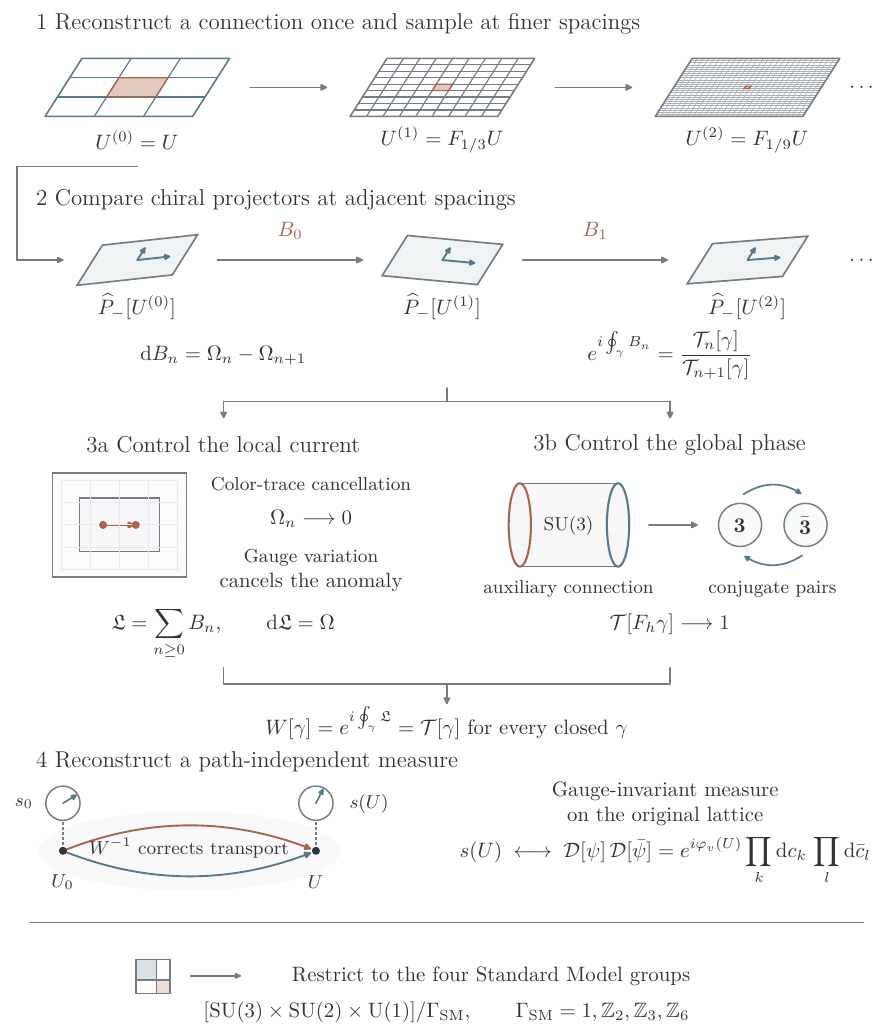}
\caption{Construction of the measure for $\mathrm{SU}(5)$ with $R=\mathbf{10}\oplus\bar{\mathbf5}$. Here $U$ denotes the original gauge field and $F_hU$ its reconstructed connection sampled at finer spacing $h$. At refinement step $n$, $\widehat P_-$ selects the allowed fermion states and $B_n$ contributes to the measure term $\mathfrak L$. The curvature $\Omega_n$ describes infinitesimal-loop phases; $\mathcal T_n[\gamma]$ and $W[\gamma]$ are phases around a closed gauge-field loop $\gamma$, obtained from the projectors and $\mathfrak L$, respectively. The correction $W^{-1}$ makes transport from $s_0$ at $U_0$ to $s(U)$ path independent. For fermion basis $v$ and a fixed antifermion basis, $s(U)$ is encoded in $e^{i\varphi_v(U)}$, multiplying the integration element for Grassmann coefficients $c_k,\bar c_l$. The subgroup $\Gamma_{\rm SM}$ specifies the Standard Model quotient.}
\label{fig:measure-construction}
\end{figure}

We instead construct the measure by comparing overlap fermions at successive lattice spacings. We begin with $\SU(5)$ and left-handed fermions in $R=\mathbf{10}\oplus\bar{\mathbf5}$, which contains one Standard Model generation, later extending our results to more general groups and fermion representations. The construction applies on finite periodic four-dimensional lattices for which the gauge fields satisfy a sufficiently small, volume-independent plaquette bound. The main steps are shown in Figure~\ref{fig:measure-construction}.

\begin{figure}[htbp]
  \centering
  \includegraphics[width=\linewidth]{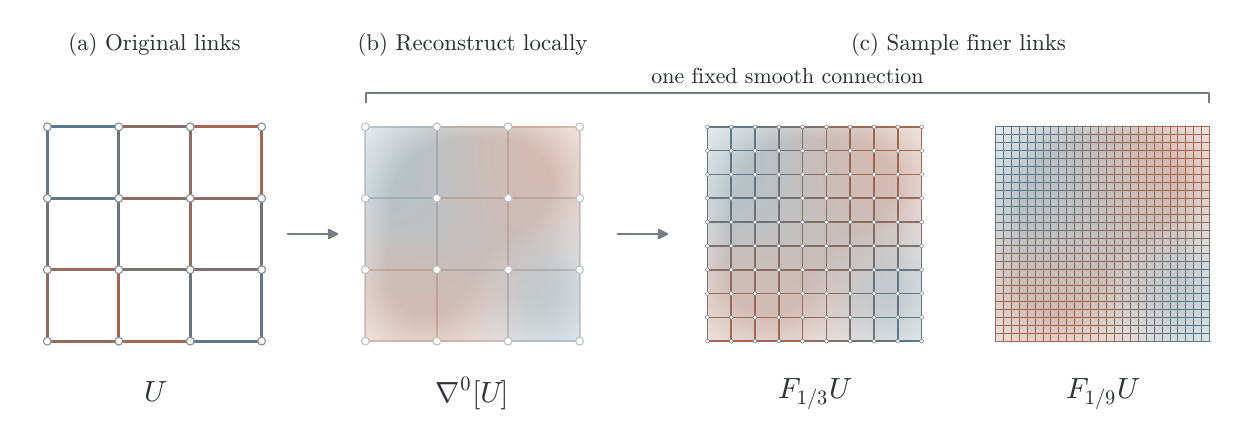}
  \caption{Reconstruction and sampling of the gauge links. (a) The original gauge links $U$, shown schematically in two dimensions. (b) We reconstruct a smooth gauge connection $\nabla^0[U]$ locally from these links, represented here by the shading. (c) We use parallel transport of the same connection along shorter edges to obtain (e.g.) the finer link fields $F_{1/3}U$ and $F_{1/9}U$, with spacings one-third and one-ninth of the original lattice spacing.}
  \label{fig:gauge_sampling}
\end{figure}

For a given gauge configuration, the original links $U$ specify parallel transport along the original lattice edges. To define links on finer lattices, we reconstruct a smooth gauge connection locally from $U$ and use its parallel transport along the finer edges. This procedure is illustrated schematically in Figure~\ref{fig:gauge_sampling}. Each sampled link field determines a Hermitian Wilson operator $H$ and a chiral projector $\widehat P_-$ by the same formulas as on the original lattice. We retain the actual original projector $\widehat P_-[U]$ first, followed by the projectors of the finer samples. During refinement, both the reconstructed connection and the original physical lattice spacing are held fixed. The physical fermion operator remains $D[U]$.

To compare the chiral projectors at two successive spacings, we first put their Hermitian Wilson operators on the same fermion space. We embed coarse fermion fields in the fine space by parallel-transporting each coarse-site spinor to nearby fine sites, with normalization preserving its norm. On the remaining fine-space degrees of freedom, we extend the coarse Hermitian Wilson operator by the fixed chirality matrix $\gamma_5$. On these added states, the chiral projector is $P_+=(1+\gamma_5)/2$, the same projector used for the antifermions. The added fermion and antifermion determinant factors therefore cancel, including their return phases.

We then interpolate linearly from the fine Hermitian Wilson operator to this extended coarse operator. We prove that the spectrum of the interpolating Hermitian operator stays away from zero throughout the interpolation. Consequently its sign, and hence the projector $(1+\operatorname{sign}H)/2$, vary smoothly between the fine and enlarged coarse chiral projectors. The intermediate Hermitian operators need not arise from gauge links. We therefore establish a gap for this operator interpolation without requiring an admissible interpolation of the gauge potential toward zero.

Next we vary the original links $U$ infinitesimally. Both endpoint operators change, giving a neighboring path of projectors; these two paths bound a narrow strip in the space of projectors. Integrating the projector curvature over this strip gives a phase angle that is linear in the original link variation, and hence a contribution to the measure term. We denote this contribution by $B_n$, with $n=0,1,\ldots$ labeling successive comparisons. Its coefficients in the original link variation give the corresponding contribution to the measure current. We propose the full measure term as the sum $\mathfrak L=\sum_{n=0}^{\infty}B_n$.

To see why this sum is useful, follow a closed loop of the original links $U$. Refinement gives a corresponding loop on each finer lattice. Each comparison term $B_n$ accounts for the change in projector return phase between two successive lattice spacings: $e^{i\oint B_n}$ is the coarser return phase divided by the finer one. The phase obtained from a finite sum of comparison terms is the product of these ratios, so all intermediate return phases cancel. It therefore equals the original return phase divided by the phase on the finest lattice included. Thus our proposed $\mathfrak L$ reproduces the original return phase required by integrability if the current converges and the remaining fine-lattice phase tends to one.

Anomaly cancellation makes the current converge. Finer sampling makes plaquettes approach the identity, but also increases the number of sites in the traces defining each $B_n$. An exact color-trace identity for $\mathbf{10}\oplus\bar{\mathbf5}$ cancels the leading terms, giving sufficient decay with the auxiliary spacing to overcome this growth. Estimates on these contributions and their original-link derivatives establish smoothness and exponential locality of the summed current. They also make the fine-lattice contribution to infinitesimal-loop phases tend to zero, giving local integrability. These estimates hold throughout the admissible domain, without assuming weak coupling or that the original links are near the identity.

The divergence equation follows from the same comparisons. Evaluating $B_n$ on an infinitesimal gauge transformation gives the anomaly contribution from the coarse overlap operator minus that from the fine overlap operator. Summing cancels the intermediate anomalies, and the fine-lattice remainder tends to zero. The covariant divergence of the resulting measure current therefore equals the original lattice anomaly. The construction also preserves gauge covariance and the required lattice transformation laws.

Full integrability requires the remaining fine-lattice projector phase to tend to one for every closed loop. We compare the reconstructed connection along each loop with an auxiliary $\SU(3)$ connection under which the charged fermions occur in conjugate pairs. For this auxiliary connection, including color parallel transport in the projector transport preserves the pairing and fixes the return phase to $+1$. We show that the corrections needed to return to the reconstructed connection and remove the extra color transport vanish under refinement, establishing the required limit.

We then choose a reference measure at one field in each connected component and transport it using $\widehat P_-[U]$. Multiplying by the inverse phase obtained by integrating the measure term $\mathfrak L$ along the transport path makes the result independent of that path. Together with gauge covariance and the divergence equation, this gives a gauge-invariant measure with $e^{i\varphi_v(U)}$ fixed up to one constant phase per disconnected component of admissible gauge fields. Physical requirements such as symmetry and cluster decomposition can constrain them, leading to familiar choices such as topological $\theta$ angles.

The Standard Model follows by mapping its gauge links into $\SU(5)$ and evaluating this measure on the resulting links. Choosing plaquette bounds that keep the image links admissible gives a measure for all four global Standard Model gauge groups shown in Fig.~\ref{fig:measure-construction}. Restricting $\mathbf{10}\oplus\bar{\mathbf5}$ gives one generation, and the generalization to three generations is immediate.

Our construction extends to other $\SU(N_{\rm c})$ theories under two sufficient conditions. Cancellation of the total cubic gauge anomaly supplies the local estimates. As for the global part, we employ a slightly more general procedure than we used for the Standard Model: under the $\SU(2)$ subgroup acting on two fundamental colors, each half-integer-isospin representation must occur with even multiplicity. Pairing these copies provides the conjugation symmetry needed for the global argument. Restricting these measures along group homomorphisms gives further compact gauge groups and their corresponding fermion representations. The additional even-multiplicity condition is tied to our specific approach to controlling the global phase; it is plausible that the measure construction can be extended to broader classes of anomaly-free chiral gauge theories, but we do not do so here. Representative theories covered by our construction are listed in
Table~\ref{tab:intro-covered-theories}.

\begin{table}[htbp]
\centering
\renewcommand{\arraystretch}{1.2}
\begin{tabular}{p{0.43\linewidth}p{0.48\linewidth}}
\hline
Gauge group & Left-handed fermion content \\
\hline
$[\SU(3)\times\SU(2)\times\mathrm{U}(1)]/\Gamma_{\rm SM}$\newline
$\Gamma_{\rm SM}=1,\mathbb Z_2,\mathbb Z_3,\mathbb Z_6$
& Any number of Standard Model generations \\
$\SU(N_{\rm c})$, $N_{\rm c}\ge5$
& One antisymmetric two-index tensor and $N_{\rm c}-4$ antifundamentals \\
$\SU(N_{\rm c})$, $N_{\rm c}\ge3$
& One symmetric two-index tensor and $N_{\rm c}+4$ antifundamentals \\
\hline
\end{tabular}
\caption{Representative theories covered by the construction on admissible
 gauge fields. The $\SU(5)$ member of the antisymmetric family has
 $R=\mathbf{10}\oplus\bar{\mathbf5}$. Section~\ref{sec:other-representations}
 gives broader families and the precise conditions on the fermion content.}
\label{tab:intro-covered-theories}
\end{table}

The paper is organized as follows: Sections~\ref{sec:lattice}--\ref{sec:measure} introduce the lattice fields and fermion measure and state L\"uscher's reconstruction conditions. Sections~\ref{sec:comparison}--\ref{sec:locality} construct the comparison current and establish its convergence, smoothness, and exponential locality. Section~\ref{sec:symmetries} proves the gauge and lattice transformation laws and the required divergence equation. Section~\ref{sec:global} establishes global integrability and reconstructs the measure. Section~\ref{sec:standard-model} restricts the result to the four Standard Model gauge groups, and section~\ref{sec:other-representations} extends it to other fermion representations and gauge groups. Section~\ref{sec:verification} discusses various checks of the construction, and section~\ref{sec:conclusions} summarizes our results and remaining open questions.

Detailed constructions and estimates are reserved for a series of appendices. Appendix~\ref{app:g:geometry} reconstructs the smooth connection from the original links. Appendix~\ref{app:s:spectral} proves that the comparison operators remain gapped. Appendix~\ref{app:a:locality} establishes the derivative and locality bounds, and appendix~\ref{app:d:cylinder-traces} supplies the phase estimates used in the global argument. Appendix~\ref{poly:section} expresses the current as an exact convergent series of polynomials in the original link matrices and their adjoints. Appendix~\ref{phys:section} studies the classical continuum behavior of the current and the convergence of polynomial approximations in fermion amplitudes at fixed cutoff. Finally, we reserve Lean verification of the main proof and supporting evidence from an analogous two-dimensional chiral gauge theory for supplementary material.

\section{Lattice fields and chiral projectors}
\label{sec:lattice}

As we have seen, the measure for overlap fermions will be defined on the fermion subspace selected by a
gauge-field-dependent chiral projector. We first fix the lattice and operator conventions that determine this subspace.

\subsection{Gauge fields and admissibility}

We set the original lattice spacing to one (until appendix~\ref{phys:section}, where we discuss aspects of the continuum limit). The Euclidean spacetime lattice is
\(\Gamma=\prod_{\mu=1}^4\mathbb Z/L_\mu\mathbb Z\), with
positive integer periods \(L_\mu\) and unit coordinate vectors
\(\hat\mu\). The continuous torus with these periods is
\(T^4=\prod_{\mu=1}^4\mathbb R/L_\mu\mathbb Z\).
Gauge fields and gauge transformations are periodic. Fermions
have independently chosen periodic or antiperiodic boundary
conditions in the four directions.

The matrix associated with the link \((x,x+\hat\mu)\) is
\(U(x,\mu)\in\SU(5)\). It transports a color vector at
\(x+\hat\mu\) to a color vector at \(x\). The reverse link
carries \(U(x,\mu)^{-1}\); only positive links are independent
variables. A gauge transformation \(g(x)\in\SU(5)\) acts as
\begin{equation}
 U^g(x,\mu)=g(x)U(x,\mu)g(x+\hat\mu)^{-1}.
\end{equation}
In geometric and operator formulas an edge may also be denoted
by \(e=(x,x+\hat\mu)\), with \(U_e=U(x,\mu)\).
Thus a sum over independent edges is a sum over \(x\in\Gamma\)
and \(\mu=1,\ldots,4\); it does not include a Lie-algebra index.

For the elementary plaquette \(p=(x;\mu,\nu)\), define
\begin{equation}
 U(p)=U(x,\mu)U(x+\hat\mu,\nu)
       U(x+\hat\nu,\mu)^{-1}U(x,\nu)^{-1}.
\end{equation}
The admissible field space is determined by the requirement
\begin{equation}
 \mathcal U_\eps=
 \{U:\|1-U(p)\|<\eps\text{ for every plaquette }p\}.
 \label{eq:admissibility}
\end{equation}
Here the norm is the operator norm in the fundamental
representation. Matrix norms below are also operator norms
unless specified otherwise. The admissibility threshold \(\eps>0\) will be
chosen independently of the lattice volume. A component means
a connected component of \(\mathcal U_\eps\), before taking
the quotient by gauge transformations. This condition does not require individual link matrices to be
close to the identity. In particular, it includes flat fields
with nontrivial parallel transport around the periodic directions.

\subsection{Fermions and the overlap operator}

We use anti-Hermitian generators \(T^a\), \(a=1,\ldots,24\),
of \(\mathfrak{su}(5)\), chosen with algebraic matrix entries
and normalized by
\begin{equation}
 -\tr_{\mathbf5}(T^aT^b)=\delta^{ab}.
\end{equation}
The symbol \(R\) denotes both the unitary group representation
\(\mathbf{10}\oplus\bar{\mathbf5}\) and its derived
Lie-algebra representation:
\begin{equation}
 R(X)=\left.\frac{\dd}{\dd u}R(e^{uX})\right|_{u=0},
 \qquad X\in\mathfrak{su}(5).
\end{equation}
For this fermion representation,
\(\|1-R[U(p)]\|\le2\|1-U(p)\|\).
Thus a sufficiently small plaquette bound in the fundamental
representation also implies the bound in \(R\) used by
L\"uscher~\cite[eq.~(2.1)]{LuscherNonabelian}.
The conjugate representation is denoted by \(R^*\), realized as
\(R^*(g)=\overline{R(g)}\) in fixed orthonormal bases.

The fermion Hilbert space \(\mathcal H\) consists of complex spinor
wavefunctions
\(\psi(x)\in\mathbb C^4\otimes\mathbb C^{15}\), where the
factors carry the spin and representation \(R\) indices.
The fields obey the chosen boundary conditions. The inner product on \(\mathcal H\) is
\((\psi,\chi)=\sum_{x\in\Gamma}\psi(x)^\dagger\chi(x)\),
antilinear in the first argument. A dagger denotes the
Hilbert-space or matrix adjoint. We use \(\Tr\) for the full
trace on a specified fermion space, including sites, spin,
color and any auxiliary factors, and \(\tr\) for a
finite-dimensional trace over the indices indicated.
The integration field \(\psi\) has Grassmann coefficients in a
basis of its allowed subspace of \(\mathcal H\).

The unitary covariant shift \(T_\mu=T_\mu[U]\) introduced above
translates a spinor by one site and transports its color components
back to \(x\):
\((T_\mu\psi)(x)=R(U(x,\mu))\psi(x+\hat\mu)\), with
the boundary sign included when a link crosses a periodic boundary.
Choose Hermitian Euclidean gamma matrices satisfying
\(\{\gamma_\mu,\gamma_\nu\}=2\delta_{\mu\nu}\), and
set \(\gamma_5=\gamma_1\gamma_2\gamma_3\gamma_4\).
With the Wilson parameter fixed to one, the Wilson operator
\(D_W\), Hermitian Wilson operator \(H\), and overlap operator \(D\) are
\begin{equation}
 \begin{split}
 D_W&=\frac12\sum_{\mu=1}^4
 \{\gamma_\mu(T_\mu-T_\mu^\dagger)
                  +2-T_\mu-T_\mu^\dagger\},\\
 H&=\gamma_5(D_W-1),\qquad
 D=1+\gamma_5\operatorname{sign}H.
 \end{split}
 \label{eq:wilson}
\end{equation}
The sign is defined by spectral calculus,
\(\operatorname{sign}H=H(H^2)^{-1/2}\). For sufficiently small \(\eps\), the spectrum of \(H\) stays a fixed
positive distance from zero, uniformly over admissible fields and
lattice volumes. The overlap operator \(D[U]\) is smooth in \(U\)
and exponentially local~\cite{HJL,Neuberger}. On each finer
auxiliary lattice, we again use dimensionless operators, with
its own spacing absorbed into \(D_W\) and \(D\). As illustrated in Fig.~\ref{fig:kernel-gap-and-overlap},
the gap in $H$ does not exclude zero modes of $D$.

As in the introduction, we impose chirality using projectors
on the fermion and antifermion fields that are compatible with \(D\). The overlap
operator satisfies the Ginsparg--Wilson relation
\cite{GinspargWilson,NeubergerGW}, which supplies the
gauge-field-dependent chirality operator \(\widehat\gamma_5\):
\begin{align}
 \gamma_5D+D\gamma_5&=D\gamma_5D,
 \label{eq:gw}\\
 \widehat\gamma_5&=\gamma_5(1-D)=-\operatorname{sign}H,
 &\widehat P_\pm&=\frac{1\pm\widehat\gamma_5}{2},
 &P_\pm&=\frac{1\pm\gamma_5}{2}.
 \label{eq:projectors}
\end{align}
In particular, \(\widehat P_-=(1+\operatorname{sign}H)/2\).
The fermion integration fields are Grassmann-valued:
\(\psi\) is a column field and \(\bar\psi\) an independent
row field. Their left-handed components obey
\begin{equation}
 \widehat P_-\psi=\psi,\qquad \bar\psi P_+=\bar\psi,
\end{equation}
and their action is \(S_F=\sum_x\bar\psi(x)(D\psi)(x)\).
The Ginsparg--Wilson relation gives
\(D\widehat\gamma_5=-\gamma_5D\), or equivalently
\(D\widehat P_-=P_+D\), so the two chirality constraints
are compatible with this action.
At unit lattice spacing the position-space kernel is defined by
\((D\psi)(x)=\sum_yD(x,y)\psi(y)\).
Because no eigenvalue of \(H\) can cross zero within the admissible
domain, the rank of \(\widehat P_-\) is constant on each connected
component of \(\mathcal U_\eps\). It can differ between components
and need not equal the fixed rank of \(P_+\).

\begin{figure}[htbp]
  \centering
  \includegraphics[width=\linewidth]{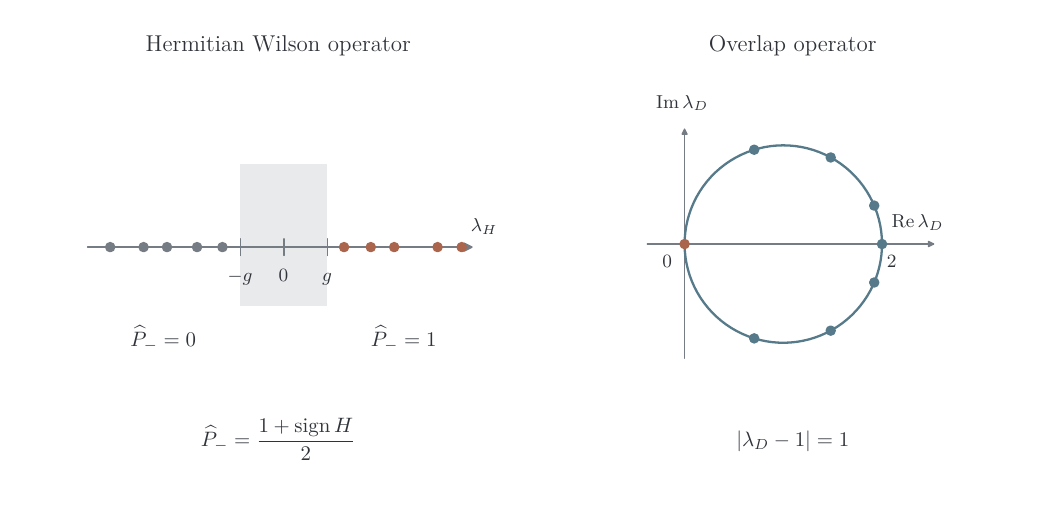}
  \caption{Schematic spectra of the Hermitian Wilson operator $H$ (left) and overlap operator $D$ (right). Here $\lambda_H$ and $\lambda_D$ denote their eigenvalues. The shaded interval $(-g,g)$ is the uniform spectral gap of $H$, with $g>0$; $\widehat P_-$ retains states with positive eigenvalues of $H$ and removes those with negative eigenvalues. The operator $\operatorname{sign}H$ assigns the sign of each eigenvalue of $H$. The overlap circle passes through zero, so the gap of $H$ does not exclude zero modes of $D$. The dots mark illustrative spectral locations, not eigenvalues computed for a particular gauge field.}
  \label{fig:kernel-gap-and-overlap}
\end{figure}

\section{Fermion integration measure}
\label{sec:fermion-measure}

Next we define the fermion integration measure and its phase
freedom, then explain how the measure term enters variations
of the fermion effective action.

\subsection{Basis dependence and the phase of the measure}

The chirality constraint fixes the allowed fermion fields,
but their integration measure still has a phase freedom
that must be fixed consistently with gauge invariance.
We first describe this freedom directly in the measure,
then give its expression in Grassmann coordinates.

On a sufficiently small neighborhood of an admissible gauge
field, choose a smooth orthonormal reference basis \(v_k(U)\),
\(k=1,\ldots,N\), of the allowed fermion subspace, where
\(N=\operatorname{rank}\widehat P_-\):
\begin{equation}
 \widehat P_-[U]v_k(U)=v_k(U),\qquad
 (v_k(U),v_l(U))=\delta_{kl}.
\end{equation}
The rank \(N\) is constant on the connected component.
The antifermion projector \(P_+\) is independent of \(U\).
Choose once an orthonormal basis of row vectors \(\bar v_l\),
\(l=1,\ldots,\bar N\), satisfying
\(\bar v_lP_+=\bar v_l\), where
\(\bar N=\operatorname{rank}P_+\).

Under a change of orthonormal basis, the product of Grassmann
integration differentials acquires a phase, the Berezin Jacobian.
The exterior product of the basis vectors acquires the same phase,
so it can represent the measure independently of Grassmann coordinates.
For an \(n\)-dimensional complex vector
space \(V\) with ordered basis
\(e_1,\ldots,e_n\), the determinant line is the
one-dimensional space
\begin{equation}
 \det V:=\bigwedge^n V
 =\{z\,e_1\wedge\cdots\wedge e_n:z\in\mathbb C\}.
\end{equation}
The wedge denotes the fully antisymmetrized product. Thus \(\det V\) denotes a
vector space, not a numerical determinant of an operator.
Its name reflects the fact that a change of basis
multiplies the exterior product by the determinant
of the change-of-basis matrix. For orthonormal bases,
this determinant is a phase.

For Grassmann fields, the fermion integration measure
is represented by an element of the determinant line
of the allowed field space. Here the fermion vectors
belong to \(\im\widehat P_-[U]\), while the antifermion
rows are linear functionals on \(\im P_+\). Together
they give the relative determinant line~\cite[sect.~6]{AdamsFamilies}
\begin{equation}
 \mathcal D(\widehat P_-[U])
 =\det\im\widehat P_-[U]\otimes(\det\im P_+)^*,
 \label{eq:measure-line}
\end{equation}
where the star denotes the complex-linear dual.

With a fixed ordering, the reference bases determine
a unit-norm element \(s_v(U)\) of this line. A general
normalized measure \(s(U)\) can therefore be written as
\begin{equation}
 \begin{split}
 s_v(U)
 &=\bigl[v_1(U)\wedge\cdots\wedge v_N(U)\bigr]
   \otimes
   \bigl[\bar v_1\wedge\cdots\wedge\bar v_{\bar N}\bigr],
 \\
 s(U)&=e^{\ii\varphi_v(U)}s_v(U).
 \end{split}
 \label{eq:measure-exterior}
\end{equation}
The angle \(\varphi_v(U)\) is real and defined modulo
\(2\pi\). It specifies the phase of the measure relative
to the reference element \(s_v(U)\).
Even with \(\varphi_v=0\), the measure depends on \(U\)
through the fermion basis vectors in the exterior product.
Section~\ref{sec:measure-reconstruction} constructs
compatible phase choices giving a smooth, gauge-invariant
measure.

To express this measure in Grassmann coordinates,
introduce Grassmann coefficients \(c_k,\bar c_l\), independent of
\(U\), and write
\begin{equation}
 \psi(x)=\sum_{k=1}^N v_k(x;U)c_k,\qquad
 \bar\psi(x)=\sum_{l=1}^{\bar N}\bar c_l\bar v_l(x).
 \label{eq:fermion-coordinate-expansion}
\end{equation}
Using the corresponding fixed ordering of the integration
differentials, the coordinate expression of \(s(U)\) is
\begin{equation}
 \mathcal D[\psi]\,\mathcal D[\bar\psi]
 =e^{\ii\varphi_v(U)}
   \left(\prod_{k=1}^N\dd c_k\right)
   \left(\prod_{l=1}^{\bar N}\dd\bar c_l\right).
 \label{eq:measure-coordinates}
\end{equation}
The Grassmann coefficients are held fixed when the gauge
field is varied.

To compare two reference bases \(v_k(U)\) and
\(\widetilde v_k(U)\) at the same gauge field, define the
\(N\times N\) overlap matrix
\begin{equation}
 \mathcal Q_{kl}(U)
   =(\widetilde v_k(U),v_l(U)),\qquad k,l=1,\ldots,N.
\end{equation}
Both bases are orthonormal and span the same subspace,
so \(\mathcal Q\) is unitary. The basis vectors and the
coordinates of the same fermion field are related by
\begin{equation}
 \widetilde v_k=\sum_{l=1}^N
                  v_l(\mathcal Q^{-1})_{lk},\qquad
 \widetilde c_k=\sum_{l=1}^N\mathcal Q_{kl}c_l.
\end{equation}
These relations compare coordinates at a fixed \(U\).
The exterior product and the Berezin Jacobian transform
by the same phase:
\begin{equation}
 \begin{split}
 s_{\widetilde v}(U)
   &=(\det\mathcal Q)^{-1}s_v(U),\\
 \prod_{k=1}^N\dd\widetilde c_k
   &=(\det\mathcal Q)^{-1}\prod_{k=1}^N\dd c_k.
 \end{split}
\end{equation}
The two descriptions therefore represent the same
measure when their relative phases obey
\begin{equation}
 e^{\ii\varphi_{\widetilde v}(U)}
   =\det\mathcal Q(U)\,e^{\ii\varphi_v(U)}.
 \label{eq:measure-phase-change}
\end{equation}
At fixed relative phase, this leaves precisely the
determinant-one freedom in the basis described by
L\"uscher~\cite[sect.~3]{LuscherNonabelian}.

Figure~\ref{fig:subspace-and-measure-phase} illustrates how
the allowed fermion subspace changes with $U$ and how
different reference bases describe the same measure.

\begin{figure}[htbp]
  \centering
  \includegraphics[width=\linewidth]{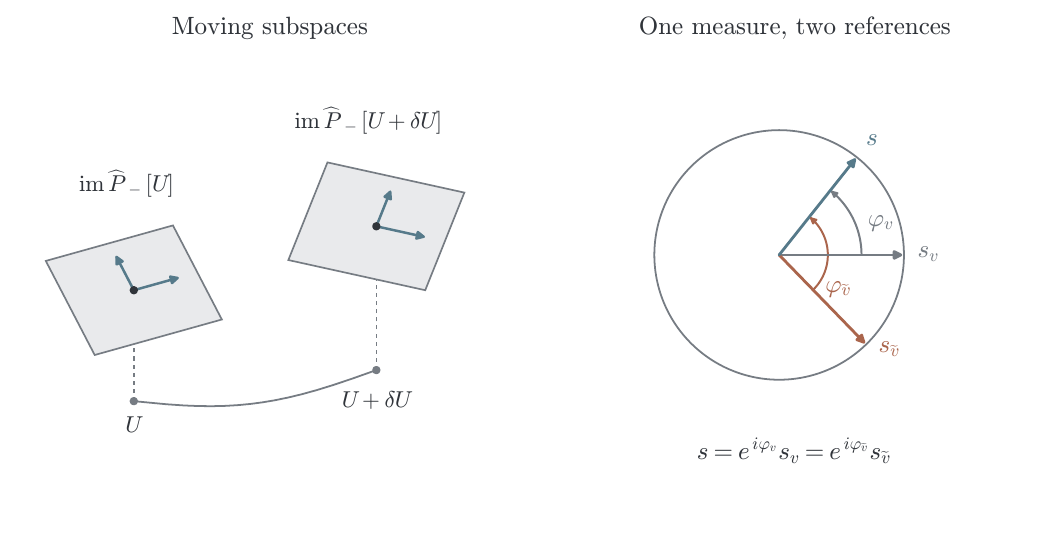}
  \caption{Left: varying the gauge field from $U$ to $U+\delta U$ moves the allowed fermion subspace $\operatorname{im}\widehat P_-$; the planes and arrows represent subspaces and their bases. Right: at fixed $U$, two orthonormal fermion bases define reference measures $s_v$ and $s_{\widetilde v}$. The angles $\varphi_v$ and $\varphi_{\widetilde v}$ specify the same measure $s$ relative to these references. Each reference includes the fixed antifermion factor.}
  \label{fig:subspace-and-measure-phase}
\end{figure}

Under a gauge transformation, the fields transform as
\begin{equation}
 \psi(x)\mapsto R(g(x))\psi(x),\qquad
 \bar\psi(x)\mapsto\bar\psi(x)R(g(x))^{-1}.
\end{equation}
Applying these transformations to each factor in the
exterior products defines the induced action on the
measure, denoted by \(g\,s(U)\).
Gauge invariance requires \(s(U^g)=g\,s(U)\), or
equivalently
\begin{equation}
 e^{\ii\varphi_v(U^g)}s_v(U^g)
 =e^{\ii\varphi_v(U)}\,g\,s_v(U).
\end{equation}
The reference element chosen at \(U^g\) can differ by a
phase from the gauge transform of the reference element
at \(U\). The relative phases must compensate this
difference, so that the full measure is gauge invariant.

A smooth measure requires smooth local phase choices
obeying \eqref{eq:measure-phase-change} on overlapping
neighborhoods. A single reference basis over the whole
component is not required. This definition also applies
when \(N\ne\bar N\) or when \(D\) has zero modes.

\subsection{The measure term and its current}

We now describe the contribution of the measure to variations
of the fermion effective action in terms of a current, which
will be the starting point of our construction.

Consider an infinitesimal variation of the links,
\begin{equation}
 \delta_\eta U(x,\mu)=\eta_\mu(x)U(x,\mu),\qquad
 \eta_\mu(x)=\sum_{a=1}^{24}\eta_\mu^a(x)T^a.
 \label{eq:link-variation}
\end{equation}
The components \(\eta_\mu^a(x)\) are real, and
\(\delta_\eta\) denotes differentiation in this direction.
The Grassmann coefficients remain fixed. The field \(\psi\)
therefore varies through \(v_k(U)\), while \(\bar\psi\) remains fixed.

As \(U\) changes, the measure vector \(s(U)\) moves in the
fixed exterior-product space built from \(\mathcal H\).
Its derivative includes motion of the allowed fermion subspace.
The component along \(s(U)\) records the phase change. We extract
this component with the inner product induced by that on
\(\mathcal H\), and define the \emph{measure term} by
\begin{equation}
 \mathfrak L_\eta
 =\ii(s,\delta_\eta s)
 =\ii\sum_{k=1}^N(v_k,\delta_\eta v_k)
                   -\delta_\eta\varphi_v.
 \label{eq:measure-basis}
\end{equation}

The two terms account for the changing reference basis and
relative phase. L\"uscher incorporates the phase choice into
the basis representing the measure~\cite{LuscherNonabelian};
keeping it explicit here produces the second term.

Orthonormality makes each \((v_k,\delta_\eta v_k)\) purely
imaginary, and \(\varphi_v\) is real, so
\(\mathfrak L_\eta\) is real. Under a change of reference
basis, the phase transformation
\eqref{eq:measure-phase-change} makes the changes in the two
terms cancel. Thus \(\mathfrak L_\eta\) depends on the chosen measure,
but not on the reference basis used to express it.

Its linearity in the components \(\eta_\mu^a(x)\) defines
the \emph{measure current}:
\begin{equation}
 \mathfrak L_\eta
 =\sum_{x\in\Gamma}\sum_{\mu=1}^4\sum_{a=1}^{24}
       \eta_\mu^a(x)j_\mu^a(x;U).
 \label{eq:measure-term}
\end{equation}

To see how the measure term enters the fermion integral,
substitute the field expansions into the action:
\begin{equation}
 \begin{split}
 S_F&=\sum_{l=1}^{\bar N}\sum_{k=1}^N
                 \bar c_l M_{lk}(U)c_k,\\
 M_{lk}(U)&=\sum_{x\in\Gamma}\bar v_l(x)(Dv_k)(x).
 \end{split}
\end{equation}
Here \(M\) contains the matrix elements of \(D\) between
the two reference bases; it has \(\bar N\) rows and \(N\)
columns. When \(N=\bar N\), we fix the Berezin normalization
so that the Gaussian integral over the coefficients \(c_k,\bar c_l\)
equals \(\det M\).
The fermion partition function is then
\begin{equation}
 Z_F(U)=\int\mathcal D[\psi]\,\mathcal D[\bar\psi]\,
                    e^{-S_F}
       =e^{\ii\varphi_v(U)}\det M(U).
\end{equation}
Where \(D\) is invertible, the relation
\(D\widehat P_-=P_+D\) makes its restriction an invertible map
from \(\im\widehat P_-\) to \(\im P_+\). Thus
\(N=\bar N\) and \(M\) is invertible. Differentiating
\(M\) changes both \(D\) and the fermion reference basis.
Since \(\bar v_lP_+=\bar v_l\), the same relation ensures
that only the component of \(\delta_\eta v_k\) within
\(\im\widehat P_-\) contributes to \(\delta_\eta M\)
through the basis variation. Its contribution
to \(\delta_\eta\log\det M\) is
\(\sum_k(v_k,\delta_\eta v_k)\). Together with the derivative
of the explicit phase in \(Z_F\), this is
\(-\ii\mathfrak L_\eta\). The remaining contribution comes
from varying \(D\), giving
\begin{equation}
 \delta_\eta\log Z_F
 =\Tr\{(\delta_\eta D)\widehat P_-D^{-1}P_+\}
       -\ii\mathfrak L_\eta.
 \label{eq:determinant-variation}
\end{equation}
The fermion
effective action is \(-\log Z_F\), so its variation receives
the measure contribution \(+\ii\mathfrak L_\eta\).
Locality of the measure current ensures that the interaction
vertices arising from this contribution are exponentially local.
The definitions of the measure term and current use no
inverse of \(D\) and remain applicable when \(D\) has zero
modes or \(N\ne\bar N\).

We will construct the current first and prove that it arises
from a smooth, gauge-invariant fermion measure. The next
section states the conditions needed for this reconstruction.

\section{Conditions on the measure current}
\label{sec:measure}

We now state the conditions under which a proposed current
determines a smooth, local, gauge-invariant fermion
measure~\cite{LuscherNonabelian}. We then state the
existence result established by our construction.

\subsection{Smoothness and locality}
\label{sec:locality-definition}

The current at an original link may depend on the whole gauge
field, but its response to a distant link must be exponentially
small. We impose this requirement also on higher derivatives,
which describe its response to several link variations. The
bounds below control the total response after summing over
the positions of those variations, uniformly in the volume.

To keep track of how far the dependence reaches, write
\begin{equation}
 j_\mu^a(x;U)=\sum_{r\ge0}j_{\mu,r}^a(x;U),
\end{equation}
where the term with integer radius \(r\) uses only links
within distance \(r\) of \((x,x+\hat\mu)\).
Distances are shortest periodic graph distances in original
lattice units; the distance between edges is the minimum
distance between their endpoints. Locality will follow if
the terms that reach farther away become sufficiently small,
including after differentiation.

For each finite-support term we require a smooth extension
to independent \(\SU(5)\) matrices on its input links.
We choose and fix these extensions for the estimates.
This lets us bound a term using only its own finite set of
inputs, without asking whether those inputs extend to an
admissible field elsewhere. The suprema below impose no
plaquette constraint on these extensions; the fermion measure
itself is still defined on \(\mathcal U_\eps\).

Let \(F\) be a scalar or matrix-valued function of the gauge
field. Varying just the link \((y,y+\hat\nu)\) in the
direction \(T^b\) defines
\begin{equation}
 D_{y,\nu,b}F(U)=
 \left.\frac{\dd}{\dd u}
 F(\ldots,e^{uT^b}U(y,\nu),\ldots)\right|_{u=0}.
 \label{eq:link-derivative}
\end{equation}
This is \eqref{eq:link-variation} with generator \(T^b\)
on the chosen link and zero elsewhere. To write several
such derivatives, let \(I=(x,\mu,a)\) label the current
component and \(J=(y,\nu,b)\) label a varied link and its
generator. Set \(j_{I,r}=j_{\mu,r}^a(x)\) and
\(D_J=D_{y,\nu,b}\).

We require a positive constant \(\lambda\) for which
\begin{equation}
 \sup_I\sum_{r\ge0}e^{\lambda r}
 \sum_{J_1,\ldots,J_\ell}
 \sup_U\left|D_{J_1}\cdots D_{J_\ell}j_{I,r}(U)\right|
 \le C_\ell
 \qquad(\ell=0,1,\ldots).
 \label{eq:current-locality}
\end{equation}
The exponential weight requires decay as the support radius
grows. The inner sum controls the response to variations at
all positions and in all generator directions, while the
outer supremum makes the bound uniform in the current
component. For \(\ell=0\), the inner sum is absent.
Derivatives on the same link need not commute, so their
displayed order is retained. Each supremum over \(U\)
ranges over the independent input matrices of the chosen
extension.

The constants \(C_\ell\) may depend on the derivative
order and the fixed construction, but not on volume,
admissible component or refinement level. One positive
\(\lambda\) must work for every derivative order, volume
and component. 
For example, a link at distance \(d\) from
\((x,x+\hat\mu)\) cannot affect any term with \(r<d\).
Only the tail \(r\ge d\) contributes to its derivative,
and \eqref{eq:current-locality} gives
\begin{equation}
 |D_{y,\nu,b}j_\mu^a(x;U)|\le C_1e^{-\lambda d}.
\end{equation}
The corresponding bounds at every fixed derivative order
allow differentiation of the series and make the current
smooth.

We use \(C\) when no derivative order is displayed; its
value may change between estimates. Lower-bound constants
denoted by \(c\) are positive.

\subsection{Gauge invariance and the anomaly}

Under an infinitesimal gauge transformation, the two
contributions in \eqref{eq:determinant-variation} must
cancel. This cancellation gives a covariant divergence equation for the measure current.

In matrix notation the current is
\(j_\mu(x;U)=\sum_a j_\mu^a(x;U)T^a\). We require it to
transform covariantly:
\begin{equation}
 j_\mu(x;U^g)=g(x)j_\mu(x;U)g(x)^{-1}.
\end{equation}
For \(X,Y\in\mathfrak{su}(5)\), write
\(\Ad_gX=gXg^{-1}\) for the adjoint action and
\(\langle X,Y\rangle=-\tr_{\mathbf5}(XY)\) for the
invariant inner product. Thus \eqref{eq:measure-term} can also be
written as \(\sum_{x,\mu}\langle j_\mu(x),\eta_\mu(x)\rangle\).
The inner product sums over the generator index only.

Let \(\omega(x)=\sum_a\omega^a(x)T^a\), with real
coefficients \(\omega^a(x)\), be a gauge parameter.
For \(g_t(x)=e^{t\omega(x)}\), differentiating the transformed
link gives
\begin{equation}
 \left.\partial_tU^{g_t}(x,\mu)\right|_{t=0}
 =\omega(x)U(x,\mu)-U(x,\mu)\omega(x+\hat\mu).
\end{equation}
The corresponding generator in \eqref{eq:link-variation} is
\(\eta_\mu(x)=-\nabla_\mu\omega(x)\), where
\begin{equation}
 \nabla_\mu\omega(x)
 =U(x,\mu)\omega(x+\hat\mu)U(x,\mu)^{-1}-\omega(x).
 \label{eq:forward-difference}
\end{equation}
The backward covariant difference on the current is
\begin{equation}
 (\nabla_\mu^*j_\mu)(x)
 =j_\mu(x)-U(x-\hat\mu,\mu)^{-1}
             j_\mu(x-\hat\mu)U(x-\hat\mu,\mu),
 \label{eq:backward-difference}
\end{equation}
with no sum over \(\mu\) in this definition.

On \(\mathcal H\), the gauge parameter acts pointwise:
\((R(\omega)v)(x)=R(\omega(x))v(x)\) for \(v\in\mathcal H\).
Gauge covariance of the overlap operator gives
\begin{equation}
 \delta_{-\nabla\omega}D=[R(\omega),D].
\end{equation}
Define the lattice gauge anomaly by
\begin{equation}
 \mathcal A^a(x)=\frac{\ii}{2}
 \tr_{\rm spin,color}\{\gamma_5R(T^a)D(x,x)\}.
\end{equation}
The trace is over spin and color at the indicated site \(x\).
When \(D\) is invertible, substituting its gauge variation
into \eqref{eq:determinant-variation} and summing the
current contribution by parts gives
\begin{equation}
 \delta_{-\nabla\omega}\log Z_F
 =\ii\sum_{x,a}\omega^a(x)
 \left\{\mathcal A^a(x)
       -\sum_{\mu=1}^4(\nabla_\mu^*j_\mu)^a(x)\right\}.
 \label{eq:gauge-determinant-variation}
\end{equation}
Vanishing of this variation for every \(\omega\) requires
the covariant divergence equation
\begin{equation}
 \sum_{\mu=1}^4(\nabla_\mu^*j_\mu)^a(x)=\mathcal A^a(x).
 \label{eq:anomaly-equation}
\end{equation}
Cancellation of the local continuum gauge anomaly in
\(R=\mathbf{10}\oplus\bar{\mathbf5}\) does not require
\(\mathcal A^a(x)\) to vanish at nonzero lattice spacing.
Equation~\eqref{eq:anomaly-equation} specifies how the
measure current cancels its contribution to the gauge
variation of the fermion integral.
The anomaly equation and the measure reconstruction
themselves remain defined in zero-mode sectors.

For the proof, we also use the divergence equation summed against
the gauge parameter \(\omega\). Discrete integration by parts
gives the equivalent condition
\begin{equation}
 \mathfrak L_{-\nabla\omega}=m_\omega,
 \qquad
 m_\omega:=\sum_{x,a}\omega^a(x)\mathcal A^a(x)
 =\ii\Tr\{(\widehat P_--P_+)R(\omega)\}.
 \label{eq:ward}
\end{equation}
The last equality follows from
\(\widehat P_--P_+=\gamma_5D/2-\gamma_5\) and
\(\tr_{\rm spin}\gamma_5=0\). Thus \(m_\omega\) is the gauge
anomaly summed against \(\omega\), and \eqref{eq:ward} is
the integrated form of the same divergence equation.

Once a smooth measure has been reconstructed, infinitesimal
gauge invariance implies invariance under finite gauge
transformations, since the site gauge group
\(\prod_{x\in\Gamma}\SU(5)\) is connected.
Section~\ref{sec:measure-reconstruction} gives this argument
for the reconstructed measure, including zero-mode sectors.

\subsection{The full integrability condition}

The next condition makes reconstruction of the measure
independent of the path used to reach a gauge field.
It compares the phase obtained from the current with
the phase produced by transporting the allowed fermion
subspace around a closed curve.

Let \(\gamma:t\mapsto U_t\), \(0\le t\le1\), be a
piecewise smooth curve in \(\mathcal U_\eps\). Define its
link generator \(\eta(t)\) and its measure Wilson line
\(W[\gamma]\), the phase obtained by integrating the measure
term along the curve and exponentiating:
\begin{equation}
 \eta_\mu(x,t)=\partial_tU_t(x,\mu)U_t(x,\mu)^{-1},
 \qquad
 W[\gamma]=\exp\left\{\ii\int_0^1
                         \mathfrak L_{\eta(t)}[U_t]\,\dd t\right\}.
 \label{eq:measure-wilson-line}
\end{equation}
For a current arising from a measure, the Wilson line
integrates both terms in \eqref{eq:measure-basis}.
In reconstruction, \(W[\gamma]\) is computed directly
from the proposed current before the relative phase
is chosen.

Write \(P_t=\widehat P_-[U_t]\), and define the transport
operator \(Q_t\) on the fixed Hilbert space \(\mathcal H\) by
\begin{equation}
 \partial_tQ_t=[\partial_tP_t,P_t]Q_t,\qquad Q_0=1.
 \label{eq:projector-transport}
\end{equation}
This is Kato transport~\cite{Kato1950}. For a piecewise smooth
curve, the Wilson integral is the sum over its finitely many smooth pieces.
Solve the transport equation on each piece, continuing from the
preceding endpoint value. Thus $Q_t$ is continuous at the joins,
where no derivative is required. The projector is differentiated
as an operator on the fixed space $\mathcal H$.
The solution is unitary and satisfies \(P_tQ_t=Q_tP_0\), so
\(v_t=Q_tv_0\) stays in \(\im P_t\) whenever
\(v_0\in\im P_0\). Moreover, $P_t\partial_t v_t=0$ on each smooth piece:
the transported vector's derivative has no component within the
instantaneous chiral subspace. This condition fixes the rotation of
vectors within that subspace. Taking their exterior product therefore gives
a definite reference transport for the measure phase.

For a closed curve \(U_1=U_0\), the determinant of this
transport on the chiral subspace is
\begin{equation}
 \mathcal T[\gamma]=\det(1-P_0+P_0Q_1).
 \label{eq:twist}
\end{equation}
The operator inside the determinant acts as the identity
on the complement of \(\im P_0\). The determinant \(\mathcal T[\gamma]\)
is the return phase acquired by the exterior product of the
transported fermion basis;
the antifermion basis remains fixed. We call \(\mathcal T\)
the projector determinant phase.

Starting from a measure at \(U_0\), reconstruction combines
projector transport with the factor \(W[\gamma]^{-1}\).
For a closed curve, the accumulated phase is
\(\mathcal T[\gamma]W[\gamma]^{-1}\); returning to the
same measure requires this phase to be one.
The full integrability condition is
\begin{equation}
 W[\gamma]=\mathcal T[\gamma]
 \qquad\text{for every closed admissible curve }\gamma.
 \label{eq:integrability}
\end{equation}
It is required even for curves that cannot be contracted to
a point while remaining in the admissible field space.
Both sides can be nontrivial phases.
L\"uscher's reconstruction then gives a smooth measure
from the current, up to a constant phase in each component.
We give the corresponding reconstruction in
section~\ref{sec:global}.

\subsection{The differential integrability condition}

Full integrability requires the current and projector
transport to assign the same phase to every closed curve.
For a small curve generated by two link variations, the
leading phase is determined by a curvature. Equating the
two curvatures gives the differential integrability
condition. It is enough for contractible curves, but
noncontractible curves require a further argument.

For gauge-field-independent link generators \(\eta\) and
\(\zeta\), differentiation of \eqref{eq:measure-basis} gives
\begin{equation}
 \delta_\eta\mathfrak L_\zeta
 -\delta_\zeta\mathfrak L_\eta
 +\mathfrak L_{[\eta,\zeta]}
 =\ii\Tr\{\widehat P_-
        [\delta_\eta\widehat P_-,\delta_\zeta\widehat P_-]\}.
 \label{eq:curvature}
\end{equation}
The left side is the curvature obtained from the current;
the right side depends only on the chiral projector.
Here \([\eta,\zeta]_\mu(x)=[\eta_\mu(x),\zeta_\mu(x)]\)
is the matrix commutator. Its measure term is needed because
the two link variations need not commute. The convention
\([\delta_\eta,\delta_\zeta]=-\delta_{[\eta,\zeta]}\)
fixes the positive sign in \eqref{eq:curvature} and makes
the terms containing \(\varphi_v\) cancel.

We will use differential forms to express the same
condition in the projector comparisons below. At each admissible
field \(U\), a one-form assigns a number linearly to a link
variation \(\eta\). The measure term is therefore the real one-form
\(\mathfrak L(\eta)=\mathfrak L_\eta\).
Its integral along \(\gamma\) is the integral already used
to define \(W\):
\(\int_\gamma\mathfrak L=\int_0^1
\mathfrak L_{\eta(t)}[U_t]\,\dd t\).
Write \(\dd\) for the exterior derivative on gauge-field
space and name the projector curvature
\begin{equation}
 \Omega(\eta,\zeta)=\ii\Tr\{\widehat P_-
        [\delta_\eta\widehat P_-,\delta_\zeta\widehat P_-]\}.
 \label{eq:projector-curvature}
\end{equation}
Equation~\eqref{eq:curvature} is then
\(\dd\mathfrak L=\Omega\).

For a smooth deformation,
this equality controls how the two phases change. The deformation sweeps out
a surface; its curvature integral gives the phase change
for each transport law. Equal curvatures therefore leave
the ratio \(W[\gamma]/\mathcal T[\gamma]\) unchanged.

The same conclusion holds for a continuous deformation between
piecewise smooth loops. First reparametrize the boundary loops to
be stationary near their finitely many joins, including the closing
join; this leaves both phases unchanged. Extend these changes to
the deformation and make it constant in the deformation parameter
near each boundary loop. The resulting deformation can be smoothed while retaining the two
boundary loops and keeping each intermediate curve closed. Its
compact image lies in the open admissible domain, so sufficiently
close smoothing remains admissible.

For a contractible curve this ratio is one, since the
curve can be shrunk to a constant curve, for which both
phases are one. For a noncontractible curve, requiring
the ratio to be one is additional information in
\eqref{eq:integrability}. Our local and global proofs
address these two requirements separately.
Figure~\ref{fig:local-and-global-integrability} illustrates why
matching curvatures does not by itself establish full integrability.

\begin{figure}[htbp]
  \centering
  \includegraphics[width=\linewidth]{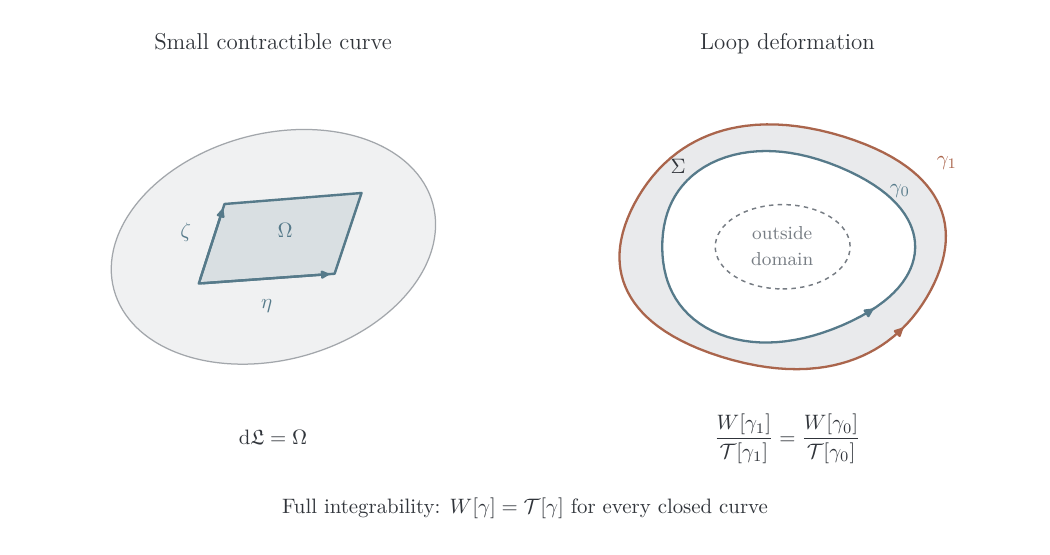}
  \caption{Left: a small contractible curve spans link variations $\eta$ and $\zeta$. Here $\mathfrak L$ is the measure term, $\Omega$ the projector curvature, and $\mathrm d$ the exterior derivative on gauge-field space. Right: deforming $\gamma_0$ into $\gamma_1$ sweeps out the shaded surface $\Sigma$. Here $W$ is the measure Wilson line and $\mathcal T$ the projector determinant phase; their ratio is unchanged by the deformation when the curvature condition at left holds. The bottom condition applies to every closed admissible curve $\gamma$. The central hole schematically represents an obstruction to contraction.}
  \label{fig:local-and-global-integrability}
\end{figure}

When the link generators \(\eta\) and \(\zeta\) depend on
\(U\), differentiating \(\mathfrak L_\zeta\) also differentiates
\(\zeta\). The exterior derivative is then
\begin{equation}
 (\dd\mathfrak L)(\eta,\zeta)
 =\delta_\eta\mathfrak L_\zeta-\delta_\zeta\mathfrak L_\eta
       -\mathfrak L([\delta_\eta,\delta_\zeta]).
\end{equation}
Here \([\delta_\eta,\delta_\zeta]\) is the commutator of the
two variations acting on functions of \(U\). Its link generator is
\(\delta_\eta\zeta-\delta_\zeta\eta-[\eta,\zeta]\), which reduces
to \(-[\eta,\zeta]\) when the generators are independent of \(U\).
In particular, denote by \(X_\omega\)
the vector field generated by a gauge transformation.
Its link generator is \(-\nabla_\mu\omega(x)\), which
depends on \(U\). In this notation the divergence equation \eqref{eq:ward} is
\(\mathfrak L(X_\omega)=m_\omega\).

\subsection{Statement of the result}

We now state the existence result for the overlap operator
and the representation \(R=\mathbf{10}\oplus\bar{\mathbf5}\)
specified in section~\ref{sec:lattice}.

\begin{theorem}
\label{thm:measure}
For the overlap operator \eqref{eq:wilson}, there is
an \(\eps>0\) and a smooth gauge-covariant current
\(j_\mu^a(x;U)\) on \(\mathcal U_\eps\) which satisfies
the covariant divergence equation \eqref{eq:anomaly-equation}
and the full integrability condition \eqref{eq:integrability}.
It has the local bounds \eqref{eq:current-locality} at every fixed
derivative order. Under lattice rotations and reflections it
transforms as an axial vector; its charge-conjugation law
relates the currents in \(R\) and \(R^*\).
These laws are stated in
\eqref{eq:axial}--\eqref{eq:axial-components}.
The admissibility threshold and positive locality rate
are independent of volume, component and derivative order.
Consequently this current determines a smooth gauge-invariant
fermion measure for \(\mathbf{10}\oplus\bar{\mathbf5}\) on
every nonempty component, up to a constant phase on each
component.
\end{theorem}

Translations, coordinate permutations and reflections act
within the same lattice theory when they preserve its periods
and fermion boundary conditions. Otherwise the transformation
laws relate the theories on the corresponding transformed tori.
The transformation laws of the current do not fix relative
phases between components or guarantee that every choice
of those phases respects discrete symmetries.

The proof of theorem~\ref{thm:measure} uses the finite
comparison identity established in section~\ref{sec:comparison}
to construct the measure current in section~\ref{sec:construction}.
Section~\ref{sec:locality} establishes smoothness, the uniform
locality bounds and differential integrability.
Section~\ref{sec:symmetries} proves gauge covariance, the
covariant divergence equation and the stated symmetry properties.
Section~\ref{sec:global} establishes the full integrability
condition for all closed admissible loops, including
noncontractible ones. These results are combined in
section~\ref{sec:measure-reconstruction} to reconstruct a
smooth gauge-invariant fermion measure on each component
and complete the proof. The supporting geometric constructions
and analytic estimates are given in
appendices~\ref{app:g:geometry}--\ref{app:d:cylinder-traces}.

\section{A finite comparison identity}
\label{sec:comparison}

We first compare two families of chiral subspaces by smoothly
interpolating between their projectors. The resulting identity
relates their curvatures and their phases around closed curves.
It will supply the contribution to the measure term from each
pair of successive lattices.

Let \(P_s(U)\), \(0\le s\le1\), be the projector onto the positive spectrum of a smooth family of
Hermitian operators on a fixed finite-dimensional Hilbert space.
Assume that their spectra stay uniformly away from zero. Here \(U\) ranges over a common open set of
original gauge fields, and \(s\) is an auxiliary comparison
parameter. The projector at an intermediate \(s\) need not
be the overlap projector of a gauge-link field.

We seek a contribution to the measure term whose exterior
derivative is the difference between the curvatures of
the endpoint projectors \(P_1\) and \(P_0\).
The interpolation supplies such a one-form: we evaluate
the projector curvature \eqref{eq:projector-curvature}
with one variation along \(s\) and the other along the
original gauge-field variation \(\eta\), then integrate
over \(s\):
\begin{equation}
 b_\eta=\ii\int_0^1\Tr\{P_s
             [\partial_sP_s,\delta_\eta P_s]\}\,\dd s,
 \qquad b(\eta)=b_\eta.
 \label{eq:transgression}
\end{equation}
Thus \(b\) is a real one-form: the trace in
\eqref{eq:transgression} is purely imaginary.
The lemma below establishes the required curvature
difference and shows that the same one-form also
relates the endpoint phases around closed curves.
These two properties allow us to assemble the measure
term from successive projector comparisons.
The gap keeps the rank of \(P_s\) fixed during the comparison.
The notation \(\dd P_s\) means the matrix-valued one-form
whose value on \(\eta\) is \(\delta_\eta P_s\).
The Hilbert space is held fixed during differentiation,
and \(\dd\) acts independently of \(\partial_s\).

For this comparison, we apply the determinant-line construction
of section~\ref{sec:fermion-measure} to the interpolating
subspaces \(\im P_s\). Let \(P_+\) be a fixed reference
projector on the same Hilbert space; in our applications
it is the positive-chirality projector. The corresponding
relative determinant line is
\begin{equation}
 \mathcal D(P_s)=\det\im P_s\otimes(\det\im P_+)^*.
\end{equation}
For a local orthonormal basis \(v_k\) of \(\im P_s\)
and a fixed orthonormal row basis of \((\im P_+)^*\),
define \(s_v\) by the exterior products in
\eqref{eq:measure-exterior}.

To compare phases at different gauge fields, we use the
projector parallel transport of
\eqref{eq:projector-transport}. Transporting each fermion
basis vector and taking their exterior product gives
parallel transport in the determinant line; the
antifermion factor remains fixed. Denote the associated
connection by \(\nabla^{\det}\). Its connection one-form
records the infinitesimal phase change of a chosen local
reference element relative to this parallel transport.

In the reference element \(s_v\), the connection one-form
is \(\sum_k(v_k,\dd v_k)\), obtained by differentiating
the factors in the exterior product. The fixed
antifermion basis contributes zero. For the original
measure, the reference element is multiplied by
\(e^{\ii\varphi_v}\). Differentiating this additional
phase adds \(\ii\dd\varphi_v\), giving
\begin{equation}
 \sum_k(v_k,\dd v_k)+\ii\dd\varphi_v
       =-\ii\mathfrak L,
\end{equation}
where the equality follows from
\eqref{eq:measure-basis}. Thus the measure term describes
how the phase of the chosen measure changes relative
to projector parallel transport.

\begin{lemma}
\label{lem:transgression}
Let \(\Omega_s(\eta,\zeta)=\ii\Tr\{P_s
[\delta_\eta P_s,\delta_\zeta P_s]\}\) be the curvature
at parameter \(s\). For a closed curve \(\gamma\) in
the common field domain, let \(\mathcal T_0[\gamma]\) and
\(\mathcal T_1[\gamma]\) be the projector determinant phases of
\eqref{eq:projector-transport}--\eqref{eq:twist} for the
endpoint projectors. Then the comparison term
\eqref{eq:transgression} satisfies
\begin{equation}
 \dd b=\Omega_1-\Omega_0,\qquad
 \exp\left\{\ii\oint_\gamma b\right\}
       =\frac{\mathcal T_1[\gamma]}{\mathcal T_0[\gamma]}.
 \label{eq:comparison-identities}
\end{equation}

\end{lemma}

\begin{figure}[htbp]
  \centering
  \includegraphics[width=\linewidth]{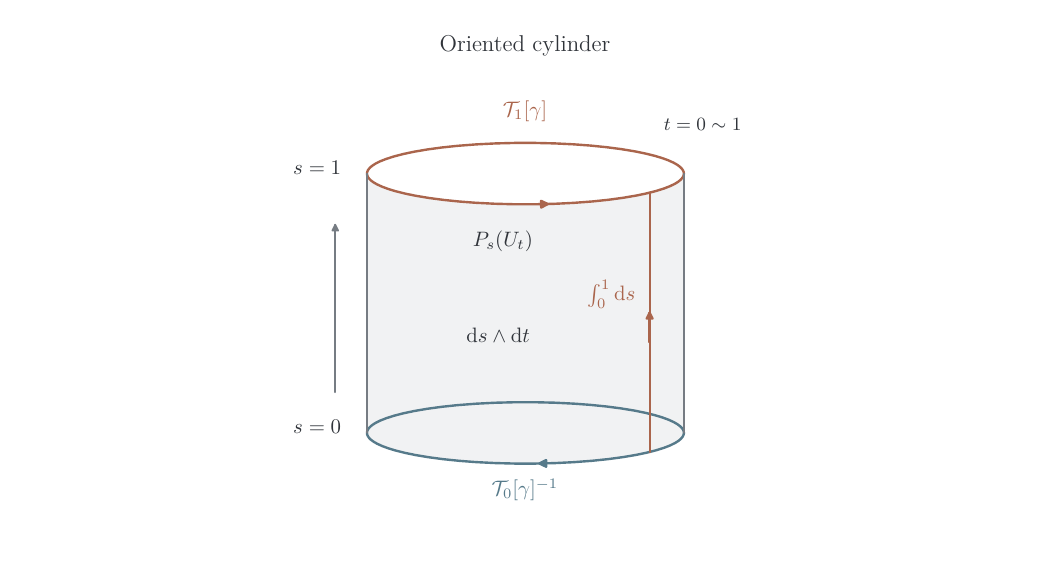}
  \caption{A schematic cylinder of interpolating projectors $P_s(U_t)$ on one fixed fermion space. The parameter $t$ runs around the closed gauge-field curve $\gamma$, with $t=0$ and $t=1$ identified; $s$ runs between the two endpoint projectors. The orientation $\mathrm d s\wedge\mathrm d t$ gives the opposite boundary arrows. Here $\mathcal T_1[\gamma]$ and $\mathcal T_0[\gamma]^{-1}$ are the endpoint determinant phases with these orientations. The highlighted vertical line marks integration over $s$. Intermediate projectors need not come from physical gauge-link fields, and the cylinder does not require a disk filling of $\gamma$.}
  \label{fig:oriented-comparison-cylinder}
\end{figure}

\begin{proof}
The curvature at fixed \(s\) and the integrand of \(b\)
are components of a single two-form on the product of the
field domain with the \(s\) interval. Put
\(\dd_{\rm tot}=\dd+\dd s\,\partial_s\).
The two-form
\begin{equation}
 \Omega_{\rm tot}=\ii\Tr\{P_s\dd_{\rm tot}P_s
                                  \wedge\dd_{\rm tot}P_s\}
\end{equation}
is closed. Differentiating \(P_s^2=P_s\) shows that
\(\dd_{\rm tot}P_s\) has only off-diagonal blocks between
\(\im P_s\) and its orthogonal complement. The exterior
derivative of \(\Omega_{\rm tot}\) is a trace of three
such off-diagonal operators and therefore vanishes.
The component of this closure equation containing one
\(\dd s\) equates \(\partial_s\Omega_s\) to the
exterior derivative of the integrand in \eqref{eq:transgression}.
Integrating from \(s=0\) to \(s=1\) gives the first identity.

For the second identity, consider the cylinder
\([0,1]\times\gamma\) shown in
Fig.~\ref{fig:oriented-comparison-cylinder}, with orientation
\(\dd s\wedge\dd t\), where \(t\) parametrizes \(\gamma\).
This cylinder exists
even when \(\gamma\) is not contractible in the field domain:
the extra direction is the projector interpolation itself.
Its boundary traverses \(\gamma\) at \(s=1\) and the reverse curve at \(s=0\),
so the boundary phase is \(\mathcal T_1/\mathcal T_0\).
The same projector transport defines \(\nabla^{\det}\) for
variations of \(s\) as well as \(U\). On this cylinder its
curvature is \(-\ii\Omega_{\rm tot}\). Since parallel transport uses
the negative of the connection one-form, Stokes' formula gives
\(\mathcal T_1/\mathcal T_0=\exp(\ii\int\Omega_{\rm tot})\).
The cylinder integral is \(\oint_\gamma b\), as required.
This identity determines the phase itself, without choosing
a branch of its logarithm.
\end{proof}

To compare different lattice spacings, we will transport coarse
spinors and their color bases together to nearby fine sites.
In these bases, the embedding is a constant matrix and its
orthogonal complement in the fine fermion space is fixed.
At the coarse endpoint, the fermion and reference subspaces
share the positive-chirality subspace of this complement.
Its determinant factor cancels against its dual in
\(\mathcal D(P_s)\). The constant endpoint identification adds
no connection term to projector transport in these coordinates.
The color bases themselves remain \(U\)-dependent, and their
derivatives enter the sampled link matrices.
Section~\ref{sec:construction} gives the construction.

\section{Construction of the current}
\label{sec:construction}

We now construct the comparisons whose sum will give the measure term
$\mathfrak L$. Starting from the original links $U$, we reconstruct one
smooth connection and sample it on successively finer lattices. At each
step we put the coarse and fine Hermitian Wilson operators on a common
fermion space and interpolate between them. The comparison identity of
section~\ref{sec:comparison} then supplies a contribution $B_n$ to the
measure term. We define these contributions here and derive the exact
identities that their finite sums satisfy. The following sections establish
convergence and the reconstruction conditions.

\subsection{A connection determined by the original links}

To control the sum over refinement steps, we need parallel transport
around a shrinking plaquette to approach the identity.
We achieve this by sampling one smooth connection, held fixed during
refinement. Earlier approaches to chiral gauge theory also used gauge-field
interpolation to couple fermions on a finer lattice~\cite{HernandezSundrumChiral,HernandezSundrumInterpolation}.
Here the finer lattices are auxiliary: we use them to construct the measure
for the original overlap operator $D[U]$. We retain $U$ itself as the
first field and use the reconstructed connection to supply the finer fields.

At an original lattice vertex, the fundamental representation of $\SU(5)$
acts on a five-dimensional complex color space. The links $U$ specify
parallel transport between these spaces along the original edges. To sample
finer edges, we need color spaces and transport between their endpoints as
well. We therefore reconstruct a five-dimensional color space $E_y(U)$ at
each point $y$ of the continuous spacetime torus $T^4$, together with a
smooth $\SU(5)$ connection $\nabla^0[U]$. This connection supplies transport
along paths in $T^4$. An orthonormal basis of $E_y(U)$ is a color basis at
$y$; a smoothly varying choice on a neighborhood is called a local frame.

Appendix~\ref{app:g:geometry} gives the reconstruction explicitly. Short
original-lattice transports between nearby vertices are combined with
weights that vary smoothly with $y$. The resulting matrix acts on the
color spaces at those vertices. The plaquette bound makes it close to a
rank-five projection. The eigenspaces of its five eigenvalues near one
together form the color space $E_y(U)$.

To differentiate a color vector as $y$ changes, we project
its derivative back into $E_y(U)$; this defines a connection. The appendix
fixes a normalized volume element, meaning a unit exterior product of five
color vectors, and removes the trace part of the connection so that
transport preserves this element. The resulting connection is $\nabla^0[U]$.
This construction uses local frames and does not require the gauge bundle
to be trivial.

We will also need to compare reconstructed color spaces when $U$ itself
varies along a smooth family. Projecting derivatives in the family
parameters and removing the trace part in the same way supplies transport
in these additional directions. Thus $\nabla^0$ has components along the
family as well as along spacetime. Section~\ref{sec:global} uses the family
component along a closed loop of original gauge fields. During refinement
at a fixed $U$, no such parameter variation is made.

Let $h_n=3^{-n}$, $n=0,1,\ldots$, denote the auxiliary lattice spacings in
original units. At spacing $h=h_n$, the integer site label $x$ denotes the
point $y=hx$ of $T^4$, and $x_\mu$ has period $3^nL_\mu$. Write $f_h(y)$
for the matrix whose five columns form the chosen orthonormal color basis
at that point. We require their exterior product to equal the volume
element fixed by the reconstruction. In such bases, transport by
$\nabla^0[U]$ is represented by an $\SU(5)$ matrix.

The choice of these bases also determines how we will compare fermions at
successive spacings. A coarse-site spinor is transported to nearby fine
sites. If its color basis is transported with it, its numerical components
stay unchanged. We can then express this transport as a constant copying
matrix, even though the transported vectors and bases depend on $U$.

To implement this choice, start at each original vertex $i$ with the
orthonormal basis $C_i(i,U)$ supplied by appendix~\ref{app:g:geometry}, and
set $f_1(i)=C_i(i,U)$. A site with integer label $x$ at spacing $h$ is
assigned the 81 sites with labels $3x+k$ at spacing $h/3$, where
$k\in\{-1,0,1\}^4$. Their spacetime positions are $hx+hk/3$; we call them
the children of $x$. Transport the parent's basis along the straight
segment to each child, and repeat at every refinement step. This determines
$f_h$ at all sampled sites.

Let $\tau(y\leftarrow z)$ denote transport by $\nabla^0[U]$ along the
straight segment from $z$ to $y$. The sampled link matrix is
\begin{equation}
 (F_hU)(x,\mu)
 =f_h(y)^\dagger
       \tau(y\leftarrow y+h\hat\mu)f_h(y+h\hat\mu),
 \qquad y=hx.
 \label{eq:sampled-link-field}
\end{equation}
Reading from right to left, this formula turns numerical components at the
terminal site into a color vector, transports it to the initial site, and
expresses it in the initial basis. It therefore has the terminal-to-initial
orientation of section~\ref{sec:lattice}. These matrices are the gauge
links entering the fine-lattice Hermitian Wilson operator.

The repeated transport of bases does not lead to paths of growing length.
Following a refined site's parents back to an original vertex gives a
total one-norm length at most $4\sum_{m\ge1}3^{-m}=2$ in original units.
Together with the original-link derivative bounds for the connection and
initial bases in appendix~\ref{app:g:geometry}, this bounds every fixed
order of derivatives of $f_h$ uniformly in the refinement level.

We construct the paths and bases first on $\mathbb R^4$, then impose the
periodic identifications. The sampled gauge links are periodic; the
chosen periodic or antiperiodic fermion boundary signs enter separately
when these links act on fermions across a boundary.

There is one distinction between reconstruction and refinement that matters
at the first step. Sampling at spacing one need not give back exactly $U$:
the construction only guarantees $\|F_1U-U\|\le C\eps$ link by link.
We must keep the actual original overlap operator, so we define
\begin{equation}
 U^{(0)}[U]=U,\qquad
 U^{(n)}[U]=F_{h_n}U\quad(n\ge1).
 \label{eq:coarse-fields}
\end{equation}
The first comparison is therefore between $U$ and $F_{1/3}U$. All later
comparisons are between successive samples of the same reconstructed
connection: $F_{1/3}U$, $F_{1/9}U$, and so on.

\subsection{Comparison of the coarse and fine operators}

At step $n$ we compare the Hermitian Wilson operators
$H(U^{(n)})$ and $H(F_{h_{n+1}}U)$ defined by \eqref{eq:wilson}.
They act on fermion spaces of different dimensions, since the fine lattice
has 81 times as many sites. To use the comparison identity, we first embed
the coarse space in the fine space and extend the coarse operator to the
remaining states. Both endpoints will then act on the same space.

The fundamental color bases just constructed induce bases in the fermion
representation $R=\mathbf{10}\oplus\bar{\mathbf5}$. In these bases,
transporting a spinor from a parent to a child leaves its color components
unchanged. The embedding $I_n$ therefore copies a coarse field $\psi$ to
its 81 children:
\begin{equation}
 (I_n\psi)(3x+k)=\psi(x)/9,
 \qquad k\in\{-1,0,1\}^4.
 \label{eq:constant-inclusion}
\end{equation}
The factor $1/9$ makes the 81 copies have the same total squared norm as
the original, so $I_n^\dagger I_n=1$ for the site-sum inner product.
There is no mixing of spin or color components, and $I_n$ commutes with
$\gamma_5$. As with the color bases, we define the assignment on the
covering space before identifying points on the torus; coarse and fine
fermions have the same prescribed boundary signs.

The matrix $I_n$ is constant in these component coordinates:
$\dd I_n=0$. This does not make the color bases independent of $U$.
Their dependence is contained in the sampled link matrices
\eqref{eq:sampled-link-field} and must be differentiated when the
original links vary.

The projector $I_nI_n^\dagger$ selects the embedded coarse space. On its
orthogonal complement, we extend the coarse operator by $\gamma_5$.
This supplies fixed eigenvalues $\pm1$ on all additional states. On these
states, the positive spectral projector is $P_+=(1+\gamma_5)/2$, the same
projector used for the antifermions. The added fermion and reference determinant
factors will therefore cancel, as described in
section~\ref{sec:comparison}.

At fixed original field $U$, define the interpolation
\begin{equation}
 H_{n,s}=(1-s)H(F_{h_{n+1}}U)
 +s\{I_nH(U^{(n)})I_n^\dagger
                    +\gamma_5(1-I_nI_n^\dagger)\},
 \qquad 0\le s\le1.
 \label{eq:direct-comparison}
\end{equation}
It starts at the fine Hermitian Wilson operator when $s=0$ and ends at
the extended coarse operator when $s=1$. Both Hermitian Wilson operators
entering this formula have the dimensionless normalization of \eqref{eq:wilson}, in their own
lattice units. The intermediate Hermitian operators need not be Wilson
operators of gauge-link fields: we interpolate the operators, not the
links. Figure~\ref{fig:centered-refinement} illustrates the sampling,
choice of bases, and embedding used in this comparison.

\begin{figure}[htbp]
  \centering
  \includegraphics[width=\linewidth]{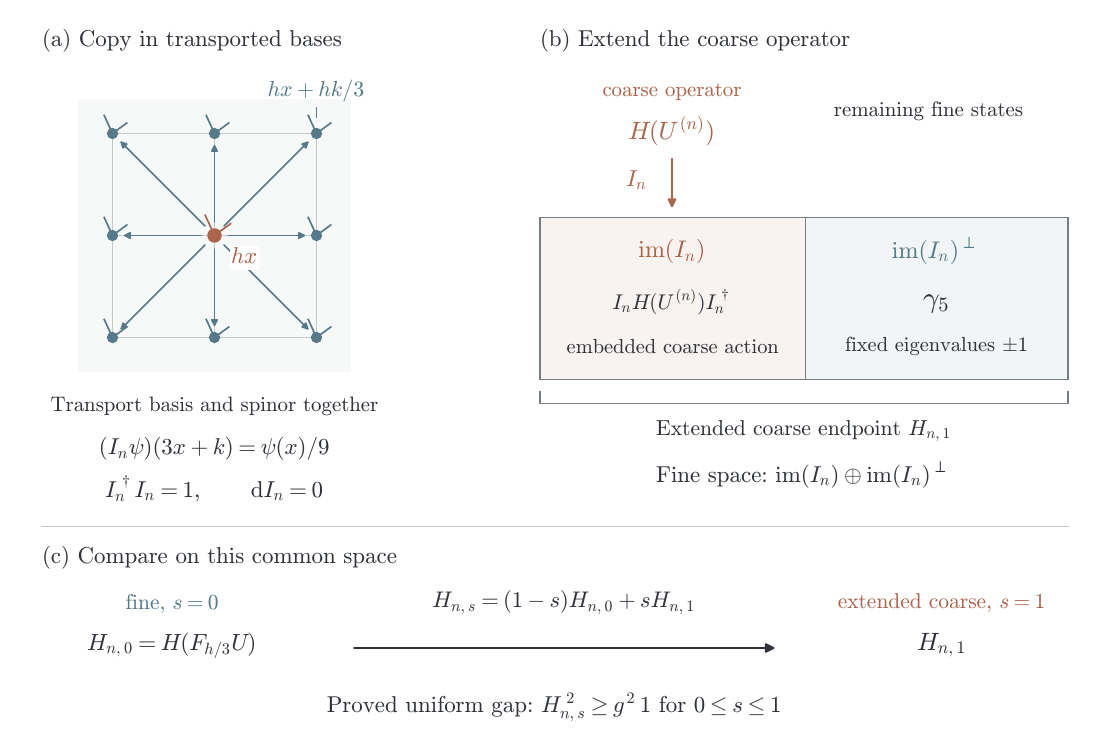}
  \caption{Comparing successive operators in one fermion space.
(a) Transporting a parent's basis and spinor together lets $I_n$ copy
the numerical components with normalization $1/9$. Nine of the 81
children are shown, with $h=h_n=3^{-n}$ and $k\in\{-1,0,1\}^4$.
The bases depend on $U$, although $\dd I_n=0$ in these coordinates.
(b) The extended coarse operator $H_{n,1}$ acts as the coarse operator
on $\im(I_n)$ and as $\gamma_5$ on the orthogonal complement.
The boxes represent subspaces, not spatial regions.
(c) Both endpoints now act on the fine space; their interpolation has
the proved uniform gap $g$. The first coarse field is $U^{(0)}=U$.}
  \label{fig:centered-refinement}
\end{figure}

Smooth endpoint projectors do not by themselves ensure a smooth
interpolation: an eigenvalue of $H_{n,s}$ could cross zero between the
endpoints. We must exclude such crossings. The next proposition states
the required uniform gap, together with the sampling properties used to
prove it and to control the current.

\begin{proposition}
\label{prop:comparison}
Choose the admissibility threshold $\eps$ sufficiently small. Each
sampled link depends on original links in a fixed finite neighborhood of
its position, whose size does not grow with refinement. Its derivatives
of every fixed order with respect to the original links are bounded
uniformly in the sampling spacing. The sampled fields satisfy
\begin{equation}
 \|(F_1U)(x,\mu)-U(x,\mu)\|\le C\eps,
 \qquad
 \|1-(F_hU)(p)\|\le C\eps h^2.
 \label{eq:refinement-estimates}
\end{equation}
Here $p$ is a plaquette of the spacing-$h$ lattice. If all original
plaquette matrices equal the identity, then $F_1U=U$, including fields
with nontrivial transport around a periodic direction.

In the transported bases, an original gauge transformation acts at each
sampled site through the gauge matrix at its original ancestor. Sampled
links obey the usual endpoint transformation law with these matrices.
Sampling also respects translations, coordinate permutations, reflections,
and complex conjugation.

There is a positive constant $g$ such that
\begin{equation}
 H_{n,s}^2\ge g^2\,1
\end{equation}
for every $n\ge0$, $0\le s\le1$, admissible component, and lattice
volume. The admissibility domain and the gap constant are chosen
independently of derivative order.
\end{proposition}

The sampling estimates follow from the construction of the connection and
bases in appendix~\ref{app:g:geometry}. The $h^2$ factor in the plaquette
estimate reflects the area enclosed by a fine plaquette. It remains after
any fixed number of original-link derivatives. The factor $\eps$, however,
need not remain after differentiation of either estimate; for derivatives
of $F_1U-U$, we use only uniform boundedness.

Appendix~\ref{app:s:spectral} first proves the gap for the free field at
every momentum and every $s$. It then compares an admissible field locally
with the free one. In each region, it chooses coarse-site color bases by
transport from the region's center, applying the same change of basis to
all children of a coarse site. This leaves $I_n$ unchanged. In these
bases, the curvature bound controls the difference between the local
operator coefficients and their free values. The first comparison also
includes the difference $U-F_1U$, since its coarse endpoint is the actual
original field.

To combine these local estimates, the proof multiplies a fermion field
by overlapping scalar weights whose squares sum to one. Because the
operator hops between sites, it does not commute with these weights.
The proof first chooses the weights to vary slowly enough that these
commutator errors are small compared with the free gap. It then chooses
$\eps$ small enough to control the local operator differences on the
fixed regions. The free gap survives both errors, giving the asserted
bound for $H_{n,s}$.

\paragraph{Properties used in the locality and symmetry proofs.}
The construction below uses the gapped interpolation just obtained. Its
later locality proof also requires a gap bound when links are specified
only near a finite region. Appendix~\ref{app:s:supported}
proves the corresponding lower bound on the norm of the operator applied
to a fermion field supported there. All output components count in this
norm, including those outside the input support; the operator is not
compressed to the region. The specified links must cover the reconstruction,
the parent-to-child transports, and the local comparisons in the gap
proof, together with every plaquette needed to control these transports.
The bound holds for any assignment satisfying the plaquette conditions
throughout this neighborhood. No admissible completion of the links on
the rest of the lattice is assumed.

We will also use the exact transformation laws of $H_{n,s}$. Its behavior
under complex conjugation enters the color-trace cancellations in
section~\ref{sec:locality}; its behavior under conjugation and changes of
lattice coordinates determines the current's transformation laws in
section~\ref{sec:symmetries}.

Let $O$ be a coordinate permutation combined with any coordinate
reflections. The field $OU$ is obtained by relabeling links and inverting
a link matrix when its orientation is reversed. Let $W_{O,n}$ be the
corresponding unitary action on position and spin at level $n$, with the
lattice periods and fermion boundary signs transformed accordingly.

For complex conjugation, choose a fixed unitary spin matrix $\mathcal B$
such that $\mathcal B\overline{\gamma_\mu}\mathcal B^{-1}=\gamma_\mu$.
An overbar denotes entrywise complex conjugation; the corresponding
antiunitary operator is $J_{\rm spin}=\mathcal B\mathsf C$, where
$\mathsf C$ performs this conjugation. With the representation displayed
as a subscript and $R^*$ denoting its conjugate as in
section~\ref{sec:lattice}, the identities are
\begin{equation}
 \begin{split}
 H_{n,s}(OU)&=(\det O)W_{O,n+1}H_{n,s}(U)W_{O,n+1}^{-1},\\
 H_{n,s;R^*}(U)&=\mathcal B\,
       \overline{H_{n,s;R}(U)}\,\mathcal B^{-1}.
 \end{split}
 \label{eq:kernel-symmetry}
\end{equation}
Appendix~\ref{app:s:spectral} constructs $\mathcal B$ and proves these
identities.

\subsection{The measure term at each refinement step}

The gap in Proposition~\ref{prop:comparison} makes
$P_{n,s}=(1+\operatorname{sign}H_{n,s})/2$ a smooth projector in both
$s$ and the original field $U$. We can therefore insert it into the
comparison formula \eqref{eq:transgression}. The two derivatives in that
formula play different roles: $\partial_s$ moves along the interpolation
at fixed $U$, whereas $\delta_\eta$ varies the original links according
to \eqref{eq:link-variation}, at fixed $s$. The latter changes both
endpoint operators, since both are determined by $U$.

Integrating over $s$ gives one contribution $B_{n,\eta}$, linear in the
original link variation. We propose its sum as the measure term:
\begin{equation}
 B_{n,\eta}=\ii\int_0^1\Tr\{P_{n,s}
       [\partial_sP_{n,s},\delta_\eta P_{n,s}]\}\,\dd s,
 \qquad
 \mathfrak L_\eta=\sum_{n=0}^\infty B_{n,\eta}.
 \label{eq:scale-current}
\end{equation}
The trace is over sites, spin, and color in the fine fermion space at
level $n+1$. We write $B_n(\eta)=B_{n,\eta}$ when treating $B_n$ as
a one-form on the original gauge-field space. Convergence of the sum,
including its original-link derivatives, will be proved in
section~\ref{sec:locality}.

The derivative $\delta_\eta$ must include the entire dependence of the
sampled links on $U$. In \eqref{eq:sampled-link-field}, it acts on the
transport $\tau$ and on both endpoint bases $f_h$. Differentiating those
bases in turn differentiates their initial values $C_i(i,U)$ and every
parent-to-child transport used to construct them. The copying matrix
alone is fixed: $\delta_\eta I_n=0$. Thus the constant matrix $I_n$
does not remove any original-link derivatives of the geometric transport.

The current is the coefficient of each original link variation in
\eqref{eq:scale-current}. Using the elementary derivative
$D_{x,\mu,a}$ defined in \eqref{eq:link-derivative}, we obtain
\begin{equation}
 \begin{split}
 j_{\mu,n}^a(x;U)
 &=\ii\int_0^1\Tr\{P_{n,s}
             [\partial_sP_{n,s},D_{x,\mu,a}P_{n,s}]\}\,\dd s,\\
 B_{n,\eta}
 &=\sum_{x,\mu,a}\eta_\mu^a(x)j_{\mu,n}^a(x;U),
 \qquad j_\mu^a(x;U)=\sum_{n\ge0}j_{\mu,n}^a(x;U).
 \end{split}
 \label{eq:comparison-current-components}
\end{equation}
Although the trace uses the fine lattice, $x$ and $\mu$ label an original
link, and $a$ labels its Lie-algebra generator. The index $n$ labels
successive comparisons; it is distinct from the support radius $r$ in
the locality decomposition \eqref{eq:current-locality}.

\subsection{Curvature and phase identities}

We next apply Lemma~\ref{lem:transgression} to each $B_n$. To identify
its two endpoint curvatures and phases with those of the coarse and fine
overlap projectors, we must first remove the fixed states added to the
coarse space.

At the fine endpoint, $P_{n,0}=\widehat P_-[F_{h_{n+1}}U]$. At the
coarse endpoint, $H_{n,1}$ acts as $\gamma_5$ outside the embedded coarse
space, so $P_{n,1}$ selects the positive-chirality states there. Denote
the reference projector on the level-$n$ fermion space by
$P_{+,n}=(1+\gamma_5)/2$. Then
\begin{equation}
 \begin{split}
 P_{n,1}&=I_n\widehat P_-[U^{(n)}]I_n^\dagger
                 +(1-I_nI_n^\dagger)P_{+,n+1},\\
 P_{+,n+1}&=I_nP_{+,n}I_n^\dagger
                 +(1-I_nI_n^\dagger)P_{+,n+1}.
 \end{split}
\end{equation}
The second term is the same in both lines. In the relative determinant
line, this subspace contributes a determinant factor for the fermions
and its dual for the reference states; the two cancel. The remaining
factors are exactly those of the coarse lattice, identified through $I_n$.
Because $I_n$ is constant in component coordinates, this identification
adds no connection term, curvature, or transport phase. No equality
between the dimensions of the fermion and reference subspaces is needed.

There is also a distinction between the lattice on which a projector acts
and the field we vary. We use $\Omega$ and $\mathcal T$ for the curvature
and projector determinant phase on each lattice, but all fields in the
construction depend on the same original $U$. Their curvatures must
therefore be evaluated on the variations induced by changing $U$.

To make this explicit, let $F$ map original fields to link fields on
another lattice. The variation $\delta_\eta U=\eta U$ induces a target
link variation with generator $F_*\eta$. For a target link $(x,\mu)$,
\begin{equation}
 \begin{split}
 (F_*\eta)_\mu(x)
 &=(\delta_\eta F(U)(x,\mu))F(U)(x,\mu)^{-1},\\
 (F^*\Omega)_U(\eta,\zeta)
 &=\Omega_{F(U)}(F_*\eta,F_*\zeta).
 \end{split}
\end{equation}
The second line defines $F^*\Omega$, the curvature pulled back to the
original gauge-field space: it is simply the target curvature evaluated
on the two induced variations. For a closed, piecewise smooth curve
$\gamma:t\mapsto U_t$ of original admissible fields, similarly write
$F\gamma$ for $t\mapsto F(U_t)$. Its phase $\mathcal T[F\gamma]$ is
computed by projector transport on the target lattice.

With these conventions, Lemma~\ref{lem:transgression} gives
\begin{equation}
 \dd B_n=(U^{(n)})^*\Omega-F_{h_{n+1}}^*\Omega,
 \qquad
 e^{\ii\oint_\gamma B_n}
  =\frac{\mathcal T[U^{(n)}\gamma]}
         {\mathcal T[F_{h_{n+1}}\gamma]}.
 \label{eq:level-telescope}
\end{equation}
Here $U^{(n)}$ denotes the map $U\mapsto U^{(n)}[U]$, and each curvature
or phase belongs to the lattice specified by its field argument. The
coarse-minus-fine difference and the coarse-to-fine ratio follow from
the interpolation's direction: $s=0$ is fine and $s=1$ is coarse.

Now sum the contributions for $n=0,\ldots,n_{\max}-1$, with
$n_{\max}\ge1$. Each intermediate field is the fine endpoint of one
comparison and the coarse endpoint of the next. Its curvature cancels
in the sum. Exponentiating the integral of that sum multiplies the phase
ratios, so its determinant phase cancels as well. Since $U^{(0)}[U]=U$,
we obtain
\begin{equation}
 \dd\sum_{n=0}^{n_{\max}-1}B_n=\Omega-F_{h_{n_{\max}}}^*\Omega,
 \qquad
 e^{\ii\oint_\gamma\sum_{n=0}^{n_{\max}-1}B_n}
       =\frac{\mathcal T[\gamma]}{\mathcal T[F_{h_{n_{\max}}}\gamma]}.
 \label{eq:telescope}
\end{equation}
These identities are exact at every finite $n_{\max}$. Only the original
projector and the finest sampled projector remain. Completing the
construction requires convergence of the sum and control of the remaining
fine-lattice terms.

Section~\ref{sec:locality} proves that the sum converges to a smooth,
exponentially local current, with bounds on every fixed order of
original-link derivatives. One common positive locality rate works
independently of derivative order and lattice volume. It also proves
$F_{h_{n_{\max}}}^*\Omega\to0$ in the corresponding norms. The first
identity in \eqref{eq:telescope} then gives the differential integrability
condition \eqref{eq:curvature}.

Section~\ref{sec:symmetries} evaluates the same comparisons on gauge
variations. Successive anomaly contributions cancel, leaving the original
anomaly and a remainder on the finest lattice. For each fixed original
field, gauge parameter, and volume, the remainder tends to zero. This
gives the covariant divergence equation \eqref{eq:anomaly-equation}.

Finally, section~\ref{sec:global} proves
$\mathcal T[F_{h_{n_{\max}}}\gamma]\to1$ for every fixed closed,
piecewise smooth curve $\gamma$, including those that cannot be contracted
within the admissible field space. Together with convergence of the
measure term, this turns the second identity in \eqref{eq:telescope}
into the full integrability condition \eqref{eq:integrability}.
\section{Convergence and locality}
\label{sec:locality}

We have defined a contribution $B_n$ from each comparison of adjacent
lattice spacings. We must now show that their sum gives a smooth,
exponentially local current on the original lattice. Small fine-lattice
plaquettes are not sufficient by themselves: the number of sites in
the trace defining $B_n$ grows as the lattice is refined. The contribution
from each trace position must decrease fast enough to compensate.

For $R=\mathbf{10}\oplus\bar{\mathbf5}$, an exact color-trace
cancellation supplies this decrease. For comparisons with $n\ge1$
and coarse spacing $h=h_n$, it
gives a factor $h^8$ in the bound at each trace position. The density of
positions grows as $h^{-4}$, leaving a summable $h^4$ bound on $B_n$.
The corresponding curvature estimate gives $h^6$ before the site sum
and $h^2$ afterward. The gap in the Hermitian operators supplies
exponential decay with distance, and both estimates survive any fixed
number of original-link derivatives. We explain these steps below;
appendix~\ref{app:a:locality} gives the complete bounds.

The first comparison is different because its coarse endpoint is the
actual original field $U$. For every later comparison, both endpoints
sample the same reconstructed connection. Only those later terms need
to decrease with $n$: one finite bound on $B_0$ is enough.

\subsection{Separating contributions by their range}

Fix an original link and generator $I=(x,\mu,a)$ and a comparison
$n\ge1$. Its contribution to the current is
$j_{\mu,n}^a(x;U)$ in \eqref{eq:comparison-current-components}.
The original-link derivative $D_I$ changes both the coarse and fine
fields, including their transported bases. We first bound that entire
induced variation, before taking the trace over lattice sites.

For a covariant shift $T_\nu[V]$, its variation acts by the same unitary
shift followed by multiplication, at each output site, by the induced
link generator in representation $R$. The norm of this multiplication
operator is the largest generator norm. It does not grow with the number
of affected links. Each Hermitian Wilson operator contains a fixed number
of shifts, and embedding the coarse variation with $I_n$ preserves its
norm. The gap then controls the corresponding projector variation. Thus
the many fine links changed by one original link do not introduce a site
count at this stage. That count enters when we take the trace.

To control the range of the trace contributions, expand the projector
in polynomials of its Hermitian operator. Write $\mathcal K=H_{n,s}$
for the comparison; the same expansion will later be used with
$\mathcal K=H(F_hU)$ for the sampled curvature. If $P_{\mathcal K}$
is its positive spectral projector, then
\begin{equation}
 \begin{gathered}
 P_{\mathcal K}=\frac12\left(1+\sum_{m\ge0}S_m\right),
 \qquad S_m=c_m\mathcal K(1-\mathcal K^2/\Lambda^2)^m,\\
 c_m=\frac{4^{-m}\binom{2m}{m}}{\Lambda}.
 \end{gathered}
 \label{eq:local-polynomial-projector}
\end{equation}
Here $m$ is a nonnegative integer and $\Lambda>13$ is a fixed number
larger than the common operator-norm bound in
appendix~\ref{app:s:spectral}. This is the convergent
inverse-square-root series for $\operatorname{sign}\mathcal K$.
The gap bounds the norm of $1-\mathcal K^2/\Lambda^2$ strictly below
one. Consequently the polynomial terms decrease exponentially in $m$.
Their propagation range grows only linearly: $S_m$ contains at most
$2m+1$ operator steps.

Substitute this series into the trace
$\ii\Tr\{P_{n,s}[\partial_sP_{n,s},D_IP_{n,s}]\}$.
The constant part of the first projector multiplies a commutator with
zero full trace. We remove that term before separating the trace into
sites; its trace need not vanish at an individual site. The remaining
terms have three polynomial factors and a prefactor $\ii/8$.

Group the 81 fine children of each coarse parent into one block, and let
$y$ denote that parent's spacetime position. At fixed $I,n,s$, write
$f_{y,M}$ for the contribution with total series order
$M=m_1+m_2+m_3$:
\begin{equation}
 f_{y,M}
 =\operatorname{Re}\frac{\ii}{8}
   \sum_{m_1+m_2+m_3=M}
   \tr\!\left\{
       \bigl(S_{m_1}[\partial_sS_{m_2},D_IS_{m_3}]\bigr)(y,y)
             \right\}.
 \label{eq:local-trace-example}
\end{equation}
The trace at $y$ includes spin, color, and all child indices. Summing
$f_{y,M}$ over $y,M$ and integrating over $s$ gives
$j_{\mu,n}^a(x;U)$. Taking the real part term by term leaves this real
current unchanged. We suppress $I,n,s$ in $f_{y,M}$ only to simplify
the notation within this argument.

At total order $M$, the polynomial contribution contains at most
$2M+3$ operator steps. The gap gives an exponentially decreasing bound
in $M$, multiplied by powers of $M+1$ from the choices of polynomial
orders and derivative positions. We estimate the polynomial products as
matrices, without taking an absolute value separately for every path.
The possible paths are used to determine which links can affect a term;
their cancellations remain in the matrix norm estimate.

For $(F_h^*\Omega)(D_I,D_J)$, use the same polynomials with
$\mathcal K=H(F_hU)$ and replace the commutator by
$[D_IS_{m_2},D_JS_{m_3}]$. Now $y$ is a sampled site, the trace is
over spin and color, and there is no comparison parameter or $s$
integral. Both variations are induced by changing the original links.

\subsection{Color cancellation and powers of the spacing}

The polynomial bound controls range but does not yet supply a small
power of $h_n$. That comes from cancellation of the cubic gauge anomaly
in the complete fermion multiplet. We first show why this condition
also controls the ordered color traces in the polynomial expansion.
For a finite-dimensional unitary representation $R$ of $\SU(N_{\rm c})$,
the cubic condition is
\begin{equation}
 \tr_R R(X)\{R(Y),R(Z)\}=0,
 \qquad X,Y,Z\in\mathfrak{su}(N_{\rm c}).
 \label{eq:local-cubic-anomaly}
\end{equation}

For anti-Hermitian generators, $R^*(X)=-R(X)^T$, so conjugating the
representation reverses an ordered trace and multiplies it by $(-1)^b$,
where $b$ is the number of factors. For zero, one, or two factors,
equality follows from equal dimensions, tracelessness, and cyclicity,
respectively. For three represented generators $A=R(X_1)$,
$B=R(X_2)$, and $C=R(X_3)$, the difference is
$\tr_R(ABC+CBA)=\tr_R A\{B,C\}$, which vanishes by
\eqref{eq:local-cubic-anomaly}. With a fourth generator $D=R(X_4)$,
cyclicity gives
\begin{equation}
 \tr_R(ABCD-DCBA)
 =\frac12\tr_R\!\left([A,B]\{C,D\}+[C,D]\{A,B\}\right)=0.
\end{equation}
Both terms vanish by the same cubic condition, because
$[R(X),R(Y)]=R([X,Y])$ is again a represented Lie-algebra element.
Thus
\begin{equation}
 \tr_R R(X_1)\cdots R(X_b)
 =\tr_{R^*}R^*(X_1)\cdots R^*(X_b),
 \qquad 0\le b\le4,
 \label{gen:ordered-traces}
\end{equation}
where an empty product is the identity. The matrix order is retained,
and the quartic traces need not vanish individually. This is the
trace-reduction consequence of cubic anomaly cancellation that we
need~\cite[sec.~3]{vanRitbergenGroupTheory}. L\"uscher used the same charge-conjugation and ordered-trace cancellation in the perturbative analysis of the lattice curvature~\cite[secs.~4.1--4.2]{LuscherPerturbative}.

For $R=\mathbf{10}\oplus\bar{\mathbf5}$, the cubic condition follows
from the character identity
\begin{equation}
 \chi_R(g)-\chi_{R^*}(g)=\det(1-g),\qquad
 \chi_R(g)=\tr R(g),\qquad g\in\SU(5).
 \label{eq:character}
\end{equation}
Expanding the determinant in the five-dimensional fundamental
representation gives this identity: the exterior powers of degrees
zero and five cancel, and the other four give the character difference.
For $g=e^{tX}$, its right-hand side is $O(t^5)$, so the cubic traces
in $R$ and $R^*$ agree. They also have opposite signs, hence
$\tr_R R(X)^3=0$. Polarization gives
\eqref{eq:local-cubic-anomaly} for arbitrary $X,Y,Z\in\mathfrak{su}(5)$.

The spin conjugation identity \eqref{eq:kernel-symmetry} gives the
opposite relation: each real contribution $f_{y,M}$ changes sign when
$R$ is replaced by $R^*$. A coefficient containing at most four
Lie-algebra factors in its color trace must therefore vanish by
\eqref{gen:ordered-traces}. The trace includes the full fermion multiplet;
taking absolute values separately for its species would lose the
cancellation. We will
apply this cancellation to an expansion of the background transports,
while retaining the original-link variation that defines the current.

To turn this cancellation into a bound, consider a polynomial term whose
transport paths fit inside a spacetime patch of fixed small radius,
together with the sampled-lattice regions and outgoing steps needed for
the local gap estimate below. The original-link inputs used to reconstruct
these local data may extend beyond the patch. Set $h=h_n$.
Choose a color basis by radial transport from the trace position $y$.
In this basis the connection has a local potential $A^{\rm rad}$, with
$A^{\rm rad}=\sum_\mu A_\mu^{\rm rad}(z)\,\dd z^\mu$ at a point
$z$ in the patch. Radial transport expresses this potential as curvature
integrated along the radius, so its norm is bounded by a constant times
$|z-y|$. This distance factor survives every fixed number of
original-link derivatives.

As in section~\ref{sec:construction}, change the basis at each coarse
site and copy that change to its children, keeping $I_n$ constant.
A fine coefficient then transports from a coarse parent to its child,
across the fine edge, and back to the other parent. These paths and
the coarse edges have length $O(h)$. At polynomial order $M$, they
remain within distance $O(h(M+1))$ of $y$. Integrating the radial
potential along one coefficient path therefore contributes $h^2$,
up to a power of $M+1$: one power of $h$ comes from the path length
and one from the distance to the radial origin.

We will vary these background transports to estimate the current, while
retaining the response to the chosen original-link variation $D_I$.
First evaluate that response at the actual field $U$. For a positive
sampled link $V_e$, transporting from $z_e$ to $y_e$, define its
induced generator
\begin{equation}
 \Xi_e=(D_I V_e)V_e^{-1}.
\end{equation}
This includes derivatives of the reconstructed connection, the initial
color bases, and all transports used to obtain the refined bases. These
generators and their fixed-order original-link derivatives are bounded
uniformly under refinement; they need not become small with $h$.

For any fixed color-basis change $g$, the same variation has generator
$\Xi'_e=g(y_e)\Xi_e g(y_e)^{-1}$ in the new basis. Transforming
the links and their generators together leaves the trace unchanged.
We then substitute the actual radial basis change $g(U)$. This transforms
the already computed variation, rather than redefining it by
differentiating the moving basis. The resulting identity depends on $U$,
so further original-link derivatives act on both $g(U)$ and $\Xi_e(U)$.

We now introduce an auxiliary parameter $0\le u\le1$ and replace the
connection on the patch by $\dd_z+uA^{\rm rad}$, where $\dd_z$
differentiates spacetime coordinates. Recompute the coarse and fine
transports along their fixed paths. They are identities at $u=0$ and
the actual transports in the copied radial bases at $u=1$. This local
interpolation estimates one finite polynomial term; it does not change
$U$, reconstruct a new connection, or require a global admissible path
to the vacuum.

Only during the auxiliary variation in $u$ do we hold $\Xi'_e$ fixed.
On the positive background link $V'_e(u)$, the prescribed variation is
$\Xi'_e V'_e(u)$. Its action on the reverse link follows by
differentiating the inverse, giving the reverse-link generator
\begin{equation}
 -V'_e(u)^{-1}\Xi'_e V'_e(u).
\end{equation}
Thus the reverse generator still varies with $u$. This prescription
keeps the two edge orientations consistent and reproduces the required
original variation at $u=1$.

The prescribed variation $D_I$ supplies one Lie-algebra factor in the
color trace. At the identity background, each $u$ derivative supplies
at most one more. The derivative $\partial_s$ supplies none: it changes
only the scalar coefficients multiplying the endpoint operators.
Consequently the first possible nonzero term needs four background
derivatives.

Let $f(u)$ denote the resulting contribution $f_{y,M}$, with its
induced generators $\Xi'_e$ on positive links held fixed during this
auxiliary variation.
The color identity sets its Taylor coefficients of degrees zero through
three to zero. Taylor's theorem therefore gives the exact expression
\begin{equation}
 f(1)=\frac1{3!}\int_0^1(1-u)^3 f^{(4)}(u)\,\dd u.
 \label{eq:radial-current-remainder}
\end{equation}
Here $f^{(4)}$ is the fourth derivative with respect to $u$.
Each derivative inserts one path integral of $A^{\rm rad}$.
The four insertions supply $h^8$, up to powers of $M+1$.
This is a bound on the full remainder. Since the cancellation holds for arbitrary
Lie-algebra generators, it remains valid after original-link derivatives
act on the actual generators and bases.
Original-link derivatives can also act on the inserted potentials, but
each retains its radial distance factor. Differentiating the other
transport factors does not remove these insertions. Thus the $h^8$
bound survives every fixed order of original-link differentiation.

A curvature coefficient has two prescribed link variations instead of
one, so only three background factors are needed to reach five color
factors. Its remainder is bounded by $h^6$. The anomaly pairing
\eqref{eq:ward} contains one prescribed site generator $R(\omega)$;
the same count gives four background factors and $h^8$ for a bounded
generator. These bounds apply at a trace position, before summing sites.

\subsection{Uniform bounds and the sum over sites}

The powers of $h$ are useful only if the exponential bound in polynomial
order also survives the entire Taylor interpolation. This cannot be
deduced by assuming that the locally interpolated links extend to a
globally admissible field. Instead we use the local gap bounds of
sections~\ref{app:s:supported} and \ref{app:s:local-family}, retaining
every output of a spinor with finite support.

Choose a finite input region containing every path allowed by the
polynomial term, independently of the values of the link matrices.
Let $K_{\rm out}$ apply the Hermitian operator to spinors supported
there, while keeping all output components, including those just outside
the region. The local gap bounds the norm of this full output from
below, uniformly for $0\le u\le1$. Together with the upper operator
bound, it makes $1-K_{\rm out}^\dagger K_{\rm out}/\Lambda^2$
a contraction with a uniform factor strictly less than one.

Retaining the outgoing components matters even for a trace inside the
region. A two-step path can leave the input region at its middle site
and return at its final site. The matrix
$K_{\rm out}^\dagger K_{\rm out}$ includes that intermediate output,
so its entries between input sites agree with those of the full
$\mathcal K^2$. Squaring an operator first compressed to the input
region would omit it. The contracting matrix and the remaining single
operator factors reproduce the chosen polynomial contribution, since
the input region contains every path permitted by its degree.
Appendix~\ref{app:a:locality} gives this finite-matrix construction and
its derivative bounds. No projector for a truncated local operator is
introduced.

This proves the required bounds when the sampled paths and regions
needed for this estimate fit in the radial patch. A term for which they
do not fit, including a path
winding around the torus, must have $h(M+1)\ge c$ for a fixed $c>0$.
We bound it at the actual
gapped field, retaining its transports and fermion boundary signs.
Since $[h(M+1)/c]^{2p}\ge1$, its polynomial-order bound also
gives a factor $h^{2p}$ after increasing only the power of $M+1$.
Taking $p=4$ for the current and $p=3$ for curvature gives the same
spacing estimates for short and long terms, with the exponential decay
in $M$ retained.

We can now sum the trace positions. The chosen original variation $D_I$
changes sampled data only in its fixed reconstruction neighborhood.
A polynomial term can connect its trace position to that neighborhood
only through its finitely many operator steps. Its trace position can
thus lie no farther from the original link than a fixed reconstruction
radius plus $O(h(M+1))$. The possible positions occupy an original
four-volume bounded by a polynomial in $M+1$. The density of coarse parent
blocks is $O(h^{-4})$, and the fixed 81 child indices have already
been traced within each block. The current thus gains $h^{8-4}=h^4$.
The sampled curvature similarly gains $h^{6-4}=h^2$. There is no
second fine-site count for the links affected by $D_I$: their combined
effect was bounded by the maximum generator norm before taking the trace.
On a fixed torus, the anomaly pairing with a uniformly bounded site
parameter also has the $h^4$ bound.

Further original-link derivatives and the sums over their indices add
polynomial factors in $M+1$. Below we include these factors in the
locality norm and show that one positive decay rate works for every
fixed derivative order. The resulting decay with spacing and the
original-link support decomposition are illustrated in
Fig.~\ref{fig:locality-refinement}.

\begin{figure}[htbp]
  \centering
  \includegraphics[width=\linewidth]{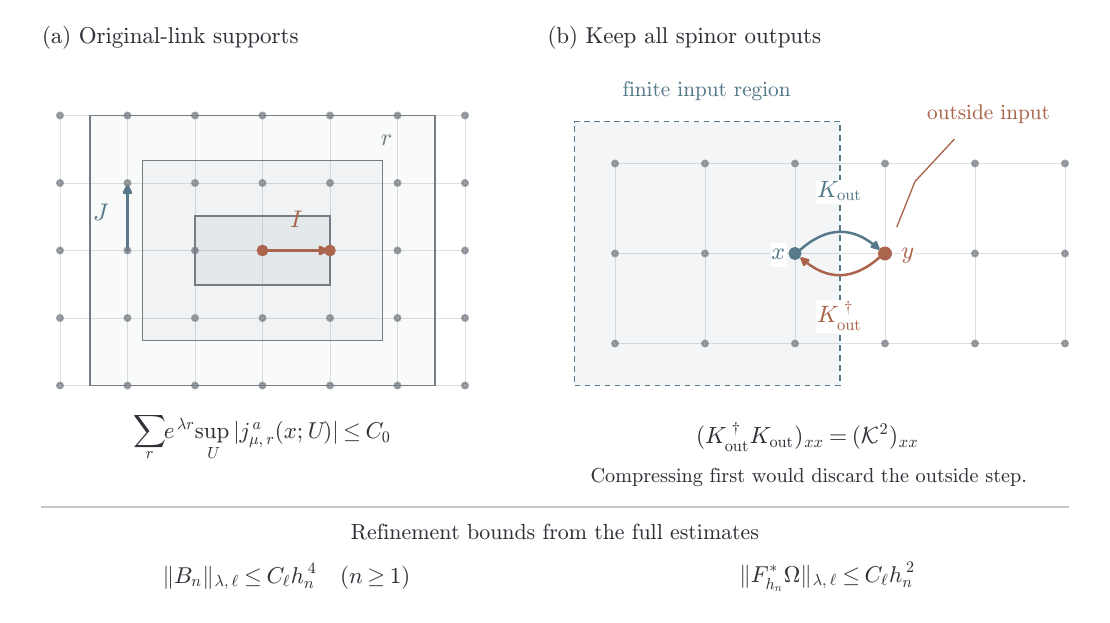}
  \caption{(a) Growing original-link supports of $j_{\mu,r}^a(x;U)$,
with output index $I=(x,\mu,a)$ and differentiated input index $J$.
The radius is $r$ and $\lambda>0$ is the locality rate.
(b) On the operator's lattice, $K_{\rm out}$ applies the Hermitian
operator $\mathcal K$ to a spinor supported in the shaded region while
keeping all outputs. The return path $x\to y\to x$ includes an
intermediate site outside that region. Compressing the operator before
squaring would lose this contribution to $\mathcal K^2$.
The bounds below summarize the separate refinement estimates:
$h_n=3^{-n}$, the norm includes $\ell$ original-link derivatives, and
$C_0,C_\ell$ are bounding constants. There is a separate finite bound
for $B_0$.}
  \label{fig:locality-refinement}
\end{figure}

\subsection{Bounds on the current and curvature}

The locality condition in section~\ref{sec:locality-definition} controls
the current at one original link and its response to changes in other
links. We need the same control for each $B_n$ and for the remaining
fine-lattice curvature $F_h^*\Omega$ in \eqref{eq:telescope}. The
curvature has two variation arguments, so its bound must also sum over
the position and generator of the second variation. This extra sum
controls the total response, rather than only a response at one chosen
pair of links.

We use one notation for these bounds. Let $\beta$ be a differential
form with $p$ variation arguments: $p=1$ for a measure term and $p=2$
for a curvature. As in the current decomposition, write
$\beta=\sum_{r\ge0}\beta_r$, where $r$ is a support radius in
original lattice units. The coefficient $\beta_{I_1\ldots I_p,r}$
is obtained by inserting the elementary variations labeled by
$I_1,\ldots,I_p$; each index specifies an original link and its
Lie-algebra generator, as in \eqref{eq:link-derivative}. Its inputs lie
within radius $r$ of the first link, and it vanishes unless the other
variation links also lie there.

Choose a smooth extension of each coefficient to independent $\SU(5)$
matrices on its finite set of input links, as required in
section~\ref{sec:locality-definition}. The extension obeys the same
support condition. For a positive locality rate
$\lambda$ and $\ell$ additional original-link derivatives, define
\begin{equation}
 \|\beta\|_{\lambda,\ell}
 =\sup_{I_1}\sum_{r\ge0}e^{\lambda r}
 \sum_{\substack{I_2,\ldots,I_p\\J_1,\ldots,J_\ell}}
 \sup_U\bigl|D_{J_1}\cdots D_{J_\ell}
              \beta_{I_1\ldots I_p,r}(U)\bigr|.
 \label{eq:locality}
\end{equation}
This is the current bound \eqref{eq:current-locality}, with the extra
variation indices included when needed. The exponential weight requires
small contributions from large supports; the index sums include all
positions and generator directions. Empty sums are omitted, and the
order of the derivatives is retained. The supremum over $U$ is over the
independent inputs of the chosen extension, without an admissibility
condition. The norm refers to that extension and support decomposition.

The estimates we need are the following.
\begin{proposition}
\label{prop:locality}
On $\mathcal U_\eps$ for sufficiently small $\eps>0$, there is a
$\lambda>0$ such that, for every fixed nonnegative derivative order
$\ell$,
\begin{equation}
 \|B_0\|_{\lambda,\ell}\le C_\ell,\qquad
 \|B_n\|_{\lambda,\ell}\le C_\ell h_n^4\quad(n\ge1),\qquad
 \|F_h^*\Omega\|_{\lambda,\ell}\le C_\ell h^2
       \quad(h=3^{-n},\ n\ge0).
 \label{eq:scale-bounds}
\end{equation}
Each finite-support coefficient has a smooth extension to all independent
input link matrices as required in \eqref{eq:locality}. The admissibility
domain and the positive rate $\lambda$ are independent of volume,
component, and derivative order. The constants $C_\ell$ may depend on
the derivative order but are independent of volume, component, and
refinement level.
\end{proposition}

\subsection{Summation and differential integrability}

To finish the locality estimate, we need the smooth extensions to
unrestricted input links used in \eqref{eq:locality}. The preceding
bounds require plaquette conditions only in the neighborhood used by
each term. Multiply that term by smooth functions of those original
plaquettes, equal to one near the target admissible domain and zero
before the local construction loses its bounds. The product then extends
smoothly by zero to all independent input matrices. These functions
change neither the current nor any of its derivatives on the target
domain.

Differentiating these cutoff functions a fixed number of times adds only
a polynomial in the number of input links. To see why one positive
locality rate works for every derivative order, compare this cost with
the decay in polynomial order. The range estimate used in the site count
also bounds the support of each coefficient. Centering that support on
the chosen original link changes only the constants, so its radius grows
at most as a fixed constant times $1+h(M+1)$. Choosing $\lambda$
sufficiently small, with $h\le1$, allows the
exponential support weight $e^{\lambda r}$ to use only part of the
exponential decay in $M$. The remaining decay sums every fixed power
of $M+1$, including those from derivatives, index sums, and cutoffs.
Increasing the derivative order changes the constant $C_\ell$, not
the required admissibility domain or the chosen positive rate.

Appendix~\ref{app:a:locality} groups these contributions by their
original-link support radius and establishes \eqref{eq:scale-bounds}.
The same polynomial estimates and extensions give the finite bound on
$B_0$, without a small-spacing factor. Since $\sum_{n\ge1}h_n^4$
converges, so do $\mathfrak L=\sum_nB_n$ and
$j_\mu^a(x)=\sum_nj_{\mu,n}^a(x)$, with every derivative controlled
by the locality norm. They define a smooth, exponentially local measure
term and current.

We can therefore differentiate the sum and take the limit in the finite
curvature identity \eqref{eq:telescope}. Its remaining term obeys
$\|F_{h_{n_{\max}}}^*\Omega\|_{\lambda,\ell}
\le C_\ell h_{n_{\max}}^2$ and tends to zero. Hence
\begin{equation}
 \dd\mathfrak L=\Omega,
\end{equation}
which is the differential integrability condition \eqref{eq:curvature}.

\paragraph{A closed extension of the curvature.}
As a further consequence, the extended coefficients define a smooth
one-form $\mathfrak L_{\rm ext}$ on unrestricted links. Its exterior derivative
agrees with $\Omega$ on $\mathcal U_\eps$ and is closed because
$\dd^2=0$. It remains exponentially local: if differentiation makes
an input link the first variation index, a support of radius $r$
about the old center lies within radius $2r+1$ of the new one.
A further fixed reduction of the positive locality rate covers this
enlargement and the index sums at every derivative order. Thus
$\dd\mathfrak L_{\rm ext}$ is a closed local extension of the
overlap curvature. This local result does not fix phases around
noncontractible loops; full integrability requires the separate argument
in section~\ref{sec:global}.

\section{Gauge invariance and lattice symmetries}
\label{sec:symmetries}

The comparison sum now defines a smooth, exponentially local current
with the required curvature. We must also establish its gauge covariance
and the covariant divergence equation. We prove these directly from
$B_n$, then derive the lattice symmetry laws and check their consistency
with the projector determinant phases.

\subsection{Gauge covariance}

Consider a fixed gauge transformation $g(x)$ of the original links.
In the transported color bases of section~\ref{sec:construction},
fermion components at a refined site transform by $R(g(i))$, where
$i$ is the original vertex reached by following that site's parents.
A coarse site and its fine children have the same original ancestor,
so the same gauge matrix acts on their components. Copying a coarse
spinor with $I_n$ therefore commutes with applying the gauge transformation.

For fixed $g$, this action on the numerical components is independent
of $U$, even though the transported color bases depend on $U$.
Consequently, transforming $U$ and its variation conjugates
$H_{n,s}$, $P_{n,s}$, and both projector derivatives in
\eqref{eq:scale-current} by the same fixed unitary operator.
Their trace is unchanged. Each $B_n$ is therefore invariant under this
simultaneous transformation, as is the convergent sum:
\begin{equation}
 \mathfrak L_{\eta^g}[U^g]=\mathfrak L_\eta[U],\qquad
 \eta^g_\mu(x)=g(x)\eta_\mu(x)g(x)^{-1}.
\end{equation}
Since the original link variations can be chosen independently,
invariance of their pairing with the current gives
\begin{equation}
 j_\mu(x;U^g)=g(x)j_\mu(x;U)g(x)^{-1}.
\end{equation}

\subsection{The covariant divergence equation}

Covariance alone does not determine the divergence of the current.
We must show that it equals the lattice gauge anomaly, as required by
\eqref{eq:anomaly-equation}. As in \eqref{eq:ward}, we do this by
evaluating the measure term on the gauge variation $\eta=-\nabla\omega$.
Each comparison will give the anomaly on its coarse lattice minus the
anomaly on its fine lattice. Summing then leaves the original anomaly
and a fine-lattice remainder.

Keep the original gauge parameter $\omega(x)$ fixed. Let $\omega_n(y)$
be its value at the original ancestor of a level-$n$ site $y$, and
let $G_n=R(\omega_n)$ act pointwise on the fermions at that level.
The anomaly summed against this copied parameter is
\begin{equation}
 m_n(U,\omega)=\ii\Tr\{
       (\widehat P_-[U^{(n)}]-P_{+,n})G_n\}.
\end{equation}
The trace is on the level-$n$ fermion space, and $P_{+,n}$ is its
fixed antifermion projector. At the original level, this is precisely
the quantity defined in \eqref{eq:ward}:
\begin{equation}
 m_0=m_\omega=\sum_x\sum_{a=1}^{24}
                  \omega^a(x)\mathcal A^a(x).
\end{equation}

The copied gauge action gives
\begin{equation}
 G_{n+1}I_n=I_nG_n,\qquad
 \delta_{-\nabla\omega}P_{n,s}=[G_{n+1},P_{n,s}].
\end{equation}
The first identity expresses the agreement of the gauge actions before
and after copying a spinor. The second follows by differentiating gauge
covariance: an infinitesimal gauge transformation acts on the comparison
projector by a commutator. It includes the gauge variation of all sampled
links and their transported bases.

Substituting this commutator into $B_{n,-\nabla\omega}$ turns the
comparison integral into an endpoint difference. For an operator $G$
and a variation $X$ of a projector $P$ on the same fermion space,
differentiating $P^2=P$ gives $PX+XP=X$.
Using this identity and cyclicity of the trace yields
\begin{equation}
 \Tr\{P[X,[G,P]]\}=\Tr(GX).
\end{equation}
Here we take $P=P_{n,s}$, $X=\partial_sP_{n,s}$, and $G=G_{n+1}$.
The integrand in \eqref{eq:scale-current} is therefore
$\ii\Tr(G_{n+1}\partial_sP_{n,s})$. Since $G_{n+1}$ is independent
of $s$, integrating simply subtracts its traces on the two endpoint
subspaces.

To identify those traces with $m_n$ and $m_{n+1}$, subtract the same
reference trace $\ii\Tr(P_{+,n+1}G_{n+1})$ at both endpoints.
This does not change their difference. At $s=0$, the projector is
the fine overlap projector, giving $m_{n+1}$. At $s=1$, it is the
embedded coarse projector together with $P_{+,n+1}$ on the added
fine-space states. The added states cancel against the corresponding
part of the reference trace. The relation $G_{n+1}I_n=I_nG_n$
then identifies the remaining trace with $m_n$ on the coarse space.
Thus
\begin{equation}
 B_{n,-\nabla\omega}=m_n-m_{n+1},\qquad
 \sum_{n=0}^{n_{\max}-1}B_{n,-\nabla\omega}=m_\omega-m_{n_{\max}} .
 \label{eq:ward-telescope}
\end{equation}
Every intermediate anomaly cancels in the sum. This includes the first
comparison, whose coarse field is the actual original $U$ and whose
fine field is $U^{(1)}=F_{h_1}U$.

It remains to remove $m_{n_{\max}}$. The trace estimate with one
inserted site generator, explained in section~\ref{sec:locality} and
proved in appendix~\ref{app:d:cylinder-traces}, gives
$m_{n_{\max}}=O(h_{n_{\max}}^4)$ for fixed original field, gauge
parameter, and volume. The copied parameter can jump between regions
with different original ancestors, but its maximum norm is the same
as that of $\omega$. This is the only norm of the parameter used in
the estimate; no spatial differences of $\omega_{n_{\max}}$ are taken.
The remainder therefore tends to zero, and convergence of the current
allows us to take the limit in \eqref{eq:ward-telescope}:
\begin{equation}
 \mathfrak L_{-\nabla\omega}
       =\sum_x\sum_{a=1}^{24}\omega^a(x)\mathcal A^a(x).
\end{equation}
Discrete integration by parts on the left gives
\eqref{eq:anomaly-equation} at each site and for each generator,
since $\omega$ was arbitrary. This is the divergence equation needed
to cancel the operator contribution to the gauge variation of the
fermion integral. Section~\ref{sec:measure-reconstruction} uses it
to prove gauge invariance of the reconstructed measure, including in
zero-mode sectors.

\subsection{Rotations, reflections, and charge conjugation}

We next transform the entire comparison under a change of lattice
coordinates. Let $O$ be a coordinate permutation combined with any
coordinate reflections, as in \eqref{eq:kernel-symmetry}. The induced
link variation is
$\eta^O_\nu(y)=\delta_\eta(OU)(y,\nu)(OU)(y,\nu)^{-1}$.
An axial vector transforms like an ordinary vector under rotations,
but has an additional factor $\det O$ under reflections. We will
derive this factor for the measure term, then translate the result
into components of the current on positive links.

The first identity in \eqref{eq:kernel-symmetry} transforms
$H_{n,s}$ by a fixed unitary change of position and spin coordinates,
with an overall factor $\det O$. For $\det O=1$, the positive
spectral projector $P_{n,s}$ is therefore unitarily conjugated, leaving
the trace defining $B_n$ unchanged. For $\det O=-1$, the sign of
the Hermitian operator is reversed, so its positive spectral projector
becomes a unitarily conjugated copy of $1-P_{n,s}$.

Replacing $P$ by $1-P$ changes the signs of both differentiated
projectors in $\Tr\{P[\partial_sP,\delta_\eta P]\}$.
These two signs cancel inside the commutator. The remaining factor
$1-P$ reverses the trace, because the full trace of the commutator
vanishes. Thus each $B_n$, and hence $\mathfrak L$, acquires exactly
the factor $\det O$.

Charge conjugation follows from the second identity in
\eqref{eq:kernel-symmetry}. Replacing $R$ by $R^*$ complex
conjugates the comparison projector and its derivatives, up to the
fixed spin matrix $\mathcal B$. The trace in $B_n$ is purely
imaginary: the projector and its derivatives are Hermitian, while
their commutator is anti-Hermitian. Conjugating the trace therefore
reverses the sign of the real measure term,
$\mathfrak L^{R^*}_\eta[U]=-\mathfrak L^R_\eta[U]$.
Equivalently, for $R=\mathbf{10}\oplus\bar{\mathbf5}$,
complex conjugating the links realizes the conjugate representation,
up to a fixed unitary color-basis change. At fixed $R$, the two laws are
\begin{equation}
 \mathfrak L_{\eta^O}[OU]=(\det O)\mathfrak L_\eta[U],\qquad
 \mathfrak L_{\bar\eta}[\bar U]=-\mathfrak L_\eta[U].
 \label{eq:axial}
\end{equation}
Translations simply translate the sampled links and comparison
operators, giving translation covariance. These transformations act
on the added fine-space states as well as the embedded coarse states;
the fixed complement introduces no extra contribution.

There is one further step in expressing the coordinate-transformation
law for $j_\mu(x;U)$: the current is assigned to a positive link
and acts on color vectors at its starting vertex.
Consider $e=(x,x+\hat\mu)$, with $U_e=U(x,\mu)$.
If $O\hat\mu=+\hat\nu$, its image is a positive link starting
at $Ox$. If $O\hat\mu=-\hat\nu$, its image points in the
negative direction. The corresponding positive link starts at
$Ox-\hat\nu$ and carries the inverse matrix $U_e^{-1}$.
Differentiating that inverse gives its variation generator:
\begin{equation}
 \eta_{e^{-1}}=(\delta_\eta U_e^{-1})U_e
              =-U_e^{-1}\eta_e U_e.
\end{equation}
The minus sign comes from inversion, and conjugation by $U_e^{-1}$
expresses the generator at the other endpoint.

Substitute this relation into the invariant pairing
$\sum_{x,\mu}\langle j_\mu(x),\eta_\mu(x)\rangle$ of
section~\ref{sec:measure}. Comparing the independent variations in
\eqref{eq:axial} gives
\begin{equation}
\begin{aligned}
 j_\nu(Ox;OU)&=(\det O)j_\mu(x;U),
       &&O\hat\mu=+\hat\nu,\\
 j_\nu(Ox-\hat\nu;OU)&=-(\det O)U_e^{-1}j_\mu(x;U)U_e,
       &&O\hat\mu=-\hat\nu .
\end{aligned}
\label{eq:axial-components}
\end{equation}
For the reflection $x_1\mapsto-x_1$, the current in direction 1
has two minus signs: one from $\det O=-1$ and one from reversing
the link. They cancel, leaving the conjugation that moves the current
to its new starting vertex. A transverse positive link is not reversed
and retains the axial minus sign. Similarly, comparing the independent
variations in the conjugation law gives
$j_\mu(x;\bar U)=-\overline{j_\mu(x;U)}$.

These transformations are symmetries of a fixed torus when they preserve
its periods and assigned fermion boundary signs. Otherwise the laws
relate the two transformed tori with their corresponding boundary
conditions. Complex conjugation preserves the real periodic or
antiperiodic signs. The same qualification applies to the finite-volume
statements below.

\subsection{Consistency of the loop phases}

The symmetry laws must also respect the full integrability condition,
which compares the measure Wilson line $W[\gamma]$ with the projector
determinant phase $\mathcal T[\gamma]$ around a closed admissible
curve $\gamma$. A sign change of $\mathfrak L$ inverts
$W[\gamma]=\exp(\ii\oint_\gamma\mathfrak L)$. We check that
the corresponding transformation of the projector inverts $\mathcal T$
as well.

For a closed path of projectors $P_t$, replacing $P_t$ by $1-P_t$
leaves the transport generator $[\partial_tP_t,P_t]$ in
\eqref{eq:projector-transport} unchanged. The two complementary
subspaces therefore have the same full transport operator $Q_t$.
Its generator has zero trace, so $\det Q_1=1$. Since the path is
closed, $Q_1$ preserves the initial subspace and its complement.
Its full determinant is the product of the determinants on those two
subspaces; their return phases must therefore be inverses.
A fixed antiunitary conjugation instead complex conjugates $Q_1$
and its subspace determinant, which also inverts the phase.

Consequently an orientation-reversing coordinate transformation or
complex conjugation inverts both $W$ and $\mathcal T$;
orientation-preserving coordinate transformations and translations leave both
unchanged on the transformed curve. This establishes consistent
transformation laws for both phases. Their equality for every closed
curve remains to be proved in section~\ref{sec:global}.
\section{Global integrability}
\label{sec:global}

The current satisfies the differential integrability condition and the
covariant divergence equation. To reconstruct the measure, we must still
show that its Wilson line $W[\gamma]$ equals the projector determinant
phase $\mathcal T[\gamma]$ for every closed admissible curve, including
curves that cannot be contracted. The finite comparison identity
\eqref{eq:telescope} reduces this question to a phase on the finest
auxiliary lattice. For a fixed loop $\gamma:t\mapsto U_t$, write
\begin{equation}
 \gamma_h:t\longmapsto F_hU_t,\qquad h=3^{-n}.
\end{equation}
The identity expresses $\mathcal T[\gamma]$ as the Wilson line of a
finite comparison sum multiplied by $\mathcal T[\gamma_h]$. The sum
already converges to $\mathfrak L$, so the remaining task is to prove
$\mathcal T[\gamma_h]\to1$ as $h\to0$.

We compare the reconstructed connection along $\gamma$ with an auxiliary
$\SU(3)$ connection for which the charged fermions occur in conjugate
pairs. Including color transport along the loop preserves this pairing
and fixes the auxiliary return phase to $+1$. We then remove the auxiliary color transport to recover the ordinary
phase at that boundary, and use the ordinary projector curvature to
compare it with the reconstructed boundary. The trace cancellations of section~\ref{sec:locality} make both
phase corrections vanish under refinement. The comparison is illustrated
in Fig.~\ref{fig:history-connection-cylinder}.

\subsection{An auxiliary connection with conjugate fermion pairs}

Fix a smooth closed admissible loop on the original finite lattice.
Reparametrize it to be stationary near $t=0$ and $t=1$, so these endpoints
can be joined smoothly. At every point $y$ of spacetime and every loop
parameter $t$, reconstruction supplies the five-dimensional color space
$E_y(U_t)$. Together these spaces form a rank-five Hermitian bundle $E$
over
\begin{equation}
 Y=T^4\times S^1.
\end{equation}
Here $S^1$ is the circle obtained by identifying the endpoints of the
loop parameter; it is not an additional physical spacetime direction.
This family of color spaces is what we call the history bundle.
The family construction in section~\ref{sec:construction} supplies a
connection $\nabla^0$ along $t$ as well as along the four spacetime
directions. It preserves the normalized volume element fixed by the
reconstruction.

For a related use of subgroup reduction in five-dimensional continuum calculations of global anomalies, see ref.~\cite[sec.~3.1.2]{GarciaEtxebarriaMontero}. To obtain an auxiliary connection acting through $\SU(3)$, we will choose
two orthonormal color vectors smoothly over $Y$ and make them parallel.
The connection will then act only on the remaining three colors.

The existence of the two vectors follows from the rank of $E$ and the dimension of
$Y$. Using local color bases and a partition of unity, choose finitely
many smooth sections that span each color space. A generic linear
combination is transverse to zero~\cite[Ch.~3, Theorem~2.7]{Hirsch}:
its zeros, if present, would have real codimension ten, the real dimension
of a five-component complex vector. Since $Y$ is only five-dimensional,
there are no zeros, and the section can be normalized everywhere.
Apply the same argument to its rank-four orthogonal complement. A generic
section there would have zeros of codimension eight, so it too is nowhere
zero. Normalizing it gives the second color vector, orthogonal to the first.

The two vectors span a trivial two-dimensional color subbundle. Let $V$
be their three-dimensional orthogonal complement. The normalized volume
element of $E$, together with the ordered pair of vectors, fixes a
normalized volume element in $V$. We can therefore use local bases of $V$
whose changes of basis lie in $\SU(3)$, and write
$E=V\oplus\mathbb C^2$. The space $V$ need not admit one global basis;
only the two distinguished color vectors have been chosen globally.

Under this $\SU(3)$, the fundamental representation decomposes as
$\mathbf5=\mathbf3\oplus\mathbf1\oplus\mathbf1$. The $\mathbf{10}$
is the antisymmetric product of two fundamental colors: two colors in
$\mathbf3$ give $\bar{\mathbf3}$, one in $\mathbf3$ and one in either
singlet give two copies of $\mathbf3$, and the two singlets give one
singlet. Adding $\bar{\mathbf5}$ therefore gives
\begin{equation}
 (\mathbf{10}\oplus\bar{\mathbf5})|_{\SU(3)}
       =2\mathbf3\oplus2\bar{\mathbf3}\oplus3\mathbf1 .
 \label{eq:real-restriction}
\end{equation}
Each triplet can be paired with an antitriplet. Complex conjugation
exchanges the two members of each pair and conjugates the singlets,
defining an antiunitary map $J_{\rm col}$ on the color space in the
fermion representation, with $J_{\rm col}^2=1$. It preserves norms and
is antilinear in complex scalars. This is the color conjugation symmetry
we will use to determine the phase.

Choose an $\SU(3)$ connection on $V$ and require the two distinguished
color vectors to be parallel. The resulting connection $\nabla^1$ on $E$
preserves $J_{\rm col}$ in the fermion representation, in spacetime and
along $t$. The two connections live on the same bundle; no gauge transformation
relating them is required. The auxiliary
$\SU(3)$ need not be the physical color subgroup.

We can now interpolate between the two connections. Their difference
$\mathfrak a=\nabla^1-\nabla^0$ is a one-form on $Y$ whose components
are traceless anti-Hermitian maps of the color space to itself. Set
\begin{equation}
 \nabla^r=\nabla^0+r\mathfrak a,\qquad 0\le r\le1.
\end{equation}
The parameter $r$ changes the connection; $t$ still runs around the
original loop. We use the same bundle and the same changes of local
basis at every $r$, and set the connection component along $r$ to zero.

To obtain matrices on a fixed fine-lattice fermion space, use the color
bases $f_h$ of section~\ref{sec:construction}, transported from original
vertices along the parent-to-child paths by $\nabla^0$. Keep these bases
fixed as $r$ varies. Sampling the spatial transports of $\nabla^r$ in
them gives links $U_h(r,t)$, with
\begin{equation}
 U_h(0,t)=F_hU_t=\gamma_h(t).
\end{equation}
The bases and paths are periodic in $t$, so each fixed $r$ gives a closed
loop of sampled links. Thus $(r,t)$ describes a cylinder with the
reconstructed connection at one boundary and the auxiliary connection
at the other. It does not contract the original loop to a point.

We choose this cylinder once for the fixed loop and then let only $h$
tend to zero. Its spatial curvature is bounded; we do not require its
samples at the original spacing to be admissible. Sufficiently fine
samples will be admissible, as shown below. The bounds and the required sampling
spacing may depend on the loop, the original spacing and finite volume,
and the auxiliary connection. We need no smooth choice of the auxiliary
connection over the whole gauge-field space: it is used to prove a phase
identity.

\subsection{Transport on the cylinder and its auxiliary boundary}

The auxiliary connection preserves the conjugate pairing, but its matrix
expression in the sampled color bases can change around the loop.
Ordinary projector transport differentiates components in these bases and
need not preserve that pairing. We therefore include color transport when
determining the auxiliary return phase. To compare the two boundaries of
the cylinder, ordinary projector transport already supplies the comparison
identity of section~\ref{sec:comparison}.

At a fine vertex $y$, let $f=f_h(t,y):\mathbb C^5\to E_{(y,t)}$
be the inherited color basis. It is independent of $r$ and periodic in
$t$. The connection matrices in the parameter directions are
\begin{equation}
 A_r=0,\qquad
 A_t(r,t)=f^\dagger\nabla^0_{\partial_t}f
                       +r f^\dagger\mathfrak a_t f,
 \qquad \mathfrak a_t=\mathfrak a(\partial_t).
 \label{eq:sampled-parameter-potentials}
\end{equation}
Both terms in $A_t$ are needed. The first measures how the inherited
basis changes relative to $\nabla^0$; transporting this basis along
spatial paths does not make it parallel around the loop.

These matrices are bounded uniformly in the sampling spacing.
Differentiating an inherited basis along $t$ gives the original-vertex
contribution $C_i(i,U_t)^\dagger\nabla^0_{\partial_t}C_i(i,U_t)$,
together with an integral of mixed spatial and $t$ curvature along its
parent-to-child path. The color factors are transported into the same
basis. The original-vertex term and curvature are bounded on the fixed
family, and the path has total length at most two in original units.
The extra matrix $f^\dagger\mathfrak a_t f$ is bounded as well.

The spatial curvature of the fixed cylinder is bounded, so its sampled
plaquettes differ from the identity by $O_\gamma(h^2)$. For sufficiently
small $h$, the ordinary Hermitian Wilson operator therefore has the gap
\eqref{app:s:ordinary-gap} uniformly on the cylinder and defines the
smooth projector $P(r,t)=\widehat P_-[U_h(r,t)]$.
The intermediate links need not be reconstructed from original fields;
the gap follows from their plaquette bound.

The inherited bases identify the color space at each fine vertex with
$\mathbb C^5$ around the entire $t$ circle. They need not extend to a
single smooth frame over spacetime. Using the same vertex bases for every
$r$, with the fixed spatial periodic identifications and real fermion
boundary signs, puts all these projectors on one numerical Hilbert space.
In particular, $P(r,1)=P(r,0)$ without an additional endpoint change of
basis. The determinant comparison below uses this fixed space.

For an edge $e=(y,z)$, with $z=y+h\hat\mu$ and terminal-to-initial
link matrix $U_e$, define its ordinary variation generators by
\begin{equation}
 \Xi_{\tau,e}=(\partial_\tau U_e)U_e^{-1},
 \qquad \tau\in\{r,t\}.
 \label{eq:ordinary-cylinder-variation}
\end{equation}
They are traceless and anti-Hermitian, and
$\partial_\tau P=\delta_{\Xi_\tau}P$.
The variation formula for parallel transport writes $\Xi_{\tau,e}$
as an integral of mixed parameter and spatial curvature along the edge,
minus $A_\tau(y)$ plus $U_eA_\tau(z)U_e^{-1}$.
The curvature is bounded on the fixed cylinder, and the edge has length
$h$. Together with the preceding endpoint bounds, this gives
$\|\Xi_\tau\|_\infty\le C_\gamma$ uniformly in $h,r,t$.
We will use only this maximum-norm bound.

We now include color transport only at the auxiliary boundary $r=1$.
Write $\mathsf A_t=R(A_t(1,t))$ and
$D_tP=\partial_tP+[\mathsf A_t,P]$ there. The covariant projector
transport is
\begin{equation}
 \partial_t\widehat Q_t
       =([D_tP,P]-\mathsf A_t)\widehat Q_t,\qquad
 \widehat Q_0=\Id.
\end{equation}
Its generator is anti-Hermitian, and it satisfies
$P(1,t)\widehat Q_t=\widehat Q_tP(1,0)$.
For an initially allowed vector $v_t=\widehat Q_tv_0$, it gives
$P(\partial_t+\mathsf A_t)v_t=0$: the vector is transported in the
allowed subspace with its change of color basis measured covariantly.

Apply the same rule to the reference subspace $\im P_+$.
Since $[\mathsf A_t,P_+]=0$, its transport generator is
$-\mathsf A_t$. The relative determinant line contains the dual of
this reference factor. Its auxiliary covariant return phase
$\widehat{\mathcal T}_h(1)$ is therefore the fermion return determinant
divided by the reference return determinant.

\subsection{Comparing the determinant phases}

Let $\mathcal T_h(r)=\mathcal T[t\mapsto U_h(r,t)]$ be the ordinary
projector phase, with the reference vectors fixed. Across the cylinder,
define the ordinary projector curvature
\begin{equation}
 \kappa_h(r,t)=\ii\Tr\{P[\partial_rP,\partial_tP]\}.
 \label{eq:ordinary-cylinder-curvature}
\end{equation}
The projectors form a smooth family on one fixed numerical Hilbert
space. The cylinder argument in the proof of the comparison identity
in section~\ref{sec:comparison} therefore applies directly, with $r$
as its interpolation parameter. It needs only this fixed-loop family,
not an auxiliary connection defined near the loop in field space.
With orientation
$\dd r\wedge\dd t$, it gives
\begin{equation}
 \frac{\mathcal T_h(1)}{\mathcal T_h(0)}
 =\exp\!\left(\ii\int_{[0,1]\times S^1}
                              \kappa_h\,\dd r\,\dd t\right).
 \label{eq:ordinary-cylinder-holonomy}
\end{equation}
The boundary loops are traversed in opposite directions. This comparison
uses their cylinder of projectors; it does not require either loop to
bound a disk in the admissible field space.

At the auxiliary boundary, we must also compare ordinary and covariant
transport. Adding color transport adds $\Tr(P\mathsf A_t)$ to the
fermion determinant connection and subtracts
$\Tr(P_+\mathsf A_t)$ for its reference factor. Set
\begin{equation}
 m_h^1(t)=\ii\Tr\{(P(1,t)-P_+)\mathsf A_t\}.
 \label{eq:auxiliary-color-correction}
\end{equation}
This is the anomaly pairing \eqref{eq:ward}, evaluated on the auxiliary
sampled field with the site parameter $A_t(1,t)$. Parallel transport
exponentiates the negative of the connection coefficient, so
\begin{equation}
 \widehat{\mathcal T}_h(1)=\mathcal T_h(1)
                 \exp\!\left(\ii\int_{S^1}m_h^1\,\dd t\right).
\end{equation}
Combining the two phase comparisons yields
\begin{equation}
 \mathcal T[\gamma_h]=\widehat{\mathcal T}_h(1)
 \exp\!\left[-\ii\int_{[0,1]\times S^1}\kappa_h\,\dd r\,\dd t
              -\ii\int_{S^1}m_h^1\,\dd t\right].
 \label{eq:phase-comparison-ordinary-cylinder}
\end{equation}
We will determine the auxiliary covariant phase exactly and bound both
integrals. The curvature trace contains two ordinary link generators,
$\Xi_r$ and $\Xi_t$; the correction has one site generator.
All traces are over fine sites, spin, and color.

\subsection{The auxiliary phase and the refinement limit}

At $r=1$, the connection preserves the color conjugation map
$J_{\rm col}$ with square $+1$. Combine it with the spin antiunitary
$J_{\rm spin}$ introduced with \eqref{eq:kernel-symmetry}.
As shown in appendix~\ref{app:s:spectral}, $J_{\rm spin}$ commutes
with the Euclidean gamma matrices and has square $-1$. The combined
antiunitary $J$ on the fine fermion space consequently obeys
\begin{equation}
 J^2=-1,\qquad [J,P]=[J,P_+]=0,\qquad
 \partial_tJ+[\mathsf A_t,J]=0 .
\end{equation}
The sampled $\nabla^1$ transports preserve the color pairing, so the
Hermitian Wilson operator and its spectral projector commute with $J$.
The last equality says that $J$ is parallel along $t$. The prescribed
periodic or antiperiodic fermion boundary signs are real and are therefore
also preserved by $J$.

We can now see why color transport was needed. Covariantly differentiating
$[J,P]=0$, and using that $J$ is parallel, gives $[J,D_tP]=0$.
The equation for $\widehat Q_t$ then implies that
$\widehat Q_t^{-1}J(t)\widehat Q_t$ is constant. Since the loop is
closed, $J(1)=J(0)$, and the return operator commutes with the initial
$J$. The same argument applies to transport in $\im P_+$.

Consider either unitary return operator. If $v$ is an eigenvector with
eigenvalue $\lambda$, antiunitarity makes $Jv$ an eigenvector with
eigenvalue $\bar\lambda$~\cite[sect.~3]{SuzukiReal}.
Moreover,
$(Jv,v)=(Jv,J^2v)=-(Jv,v)$, where the first equality uses
antiunitarity and the second uses $J^2=-1$. Thus the two vectors are
orthogonal, even for the real eigenvalues $\lambda=\pm1$.
Their orthogonal complement is preserved by both the return operator
and $J$, so repeating the argument gives a complete paired eigenbasis.
Every pair contributes $\lambda\bar\lambda=1$ to the determinant.
Both return determinants are therefore $+1$, and
\begin{equation}
 \widehat{\mathcal T}_h(1)=1 .
\end{equation}
This fixes the sign as well as the phase. It is a statement about the
auxiliary covariant transport, not a reality condition on the physical
chiral determinant. The ordinary projector phase still requires the two
corrections in \eqref{eq:phase-comparison-ordinary-cylinder}.

\begin{lemma}[Trace bounds on the comparison cylinder]
\label{lem:fixed-family-traces}
For all sufficiently small $h$, uniformly in $(r,t)$,
\begin{equation}
 |\kappa_h(r,t)|\le C_\gamma h^2,\qquad
 |m_h^1(t)|\le C_\gamma h^4 .
 \label{eq:phase-bounds}
\end{equation}
The constant may depend on the fixed loop, its original lattice spacing
and finite volume, and the chosen cylinder, but not on $h$.
\end{lemma}

These bounds use the same color cancellation as
section~\ref{sec:locality}, now applied to the connections on this fixed
cylinder. At a given $(r,t)$, expand the overlap projector in the
polynomials \eqref{eq:local-polynomial-projector}. For a contribution
of total order $M$, choose a spatial radial color basis centered on its
trace position. Each transport path has steps of length $O(h)$ and
reaches distance $O(h(M+1))$ from that center. The radial potential is
bounded by the distance from the center. Integrating it along a step
of length $O(h)$ therefore supplies $h^2$, up to powers of $M+1$, for
each background insertion in the exact Taylor remainder.

Together with spin conjugation, the ordered-trace identity \eqref{gen:ordered-traces}
cancels these real trace contributions whenever they contain fewer than
five Lie-algebra factors. The curvature $\kappa_h$ already contains two prescribed ordinary link
variations, $\Xi_r$ and $\Xi_t$. It therefore needs three additional
background factors, leaving a bound proportional to $h^6$ at each trace
position. The correction $m_h^1$ has one site generator $A_t(1,t)$,
so it needs four background factors and is bounded by $h^8$ before the
site sum. There are $O(h^{-4})$ sites on the fixed torus. Summing them,
together with the exponentially decreasing polynomial orders, gives
$h^2$ for $\kappa_h$ and $h^4$ for $m_h^1$.
The Taylor remainder is exact, so these are bounds on the full traces,
not on a truncated weak-field expansion.

The actual matrices $\Xi_r,\Xi_t,A_t(1,t)$ enter these estimates
through their uniformly bounded maximum norms. No spatial differences
of these insertion fields are taken. In using a radial basis to estimate
a finite trace, we conjugate these matrices as already specified
insertions; we do not redefine $A_t(1,t)$ as the connection coefficient in
a moving radial basis. Appendix~\ref{app:d:cylinder-traces} establishes
the needed polynomial bounds on the cylinder, including the terms whose
paths leave a radial patch. It does not assume that the cylinder is
reconstructed from original lattice fields.

Insert the auxiliary value $\widehat{\mathcal T}_h(1)=1$ into
\eqref{eq:phase-comparison-ordinary-cylinder}. Removing the color
transport at the auxiliary boundary and comparing the ordinary phases
across the cylinder gives
\begin{equation}
 \mathcal T[\gamma_h]=
 \exp\!\left[-\ii\int_{[0,1]\times S^1}\kappa_h\,\dd r\,\dd t
             -\ii\int_{S^1}m_h^1\,\dd t\right]
 \longrightarrow1 .
 \label{eq:limiting-twist}
\end{equation}
Both parameter domains are fixed, so \eqref{eq:phase-bounds} sends both
integrals to zero. The identity determines the phase itself; there is
no choice of logarithm and no remaining sign ambiguity.

\begin{figure}[htbp]
  \centering
  \includegraphics[width=\linewidth]{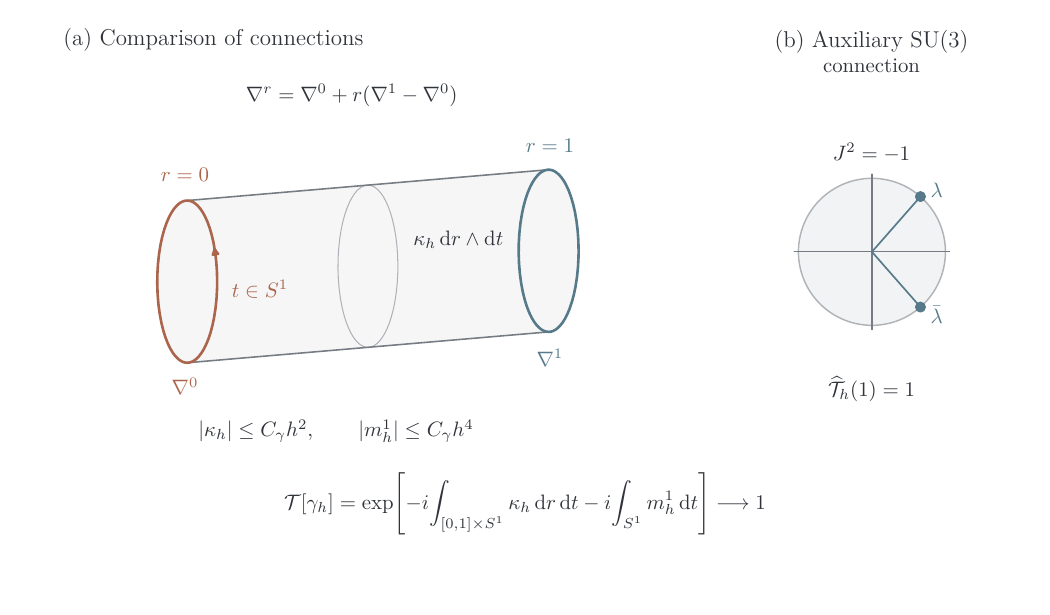}
  \caption{Comparison of phases around a fixed loop $\gamma$ of original gauge fields. The parameter $t$ runs around the loop, while $r$ interpolates between the reconstructed connection $\nabla^0$ and the auxiliary $\SU(3)$ connection $\nabla^1$. At spacing $h$, the original boundary is the sampled loop $\gamma_h=F_h\gamma$. The ordinary projector curvature $\kappa_h$ compares the boundary phases; $m_h^1$ corrects for color transport at the auxiliary boundary $r=1$. In (b), the antiunitary $J$ pairs return eigenvalues $\lambda,\bar\lambda$, fixing the auxiliary covariant phase $\widehat{\mathcal T}_h(1)=1$. The two corrections give the ordinary phase $\mathcal T[\gamma_h]$. The constant $C_\gamma$ may depend on the fixed loop and cylinder, but not on $h$. Only the sampling spacing tends to zero; the cylinder does not contract the original loop.}
  \label{fig:history-connection-cylinder}
\end{figure}

Finally take the refinement limit in \eqref{eq:telescope}. Convergence
of the current allows us to take its sum through the integral along the
fixed loop. Since the remaining projector phase tends to one,
\begin{equation}
 W[\gamma]
 =\lim_{n_{\max}\to\infty}\exp\!\left(
      \ii\int_0^1\sum_{n=0}^{n_{\max}-1}B_{n,\eta(t)}\,\dd t\right)
 =\mathcal T[\gamma],
\end{equation}
where $\eta_\mu(x,t)=\partial_tU_t(x,\mu)U_t(x,\mu)^{-1}$.
This proves full integrability for every smooth closed admissible loop,
including noncontractible loops. For a piecewise smooth loop,
reparametrize each piece to be stationary near its joins. This leaves
both transport phases unchanged and gives the same result.

The same argument will be used for other representations in
section~\ref{sec:other-representations}. Once an auxiliary connection
preserving a color antiunitary with square $+1$ has been found, the phase
argument needs only the ordered-trace cancellation
\eqref{gen:ordered-traces} and the gap on sufficiently fine samples of
the fixed cylinder. The spin antiunitary then fixes the auxiliary
covariant phase, and the two trace bounds remove the corrections.
The particular subgroup used to construct the auxiliary connection does
not enter these remaining steps.

\subsection{Reconstruction of the measure}
\label{sec:measure-reconstruction}

The current now determines a measure at every admissible original field.
Choose a base field $U_0$ in each connected component and a normalized
measure $s_0$ there. Projector transport along a path from $U_0$ to $U$
provides a measure at $U$, but its phase can depend on the path.
Multiplying by the inverse Wilson line removes that dependence, because
$W[\gamma]=\mathcal T[\gamma]$ on every closed curve.

Using the determinant connection $\nabla^{\det}$ of
section~\ref{sec:comparison}, take
$s_0\in\mathcal D(\widehat P_-[U_0])$ with unit norm and define
\begin{equation}
 s(U)=W[\gamma_{0U}]^{-1}
            \operatorname{PT}_{\nabla^{\det}}(\gamma_{0U})s_0 .
 \label{eq:section}
\end{equation}
Here $\gamma_{0U}$ is any path from $U_0$ to $U$, and
$\operatorname{PT}_{\nabla^{\det}}$ transports the fermion basis
exterior product while keeping the antifermion factor fixed. The scalar
correction is
$W[\gamma_{0U}]^{-1}=\exp(-\ii\int_{\gamma_{0U}}\mathfrak L)$.

To compare two paths, follow one to $U$ and return along the reverse
of the other. Their projector transports differ by
$\mathcal T[\gamma]$ for this closed curve, and their correcting
factors differ by $W[\gamma]^{-1}$. The product is one. Individual
transported basis vectors can differ, but the measure, represented by
their exterior product and the fixed antifermion factor, is the same.
Thus \eqref{eq:section} is independent of the path.

Equivalently, the measure is parallel for the corrected connection
$\nabla^{\det}+\ii\mathfrak L$. Its curvature is
$-\ii\Omega+\ii\dd\mathfrak L=0$, and its phase around every
closed curve is $\mathcal T[\gamma]/W[\gamma]=1$.
Smoothness follows locally by choosing paths that vary smoothly with
the endpoint; path independence makes these local constructions agree.
No global choice of individual fermion basis vectors is required.

This also gives a concrete prescription for the phase in Grassmann
coordinates. Choose an orthonormal fermion basis at $U_0$ and express
$s_0$ as its reference measure times a phase. Transport this basis to
$U$ and express the transported vectors in the endpoint reference basis
$v_k(U)$. The determinant of that unitary change-of-basis matrix gives
the transported reference measure relative to $s_v(U)$. Multiplying
by the initial phase specifying $s_0$ and by $W[\gamma_{0U}]^{-1}$
gives $e^{\ii\varphi_v(U)}$ in
$s(U)=e^{\ii\varphi_v(U)}s_v(U)$. Thus the fermion integration measure is
\begin{equation}
 \boxed{
 \mathcal D[\psi]\,\mathcal D[\bar\psi]
 =e^{\ii\varphi_v(U)}
   \left(\prod_{k=1}^N\dd c_k\right)
   \left(\prod_{l=1}^{\bar N}\dd\bar c_l\right)
 },
 \label{eq:reconstructed-measure}
\end{equation}
where $N,\bar N$, the Grassmann coefficients, and their ordering are
those of section~\ref{sec:fermion-measure}. The equation
$(\nabla^{\det}+\ii\mathfrak L)s=0$ gives exactly the measure term
\eqref{eq:measure-basis}. Changes of reference basis obey
\eqref{eq:measure-phase-change}, so these coordinate expressions describe
one measure on overlapping neighborhoods. For the fixed current,
changing $s_0$ changes the measure only by a constant phase on its
connected component.

It remains to check gauge invariance. Along a gauge orbit, transform
the fermion basis columns by $R(g)$ and the antifermion rows by
$R(g)^{-1}$. For an infinitesimal gauge parameter $\omega$, the
connection coefficient of their relative determinant element is
\begin{equation}
 \Tr\{(\widehat P_--P_+)R(\omega)\}=-\ii m_\omega .
\end{equation}
The current adds $\ii\mathfrak L_{-\nabla\omega}=\ii m_\omega$
by \eqref{eq:ward}, canceling this coefficient. The natural gauge action
therefore transports the measure by the same corrected connection as
\eqref{eq:section}. Every element of the full site gauge group
$\prod_{x\in\Gamma}\SU(5)$ can be joined to the identity. Following
such a path gives $g\,s(U)=s(U^g)$, the gauge-invariance condition of
section~\ref{sec:fermion-measure}.

This also covers a gauge transformation that leaves $U$ fixed. Even if
it cannot be connected to the identity through transformations that fix
$U$, it can be connected within the full site gauge group. Acting on $U$
along this path produces a closed admissible curve. The corrected
transport has unit phase on that curve, so the transformation acts
trivially on the measure at $U$.

Together with the smoothness and locality bounds of
section~\ref{sec:locality} and the transformation laws of
section~\ref{sec:symmetries}, this proves theorem~\ref{thm:measure}.
The reconstruction uses no inverse of $D$ and applies in zero-mode
sectors, including when the fermion and antifermion ranks differ.
The current's discrete symmetry laws leave the relative constant phases
of different connected components to be chosen; they do not make every
choice of those phases symmetric.

\section{The Standard Model groups}
\label{sec:standard-model}

The $\SU(5)$ measure also gives us a measure for the Standard Model.
One left-handed generation is the restriction of
$\mathbf{10}\oplus\bar{\mathbf5}$~\cite{GeorgiGlashow}, so we can map
each Standard Model link into $\SU(5)$ and evaluate the measure on
these image links. The gauge links themselves remain Standard Model
variables; using the $\SU(5)$ measure introduces no additional gauge
fields. We must choose the plaquette bound so that the image links are
admissible, and check that restricting the measure preserves the
reconstruction conditions.

\subsection{The four gauge groups and their fermions}

The four global forms of the Standard Model gauge group are~\cite{HucksGlobal,TongGlobal}
\begin{equation}
 G_q=[\SU(3)\times\SU(2)\times\U(1)]/\mathbb Z_q,
 \qquad q\in\{1,2,3,6\}.
 \label{eq:sm-groups}
\end{equation}
Here $\mathbb Z_q$ is generated by $c_6^{6/q}$, where
\begin{equation}
 c_6=(e^{2\pi\ii/3}\Id_3,-\Id_2,e^{\pi\ii/3}),
\end{equation}
and $\Id_m$ is the $m$-dimensional identity matrix. These are the
choices $\Gamma_{\rm SM}=\mathbb Z_q$ introduced above. We work in
terms of integer hypercharge $y=6Y$. In the notation $(r_3,r_2)_y$,
the entries specify the $\SU(3)$ and $\SU(2)$ representations, and
$z\in\U(1)$ acts as $z^y$.

For $g_3\in\SU(3)$, $g_2\in\SU(2)$, and $z\in\U(1)$, the map
\begin{equation}
 \phi(g_3,g_2,z)=\operatorname{diag}(z^{-2}g_3,z^3g_2)\in\SU(5)
 \label{eq:sm-map}
\end{equation}
has kernel $\mathbb Z_6$. It therefore defines a homomorphism
$\phi_q:G_q\to\SU(5)$ by the same formula on quotient elements
$[g_3,g_2,z]$, independent of the representative. Its kernel is
$\mathbb Z_6/\mathbb Z_q$, of order $6/q$, so only $\phi_6$ is
injective. For the other three groups, we must account for this kernel
when choosing the plaquette bound.

The familiar $\SU(5)$ branching rules give one left-handed Standard
Model generation~\cite{GeorgiGlashow}:
\begin{equation}
 \begin{split}
 \mathbf{10}|_{\phi_q}
    &=(\bar{\mathbf3},\mathbf1)_{-4}\oplus(\mathbf3,\mathbf2)_1
      \oplus(\mathbf1,\mathbf1)_6,\\
 \bar{\mathbf5}|_{\phi_q}
    &=(\bar{\mathbf3},\mathbf1)_2\oplus(\mathbf1,\mathbf2)_{-3}.
 \end{split}
 \label{eq:sm-matter}
\end{equation}
Here $|_{\phi_q}$ means that $G_q$ acts through its image in $\SU(5)$;
the direct sum of the two lines is $R_{\rm SM}$ for every $q$.

A fixed unitary change of color basis identifies these multiplets,
independently of the links and compatibly with complex conjugation.
In these bases, the Standard Model Wilson operator is the $\SU(5)$
Wilson operator in $\mathbf{10}\oplus\bar{\mathbf5}$ evaluated on the
image links. It follows that the overlap operator $D$ and chiral
projector $\widehat P_-$ agree in the two descriptions, as does the fixed
antifermion projector $P_+$. We can therefore use the same fermion
transport and measure phases. Since the change of basis is independent
of the links, it contributes no basis derivative to $\mathfrak L$.

\subsection{Admissibility}

For $q<6$, a nonidentity plaquette in $\ker\phi_q$ is mapped to the
identity in $\SU(5)$. A bound on the image plaquette therefore does not
ensure that the Standard Model plaquette itself is close to the
identity. To impose this requirement, we use a faithful unitary
representation, which distinguishes all elements of $G_q$.

A convenient choice is $\iota_q:G_q\to\U(6)$, defined by
\begin{equation}
 \iota_q([g_3,g_2,z])
       =\operatorname{diag}\bigl(\phi_q([g_3,g_2,z]),z^q\bigr).
\end{equation}
The extra scalar entry $z^q$ is unchanged by multiplication by
$c_6^{6/q}$ and is therefore well-defined on the quotient. Among the
six powers of $c_6$, exactly the subgroup $\mathbb Z_q$ also has
$z^q=1$. Thus the two blocks together have trivial kernel on $G_q$,
and $\iota_q$ is faithful. The extra scalar entry lets us test
admissibility in the full group,
but it does not add a fermion: the overlap fermions still couple through
$R\circ\phi_q$, with $R=\mathbf{10}\oplus\bar{\mathbf5}$.

Because the first block of $\iota_q$ is $\phi_q$, the operator norms obey
\begin{equation}
 \|1-\phi_q(g)\|\le\|1-\iota_q(g)\|,\qquad g\in G_q .
\end{equation}
Thus imposing $\|1-\iota_q(U(p))\|<\eps$ on every plaquette,
with the same $\eps$ as in theorem~\ref{thm:measure}, makes the image
field admissible in $\SU(5)$. The threshold remains independent of
lattice volume. We could equally use another fixed faithful unitary
representation, with a fixed reduction of its threshold. Near the
identity, its injective Lie-algebra map controls distances; compactness
ensures that elements far from the identity cannot have representative
matrices approaching it.

\subsection{Restricting the current and measure}

Let $U$ now denote a $G_q$ link field satisfying this plaquette bound,
and write $(\phi_qU)(x,\mu)=\phi_q(U(x,\mu))$. Under the fixed
fermion identification above, define
\begin{equation}
 s_q(U)=s(\phi_qU).
\end{equation}
The phase $e^{\ii\varphi_v}$ of the $\SU(5)$ measure now multiplies
the Grassmann integration differentials of the five Standard Model
multiplets in the corresponding reference basis, as in
\eqref{eq:reconstructed-measure}. A $G_q$ gauge transformation of $U$
induces an $\SU(5)$ gauge transformation of $\phi_qU$, so gauge
invariance of $s$ ensures gauge invariance of $s_q$.

To obtain its measure term, vary an original $G_q$ link by
$\delta_\eta U=\eta U$. The Lie algebra is the same for all four
groups. Write one link generator as
$\eta=(\eta_3,\eta_2,\ii\theta)$, where
$\eta_3\in\mathfrak{su}(3)$, $\eta_2\in\mathfrak{su}(2)$, and
$\theta$ is real. Differentiating \eqref{eq:sm-map} gives the image
generator
\begin{equation}
 \phi_{q,*}\eta
   =\operatorname{diag}(\eta_3-2\ii\theta\Id_3,
                        \eta_2+3\ii\theta\Id_2).
\end{equation}
Here $\phi_{q,*}$ denotes the induced Lie-algebra map. Applying it
at every link and using the $\SU(5)$ current gives
\begin{equation}
 \mathfrak L^{(q)}_\eta[U]
 =\sum_x\sum_{\mu=1}^4
       \left\langle j_\mu(x;\phi_qU),
                    \phi_{q,*}\eta_\mu(x)\right\rangle
 =\mathfrak L_{\phi_{q,*}\eta}[\phi_qU].
\end{equation}
The pairing is the $\mathfrak{su}(5)$ inner product defined in
section~\ref{sec:measure}. As before, the coefficients of the independent
Standard Model link variations define its measure current.

The map $\phi_q$ is smooth and acts on each link separately.
Varying a Standard Model link therefore changes only the corresponding
image link, and differentiating the current introduces only fixed factors
in the bounds of section~\ref{sec:locality}. Its smoothness and exponential
locality follow without enlarging the support of any contribution.

Gauge transformations and lattice transformations of the original links
induce the corresponding transformations of their images. Applying the
$\SU(5)$ current identities to these images gives gauge covariance and
the lattice transformation laws for the restricted current. Likewise,
the $\SU(5)$ divergence equation applies to every image of a $G_q$
site gauge parameter. Choosing the integer-hypercharge direction, for
example, gives the hypercharge divergence equation from \eqref{eq:ward}.

To check full integrability, consider a closed admissible $G_q$ loop
$\gamma$. Its image $\phi_q\gamma$ is a closed admissible $\SU(5)$
loop, and exponentiating $\ii\int_\gamma\mathfrak L^{(q)}$ gives
$W[\phi_q\gamma]$. Since we have identified the fermion spaces by a
fixed change of basis, the Standard Model projector determinant phase
around $\gamma$ is $\mathcal T[\phi_q\gamma]$. The equality of these
two phases, proved in section~\ref{sec:global}, therefore gives the
required integrability, including for noncontractible loops. The auxiliary
connections used in that proof may leave $\phi_q(G_q)$, since their
purpose is to establish the phase identity for the image loop. They do
not enter the measure as additional Standard Model gauge configurations.
A gauge transformation in $\ker\phi_q$ acts trivially on both the
fermions and $s_q$.

This construction also applies when the gauge bundle cannot be lifted
to $\SU(3)\times\SU(2)\times\U(1)$. For the nontrivial quotients,
it may be impossible to choose product-group transition functions
consistently on all overlapping spacetime patches. We can still apply
$\phi_q$ directly to the $G_q$ transition functions, obtaining an
$\SU(5)$ bundle without making such a choice. Because the fermion
representations agree, its associated fermion bundle in
$\mathbf{10}\oplus\bar{\mathbf5}$ is the original Standard Model
fermion bundle.

To express a connection on such a bundle as lattice links, choose a gauge
frame for the $G_q$ bundle independently at each lattice vertex. Its edge
transports are then $G_q$ elements in these frames. Applying $\phi_q$
to these link elements gives the edge transports of the associated
$\SU(5)$ connection in the induced frames.
Choosing representatives for individual links does not require a globally
consistent lift of the bundle. Thus nonliftable sectors are included
whenever their admissible lattice components are nonempty.

We retain Wilson mass one and the common real periodic or antiperiodic
fermion boundary signs chosen at the outset. Independent generations
have a product measure, so three copies give the usual Standard Model
matter content. Their determinant lines form a tensor product, and their
measure terms and currents add. As before, each connected admissible
component permits an overall constant phase. Gauge singlets contribute
no gauge-field-dependent measure term.

\section{Other gauge groups and fermion representations}
\label{sec:other-representations}

The same measure construction applies to other $\SU(N_{\rm c})$ theories
when their fermion representations have two properties: cancellation of
the cubic gauge anomaly gives the color-trace cancellations needed for
convergence and locality, while a further condition on the fermions under an
$\SU(2)$ subgroup lets us choose an auxiliary connection whose return
phase is fixed by color conjugation. We can then follow the global
argument of section~\ref{sec:global}. The comparison current is still
given by \eqref{eq:scale-current}; the subgroup enters only in the
proof of integrability, and does not restrict the physical gauge fields.

\subsection{Conditions on the fermion representation}

Let $N_{\rm c}\geq2$ be the dimension of the fundamental representation
of $\SU(N_{\rm c})$, and let $R$ be a fixed finite-dimensional unitary
representation containing the complete left-handed fermion multiplet.
We use $R(g)$ for the action of a group element and $R(X)$ for that of
an anti-Hermitian Lie-algebra generator $X$. The first condition is
cancellation of the total cubic gauge anomaly:
\begin{equation}
 \tr_R R(X)\{R(Y),R(Z)\}=0,
 \qquad X,Y,Z\in\mathfrak{su}(N_{\rm c}).
 \label{gen:anomaly}
\end{equation}
The trace includes all fermion species in $R$. For $N_{\rm c}=2$,
this cubic anomaly vanishes identically.

For the second condition, let $\SU(2)$ act on two fundamental colors
and leave the other $N_{\rm c}-2$ fixed. Its embedding is
$g\mapsto\operatorname{diag}(g,\Id_{N_{\rm c}-2})$, with
$g\in\SU(2)$. At $N_{\rm c}=2$ this is the identity map. Decompose
the fermion representation under this subgroup:
\begin{equation}
 R|_{\SU(2)}=\bigoplus_{j=0,\frac12,1,\ldots}m_j V_j,
 \qquad m_j\ \hbox{even for }j=\tfrac12,\tfrac32,\ldots .
 \label{gen:restriction}
\end{equation}
Here $V_j$ is the $(2j+1)$-dimensional irreducible representation and
$m_j V_j$ means $m_j$ identical copies. Integer-isospin copies
need not occur in pairs. Half-integer-isospin copies must occur with even
multiplicity so that we can pair them to construct the color conjugation
map $J_{\rm col}$ with $J_{\rm col}^2=1$ used in the global argument.

\begin{proposition}
\label{gen:measure}
For each fixed $N_{\rm c}$ and $R$ satisfying
\eqref{gen:anomaly} and \eqref{gen:restriction}, the conclusions of
theorem~\ref{thm:measure} hold with gauge group $\SU(N_{\rm c})$ and
fermion representation $R$. The plaquette norm and generator normalization
use the fundamental representation. The admissibility threshold and
positive locality rate may depend on $N_{\rm c}$ and $R$, but not on
volume, component, or derivative order. Choose one periodic or antiperiodic
fermion boundary sign in each coordinate direction, common to the full
multiplet.
\end{proposition}

Our global argument uses the even-multiplicity requirement. We do not
claim that it is necessary for a chiral gauge theory to admit a measure.
The local estimates require only cancellation
of the cubic anomaly. As in theorem~\ref{thm:measure}, lattice symmetries
may relate different tori, and one constant measure phase remains free
on each disconnected admissible component.

\subsection{Applying the preceding construction}

Anomaly cancellation lets us reuse the local estimates of
section~\ref{sec:locality}. By \eqref{gen:ordered-traces}, the ordered
color traces in $R$ and $R^*$ agree through degree four whenever
\eqref{gen:anomaly} holds. Together with the opposite sign under spin
conjugation, this cancels the same Taylor coefficients as before and
leaves the same powers of the auxiliary spacing.
The remaining check is that the reconstruction and operator bounds
hold on one admissible domain for the chosen group and representation.

We reconstruct the connection in the fundamental representation using
appendix~\ref{app:g:geometry}, with $1/N_{\rm c}$ replacing $1/5$ in
the determinant normalizations and trace subtraction. The paths and
neighborhoods of original links remain the same. The group and
representation dependence is controlled by the bounds
\begin{equation}
 \begin{aligned}
 |\det V-\det W|&\le N_{\rm c}\|V-W\|,
       &&V,W\in\mathrm{U}(N_{\rm c}),\\
 \|R(V)-R(W)\|&\le C_R\|V-W\|,
       &&V,W\in\SU(N_{\rm c}).
 \end{aligned}
\end{equation}
where $C_R$ is finite because $R$ is smooth and the group is compact.
These bounds let us choose an outer plaquette threshold $\eps_{\rm out}$
small enough to keep the logarithms on their chosen branches, the
required matrices invertible, and the interacting comparison operators
gapped. Applying $R$ to the fundamental transport matrices in
determinant-one bases gives the fermion transports. At the free field,
this changes only the number of color copies of the operator, so its
comparison gap is unchanged.

We also need the coefficient extensions of appendix~\ref{app:a:locality}
to agree with the original coefficients near the admissible fields.
Use the same cutoff as before, now applied to
$z_p=\tr_{\mathbb C^{N_{\rm c}}}[(1-U(p))^\dagger(1-U(p))]$,
and choose $\eps=\eps_{\rm out}/(4N_{\rm c})$ for the admissibility
threshold. On this domain,
$z_p\le N_{\rm c}\eps^2<\eps_{\rm out}^2/16$, so the cutoff leaves
the coefficients unchanged. Wherever the cutoff is nonzero,
$\|1-U(p)\|^2\le z_p\le\eps_{\rm out}^2/8$, which keeps its support
inside the outer domain where the construction is smooth.

We make these choices once, before taking any derivatives. Further
derivatives change the constants and introduce polynomial factors in
the expansion order, but the exponential bounds still hold at one
fixed positive locality rate. The arguments of
appendices~\ref{app:s:spectral} and~\ref{app:a:locality} therefore give
the current and curvature bounds, including the separate finite bound
for $B_0$. Summing the finite comparison identities then gives
differential integrability and the required divergence equation.

The lattice transformation laws follow as before. Charge conjugation
requires comparing $R(\bar g)$ with $\overline{R(g)}$. These are related
by one constant unitary change of basis: their characters agree because
$\bar g$ and $g^{-1}$ are conjugate in $\SU(N_{\rm c})$. The traces
in the current are unchanged by this basis transformation, so the proof
does not require $R$ to be real.

For the global argument, we choose an auxiliary connection that acts
only on two color directions. Along a fixed loop of gauge fields, the
fundamental color spaces form a rank-$N_{\rm c}$ bundle over
$T^4\times S^1$. When $N_{\rm c}>2$, the argument of
section~\ref{sec:global} lets us choose $N_{\rm c}-2$ orthonormal
sections of this bundle. At each step the remaining complex rank is at
least three, so the zeros of a generic section would have real
codimension greater than the five-dimensional base and are absent.
The two-dimensional orthogonal complement inherits a normalized
determinant from the full bundle. We choose an $\SU(2)$ connection on
this complement and make the selected sections parallel, both in
spacetime and around the loop. For $N_{\rm c}=2$, we simply choose the
connection on the full rank-two bundle.

The condition on subgroup isospins now gives the color conjugation
symmetry. An integer-isospin representation admits an invariant
antiunitary with square $+1$. For a half-integer isospin, the invariant
antiunitary $C_j$ instead has $C_j^2=-1$. We obtain square $+1$ by
acting on two identical copies of $V_j$ together:
\begin{equation}
 J_{\rm col}(v,w)=(C_jw,-C_jv),\qquad v,w\in V_j.
\end{equation}
The exchange of copies and the minus sign give $J_{\rm col}^2=+1$.
Because each half-integer isospin has even multiplicity, every copy
can be paired in this way. Combining these pairings with the
integer-isospin conjugation maps gives a color antiunitary preserved
by the auxiliary connection and its transition functions.
Figure~\ref{fig:pseudoreal-pairs} illustrates the pairing.

\begin{figure}[htbp]
  \centering
  \includegraphics[width=\linewidth]{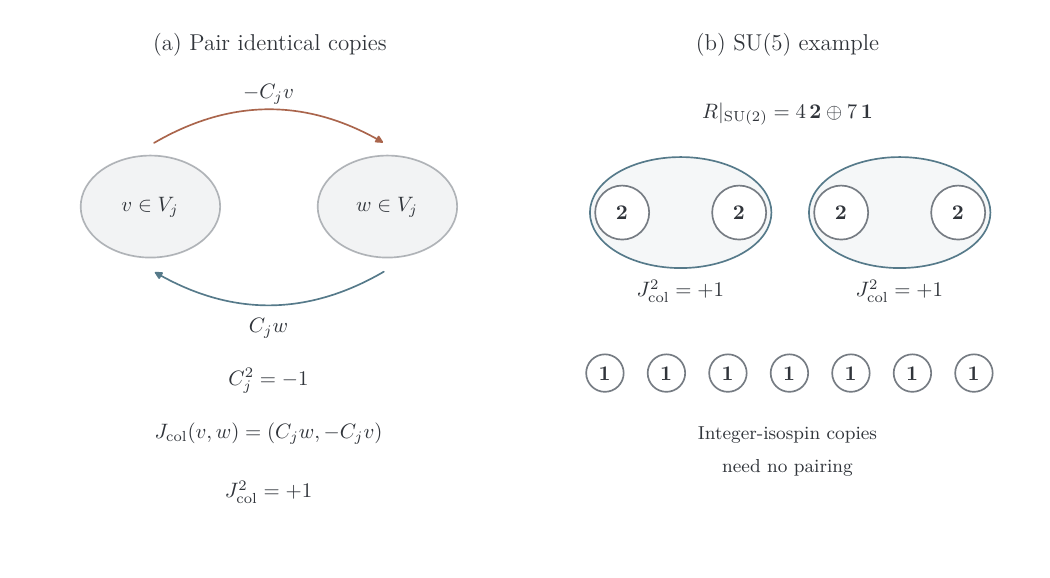}
  \caption{Color conjugation for the auxiliary $\mathrm{SU}(2)$ connection. (a) The vectors $v,w$ belong to two copies of the isospin-$j$ representation $V_j$, with $j$ half-integer. Its invariant antiunitary $C_j$ has square $-1$. Exchanging the copies with the indicated sign gives $J_{\rm col}(v,w)=(C_jw,-C_jv)$ with square $+1$. (b) Restricting $R=\mathbf{10}\oplus\bar{\mathbf5}$ to this subgroup gives four doublets and seven singlets. The blue enclosures pair the doublets; the singlets need no pairing.}
  \label{fig:pseudoreal-pairs}
\end{figure}

Combining $J_{\rm col}$ with spin conjugation fixes the auxiliary return
phase to $+1$ when color transport is included, as in
section~\ref{sec:global}. We join this auxiliary connection to the
reconstructed one by a fixed cylinder. Sufficiently
fine samples have the Wilson gap throughout the cylinder, and the
prescribed link and site generators have bounded maximum norms.
The ordered-trace identity makes both phase corrections in
\eqref{eq:phase-comparison-ordinary-cylinder} vanish under refinement:
the correction from changing the connection, and that from removing
the added color transport. It follows that
$\mathcal T[F_h\gamma]\to1$, so \eqref{eq:telescope} gives full
integrability. The cylinder is chosen for each fixed loop; no bound
uniform over all loops is required. The measure is then reconstructed
as in section~\ref{sec:measure-reconstruction}, including in components
with nonzero overlap index.

\subsection{Covered theories and restriction to other groups}

To apply the proposition to a theory, we test its complete left-handed
fermion content: the separate complex irreducible representations need
not satisfy either condition on their own. Even for $\SU(3)$,
$\SU(4)$, and $\SU(5)$, infinitely many choices pass these tests.
We begin with some useful families, then give a finite test that applies
to any fermion multiplet, including those with higher representations.

For $N_{\rm c}\ge3$, consider fermions in the fundamental representation
$F=\mathbf{N_{\rm c}}$, its antisymmetric and symmetric two-index
tensors $A_2$ and $S_2$, and their conjugates. Let $f,a,s$ count copies
of $F,A_2,S_2$, respectively, minus copies of their conjugates. Anomaly
cancellation requires
\begin{equation}
 f+(N_{\rm c}-4)a+(N_{\rm c}+4)s=0.
 \label{gen:two-index-condition}
\end{equation}
For these representations, anomaly cancellation also guarantees the
required pairing. Under our $\SU(2)$ subgroup, $F,A_2,S_2$ contain
$1,N_{\rm c}-2,N_{\rm c}-2$ doublets, respectively, and only integer
isospins otherwise. These doublet counts have the same parity as the
coefficients in \eqref{gen:two-index-condition}. A vanishing anomaly
therefore forces an even total number of doublets. Adding real
multiplets or vectorlike pairs preserves both conditions.

Table~\ref{tab:small-rank-tensors} gives all solutions built from these
representations for $\SU(3)$, $\SU(4)$, and $\SU(5)$: take finite sums
of the listed multiplets and their conjugates, and add any real
multiplets or vectorlike pairs within this class. For $\SU(3)$,
$A_2=\bar{\mathbf3}$ is already an antifundamental, so it is included
in $f$. For $\SU(4)$, $A_2=\mathbf6$ is real and can occur any number
of times.

\begin{table}[htbp]
\centering
\renewcommand{\arraystretch}{1.2}
\begin{tabular}{p{0.13\linewidth}p{0.25\linewidth}p{0.53\linewidth}}
\hline
Group & Cubic condition & Chiral generating multiplets \\
\hline
$\SU(3)$ & $f+7s=0$ & $\mathbf6\oplus7\bar{\mathbf3}$ \\
$\SU(4)$ & $f+8s=0$ & $\mathbf{10}\oplus8\bar{\mathbf4}$ \\
$\SU(5)$ & $f+a+9s=0$ & $\mathbf{10}\oplus\bar{\mathbf5}$ and
$\mathbf{15}\oplus k\bar{\mathbf{10}}\oplus(9-k)\bar{\mathbf5}$,
$k=0,\ldots,9$ \\
\hline
\end{tabular}
\caption{Complete classification within the fundamental and two-index
class for $\SU(3)$, $\SU(4)$, and $\SU(5)$. The
$\SU(3)$ $\mathbf6$, $\SU(4)$ $\mathbf{10}$, and $\SU(5)$
$\mathbf{15}$ are symmetric tensors. The
$\SU(5)$ $\mathbf{10}$ is antisymmetric. Higher representations are
not excluded from the theorem; the test below determines whether their
complete multiplet satisfies its conditions.}
\label{tab:small-rank-tensors}
\end{table}

The same approach gives families for arbitrary $N_{\rm c}\ge3$,
including higher tensors. Let $A_k$ and $S_k$ be the antisymmetric and
symmetric rank-$k$ tensors of $F$. Their cubic anomaly coefficients,
normalized to one for $F$, are~\cite[appendix~B]{ShiShrock2015}
\begin{equation}
 a_k=\frac{(N_{\rm c}-3)!(N_{\rm c}-2k)}
 {(N_{\rm c}-k-1)!(k-1)!},\qquad
 b_k=\frac{(N_{\rm c}+k)!(N_{\rm c}+2k)}
 {(N_{\rm c}+2)!(k-1)!}.
\end{equation}
The first formula holds for $1\le k\le N_{\rm c}-1$ and the second
for $k\ge1$. They tell us how many antifundamentals cancel the anomaly
of each tensor when its coefficient is nonnegative. For
$A_k\oplus a_k\bar F$ with $k\ge2$ and $N_{\rm c}\ge2k$,
the $\SU(2)$ decomposition contains
$2\binom{N_{\rm c}-3}{k-1}$ doublets and otherwise only singlets.
The pairing condition is therefore satisfied. At $N_{\rm c}=2k$,
$a_k=0$, so no antifundamentals are needed and the tensor is
self-conjugate.

Table~\ref{tab:general-families} lists these and other families.
Here $Q$ denotes a fixed finite-dimensional unitary multiplet and $Q^*$
its conjugate. A real multiplet has even multiplicities of the
half-integer subgroup isospins, and so satisfies the pairing condition.
This includes every vectorlike pair $Q\oplus Q^*$. Alternatively,
taking $Q\oplus Q$ doubles all subgroup multiplicities and gives a
measure whenever the cubic anomaly of $Q$ vanishes. Applying this to
$Q=S_k\oplus b_k\bar F$ gives the doubled symmetric family in the
table. A single copy must be checked separately.

\begin{table}[htbp]
\centering
\renewcommand{\arraystretch}{1.2}
\begin{tabular}{p{0.53\linewidth}p{0.38\linewidth}}
\hline
Fermion multiplet & Group and conditions \\
\hline
Fundamentals and two-index tensors &
$N_{\rm c}\ge3$, equation~\eqref{gen:two-index-condition} \\
$A_k\oplus a_k\bar F$ & $k\ge2$, $N_{\rm c}\ge2k$ \\
$S_2\oplus(N_{\rm c}+4)\bar F$ & $N_{\rm c}\ge3$ \\
$2S_k\oplus2b_k\bar F$ & $N_{\rm c}\ge3$, $k\ge2$ \\
Any real multiplet; $Q\oplus Q^*$ & $N_{\rm c}\ge3$, arbitrary $Q$ \\
$Q\oplus Q$ & {\raggedright $N_{\rm c}\ge3$, cubic anomaly of $Q$ vanishes\par} \\
{\raggedright Integer isospins arbitrary; each half-integer isospin with even multiplicity\par}
 & $\SU(2)$ \\
Any number of Standard Model generations &
$[\SU(3)\times\SU(2)\times\mathrm{U}(1)]/\Gamma_{\rm SM}$ \\
\hline
\end{tabular}
\caption{Covered families. The gauge group is $\SU(N_{\rm c})$
unless otherwise specified, and
$\Gamma_{\rm SM}=1,\mathbb Z_2,\mathbb Z_3,\mathbb Z_6$.
The families can be combined, conjugated, and supplemented with gauge
singlets. The test below determines whether arbitrary higher
representations satisfy the proposition's conditions.}
\label{tab:general-families}
\end{table}

For $\SU(2)$, the proposition covers arbitrary integer-isospin
representations and even numbers of each half-integer isospin. These are exactly the real multiplets, for which a measure is
already known~\cite{SuzukiReal}. The Standard Model entry follows from
section~\ref{sec:standard-model} and includes all four global forms of
the gauge group.

There is also a simple way to test a multiplet that does not appear in
the tables. For $N_{\rm c}\ge3$, we can check both conditions by counting
its charges under two one-parameter subgroups, without finding its full
$\SU(2)$ decomposition. Use the group elements
\begin{equation}
 \begin{split}
 g_3(\theta)&=\operatorname{diag}
 (e^{i\theta},e^{i\theta},e^{-2i\theta},1,\ldots,1),\\
 g_2(\theta)&=\operatorname{diag}
 (e^{i\theta},e^{-i\theta},1,\ldots,1),
 \end{split}
\end{equation}
where $\theta$ is real. A state of charge $q$ acquires a phase
$e^{iq\theta}$ under the corresponding subgroup. We can therefore count
the charges by taking the character, or trace in the representation:
\begin{equation}
 \chi_R(g_r(\theta))=\tr R(g_r(\theta))
 =\sum_{q\in\mathbb Z}c_q^{(r)}e^{iq\theta},\qquad r=2,3.
\end{equation}
Here $c_q^{(r)}$ is the number of color states of charge $q$ under
$g_r$, counted across all fermion species, without spin or lattice-site
multiplicities. For a tensor, add the charges of its indices; an index
in the conjugate representation contributes the opposite charge.
The two conditions of the proposition are equivalent to
\begin{enumerate}
\item the sum of the cubes of the charges for $g_3$ is zero;
\item for $g_2$, every odd integer charge occurs with even multiplicity.
\end{enumerate}
For any chosen finite multiplet, these two checks decide whether the
proposition applies, including for all higher representations of
$\SU(3)$, $\SU(4)$, and $\SU(5)$. Failing the test means that the
present proposition does not apply; it does not rule out a measure
obtained by another argument.

The first test is sufficient because $\SU(N_{\rm c})$ has only one
cubic invariant for $N_{\rm c}\ge3$. The $g_3$ charges detect its
coefficient: in the fundamental representation their cubes sum to
$1+1-8=-6$, rather than zero. For the second test, write
$c_\ell=c_\ell^{(2)}$ for the number of states of nonnegative
$g_2$ charge $\ell$. An isospin-$j$ representation contains one state
of each charge $2j,2j-2,\ldots,-2j$, so subtracting consecutive counts
gives $m_{\ell/2}=c_\ell-c_{\ell+2}$. Thus every odd charge has even
multiplicity exactly when every half-integer isospin does.

For example, the $\SU(3)$ multiplet
$\mathbf{10}\oplus27\bar{\mathbf3}$ passes the cubic test but fails
the pairing test: under $g_2$, the symmetric rank-three tensor has a
single charge-$3$ state, and the antifundamentals have none. Taking two copies makes all subgroup multiplicities
even, so the doubled multiplet passes both tests.

Finally, we can obtain theories with other compact gauge groups in the
same way that we obtained the Standard Model. Given a fixed smooth group
homomorphism $\phi:H\to\SU(N_{\rm c})$, map the $H$ links into
$\SU(N_{\rm c})$ and evaluate the constructed measure on them. This
gives the measure for fermions in $R\circ\phi$, since their overlap
operator and projector are exactly those of the image field.

On fields whose images are admissible, gauge invariance and integrability
follow as in section~\ref{sec:standard-model}. The measure term is
obtained by applying the induced Lie-algebra map to each link variation.
Because the map acts on each link separately, it preserves the support
of every contribution and changes only the constants in the locality
bounds. It also acts directly on transition functions, so quotient
groups require no lift to a covering group. A prescribed charge
conjugation on $H$ is inherited provided its image agrees with complex
conjugation in $\SU(N_{\rm c})$, up to one fixed change of basis in
the fundamental representation.

\section{Verification}
\label{sec:verification}

In writing this paper, we have been inspired by two quotes from Carl Sagan. The first is his observation that ``\dots the human use of artificial intelligence seems to be one of the two practicable major advances in human intelligence available in the near future.''\footnote{In his view the second practicable major advance ``is enrichment of the preschool and school learning environments of children''\cite[chap.~8]{SaganDragonsOfEden}. Make of that what you will!} The second is the far more famous ``extraordinary claims require extraordinary evidence.'' To supplement human scrutiny of the results in this paper, we have pursued  three complementary avenues for verification.

The first is a Lean formalization of the mathematical statements in the paper. The accompanying formalized \emph{Mathematical companion} contains the mathematical statements associated with the $\SU(5)$ construction and its restriction to the four Standard Model gauge groups. All sixteen modules of the corresponding Lean~4 formalization compiled without errors
or compiler warnings. Separate checks of the axioms and a replay of the
proofs in Lean's kernel accepted the formal proofs, using only Lean's
standard foundational axioms and no unproved placeholders. These checks verify the formal statements under their recorded assumptions. Their correspondence with the mathematical companion, including those assumptions, was reviewed separately by AI agents. The formal-verification supplement \emph{Lean verification of the mathematical companion} provides the source, software versions,
verification records, and correspondence review. 

This is not to suggest, however, that Lean verification is the final word in physics, where there is a great deal more that can go wrong beyond the failure of a formally stated conclusion to follow from its specified definitions and assumptions. To build greater confidence in the physical results, in the supplement \emph{A two-dimensional chiral $\SU(5)$ measure} we develop and analyze an analogous construction in two dimensions. This gives a comparable measure construction on an explicit admissible domain and compares directly evaluated currents with gauge responses and full determinant phases in specified examples. 

Finally, perhaps the most convincing check of any result in theoretical physics is to obtain the same result by different means. In this paper, we first construct a local rule for how the measure’s phase changes when the gauge field changes, then show that this rule is globally consistent and use it to construct the measure. Alternately, one can build the fermion measure directly from a simple reference state, using auxiliary local choices and averaging their phase corrections; this effectively amounts to doing things in the opposite order. This second, independent approach yields a fermion measure that agrees with the one presented here up to exponentially local counterterms, which is precisely the level of agreement one should expect for equivalent measures.

\section{Conclusions}
\label{sec:conclusions}

The Standard Model is a triumph of human understanding, providing a wildly successful description of nature on scales stretching from quarks to the cosmos. Yet the elusiveness of a nonperturbative lattice definition makes one wonder whether we have overlooked something fundamental. In this paper, we have sought to close this gap for the Standard Model (and other chiral gauge theories) by constructing a gauge-invariant measure for overlap chiral fermions on admissible nonabelian gauge fields. Our approach builds on L\"uscher's method and delivers his reconstruction criteria for a smooth, gauge-invariant measure in a theory with lattice chiral fermions.

There are numerous directions for further study. More elegant solutions to the same problem undoubtedly exist, and comparison to other existing proposals is bound to be fruitful. It goes without saying that our formulation of lattice chiral fermions is certainly one of principle, rather than practice. It would be useful to establish optimal numerical admissibility bounds, distill lessons for realistic lattice simulations, and work towards a practical implementation of the ideas presented here.  It would also be worthwhile to explore the extension to additional chiral gauge theories; fix the remaining phase freedom of disconnected components via physical considerations; and establish reflection positivity compatible with admissibility~\cite{CreutzPositivity,BietenholzTopology}. And, of course, a definition of the interacting continuum theory remains open. But a candidate lattice theory is certainly a start.

\FloatBarrier
\acknowledgments
OpenAI's GPT-6 Astra model was essential to this result. I proposed a strategy to extend L\"{u}scher's abelian result using refinement methods that underwent several rounds of iteration with Astra and outside experts before arriving at the current Astra-led proof.  The manuscript combines author-written sections with extensive rewrites of the Astra-composed proof. My colleagues and I proposed tests of the method, including Lean verification, that were carried out by Astra; supplementary materials documenting some of these checks were written by Astra.

Human experts were also essential to this result. I am deeply indebted to Maarten Golterman and Yigal Shamir for generously contributing their time, expertise, and probing questions throughout the entire arc of this project. I would also like to thank Xiao-Liang Qi and Eva Silverstein for helpful conversations; Andreas Kronfeld for providing his expertise and comments on the manuscript; and Kevin Weil, Zak Weiler, and especially Dan Roberts for their encouragement.

\appendix

\section{A local connection from the lattice links}
\label{app:g:geometry}

Section~\ref{sec:construction} uses the original links $U$ to construct
one smooth connection $\nabla^0[U]$, which is then sampled at every
auxiliary spacing. We give that construction here and establish the
bounds needed for the comparison current. The links initially specify
transport only between original lattice vertices. We must extend this
to arbitrary points of the spacetime torus, without choosing a global
gauge or introducing dependence on distant links.

The construction has three steps. First, we combine short lattice
transports to define a five-dimensional color space $E_y(U)$ at each
continuous point $y$. The plaquette bound makes the different short
routes nearly consistent; a spectral projection turns this approximate
consistency into an exactly defined color space. Second, we obtain a
connection by projecting derivatives into these spaces and removing its
trace part. Finally, we transport color bases from each original vertex
to its descendants on the finer lattices. Sampling the connection in
these bases gives $F_hU$ and the constant copying matrix used in
section~\ref{sec:construction}.

We will show that the spatial curvature is $O(\eps)$ in original lattice units,
and that the sampled plaquette deviations are $O(\eps h^2)$. Derivatives
with respect to the original links need not retain the factor $\eps$,
but they remain bounded and preserve the plaquette's $h^2$ area factor.
All these bounds are uniform in the refinement level. Earlier
constructions of bundles and gauge-covariant continuum connections from
lattice gauge fields appear in
refs.~\cite{PhillipsStone1986,GockelerKronfeldSchierholzWiese}.

We work first on $\mathbb R^4$, before identifying points separated by
a torus period. The point $y\in\mathbb R^4$ and vertices
$i,j,k\in\mathbb Z^4$ are measured in original lattice units, as in
section~\ref{sec:lattice}. Throughout this appendix, a spatial derivative
means a derivative in any of these four Euclidean spacetime directions;
it is distinct from a derivative that changes $U$. We impose the torus
identifications only after constructing the color spaces and connection.
For an edge $e=(i,i+\hat\mu)$, $U_e$ transports from its terminal vertex
to its initial vertex. Write
$|v|_\infty=\max_\mu|v_\mu|$ and $|v|_1=\sum_\mu|v_\mu|$;
unmarked lengths are Euclidean, and $\mu,\nu=1,\ldots,4$.

The geometric construction requires
$\|1-U(p)\|<\eps_{\rm g}$ on the original plaquettes used below.
The norm is the fundamental-representation operator norm of
section~\ref{sec:lattice}. We choose one positive threshold $\eps_{\rm g}$
that keeps all inverse matrices and phase choices well defined. In the
estimates, $\eps\le\eps_{\rm g}$ bounds the required plaquette deviations.
Neither individual links nor transports around the periodic directions
need be close to the identity.

We use the original-link derivatives of \eqref{eq:link-derivative}.
The label $I=(x,\mu,a)$ specifies the link $(x,\mu)$ and generator
$T^a$, $a=1,\ldots,24$; $D_I$ differentiates
$U(x,\mu)\mapsto e^{tT^a}U(x,\mu)$ at $t=0$, holding the other
links fixed. For an ordered list $\boldsymbol I=(I_1,\ldots,I_\ell)$,
write $D^{\boldsymbol I}=D_{I_1}\cdots D_{I_\ell}$. The empty list
means no differentiation. The order matters because derivatives on the
same link need not commute. Constants denoted by $C$ or $C_\ell$ may
change between estimates. They may depend on the derivative order and the fixed construction,
but are uniform in the derivative labels, lattice volume, refinement
level, and admissible component.

\subsection{Local weights and transporters}
\label{app:g:weights}

To define a color space at $y$, we will compare vectors at nearby
vertices. Two ingredients are needed: smooth weights that determine how
much each vertex contributes, and short transport matrices that compare
their color vectors. Both must respect translations, coordinate
permutations, and reflections.

\paragraph{Smooth weights.}
The weights must vanish outside a fixed neighborhood and sum to one.
They must also change smoothly when $y$ crosses a cell boundary. We will
copy a vector into several vertex spaces, using the square roots of
these weights. The normalization then preserves the vector's norm.

Let $b$ be the indicator function of $[-1/2,1/2]$, and define convolution
by $(f*g)(t)=\int_{\mathbb R}f(t-u)g(u)\,\dd u$. Fix
$\delta_b=1/8$ and the smooth even function
\begin{equation}
 \beta(t)=
 \begin{cases}
 C_\beta\exp[-1/(\delta_b^2-t^2)],&|t|<\delta_b,\\
 0,&|t|\ge\delta_b,
 \end{cases}
 \qquad \int_{\mathbb R}\beta(t)\,\dd t=1 .
\end{equation}
The integral fixes $C_\beta>0$. Define
\begin{equation}
 \kappa=b*b*b*\beta,\qquad
 \phi(t)=\sqrt{\kappa(t)},\qquad
 \phi_i(y)=\prod_{\mu=1}^4\phi(y_\mu-i_\mu),\qquad
 \rho_i(y)=\phi_i(y)^2 .
 \label{app:g:partition}
\end{equation}
Here $\rho_i(y)$ is the weight assigned to vertex $i$, and $\phi_i(y)$
is the factor used to copy a vector into its color space. 

The translates of $b$ tile the line except at endpoints, which do not
affect the convolution integral. Since $b*b*\beta$ has integral one,
\begin{equation}
 \sum_{j\in\mathbb Z}\kappa(t-j)
 =\int_{\mathbb R}(b*b*\beta)(u)
       \sum_{j\in\mathbb Z}b(t-u-j)\,\dd u=1 .
\end{equation}
Taking the product over the four coordinates therefore gives
\begin{equation}
 \sum_{i\in\mathbb Z^4}\rho_i(y)
 =\prod_{\mu=1}^4\sum_{j\in\mathbb Z}\kappa(y_\mu-j)=1 .
\end{equation}

We also need the copying factors $\phi_i$ to be smooth where they vanish. The support of $\kappa$ is
$[-13/8,13/8]$. Near its right endpoint, put $t=13/8-r$ with $r>0$.
Using $(b*b*b)(t)=\tfrac12(3/2-t)^2$ near $t=3/2$, we obtain
\begin{equation}
 \kappa(13/8-r)=\frac{C_\beta}{2}
       \int_0^r(r-v)^2e^{-1/[v(2\delta_b-v)]}\,\dd v.
\end{equation}
This formula controls the square root as well as $\kappa$. For small
$r$, the negative derivative of $1/[v(2\delta_b-v)]$ is bounded above
and below by positive constants times $v^{-2}$. Integrating from $v$ to
$r$ gives
$1/[v(2\delta_b-v)]-1/[r(2\delta_b-r)]\ge c(r-v)/r^2$.
The appreciable contribution thus comes from an interval of width
$O(r^2)$ next to $v=r$. Integrating this exponential bound gives an
upper bound $Cr^6e^{-1/[r(2\delta_b-r)]}$. Restricting instead to
$r-r^2\le v\le r-r^2/2$ gives a lower bound of the same order:
the quadratic factor is bounded above and below by positive constants
times $r^4$, the interval has width $r^2/2$,
and the exponential differs from its value at $r$ by bounded factors.

Differentiating the integral and then its square root gives, for each
fixed $m\ge0$,
\begin{equation}
 \begin{split}
 |\kappa^{(m)}(13/8-r)|&\le C_mr^{6-2m}e^{-1/[r(2\delta_b-r)]},\\
 |\phi^{(m)}(13/8-r)|&\le C_mr^{3-2m}e^{-1/[2r(2\delta_b-r)]}.
 \end{split}
\end{equation}
The exponential dominates every inverse power of $r$, so all derivatives
of $\phi$ tend to zero at the endpoint. Evenness gives the same result
at the other endpoint. The extension by zero is therefore smooth.

Only vertices with $\phi_i(y)\ne0$ contribute at $y$; we call these
vertices active. There are at most four in each coordinate and hence
at most $4^4$ in total. Around each vertex use the open neighborhood
\begin{equation}
 V_i=\{y:|y-i|_\infty<7/4\}.
\end{equation}
We call $V_i$ a chart because it will carry a local color basis.
The support of $\phi_i$ lies strictly inside it. For $y\in V_k$,
the vertex $k$ and all active vertices lie in one integer-coordinate
box of side at most three. On an overlap $V_k\cap V_j$, this box can
also include both $k$ and $j$: their coordinate separations are integers
strictly smaller than $7/2$. Thus all required comparisons take place
in boxes of fixed size. A vertex used to label a chart will be called
its anchor; it need not itself have nonzero weight at every point of
the chart.

\paragraph{Short transport matrices.}
For two vertices with $|i-j|_\infty\le3$, take lattice paths from
$i$ to $j$ that complete all steps in one coordinate before moving in
the next. Form each ordered link product with the convention of
section~\ref{sec:lattice}: a step $x\to x+\hat\mu$ contributes
$U(x,\mu)$, while a step $x\to x-\hat\mu$ contributes
$U(x-\hat\mu,\mu)^{-1}$. The product transports color vectors
back from $j$ to $i$. Average these products over all $4!=24$ coordinate
orders and call the result $M_{ij}$. A zero displacement uses no step.
If several orders give the same path, each order still contributes to the
average. Every path contains at most twelve links.

This average treats all coordinate directions equally, but it need not
be unitary. First restore unitarity by the polar decomposition, then
remove the determinant phase:
\begin{equation}
 \begin{aligned}
 V_{ij}&=M_{ij}(M_{ij}^{\dagger}M_{ij})^{-1/2},\\
 u_{ij}&=V_{ij}\exp\{-\tfrac15\operatorname{Log}\det V_{ij}\},
 \qquad u_{ii}=\Id .
 \end{aligned}
 \label{app:g:polar}
\end{equation}
The inverse square root is positive. The logarithm uses the branch near
one, with $\operatorname{Log}1=0$. The first line makes $V_{ij}$
unitary, and the second makes $u_{ij}$ special unitary. To justify
these operations, we now bound the differences between the paths.

Exchanging two neighboring coordinate runs sweeps a rectangle of at
most $3\times3$ plaquettes. Any two orders are related by at most six
such exchanges, so their link products differ by at most $54\eps$.
The other transport factors are unitary and do not enlarge this bound.
Thus $M_{ij}$ is within $54\eps$ of each path product. For sufficiently
small $\eps$, $M_{ij}^\dagger M_{ij}$ stays uniformly positive and
$\det V_{ij}$ stays in a fixed disk about one. Both normalizations are
then well defined and smooth. Reversing all paths gives
$M_{ji}=M_{ij}^\dagger$, and the chosen matrix functions give
$u_{ji}=u_{ij}^\dagger$.

These normalized matrices compare color vectors at different vertices,
but they need not compose exactly. For three vertices in the same
local box, their failure to compose obeys
\begin{equation}
 \|u_{ij}-u_{ik}u_{kj}\|\le C\eps .
 \label{app:g:triangle}
\end{equation}
To see this, compare direct transport from $j$ to $i$ with transport through
$k$. Together they make a closed path of at most 36 links. Sort its
steps by coordinate while retaining their order within each coordinate.
At most $\binom{36}{2}=630$ exchanges are needed; each inserts one
oriented plaquette. Every intermediate coordinate is the endpoint of
a prefix of the original walk in that coordinate, so the exchanged
paths and inserted plaquettes remain in the same box. After sorting,
each coordinate run is a closed one-dimensional walk, whose forward
and reverse link factors cancel. The original loop differs from the
identity by at most $630\eps$. The three polar and determinant
corrections add only $O(\eps)$, proving
\eqref{app:g:triangle}. This is the local consistency estimate used
in the next step.

\subsection{The color spaces, their bases, and the connection}
\label{app:g:projection}

We now use the weights and short transports to define one color space at
$y$, independently of which vertex is used to describe it. To compare
all the candidate vectors in a fixed space, place one copy of
$\mathbb C^5$ at each original vertex:
$\mathcal H_{\rm g}=\ell^2(\mathbb Z^4)\otimes\mathbb C^5$.
A vector in this space is a list of five-component vectors, one per
vertex, with the sum inner product. This space is only a device for
constructing the color fibers; it is not the fermion Hilbert space.
Only finitely many of its vertex components occur near any given $y$.

\paragraph{Selecting five color directions.} The weighted projection construction for approximate cocycles and the recovery of exact frames used here have precedents in refs.~\cite{ManuilovYou,ScoccolaPerea}.
Choose $y\in V_k$. A vector at the anchor $k$ can be transported to
each active vertex $i$ and multiplied by $\phi_i(y)$. This defines a
map $W_k:\mathbb C^5\to\mathcal H_{\rm g}$. We also form a matrix
$K$ that compares all active vertices without selecting an anchor:
\begin{equation}
 K_{ij}=\phi_i u_{ij}\phi_j,\qquad
 (W_k)_i=\phi_i u_{ik},\qquad \Pi_k=W_kW_k^{\dagger} .
 \label{app:g:gram}
\end{equation}
The indices $i,j$ label five-by-five blocks of $K$. A block with a
zero weight is zero, so only the short transports already constructed
are needed. The identity $u_{ji}=u_{ij}^\dagger$ makes $K$ Hermitian.
The weight normalization gives $W_k^\dagger W_k=\Id$: copying a
vector in this way preserves its norm. Consequently $\Pi_k$ is the
orthogonal projection onto the five-dimensional image of $W_k$.

If all short transports composed exactly, $K$ would equal $\Pi_k$.
It would retain precisely the weighted copies of a single anchor
vector. For an admissible field, the difference has blocks
$\phi_i\phi_j(u_{ij}-u_{ik}u_{kj})$. Applied to a vector
$v=(v_j)\in\mathcal H_{\rm g}$, \eqref{app:g:triangle} and
Cauchy--Schwarz therefore give
\begin{equation}
 \|(K-\Pi_k)v\|^2
 \le C^2\eps^2
       \left(\sum_j|\phi_j|\|v_j\|\right)^2
 \le C^2\eps^2\|v\|^2 .
\end{equation}
The sum over output components contributes $\sum_i\phi_i^2=1$;
the second inequality uses the same normalization for the input
components. The bound does not grow with the number of vertex spaces.
Choose $\eps_{\rm g}$ so that $\|K-\Pi_k\|<1/8$.

The image of $K$ itself is not suitable: a small perturbation can
produce extra nonzero eigenvalues near zero. Instead, keep only the
five eigenvalues near one. With the counterclockwise contour
$\mathcal C=\{z:|z-1|=1/2\}$, their spectral projection is
\begin{equation}
 \Pi=\frac1{2\pi\ii}\int_{\mathcal C}(z-K)^{-1}\,\dd z .
 \label{app:g:projection-formula}
\end{equation}
The contour remains separated from both parts of the spectrum. Eigenvalue
crossings within the selected group therefore do not spoil smoothness.
Since $K$ is independent of the anchor and acts on the same space in
every chart, so does $\Pi$. Define $E_y(U)=\im\Pi(y,U)$.
These smoothly varying five-dimensional spaces form the color bundle
$E(U)$. The projector $\Pi$ acts on color vectors in $\mathcal H_{\rm g}$;
it is distinct from the chiral projector $\widehat P_-$ on fermion fields.
The resolvent identity gives $\|\Pi-\Pi_k\|\le C\eps$.
Decrease $\eps_{\rm g}$ once more so that this bound is also below $1/8$.

\paragraph{Local orthonormal bases.}
To obtain a basis of $E_y(U)$ on $V_i$, project the five columns of
$W_i$ into that space and orthonormalize them:
\begin{equation}
 w_i=\Pi W_i(W_i^{\dagger}\Pi W_i)^{-1/2}.
 \label{app:g:safe-frame}
\end{equation}
The Gram matrix $W_i^\dagger\Pi W_i$, whose entries are their inner
products, has spectrum in $[7/8,9/8]$. Its positive inverse square
root is therefore smooth. The columns of $w_i$ obey
$w_i^\dagger w_i=\Id$ and $w_iw_i^\dagger=\Pi$. They define a local
frame throughout $V_i$, including points where $\phi_i$ itself vanishes.

On an overlap, the matrix $t_{ij}=w_i^\dagger w_j$ converts components
in frame $w_j$ to components in frame $w_i$. The two frames span the
same space, so $t_{ij}t_{jk}=t_{ik}$ exactly on triple overlaps.
The triangle estimate also gives $t_{ij}=u_{ij}+O(\eps)$.
Thus the $u_{ij}$, which only approximately compose, have supplied
color spaces whose actual changes of basis compose exactly.

\paragraph{Fixing the determinant phase.}
The frames $w_i$ are orthonormal, but their changes of basis lie in
$\U(5)$ rather than necessarily in $\SU(5)$. We must make their
ordered exterior products agree. Write $s_i=\bigwedge^5w_i$ for the
exterior product of the five columns of $w_i$. It is a unit vector
in the one-dimensional space $\det E_y=\bigwedge^5E_y$.
This determinant space is built from the fundamental color space;
it is not the fermion determinant line of section~\ref{sec:fermion-measure}.
On overlaps, $s_j=s_i\det t_{ij}$. Since
$\det t_{ij}=1+O(\eps)$, these unit vectors have nearly the same phase.

We can therefore average them without cancellation and normalize the
result to define a common volume element $\sigma$. Adjust the phase
of each frame to have that exterior product:
\begin{equation}
 \begin{gathered}
 S=\sum_i\rho_i s_i,\qquad \sigma=S/\|S\|,\\
 r_i=\sum_j\rho_j\det t_{ij},\qquad
 d_i=r_i/|r_i|,\qquad
 C_i=w_i\exp\{\tfrac15\operatorname{Log}d_i\}.
 \end{gathered}
 \label{app:g:determinant}
\end{equation}
The scalar $r_i$ expresses $S$ in frame $s_i$: $S=s_i r_i$.
Each weighted vector $\rho_j s_j$ extends smoothly by zero outside
its chart. The identities $\det u_{ij}=1$ and $\sum_j\rho_j=1$
give $r_i=1+O(\eps)$. Taking $\eps_{\rm g}$ small enough that
$|r_i-1|<1/2$ ensures $S\ne0$ and keeps the logarithm of $d_i$ on
the branch near one. Multiplying each of the five columns by the
same scalar multiplies their exterior product by $d_i$. Hence
$\bigwedge^5C_i=s_i d_i=\sigma$, with the positive sign in
\eqref{app:g:determinant}. All transitions $C_i^\dagger C_j$ now
have determinant one. This fixes a volume element, not a global color
basis; nontrivial special-unitary transitions are still allowed.

\paragraph{The reconstructed connection.}
We have defined color spaces at all $y$ and bases on overlapping
neighborhoods. A connection tells us how to differentiate a color-vector
field taking values in those spaces. First differentiate in the fixed
space $\mathcal H_{\rm g}$, then project back with $\Pi$. For a
smooth family of gauge fields, apply the same operation in both
spacetime and family directions. Write
$\dd_{\rm tot}=\dd_y+\dd_U$. Here $\dd_y$ changes the spacetime
point at fixed $U$, while $\dd_U$ changes the original links as the
family parameters vary.

The projected derivative preserves inner products but can change the
phase of the volume element $\sigma$. Its action on that phase is
measured by the scalar one-form $\theta$. Subtracting one fifth of
$\theta$ from each color direction removes the total phase change:
\begin{equation}
 \theta=(\sigma,\dd_{\rm tot}\sigma),\qquad
 \nabla^0=\Pi\dd_{\rm tot}-\tfrac15\theta\,\Id_E .
 \label{app:g:connection}
\end{equation}
Here $\nabla^0$ acts on fields taking values in $E$, and $\Id_E$
is the identity on each color space. The inner product on exterior
products is inherited from $\mathcal H_{\rm g}$ and is antilinear
in its first argument. Differentiating $(\sigma,\sigma)=1$ shows
that $\theta$ is imaginary-valued. Its subtraction preserves the
Hermitian norm and makes the induced derivative of $\sigma$ vanish.
In the corrected frame $C_i$, the connection potential is
\begin{equation}
 C_i^{\dagger}\dd_{\rm tot} C_i
       -\tfrac15\tr(C_i^{\dagger}\dd_{\rm tot} C_i)\Id .
\end{equation}
This is traceless and anti-Hermitian, so $\nabla^0$ is an $\SU(5)$
connection. Its spatial part supplies the sampled links; its family
components will also be used in section~\ref{sec:global}.

\paragraph{Bounds on the connection.}
Every fixed number of spatial and ordered original-link derivatives of
the construction is bounded. On one chart only a fixed finite set of
vertex components contributes, and all inverses, normalizations, and
logarithms stay away from their singularities. For spatial derivatives
alone we need more: the potential in frame $C_i$ must remain $O(\eps)$,
with the same small factor in its spatial derivatives.

To prove this, hold $U$ fixed. Spatial differentiation acts only on the
weights in $K-\Pi_i$, leaving its transporter defect $O(\eps)$.
The difference of spectral projections can be written as
\begin{equation}
 \Pi-\Pi_i=\frac1{2\pi\ii}\int_{\mathcal C}
 (z-K)^{-1}(K-\Pi_i)(z-\Pi_i)^{-1}\,\dd z .
 \label{app:g:projection-difference}
\end{equation}
Every spatial derivative still contains $K-\Pi_i$ or one of its
spatial derivatives; the other factors are bounded. The same
$O(\eps)$ estimate therefore holds for $\Pi-\Pi_i$ and, after
the smooth frame normalizations, for $C_i-W_i$, with every fixed
number of spatial derivatives.

The reference frame $W_i$ has zero spatial connection:
\begin{equation}
 W_i^{\dagger}\dd_yW_i
    =\sum_j\phi_j\,\dd_y\phi_j\,\Id=0.
\end{equation}
The last equality follows by differentiating $\sum_j\phi_j^2=1$.
Replacing $W_i$ by $C_i=W_i+O(\eps)$ gives an $O(\eps)$ spatial
potential and fixed-order spatial derivatives. Its curvature, formed
from a derivative of the potential and its commutator, has the same
bound, as do its fixed-order spatial derivatives.

An original-link derivative behaves differently: it can differentiate
the transporter defect itself and remove its small factor $\eps$.
These derivatives are bounded, but are not asserted to be $O(\eps)$.
We retain this distinction when bounding the sampled plaquettes.

\subsection{The original lattice and the periodic identifications}
\label{app:g:edges}

We next relate the reconstructed connection to the original links and
check its transformation laws. At an original vertex $i$, choose
$C_i(i,U)$ as the color basis. Let
$\tau(y_2\leftarrow y_1):E_{y_1}\to E_{y_2}$ denote straight-line
spatial transport by $\nabla^0$, with $U$ fixed. The reconstructed
sample on an original edge $e=(i,j)$, $j=i+\hat\mu$, is
\begin{equation}
 (F_1U)_e=C_i(i)^{\dagger}\tau(i\leftarrow j)C_j(j).
\end{equation}
The bases at both endpoints depend on $U$ and must be differentiated
along with $\tau$ whenever an original link varies.

The whole edge lies in $V_i$. In its frame $C_i$, the connection
potential is $O(\eps)$, so transport along that edge is
$\Id+O(\eps)$. At the terminal vertex we must also convert the
basis $C_j(j)$ to $C_i(j)$; this contributes
$C_i(j)^\dagger C_j(j)=u_{ij}+O(\eps)$. For this single-axis edge
all 24 paths coincide, giving $u_{ij}=U_e$ exactly. Thus
\begin{equation}
 \|(F_1U)_e-U_e\|\le C\eps .
 \label{app:g:root-consistency}
\end{equation}
The reconstruction approximates the original links; it does not replace
them as the first field in the comparison sequence. As in
\eqref{eq:coarse-fields}, that sequence starts with $U$ itself and
then $F_{1/3}U$. Every fixed ordered derivative of $F_1U-U$ is
bounded, but need not contain $\eps$.

For flat input, the short transports on the cover compose exactly.
Then $K=\Pi_i$, the spatial connection vanishes in the frames $W_i$, and
$F_1U=U$ exactly. This includes nontrivial flat transport around the
torus: that information remains in the identifications between local
color bases, rather than in their local connection potentials.

\paragraph{Gauge transformations.}
For a fixed original gauge transformation $g_i\in\SU(5)$, let
$\mathcal G=\bigoplus_i g_i$ act on the vertex components of
$\mathcal H_{\rm g}$. Every short transporter transforms as
$u_{ij}\mapsto g_i u_{ij}g_j^\dagger$. The spectral projection
and determinant normalization then give
\begin{equation}
 \Pi\mapsto\mathcal G\Pi\mathcal G^{\dagger},\qquad
 C_i\mapsto\mathcal G C_i g_i^{\dagger},\qquad
 \sigma\mapsto(\bigwedge^5\mathcal G)\sigma .
\end{equation}
These identities imply covariance of spatial transport.

If the gauge transformation also varies along a family of input fields,
its parameter derivative contributes to the connection in that direction.
Expressing the transformed connection on the original bundle gives
\begin{equation}
 \begin{split}
 \mathcal G^{-1}\nabla^0[U^g]\mathcal G
 ={}&\nabla^0[U]
   +\left.\Pi\mathcal G^{-1}\dd_U\mathcal G\Pi\right|_E\\
 &-\tfrac15\tr_E(\Pi\mathcal G^{-1}\dd_U\mathcal G\Pi)\Id_E .
 \end{split}
 \label{app:g:parameter-gauge}
\end{equation}
Here $\tr_E$ is the trace on the five-dimensional image of $\Pi$.
The additional term is the derivative of the change of vertex bases,
projected into the reconstructed color space; its trace is subtracted
to preserve the volume element. The matrices $g_i$ depend on the vertex
label and family parameters, but not on the continuous position $y$.
Thus $\dd_y\mathcal G=0$, and there is no such additional term in
the spatial connection. This is why parameter-dependent color transport
must be retained in section~\ref{sec:global}, even though sampled
spatial links obey the usual gauge transformation law.

\paragraph{Periodicity and lattice symmetries.}
For periodic input, we now identify points separated by a torus period.
Let $\Delta$ be a vector with $\Delta_\mu\in L_\mu\mathbb Z$, and
let $(\mathcal D_\Delta v)_i=v_{i-\Delta}$ shift the vertex components.
Translating the point and all vertex labels gives
\begin{equation}
 \Pi(y+\Delta)=\mathcal D_\Delta\Pi(y)\mathcal D_\Delta^{\dagger} .
\end{equation}
The volume element and connection transform in the same way. We can
therefore identify $(y,v)$, with $v\in E_y$, and
$(y+\Delta,\mathcal D_\Delta v)$. These identifications define the
color bundle and connection on $T^4$. If necessary, restrict the charts
to smaller neighborhoods whose projection to the torus is one-to-one.
For periods one or two, distinct lifted vertices may represent the same
torus vertex. We keep their components and paths separate until these
identifications are made. Fermion boundary signs are imposed separately
on the sampled fermion spaces.

Integer translations, signed coordinate permutations, and complex
conjugation preserve the weights, permute the paths, and commute with
the chosen matrix functions. The reconstructed bundle and connection
therefore have the corresponding symmetry laws. On the torus, a coordinate
permutation also permutes the periods and the prescribed fermion boundary
signs, as in section~\ref{sec:symmetries}. If the links preserve
a fixed $3+2$ color splitting, every path, projection, and frame preserves
it. The connection is then in $S(\U(3)\times\U(2))$, as required
in section~\ref{sec:standard-model}. Removing the total trace does
not set the traces of the two blocks to zero separately.

\subsection{Sampling and plaquette bounds}
\label{app:g:sampling}
\label{app:g:comparison}

We now choose the color bases at the refined sites. Transporting a
parent's basis to its children makes the coarse-to-fine inclusion a
constant copying matrix in component coordinates. The bases themselves
still depend on $U$. We must therefore control their original-link
derivatives, including the transport at every preceding refinement step.
The reason these derivatives remain bounded is that the successive
transport segments become shorter: their total length stays bounded as
the number of refinements increases.

First specify the parent of each refined site on the covering space $\mathbb R^4$. At
spacing $h=3^{-n}$, $n=0,1,\ldots$, let $\pi_h(y)$ be the
coordinatewise nearest point of $h\mathbb Z^4$ to a point
$y\in h3^{-m}\mathbb Z^4$, $m\ge1$. For $m=1$ it is the
immediate parent. In one coordinate, a child integer has the unique form
$3p+r$, $r\in\{-1,0,1\}$, and its parent integer is $p$.
There is no nearest-point tie, and reflecting the coordinate reflects
the parent as well.

Start with the original bases $f_1(i)=C_i(i,U)$. At each step,
transport the parent's five basis vectors along the straight segment to
the child:
\begin{equation}
 f_{h/3}(y)=\tau(y\leftarrow\pi_h(y))f_h(\pi_h(y)).
\end{equation}
The connection $\nabla^0[U]$ is the same at every step. A color vector
and its basis undergo the same transport, so the vector's numerical
components do not change. This is the reason no color matrix appears in
the inclusion \eqref{eq:constant-inclusion}.

To bound the transport accumulated in this recursion, follow a site's
ancestors from the original lattice to level $n$. The $m$-th step
changes each coordinate by at most $3^{-m}$, so the total variation
in that coordinate is at most
\begin{equation}
 \sum_{m=1}^n3^{-m}=\frac{1-3^{-n}}2<\frac12 .
\end{equation}
Thus the original ancestor is $\pi_1(y)$, and the entire piecewise
straight path lies in $\pi_1(y)+[-1/2,1/2]^4$. Summing the
coordinate variations gives a one-norm path length less than two.
Applying the same argument in units of a coarser spacing $h$ also
shows that repeated immediate-parent choices give the same ancestor as
$\pi_h(y)$.

To differentiate this transport, we need a description that includes the
variation of its endpoint fibers $E_y(U)$. The ambient space
$\mathcal H_{\rm g}$ is fixed as $U$ varies, so we first represent
transport there. Along a prescribed curve $y(t)$, let $T(t)$ be
the unitary ambient operator with $T(0)=\Id$ satisfying
\begin{equation}
 \dot T=G_y(\dot y)T,\qquad
 G_y(v)=[\partial_v\Pi,\Pi]+\tfrac15\theta(v)\Pi .
 \label{app:g:transport}
\end{equation}
Here $t$ parametrizes this spacetime curve, a dot differentiates along it, and
$\partial_v=\sum_\mu v^\mu\partial/\partial y_\mu$. The
commutator moves the range of $\Pi$ along the curve; the scalar term
implements the determinant correction in $\nabla^0$. The generator
is anti-Hermitian. Using $\Pi\dot\Pi\Pi=0$, one finds that
$T(t)$ maps the initial fiber to the current fiber and, on vectors in
the initial fiber, obeys
$\Pi\dot T=\theta(\dot y)T/5$. Its restriction to that fiber is
therefore precisely the parallel transport of
\eqref{app:g:connection}.

Now keep the geometric curve fixed and vary the original links. At each
fixed derivative order the preceding construction gives
$\|D^{\boldsymbol I}G_y(v)\|\le C_\ell|v|$. A first derivative
of transport inserts the differentiated generator along the path. With
$T(t,u)=T(t)T(u)^\dagger$, this is
\begin{equation}
 D_I T(t)=\int_0^t T(t,u)
       D_I\{G_{y(u)}(\dot y(u))\}T(u,0)\,\dd u .
\end{equation}
Both surrounding propagators are unitary, so the norm is bounded by a
constant times the path length, denoted by $L_{\rm path}$. Each further derivative either acts on
an existing insertion or inserts another differentiated generator in a
propagator. At order $\ell\ge1$, there are finitely many terms with
$r$ insertions, $1\le r\le\ell$. Their ordered integrals of
the speed factors are bounded by $L_{\rm path}^r/r!$, giving
$\|D^{\boldsymbol I}T\|\le
C_\ell\sum_{r=1}^{\ell}L_{\rm path}^r/r!$.
This argument keeps the prescribed order of the link derivatives; it
does not require them to commute. For the ancestral paths,
$L_{\rm path}\le2$. Combining this bound with the bounded
derivatives of $C_i(i,U)$ proves that every fixed ordered derivative
of $f_h$ is bounded independently of the refinement level.

We can now express the straight-edge transport in these bases. As in
section~\ref{sec:construction}, let $x$ be the integer site label
and $y=hx$ its position in original units. The sampled link is
\begin{equation}
 (F_hU)(x,\mu)=
 f_h(y)^{\dagger}\tau(y\leftarrow y+h\hat\mu)f_h(y+h\hat\mu),
 \qquad y=hx.
 \label{app:g:sample}
\end{equation}
It transports components from the terminal site to the initial site.
The two bases and the transport preserve $\sigma$, so the matrix
lies in $\SU(5)$. Its dependence on $U$ includes the initial
bases, their two ancestral paths, and the edge between the sampled
sites. These paths have bounded total length and stay within a fixed
neighborhood of $(\pi_1(y),\pi_1(y)+\hat\mu)$. The geometric
formulas along them use only a bounded number of original links.
Consequently, each sampled link has bounded local dependence on the
original field, with every fixed ordered derivative bounded uniformly
in $h$.

At $h=1$, \eqref{app:g:sample} gives the reconstructed sample
$F_1U$ in \eqref{app:g:root-consistency}.

Covariance of the connection and the symmetry of the straight segments
and nearest-parent rule pass to $F_h$. In particular, an original
gauge transformation acts on sampled components at $y$ through
$g_{\pi_1(y)}$. A parent and its children have the same original
ancestor, so they transform by the same matrix. This is the covariance
condition needed for the constant inclusion.

The bounds just obtained for an individual sampled link do not require
any small power of $h$. A plaquette has an additional cancellation:
the endpoint bases on its successive edges cancel, leaving transport
around a small closed loop. Its departure from the identity is governed
by curvature through the enclosed area. We express this cancellation
before taking original-link derivatives, so that the area factor remains
even when differentiation removes the small factor $\eps$ in the
curvature.

To estimate a sampled plaquette, take its base point $y$ on the cover
and extend its inherited basis radially across the square:
\begin{equation}
 e(z,U)=\tau(z\leftarrow y;U)f_h(y,U).
\end{equation}
The connection written in this frame is $\dd_z+A^{\rm rad}$.
Radial transport sets the potential along each ray to zero; its
remaining components are determined by the curvature:
\begin{equation}
 A_\mu^{\rm rad}(y+r,U)
 =\int_0^1 t\sum_\nu r^\nu
 F_{\nu\mu}^{\rm rad}(y+tr,U)\,\dd t .
 \label{app:g:radial-identity}
\end{equation}
This is the radial-gauge identity~\cite[eq.~(1.2)]{LeupoldWeigert}.
Here the square and its radial segments are fixed while $U$ varies.

This choice also accounts for the changing color fibers. The ambient
transport equation above gives bounded original-link derivatives of
$e$, including those of its initial value $f_h(y,U)$. Thus the
curvature matrix in this frame has norm at most $C\eps$, and every
positive ordered link derivative has norm at most $C_\ell$.
Since $|r|\le Ch$ on the square, the radial identity gives
\begin{equation}
 \|A^{\rm rad}\|\le C\eps h,\qquad
 \|D^{\boldsymbol I}A^{\rm rad}\|\le C_\ell h
 \quad(|\boldsymbol I|=\ell\ge1).
\end{equation}

In this frame the plaquette is transport around a loop of length $4h$,
based in the original inherited basis. Integrating its transport
equation gives a defect $C\eps h^2$. At positive derivative order,
use the insertion formula above with generator
$-A^{\rm rad}(\dot\gamma)$. Each insertion has integral norm at most
$C_\ell h^2$; a term with $k$ insertions is therefore bounded by
$C_\ell h^{2k}$. There are finitely many such terms, with
$1\le k\le\ell$, and $h\le1$. This proves the $C_\ell h^2$
derivative bound. The bases at intermediate vertices cancel in the
plaquette product, so the calculation includes exactly the base-frame
dependence of $\mathcal P_h$.

\begin{lemma}
\label{app:g:plaquette-derivatives}
Let $\mathcal P_h$ be a plaquette product of $F_hU$. Assume the
original plaquettes required by its construction and by these radial
segments satisfy the common geometric bound. Then
\begin{equation}
 \|\mathcal P_h-\Id\|\le C\eps h^2,\qquad
 \|D^{\boldsymbol I}(\mathcal P_h-\Id)\|\le C_{\ell}h^2
 \quad(|\boldsymbol I|=\ell\ge1).
 \label{app:g:plaquette-jets}
\end{equation}
The bounds hold for every ordered derivative list and are uniform in
refinement level and lattice periods.
\end{lemma}
At order zero, the spatial curvature bound supplies $\eps$.
Original-link derivatives of curvature need only be bounded, so the
differentiated estimate does not assert that factor.

The cylinder in section~\ref{sec:global} requires the same area
estimate for connections other than $\nabla^0$. Fix a smooth family
over a compact parameter set and the fixed spacetime torus, and write
$\nabla=\nabla^0+\mathfrak a$, where $\mathfrak a$ is a
one-form whose value in each spacetime or family direction is a traceless
anti-Hermitian map from the color space to itself. The corresponding
spatial transport uses the generator
$G_y(v)-\widehat{\mathfrak a}_y(v)$, with the hat denoting
extension by zero on $E^\perp$. Bounded family derivatives give
the same insertion estimates for this generator. Bounded spatial
curvature then gives plaquette defects $Ch^2$. These constants may
depend on the chosen family; this general case has no asserted
plaquette-derived factor $\eps$.

\subsection{Domain of the local construction}
\label{app:g:local-domain}

The reconstruction and sampling formulas need only the original links
in a bounded neighborhood of the region where they are evaluated. We
now specify that input neighborhood. This matters in the later
finite-matrix estimates, where the links are given only locally: the
formulas and bounds must apply to every allowed local assignment,
including one with no global admissible extension.

Fix a bounded region $\mathcal O$ on the cover and a larger bounded
open neighborhood containing its closure. To evaluate the reconstruction
throughout this neighborhood, include every original edge used by its
chart boxes, coordinate-order paths, and plaquette fillings. For
sampling, also include the edges needed by the initial frames and by
the original parent cubes visited by the sampling paths. Include all
four edges of every required plaquette. The region, paths, and resulting
input set are fixed before the links vary.

This finite input set suffices to construct the fibers themselves. Only finitely many ambient vertex rows
are active anywhere in the bounded open neighborhood. Restrict the
ambient space to a finite set containing all these rows. An omitted row
contributes only $z^{-1}\Id$ to the resolvent of $K$, whose
integral around $\mathcal C$ vanishes. Thus the restriction leaves
the spectral projection and its range unchanged. This removes only
inactive components from the auxiliary color space; no fermion operator
is being truncated. On overlaps, the frames still span the same $E_y$ and have the same determinant volume
$\sigma(y)$, so they have the previously constructed
special-unitary transitions. The connection on $\mathcal O$ is
therefore defined directly from the finite input set.

One admissibility threshold suffices for all the derivatives used below.
Choose a positive $\eps_{\rm g}$ satisfying all the conditions
on inverse matrices, normalizations, and logarithm branches above.
Differentiating these smooth formulas changes the constants in their
bounds, but introduces no new branch condition. It also introduces no
new input links. For a frame or sampled link, the ancestral paths stay
in their original parent cubes and have length at most two. Thus the
input-neighborhood size and the threshold are independent of both
refinement level and derivative order.

For later use, consider the variation of one sampled link. Use its
position in original lattice units as the argument by writing
$V_h(y,\nu)=(F_hU)(y/h,\nu)$, $y\in h\mathbb Z^4$.
This is the same sampled field as \eqref{app:g:sample}, with the integer
label replaced by its position. For fixed original generators
$\eta_\mu(x)=\sum_a\eta_\mu^a(x)T^a$, define the induced
generator by
\begin{equation}
 \Xi_{\eta,h,\nu}(y)
  =[\delta_\eta V_h(y,\nu)]V_h(y,\nu)^{-1}.
\end{equation}
Multiplication by $V_h^{-1}$ puts the variation in the Lie algebra,
just as for an original link. The derivative includes both ancestral
bases and the edge transport. Since only a bounded number of original
links enters and their fixed ordered derivatives are bounded, for
$|\boldsymbol I|=\ell$ we have
\begin{equation}
 \max_{y,\nu}\|D^{\boldsymbol I}\Xi_{\eta,h,\nu}(y)\|
       \le C_\ell\max_{x,\mu}\|\eta_\mu(x)\|.
\end{equation}
The generators on both sides are dimensionless. This bound has no
additional power of $h$; the stronger area suppression proved above
belongs to closed small loops. For example, at $h=1/3$ the bases at
$y=\hat\mu/3$ and $y=2\hat\mu/3$ are inherited from the original
vertices $0$ and $\hat\mu$. Their transports and the intervening edge
are collinear, so the sampled link matrix is exactly $(F_1U)(0,\mu)$.
The edge is shorter, but its component matrix can retain the full
original-scale transport because the endpoint bases have different parents.

For a torus, construct these formulas on the cover and then identify
repeated link inputs. Differentiating one torus link sums the derivatives
with respect to its lifted copies. Only a uniformly bounded number of
copies enters near one sampled link. More generally, a region of radius
$r_{\rm geo}$, measured in original lattice units, contains at most
$C(1+r_{\rm geo})^4$ required links and plaquettes. The same bounds therefore hold for short periods.

Finally, the supported comparison-gap argument requires the connection on
the paths used to compare operator coefficients, as well as on the
original support. For that application, enlarge $\mathcal O$ to
contain every localization box and radial segment in
section~\ref{app:s:supported}, together with their filling plaquettes.
Plaquettes inside holes of the original support must be included too.
This enlarges the region by a fixed width in original units. Every link
assignment satisfying the plaquette bound throughout this complete
neighborhood then has the geometric construction and derivative bounds
used by the spectral and polynomial estimates below.

\section{The comparison operator and its spectral gap}
\label{app:s:spectral}

We prove a uniform spectral gap for the Hermitian operator $H_{n,s}$
used in section~\ref{sec:construction}. As $s$ increases from zero to
one, this operator interpolates from the fine Hermitian Wilson operator
to the extended coarse operator. Its eigenvalues must stay away from
zero so that $P_{n,s}=(1+\operatorname{sign}H_{n,s})/2$ remains
smooth. Gaps at the two endpoints alone do not exclude a zero crossing
between them. We establish a lower bound valid throughout the
interpolation, independently of refinement level and lattice volume.

We first exclude zero eigenvalues for the free comparison at every
momentum. For an admissible interacting field, we then choose color bases
in each finite box so that the operator coefficients are close to their
free values. Dividing a fermion field into overlapping pieces lets us
apply the free bound locally and control the error from this division.
This avoids requiring a gauge in which the field is close to the free
field throughout the torus.

The locality proof also needs this estimate when the gauge links are
specified only near a finite set of sites. We will state which
links are needed and retain every output component when the operator
acts on a spinor supported there. The last subsections adapt the bound
to the other operators used in the paper and check the determinant
factors at the comparison endpoints.

As in section~\ref{sec:lattice}, we use the site-sum inner product,
without a factor of the spacing to the fourth power. On the infinite
cover the coarse fermion space is
$\mathcal H_{\rm c}=\ell^2(\mathbb Z^4)\otimes\mathbb C^4\otimes V_R$;
on the torus we impose the prescribed fermion boundary signs.
The integer argument labels a coarse site, $\mathbb C^4$ carries spin,
and $V_R\simeq\mathbb C^{15}$ carries the color representation
$R=\mathbf{10}\oplus\bar{\mathbf5}$. A gap $g>0$ for a Hermitian
operator means $H^2\ge g^2 1$, or equivalently
$\|H\psi\|\ge g\|\psi\|$ for every vector. Identity factors on
indices on which an operator does not act are left implicit. As in
appendix~\ref{app:g:geometry}, spatial positions refer to all four
Euclidean spacetime directions.

\subsection{Centered blocks and exact symmetries}

We group the 81 children of each coarse site into one block. This is
only a relabeling of the fine fermion components, but it makes the
inclusion $I_n$ and the range of $H_{n,s}$ explicit. We first use
these coordinates to verify the transformation laws
\eqref{eq:kernel-symmetry}.

At coarse spacing $h=3^{-n}$, write $y=hm$,
$m\in\mathbb Z^4$. As in the
physical-position notation at the end of appendix~\ref{app:g:geometry},
write the coarse links as $V(y,\mu)\in\SU(5)$. The corresponding
covariant shift is
\begin{equation}
 (T_\mu\psi)(m)=R(V(hm,\mu))\psi(m+\hat\mu).
\end{equation}
The Hermitian Wilson operator $H=\gamma_5(D_W-1)$ of
\eqref{eq:wilson} is
\begin{equation}
 H(T)=\gamma_5\left\{
 \sum_{\mu=1}^4\gamma_\mu\frac{T_\mu-T_\mu^{\dagger}}{2}
 +\sum_{\mu=1}^4\left(1-\frac{T_\mu+T_\mu^{\dagger}}{2}\right)-1
 \right\}.
 \label{app:s:wilson}
\end{equation}
Both the coarse and fine Wilson operators have this dimensionless
normalization, with mass shift $-1$ in their own lattice units.
The change of spacing therefore introduces no factor of three in
\eqref{eq:direct-comparison}.

A fine vertex has a unique description $h(m+r/3)$, with nearest
parent $m$ and child label $r\in\{-1,0,1\}^4$. Regrouping the
fine values gives $\mathcal H_{\rm f}=\mathbb C^{81}\otimes
\mathcal H_{\rm c}$. At this scale abbreviate $I_n$ to $I$:
\begin{equation}
 (I\psi)_r(m)=\psi(m)/9,\qquad I^{\dagger} I=1.
\end{equation}
The 81 copies each contribute $1/81$ of the parent's squared norm.
Thus $I$ preserves the norm and does not mix spin or color.
The projector $II^\dagger$ selects fields constant within each child
block; its orthogonal complement consists of fields whose sum over the
children vanishes at each parent.

\paragraph{Coordinate transformations.}
Consider a signed coordinate permutation
$O\hat\mu=\epsilon_\mu\widehat{\pi(\mu)}$,
$\epsilon_\mu=\pm1$. Here $\pi$ permutes the axes and $\epsilon_\mu$ records their
orientations. The induced action on child labels
$R_O|r\rangle=|Or\rangle$ leaves the constant inclusion vector
unchanged. Choose the unitary spin transformation $S_O$ to satisfy
\begin{equation}
 S_O\gamma_\mu S_O^{-1}=\sum_\nu O_{\nu\mu}\gamma_\nu,
 \qquad S_O\gamma_5S_O^{-1}=(\det O)\gamma_5.
\end{equation}
For a reflection with real unit normal $v$, one may take
$S_O=\gamma_5\gamma(v)$, where
$\gamma(v)=\sum_\mu v_\mu\gamma_\mu$. Products of these reflection matrices implement the required coordinate
transformations. Their overall sign does not affect conjugation of an
operator. A reflected covariant shift becomes its adjoint. Together
with the spin transformation, this gives the Hermitian Wilson operator
an overall factor $\det O$, coming from $\gamma_5$. The centered
parent rule transforms the coarse and fine sites in the same way.
Copying a spinor and then transforming it is therefore the same as
transforming it and then copying it. Both endpoints, including the
$\gamma_5$ complement, obey the first identity in
\eqref{eq:kernel-symmetry}, and so does their interpolation at every $s$.

\paragraph{Complex conjugation and gauge transformations.}
To verify the conjugation identity, choose Pauli matrices
$\sigma_1,\sigma_2,\sigma_3$, set
$\gamma_\mu=\sigma_1\otimes\sigma_\mu$ for $\mu=1,2,3$,
and $\gamma_4=\sigma_2\otimes1$. Put
$B_{\rm sp}=\gamma_2\gamma_4$. Choose the reflection products
$S_O$ above with real coefficients in the gamma matrices and no
added complex phase. Direct multiplication gives
\begin{equation}
 B_{\rm sp}\overline{B_{\rm sp}}=-1,\qquad
 B_{\rm sp}\overline{\gamma_\mu}B_{\rm sp}^{-1}=\gamma_\mu,
 \qquad B_{\rm sp}\overline{S_O}B_{\rm sp}^{-1}=S_O.
\end{equation}
Extending $B_{\rm sp}$ by the identity on position, child, and color
indices gives the matrix $\mathcal B$ used in
section~\ref{sec:construction}. With $\mathsf C$ denoting entrywise
complex conjugation, the corresponding antiunitary is
$J_{\rm spin}=\mathcal B\mathsf C$, with square $-1$.
The reconstruction commutes with conjugation, and $I$ is real.
Consequently, conjugating all matrix entries and applying $\mathcal B$
replaces $R$ by $R^*$ without otherwise changing the comparison.
This is the second identity in \eqref{eq:kernel-symmetry}.

Under the inherited gauge action of appendix~\ref{app:g:geometry},
the same gauge matrix is assigned to a parent and its children. Transforming and copying
a spinor therefore commute in the same sense as for the coordinate
actions above. The spin matrices carry no color indices, so both Wilson
operators and the added $\gamma_5$ term transform by conjugation
with the inherited gauge action. This proves gauge covariance of the
whole interpolation.

\subsection{The free comparison}

Set all gauge links to the identity. Fourier transformation then reduces
the comparison to a finite matrix at each coarse momentum. This matrix
contains 81 fine momenta, each with four spin components. The mixing
between these momenta occurs only through the four-dimensional embedded
coarse space. We exploit this fact to reduce the zero-mode equation to
four components. A sign argument excludes a solution at momenta with a
nonzero sine component; the remaining momenta are checked separately.

\paragraph{The free Wilson operators and the inclusion.}
Let $p\in(\mathbb R/2\pi\mathbb Z)^4$ be dimensionless momentum
conjugate to the parent label. Fine momentum $k$ is in fine-lattice
units. The Fourier matrix of the fine Hermitian Wilson operator is
\begin{equation}
 H(k)=\sum_{\mu=1}^4\Gamma_\mu\sin k_\mu+\gamma_5 M(k),
 \qquad M(k)=\sum_{\mu=1}^4(1-\cos k_\mu)-1,
 \qquad \Gamma_\mu=\ii\gamma_5\gamma_\mu .
\end{equation}
The five Hermitian matrices
$\Gamma_1,\ldots,\Gamma_4,\gamma_5$ square to one and
anticommute. Thus, with $d_\mu=1-\cos k_\mu\ge0$,
\begin{equation}
 H(k)^2=
 \left(1+2\sum_{\mu<\nu}d_\mu d_\nu\right)\Id\ge\Id .
\end{equation}
Thus each free Wilson endpoint has gap at least one. The interpolation
requires a separate argument, because the coarse and fine operators
are diagonal in different momentum bases.

Translation by one parent spacing is translation by three fine
spacings. For a given coarse momentum $p$, the fine operator therefore
has the 81 momentum blocks with $3k=p$ modulo $2\pi$:
$\mathcal M(p)=\{(p+2\pi q)/3:q\in\{0,1,2\}^4\}$, with each
component taken modulo $2\pi$. Write $\sum_{3k=p}$ for this
set. In the fine-momentum basis the constant inclusion is
\begin{equation}
 I_k=w(k)\Id_4,\qquad
 w(k)=\prod_{\mu=1}^4\frac{1+2\cos k_\mu}{3},\qquad
 \sum_{3k=p}|w(k)|^2=1.
\end{equation}
In one coordinate the normalized child plane wave has components
$e^{\ii k_\mu r}/\sqrt3$, $r=-1,0,1$. Its overlap with the
normalized constant child vector is $(1+2\cos k_\mu)/3$; multiplying
the four overlaps gives $w(k)$. The amplitudes $w(k)$ are real but need not be positive. Their squares
give the fractions of the embedded spinor's squared norm in the fine momentum
blocks. If $w(k)=0$, a copied coarse spinor has no component in that
block, but the block remains part of the fine fermion space.

With $H_{\rm f}=\bigoplus_{3k=p}H(k)$, the comparison on the
$81\times4=324$ child and spin components is
\begin{equation}
 \widehat H_s(p)=(1-s)H_{\rm f}
       +s\{IH(p)I^{\dagger}+\gamma_5(1-II^{\dagger})\}.
\end{equation}
We have suppressed the identity factor on color. We seek the uniform
bound
\begin{equation}
 g_*:=\min_{p,\ 0\le s\le1}\sigma_{\min}(\widehat H_s(p))>0,
 \qquad \widehat H_s(p)^2\ge g_*^2\Id .
 \label{app:s:free-gap}
\end{equation}
Here $\sigma_{\min}$ is the least singular value, equal to the
smallest absolute eigenvalue for a Hermitian matrix. Once we exclude
zero eigenvalues for every $p,s$, continuity on their compact domain
will give the positive minimum. No numerical gap estimate is required.

\paragraph{Reducing the zero-mode equation.}
Separate the part diagonal in fine momentum from the part acting
through the embedded coarse spinor. With $\bar s=1-s$, define
\begin{equation}
 L_s(k)=\bar s H(k)+s\gamma_5,\qquad
 C_p=H(p)-\gamma_5 .
\end{equation}
Then $L_s=\bigoplus_{3k=p}L_s(k)$ and
$\widehat H_s=L_s+sIC_pI^\dagger$. The only mixing of fine momenta is in $sIC_pI^\dagger$. Suppose
first that every diagonal block $L_s(k)$ is invertible. The
four-by-four matrix $\mathcal R_s=I^\dagger L_s^{-1}I$ applies
this inverse to an embedded coarse spinor and projects the result
back to the coarse space. For a possible zero mode
$\widehat H_s\psi=0$, let $\zeta=I^\dagger\psi$ be its
coarse component. The full equation and its coarse component then give
\begin{equation}
 \psi=-sL_s^{-1}IC_p\zeta,\qquad
 (1+s\mathcal R_sC_p)\zeta=0.
\end{equation}
If $\zeta=0$, the first equation forces $\psi=0$. A nonzero
zero mode therefore requires a nonzero solution of the second,
four-component equation. We check that $L_s$ is invertible in each
case before using this reduction.

The anticommuting gamma matrices make the inverse explicit:
\begin{equation}
 L_s(k)^2=\Delta_s(k)\Id,\qquad
 \Delta_s(k)=\bar s^2\sum_\mu\sin^2 k_\mu+
                  \{\bar s M(k)+s\}^2 .
\end{equation}
Where $\Delta_s(k)>0$, one has
$L_s(k)^{-1}=L_s(k)/\Delta_s(k)$. Consequently
\begin{equation}
 \mathcal R_s=\sum_\mu\Gamma_\mu V_{s,\mu}+\gamma_5 w_s,\qquad
 V_{s,\mu}=\sum_{3k=p}|w(k)|^2
                  \frac{\bar s\sin k_\mu}{\Delta_s(k)},
 \qquad
 w_s=\sum_{3k=p}|w(k)|^2
                  \frac{\bar s M(k)+s}{\Delta_s(k)} .
\end{equation}
The real numbers $V_{s,\mu}$ and $w_s$ are the coefficients of
$\Gamma_\mu$ and $\gamma_5$ in this reduced inverse; $w_s$
is distinct from the inclusion amplitude $w(k)$. Define the list of
gamma matrices
$\boldsymbol\Gamma=(\Gamma_1,\allowbreak\ldots,\allowbreak\Gamma_4,\allowbreak\gamma_5)$
and the real coefficient vectors
$b=(V_{s,1},\allowbreak\ldots,\allowbreak V_{s,4},\allowbreak w_s)$, and
$c=(\sin p_1,\allowbreak\ldots,\allowbreak\sin p_4,\allowbreak M(p)-1)$. Then
$\mathcal R_s=b\cdot\boldsymbol\Gamma$ and
$C_p=c\cdot\boldsymbol\Gamma$. The gamma-matrix anticommutation relations now give
\begin{equation}
 (1+s\mathcal R_sC_p)^{\dagger}(1+s\mathcal R_sC_p)
 =\left\{(1+s\,b\cdot c)^2+
       s^2\bigl[|b|^2|c|^2-(b\cdot c)^2\bigr]\right\}\Id .
\end{equation}
The right-hand side is a nonnegative scalar times the identity.
At $s>0$, a zero would require
$b\cdot c=-1/s$ and equality in Cauchy--Schwarz: $b$ and $c$
would have to be nonzero and point in opposite directions. Thus it suffices to find one coordinate in which the two vectors have
the same nonzero sign. Their first four coordinates are
$V_{s,\mu}$ and $\sin p_\mu$, which we compare next.

\paragraph{Momenta with a nonzero sine component.}
Suppose at least one component of $\sin p$ is nonzero. Every fine
momentum then has a nonzero sine component, because
$\sin k_\mu=0$ would imply $\sin p_\mu=\sin(3k_\mu)=0$.
Hence $\Delta_s(k)>0$ for $s<1$, while at $s=1$ each
$L_s(k)=\gamma_5$. The reduction is valid.

For $s<1$, it remains to prove the common-sign relation
\begin{equation}
 \operatorname{sign}V_{s,\mu}
       =\operatorname{sign}\sin p_\mu
       \quad\hbox{whenever }\sin p_\mu\ne0 .
\end{equation}
We first treat $0<p_\mu<\pi$. Hold the other three fine
components fixed and compare the three possible values of $k_\mu$.
The three terms share the squared inclusion weights from those
other coordinates. We factor these weights out and show that the
remaining sum is positive. Put
\begin{equation}
 c_{\mu,s}=\bar s\sum_{\nu\ne\mu}(1-\cos k_\nu)+s\ge0 .
\end{equation}
As a function of $t=k_\mu$, the denominator in the weighted sum is
\begin{equation}
 \Delta_s(t)=\bar s^2\left(1+\sum_{\nu\ne\mu}\sin^2 k_\nu\right)
                    +c_{\mu,s}^2-2\bar s c_{\mu,s}\cos t .
\end{equation}
It is positive and nondecreasing for $0<t<\pi$.

The three fine values contributing in this coordinate can be written
$t_1,-t_2,t_3$, where
\begin{equation}
 t_1=p_\mu/3,\qquad t_2=(2\pi-p_\mu)/3,\qquad
 t_3=(2\pi+p_\mu)/3,\qquad
 0<t_1<t_2<\pi,\quad 0<t_3<\pi .
\end{equation}
There is one negative sine contribution, at $-t_2$. We compare it
directly with the positive contribution at $t_1$.

The inclusion weight identity
\begin{equation}
 \left(\frac{1+2\cos t}{3}\right)^2
       =\frac{\sin^2(p_\mu/2)}{9\sin^2(t/2)}
       \quad (3t=p_\mu\ {\rm mod}\ 2\pi)
\end{equation}
and $\sin t/\sin^2(t/2)=2\cot(t/2)$ write their three-term sum
as a positive factor times
\begin{equation}
 \frac{\cot(t_1/2)}{\Delta_s(t_1)}
 -\frac{\cot(t_2/2)}{\Delta_s(t_2)}
 +\frac{\cot(t_3/2)}{\Delta_s(t_3)}.
\end{equation}
For the negative fine value the squared weight depends only on
$\sin^2(t_2/2)$, and its denominator equals that at $t_2$;
the sine supplies the displayed minus sign.

On $0<t<\pi$,
$\cot(t/2)$ is positive and strictly decreasing, while
$\Delta_s(t)$ is positive and nondecreasing. Their ratio is
therefore strictly decreasing. Since $t_1<t_2$, the first term
already exceeds the magnitude of the negative term, and the third is
positive. Summing over the other three coordinates preserves strict
positivity: their squared inclusion weights are nonnegative and sum
to one.

Changing $p_\mu,k_\mu$ to their negatives reverses
$V_{s,\mu}$, proving the negative-sine case. If $p_\mu=0$,
fine values with nonzero sine have zero inclusion weight. If
$p_\mu=\pi$, opposite sine values have equal weights and
denominators and cancel. Thus $V_{s,\mu}=0$ in those two cases whenever the reduced
inverse is defined.
For a momentum with at least one nonzero $\sin p_\mu$ and
$0<s<1$, the common sign prevents $b$ and $c$ from pointing in
opposite directions. The comparison therefore has no zero mode. At $s=0$, the operator
is $H_{\rm f}$, with gap at least one. At $s=1$, it is
$H(p)$ on the coarse subspace and $\gamma_5$ on its orthogonal
complement; both blocks also have gap at least one.

\paragraph{Corners of the coarse momentum domain.}
The remaining coarse momenta have only zero or $\pi$ components.
First let $N_\pi\ge1$ be the number of $\pi$ components.
Each contributes at least $1/2$ to $1-\cos k_\mu$, so every
fine momentum has $M(k)\ge N_\pi/2-1$. At $s<1$, a
nonzero fine sine makes $L_s(k)$ invertible. If all fine sines
vanish, the fine momentum is the corresponding corner and
$M(k)=2N_\pi-1>0$, again giving invertibility. At $s=1$,
$L_s(k)=\gamma_5$. The reduced equation may therefore be used,
and the cancellation of all its sine components gives
\begin{equation}
 C_p=(2N_\pi-2)\gamma_5,\qquad (V_{s,1},\ldots,V_{s,4})=0.
\end{equation}
For $N_\pi=1$, the correction $C_p$ vanishes, so the full
comparison is just the already invertible $L_s$. For
$N_\pi\ge2$, the lower bound on $M(k)$ is nonnegative and
hence $w_s\ge0$. The remaining scalar reduced equation is positive:
\begin{equation}
 1+s\mathcal R_sC_p=\{1+sw_s(2N_\pi-2)\}\Id>0 .
\end{equation}

The last case is $p=0$, where one block of the diagonal part does
become singular:
$L_{1/2}(0)=0$. This is a singularity of the chosen decomposition,
not of the full comparison. Here the inclusion has weight one at
$k=0$ and zero at every other fine momentum. On the entire
$k=0$ four-spinor block the full comparison is therefore
$H(0)=-\gamma_5$ for every $s$, so no inverse of
$L_{1/2}(0)$ is taken. Every other fine momentum remains in the
fine space; the inclusion vanishes there and its comparison block is
$L_s(k)$. Such a momentum has a nonzero sine component and is
invertible for $s<1$, and equals $\gamma_5$ at $s=1$.

\paragraph{The uniform free gap.}
We have excluded a zero mode at every $p,s$. In the original
centered-block coordinates, $\widehat H_s(p)$ is continuous on the
compact momentum torus times $[0,1]$. Changing the fine-momentum
labels at a torus seam is only a unitary coordinate change. Its least
singular value thus has a strictly positive minimum, proving
\eqref{app:s:free-gap}. Finite periods and periodic or antiperiodic
fermion signs select subsets of the same coarse momenta, so this gap is
volume independent. Fourier transformation gives the same lower bound
for finitely supported vectors on the cover. Only positivity of $g_*$,
not its numerical value, is used below.

\subsection{The gap for admissible gauge fields}
\label{app:s:local-comparison}
\label{app:s:supported}

The free gap gives a lower bound for an interacting operator wherever its
coefficients can be made sufficiently close to their free values. Small curvature
provides this comparison on a finite box: transport a color basis from
the center, and the connection in the resulting radial frame is bounded
by curvature times distance. We must choose the fine-site bases with
care, however. The interpolation contains the fixed copying matrix
$I_n$, so a change of coarse components must be applied equally to
all children of that coarse site.

Let $H_s^{\rm free}$ have Fourier symbol $\widehat H_s(p)$.
At coarse spacing $h=h_n$, choose a positive integer $L$, measured
in coarse steps. We will fix it from $g_*$ in the localization
argument below. A box on the cover of radius $Ch(L+1)$ contains
one localization window and every coarse and fine path needed to
apply the operator to it. If $z$ is the central coarse site, transport
the inherited basis $f_h(z,U)$ along the straight segment from $z$
to each point $y$ in this box. Denote the resulting radial frame by
$e_{\rm rad}(y,U)$.

At a coarse site, the matrix
$g_{\rm c}(y)=e_{\rm rad}(y,U)^\dagger f_h(y,U)$ converts
inherited components into radial components. At each child $y'$,
use the parent's matrix:
$g_{\rm f}(y')=g_{\rm c}(\pi_h(y'))$. Let
$G_{\rm c}$ and $G_{\rm f}$ multiply fermions by the corresponding
representation matrices. Because the same color matrix acts before and
after copying, $G_{\rm f}I_n=I_nG_{\rm c}$. Thus the transformed
operator $H_{n,s}^G=G_{\rm f}H_{n,s}G_{\rm f}^\dagger$ has the
same numerical inclusion $I_n$. A coarse link becomes
\begin{equation}
 V^{g_{\rm c}}(y,\mu)
   =g_{\rm c}(y)V(y,\mu)g_{\rm c}(y+h\hat\mu)^{-1}.
\end{equation}
Fine links transform in the same way, with $g_{\rm f}$ and spacing
$h/3$.

Copying the parent's component change gives the child the basis obtained
by transporting the new parent basis along the parent-child segment.
In general this differs from transporting directly from $z$ to the
child. We therefore estimate a fine coefficient by following its path
between the two parents, whose bases are radial.

For a positive fine edge from $a$ to $b=a+(h/3)\hat\mu$, put
$a_{\rm c}=\pi_h(a)$ and $b_{\rm c}=\pi_h(b)$. Substituting
the inherited child bases into the sampled link and making the copied
component changes gives
\begin{equation}
 e_{\rm rad}(a_{\rm c})^\dagger
 \tau(a_{\rm c}\leftarrow a)\,
 \tau(a\leftarrow b)\,
 \tau(b\leftarrow b_{\rm c})
 e_{\rm rad}(b_{\rm c}).
\end{equation}
Read from right to left, this matrix starts with components at the
terminal parent, transports to its child, across the fine edge to the
initial child, and back to the initial parent. It then expresses the
result in that parent's radial basis. The inherited bases at the
intermediate sites have canceled. Each parent-child segment has
Euclidean length at most $2h/3$, and the fine edge has length $h/3$,
so the total length is at most $5h/3$. A coarse sampled coefficient
uses one edge of length $h$. Both coefficients are therefore
transports along paths of length $O(h)$, expressed in radial bases
at their endpoints.

Write the connection in this radial frame as
$\dd_y+A^{\rm rad}$. The radial identity
\eqref{app:g:radial-identity}, now based at $z$, expresses this
potential in terms of the reconstructed curvature, which is at most
$C\eps$ in original units.
Since the box has radius $Ch(L+1)$, this gives
$\max_\mu\|A^{\rm rad}_\mu\|\le C\eps h(L+1)$ there.
Integrating the potential along either kind of coefficient path, with
unitary transport on both sides of the insertion, bounds the link
matrix's departure from the identity by $C\eps h^2(L+1)$.

At $n=0$ the coarse endpoint uses the original $U$. To compare it
with the free coefficient, first compare it with the reconstructed
sample $F_1U$ and then apply the radial estimate at $h=1$.
The unit-spacing consistency bound is unchanged by the unitary endpoint
matrices, so
\begin{equation}
 \|U^{g_{\rm c}}_e-1\|
 \le \|U_e-(F_1U)_e\|+\|(F_1U)^{g_{\rm c}}_e-1\|
 \le C\eps(L+1).
\end{equation}

These bounds are in the fundamental representation. For $\mathbf{10}$, the antisymmetric square of the fundamental,
replacing $V$ by $W\in\SU(5)$ changes
its two fundamental factors by at most $2\|V-W\|$ in total.
For $\bar{\mathbf5}$ the bound is $\|V-W\|$. Hence
$\|R(V)-R(W)\|\le2\|V-W\|$ for their direct sum.
Apply this to the fixed finite number of hopping coefficients acting
on a vector $v$ supported in the box. Counting every component of the
output, including components outside its support, gives
\begin{equation}
 \|(H_{n,s}^G-H_s^{\rm free})v\|
       \le C_{\rm loc}\eps(L+1)\|v\|.
 \label{app:s:radial-perturbation}
\end{equation}
Here $C_{\rm loc}$ is independent of $n,s$ and volume. The estimate
uses only the original links required by reconstruction, the inherited
bases, and the specified coefficient paths in this box. Each box may
use its own radial frame.

We now divide a fermion vector into overlapping pieces supported in
these boxes. Each piece has the free lower bound, reduced by the local
error in \eqref{app:s:radial-perturbation}. To recover a bound for the
original vector, we must also control the difference between applying
the operator before and after making this division. That difference is
small when the scalar weights vary slowly between neighboring parents.

Use the one-dimensional triangular weight and its squared site-sum norm
\begin{equation}
 w_L(m)=\max(0,1-|m|/L),\qquad
 s_L=\sum_m w_L(m)^2=\frac{2L^2+1}{3L}.
\end{equation}
The window centered at $hm_0$ is
\begin{equation}
 \chi_{m_0}(m)=s_L^{-2}\prod_{\mu=1}^4w_L(m_\mu-(m_0)_\mu).
\end{equation}
It multiplies every child, spin and color component at a parent by
the same scalar. The normalization gives
$\sum_{m_0}\chi_{m_0}(m)^2=1$, so the sum of the squared norms of
the windowed vectors equals the original squared norm.

A hopping term sees the difference of the window values at its two
parents. For one parent step the one-dimensional difference satisfies
\begin{equation}
 s_L^{-1}\sum_m [w_L(m+1)-w_L(m)]^2
                  =\frac6{2L^2+1}.
\end{equation}
In four dimensions the other three coordinates contribute their
normalized squared sums. Thus, for every parent $m$,
\begin{equation}
 \sum_{m_0}\bigl[\chi_{m_0}(m+\hat\mu)-\chi_{m_0}(m)\bigr]^2
 =\frac6{2L^2+1}\le\frac3{L^2}.
\end{equation}

Group the hops between parent blocks according to the eight oriented
coordinate directions. In each direction, the coefficient contains the
fine hops across that block boundary and the embedded coarse hop, with
weights $1-s$ and $s$. Each has norm at most one: the color transport
is unitary, the Wilson spin factor is a projection times a sign and
$\gamma_5$, and the inclusion preserves norms. Their weighted sum
therefore has norm at most $(1-s)+s=1$. All within-block terms,
including $\gamma_5(1-II^\dagger)$, commute with the scalar window.

For a single hopping term, sum the squared commutator norm over window
centers and then over parent sites. The preceding window difference
bound gives $3\|v\|^2/L^2$. Its contribution to the norm in the
direct sum over windows is therefore at most $\sqrt3\|v\|/L$.
Adding the eight contributions by the triangle inequality gives
\begin{equation}
 \left(\sum_{m_0}\|[H_{n,s},\chi_{m_0}]v\|^2\right)^{1/2}
       \le \frac{8\sqrt3}{L}\|v\|.
\end{equation}

For each piece $\chi_{m_0}v$, make the unitary change to its box's
radial components and apply the free gap together with
\eqref{app:s:radial-perturbation}. Sum the resulting lower bounds in
squares. To compare this sum with the full image, write
$H_{n,s}\chi_{m_0}v$ as $\chi_{m_0}H_{n,s}v$ plus the
commutator. The window normalization applies to both $v$ and
$H_{n,s}v$, and yields
\begin{equation}
 [g_*-C_{\rm loc}\eps(L+1)]\|v\|
 \le\left(\sum_{m_0}\|H_{n,s}\chi_{m_0}v\|^2\right)^{1/2}
 \le\|H_{n,s}v\|+\frac{8\sqrt3}{L}\|v\|.
\end{equation}
There are two losses from the free gap. The commutator error
$8\sqrt3/L$ decreases with $L$, whereas the radial comparison error
$C_{\rm loc}\eps(L+1)$ increases with it. First fix $L$ so that
the commutator error is at most $g_*/4$. Then choose
$0<\eps_{\rm out}<\eps_{\rm g}$ so that
$C_{\rm loc}\eps_{\rm out}(L+1)\le g_*/4$. This also keeps the
reconstruction on its smooth branches. We obtain
\begin{equation}
 \|H_{n,s}v\|\ge g\|v\|,\qquad
 g=g_*/2>0,\qquad \|H_{n,s}\|\le13.
 \label{app:s:gap}
\end{equation}
For the upper bound, each of the four kinetic terms has norm at most
one, each of the four Wilson terms at most two, and the mass term has
norm one. The coarse endpoint is an orthogonal sum of its Wilson
operator and $\gamma_5$, so it and the convex interpolation obey the
same bound 13. The free gap and the choices of $L$ and
$\eps_{\rm out}$ are independent of any later link-derivative order.
The larger admissible domain defined by $\eps_{\rm out}$ will be
used for the smooth coefficient extensions in
appendix~\ref{app:a:locality}. The measure itself uses a strictly
smaller plaquette threshold.

The cover bound also gives the same gap on every finite torus.
Lift a torus spinor $v$ to $\widetilde v$ with its prescribed
boundary signs, and let $\xi_M$ be the indicator of a rectangular
array of $M^4$ translated copies of one fundamental domain of the
torus, with $M$ full torus periods in each parent coordinate.
The lifted coefficients are periodic and respect those
signs, so
\begin{equation}
 \begin{gathered}
 H_{n,s}\widetilde v=\widetilde{H_{n,s}v},\\
 \|\xi_M\widetilde v\|=M^2\|v\|_{\rm torus},\qquad
 \|\xi_M\widetilde{H_{n,s}v}\|
       =M^2\|H_{n,s}v\|_{\rm torus}.
 \end{gathered}
\end{equation}
Because the operator has fixed range in parent blocks, its commutator
with $\xi_M$ involves only the $O(M^3)$ cells next to the boundary.
Its coefficients are uniformly bounded, giving
\begin{equation}
 \|[H_{n,s},\xi_M]\widetilde v\|
       \le C M^{3/2}\|v\|_{\rm torus}.
\end{equation}
Apply the cover gap to the finitely supported vector
$\xi_M\widetilde v$, divide by $M^2$, and let $M$ tend to infinity.
The boundary contribution vanishes, leaving
$\|H_{n,s}v\|_{\rm torus}\ge g\|v\|_{\rm torus}$.
In the lifted image, hops reaching the same torus site have already
been added with their boundary signs. The argument therefore includes
periods one and two, with the same $g$.

The cover localization proof also applies when original links are
specified only near a finite input support. To state precisely which
links are needed, let
$\mathcal S$ be a fixed finite union of full parent blocks, and let
$\mathcal S^+$ contain every output block reached in one application
of the operator. The proof uses more than the links needed just to form
these output coefficients. A window meeting $\mathcal S$ can have
its center up to $L$ coarse steps away and extend another $L$ steps;
its radial comparisons also need the segments from that center.

Enlarge $\mathcal S$ in coarse coordinates by a box of radius
$2L+c$, where $c$ accounts for the fixed hopping range. Then include
the fixed original-link neighborhoods required for reconstruction and
ancestry. Include every parent-child path and every filling plaquette
in these boxes, including plaquettes inside holes of $\mathcal S$.
Impose the outer plaquette bound throughout this complete neighborhood.
For every local assignment of links satisfying these inequalities,
the geometric formulas of appendix~\ref{app:g:geometry} define all
coefficients used in the proof.

Let $\mathcal H_{\mathcal S}$ contain all child, spin and color
components on the input blocks. Retain each matrix entry
$H_{n,s}(m,m')$ with $m'\in\mathcal S$ and
$m\in\mathcal S^+$ to form $K_{\rm out}$. For a vector supported
in $\mathcal S$, this matrix keeps its entire operator image. In
particular, it keeps a component carried by a hopping term outside
$\mathcal S$. Discarding that component would reduce the norm on
which the free lower bound depends. The input and output spaces must
therefore remain distinct, and the local proof gives
\begin{equation}
 g\|v\|\le\|K_{\rm out}v\|\le13\|v\|,
 \qquad
 K_{\rm out}:\mathcal H_{\mathcal S}
                    \longrightarrow\mathcal H_{\mathcal S^+}.
\end{equation}
The bound holds for every local assignment satisfying the stated plaquette
bounds, whether or not it extends to an admissible field on the whole
lattice. A plaquette bound only on
$\mathcal S^+$ is insufficient: the larger neighborhood is what
controls the reconstruction and radial paths used to prove the bound.

\subsection{Gaps for sampled and interpolated connections}
\label{app:s:ordinary-cylinder}
\label{app:s:local-family}

For ordinary Wilson operators the required gap follows directly from
the plaquette estimate. In the color representation entering the
kernel, put $\delta=\sup_p\|1-R(U_h(p))\|$.
In dimensionless units, the operator called $A$ by
Hern\'andez, Jansen and L\"uscher is $1-D_W$ at their mass parameter
$s=0$. Thus $A^\dagger A=H^2$ for our $H=\gamma_5(D_W-1)$;
their mass parameter is unrelated to our comparison parameter $s$.
Their bound~\cite[eq.~(2.16), appendix~C]{HJL} gives
\begin{equation}
 H[U_h]^2\ge(1-30\delta)\Id .
\end{equation}
The norm estimate for $R=\mathbf{10}\oplus\bar{\mathbf5}$ gives
$\delta\le2C\eps_{\rm out}h^2$ for reconstructed samples.
Decrease the common $\eps_{\rm out}$ once so that $\delta\le1/40$
even at $h=1$. Every reconstructed sample then has gap at least $1/2$.

The same bound applies to a vector of finite support when every output
component is retained. The proof of the quoted inequality consists of
finite-shift identities: after pairing with the vector, its nonnegative
terms remain nonnegative and its commutator terms use only plaquettes
in a fixed neighboring halo. Impose the original plaquette bound on
all links needed to construct this halo, within the complete input
neighborhood specified above. The conclusion therefore holds for
local assignments without an admissible extension.

On the comparison cylinder, the connections form a fixed smooth family
over a compact parameter set and a fixed torus. Their spatial curvature
has a common bound, so their sampled represented plaquettes satisfy
$\delta\le2C_{\rm cyl}h^2$. Choosing $h$ sufficiently small to make
this at most $1/40$ gives
\begin{equation}
 \|H[U_h]v\|\ge\tfrac12\|v\|,
 \label{app:s:ordinary-gap}
\end{equation}
uniformly over the family. The same finite-halo argument gives the
supported-input version wherever the connection is defined on the
required paths. This spacing may depend on the chosen family.

The Taylor estimate of appendix~\ref{app:a:locality} requires one
further lower bound. There we vary the background transports of a
finite matrix from their actual values to the identity. For $n\ge1$,
both comparison endpoints sample the same connection. We can therefore
vary them together by multiplying its local radial potential by a
parameter, while retaining the fixed copying matrix and the full
output space.

Fix the original field $U$ and an open radial patch containing every
coarse segment and every three-segment fine path used by the retained
matrix entries, including entries whose output lies outside
$\mathcal S$. Keep the radial
bases, the numerical inclusion $I_n$, and the sets
$\mathcal S,\mathcal S^+$ fixed. Replace $\dd_y+A^{\rm rad}$
on this patch by $\dd_y+uA^{\rm rad}$, with $0\le u\le1$.
This parameter is separate from $s$: $s$ still interpolates the two
endpoint operators, while $u$ changes the connection sampled in both.

For each $u$, recompute the short coarse transports and all three
transports in each fine coefficient path. These local formulas define
$K_{\rm out}(u)$ directly. They give the free coefficients at $u=0$
and the actual coefficients in the copied radial bases at $u=1$.
Throughout this operation $U$ is held fixed; no new connection is
reconstructed from altered original links.

Choose a positive radius $r_0$ in original units, independent of scale
and derivative order. Use this interpolation only when all retained
paths lie in a patch of radius at most $r_0$. The radial bound gives
$\max_\mu\|A^{\rm rad}_\mu\|\le Cr_0$, and multiplication by
$0\le u\le1$ cannot increase it. Each coefficient path has length
$O(h)$, so the resulting operator differs from the free one by at
most $Chr_0\le Cr_0$. The free lower norm bound applies to a
supported vector when all its output components are retained.
Choosing $r_0$ sufficiently small consequently gives
\begin{equation}
 \|K_{\rm out}(u)v\|\ge(g_*/2)\|v\|,
 \qquad 0\le u,s\le1.
 \label{app:s:weak-gap}
\end{equation}
For an ordinary Wilson operator, sample $\dd_y+uA^{\rm rad}$ along
its edges in the same chart; the corresponding argument starts from
free gap one.

These statements bound the full output norms of finite matrices using
only the connection on the patch. They do not define an exact spectral
projector for the local input, and no global admissible interpolation
to the vacuum is needed. Only contributions whose complete coefficient
paths fit in the patch use this construction. As described in
section~\ref{sec:locality}, longer contributions are bounded at the
actual field by exponential decay in polynomial order.

\subsection{Finite periods and endpoint determinant factors}

The gap proof used finite-period fermions through their lifts to the
cover. We finish by checking that this regrouping also preserves the
coarse inclusion and the cancellation of the extra endpoint determinant
factors. The boundary signs must be kept when identifying the fine and
coarse component spaces.

At spacing $h$, the original period $L_\mu$ gives parent period
$N_\mu(h)=L_\mu/h$ and fine period $3N_\mu(h)$. Fine spinors
obey $\psi(y+L_\mu\hat\mu)=\varsigma_\mu\psi(y)$, with the
fixed signs $\varsigma_\mu=\pm1$. Group fine sites by nearest
parent on the infinite cover, then identify translated values with
these signs. The resulting change of components is a permutation
multiplied by real signs, hence is unitary for the site-sum norm and
independent of $U$. It identifies the fine space with the
finite-period version of $\mathbb C^{81}\otimes\mathcal H_{\rm c}$.
Children across an antiperiodic seam retain their signs in the coarse
shifts $T_\mu$. This remains true when a coarse period is one and
the shift returns to the same parent label.

The constant child vector defining $I_n$ is unchanged by permutations
of the children. The position, spin, and copied gauge transformations
therefore respect the inclusion, just as on the cover. Signed coordinate permutations
either preserve the periods and signs or relate the corresponding
transformed tori, as in section~\ref{sec:measure}; original integer
translations are coarse translations at every level. Thus these
finite-period identifications preserve the embedded coarse space and
its orthogonal complement.

At the coarse endpoint, the projector identities of
section~\ref{sec:construction} put the same constant
positive-chirality subspace of $(\im I_n)^\perp$ in both
$\im P_{n,1}$ and the reference space $\im P_{+,n+1}$.
Using the same ordered basis for this common subspace, its exterior
product contributes to the fermion determinant factor and its dual
contributes to the reference factor. They cancel in the relative
determinant line, as in section~\ref{sec:comparison}. The inherited
gauge action and orientation-preserving spin transformations act equally
on the two copies, so their determinant factors cancel as well.

A reflection exchanges the reference chiralities. Its relation between
determinant phases instead follows from the complementary-projector
identity in section~\ref{sec:symmetries}; it does not preserve this
$P_+$ factor.

Neither the inclusion nor the period regrouping depends on $U$, so
neither adds a parameter connection term. The inherited color frames
do depend on $U$; their parameter-direction connection is included
in the covariant transport of section~\ref{sec:global}. After removing
the common complementary determinant factor, the ratio in
\eqref{eq:level-telescope} has the coarse phase in the numerator
and the fine phase in the denominator. Reconstructing the measure still
leaves one constant phase on each connected component. The comparison
does not determine how these constants transform under reflections.

\section{Derivative and exponential locality estimates}
\label{app:a:locality}

The current constructed in section~\ref{sec:construction} is a sum of
comparison terms $B_n$. This appendix proves that the sum and all its
fixed-order derivatives are exponentially local in the original links,
as stated in Proposition~\ref{prop:locality}. We also bound the sampled
curvature $F_h^*\Omega$, whose vanishing under refinement gives
differential integrability.

There are two estimates to combine. The spectral gap makes contributions
with many operator steps exponentially small. Color cancellation supplies
powers of the auxiliary spacing: $h^8$ for a current contribution at
one trace position and $h^6$ for a curvature contribution. These powers
must survive differentiation with respect to the original links. After
summing the trace positions, whose density grows as $h^{-4}$, they give
the required $h^4$ and $h^2$ bounds. The first comparison $B_0$ needs
only a finite bound, since it occurs once at the original spacing.

We prove these statements for finite polynomial contributions before
summing them. This also lets us specify exactly which original links
each contribution needs. Its estimate will require admissibility only
on that finite set, without assuming that the links extend to a globally
admissible field. Smooth cutoff functions will then extend each
coefficient to unrestricted inputs, as required by the locality norm
\eqref{eq:locality}. The last two subsections treat the exterior
derivative and finite-volume corrections.

The original spacing is one, and $h_n=3^{-n}$ is the coarse spacing
at comparison level $n$. We use the elementary original-link derivatives
of \eqref{eq:link-derivative}. An index $I=(x,\mu,a)$ specifies the
positive link $(x,x+\hat\mu)$ and its Lie-algebra generator $T^a$.
The link matrix transports from $x+\hat\mu$ to $x$:
\begin{equation}
 \eta^{[I]}_\nu(z)=\delta_{zx}\delta_{\nu\mu}T^a,
 \qquad D_I=\delta_{\eta^{[I]}}.
\end{equation}
For a list $\boldsymbol K=(K_1,\ldots,K_\ell)$ of such indices,
write $D^{\boldsymbol K}=D_{K_1}\cdots D_{K_\ell}$. The displayed
order is retained because derivatives on the same link need not commute.
Each estimate fixes the derivative order $\ell$. Its constant
$C_\ell$ may depend on that order, but not on volume, admissible
component, or refinement level. For a field of matrices,
$\|\cdot\|_\infty$ is the maximum over sites or links of the
operator norm.

The original link at $x$ labels the current we evaluate. The position
$y$ instead labels a term in the operator trace. For a sampled Wilson
operator, $y\in h\mathbb Z^4$ is a sampled site and the diagonal
trace includes spin and color. For a comparison operator $H_{n,s}$,
we group the 81 fine children of each coarse parent into a block labeled
by $y\in h_n\mathbb Z^4$; its diagonal trace also includes the child
indices. Coordinates are taken on the covering lattice and identified
modulo the original periods. A fixed number of operator steps covers
a distance $O(h)$ in original units, where $h=h_n$ for a comparison.
The factor of three between its two spacings is included in constants.

\subsection{Induced link variations and their domain}

Changing one original link changes every sampled link that uses it in
the reconstruction. The variation includes the reconstructed connection
and the bases transported to the two ends of each sampled link. We first
bound this entire response; counting the affected links separately
would introduce an unnecessary factor from the fine-lattice density.

At comparison level $n$, the coarse and fine endpoint fields are
\begin{equation}
 V_n^{\rm c}=U^{(n)},\qquad V_n^{\rm f}=F_{h_{n+1}}U.
\end{equation}
For $q={\rm c},{\rm f}$, the variation $D_I$ of a positive endpoint
link has generator
\begin{equation}
 \Xi^{I,q}_{n,\nu}(y)
  =[D_I V_n^q(y,\nu)]V_n^q(y,\nu)^{-1}.
 \label{app:a:receiving-source}
\end{equation}
Here $y$ is the position in original lattice units and $\nu$ is the
link direction. The integer site label is $y/h_n$ at the coarse endpoint
and $y/h_{n+1}$ at the fine endpoint. Thus
$D_I V_n^q=\Xi^{I,q}_n V_n^q$. For a single sampled field, write
$V_h(y,\nu)=(F_hU)(y/h,\nu)$ as in appendix~\ref{app:g:geometry}
and use the same definition,
$\Xi^I_{h,\nu}(y)=[D_IV_h(y,\nu)]V_h(y,\nu)^{-1}$.
These dimensionless Lie-algebra fields include the derivatives of the
initial bases and every transport used to inherit the finer bases.
A second original variation $D_J$ gives its own field $\Xi^J$.

Appendix~\ref{app:g:geometry} bounds the dependence of each sampled
link on its original inputs. Choose fixed numbers $N_{\rm in}$ and
$R_0$ so that every endpoint link and its induced generator use at most
$N_{\rm in}$ original links within radius $R_0$. Their derivatives
of every fixed order are bounded uniformly under refinement. Enlarging
these constants covers the original coarse endpoint $U^{(0)}=U$ as
well. Since $\mathfrak{su}(5)$ has 24 generators, at most
$24N_{\rm in}$ elementary derivative indices can act on any one of
these functions. Summing a list of $\ell$ derivative indices therefore
costs at most $(24N_{\rm in})^\ell$.

For a general original variation $\eta$, the bound on the induced
generators is proportional to
$\max_{x,\mu}\|\eta_\mu(x)\|$. It is a maximum over sampled
links, not a sum over them, and is consequently uniform in their number.
The number of sampled sites will enter later, when we sum diagonal
traces.

We also need these estimates when only finitely many original links
are specified. Let $\eps_{\rm out}$ be the final outer admissibility
threshold of appendix~\ref{app:s:spectral}, including the choice for
sampled Wilson operators in section~\ref{app:s:ordinary-cylinder}.
When collecting the inputs, use its comparison localization width $L$
and the fixed sampled-plaquette halo needed by the Wilson estimate.
Decrease $\eps_{\rm out}$ to a
positive rational number if necessary, and choose the target threshold
\begin{equation}
 \eps=\eps_{\rm out}/10 .
 \label{app:a:domains}
\end{equation}
The two thresholds, localization width, and fixed enlargements of
supports are chosen once. They will not depend on polynomial order,
refinement, volume, or derivative order. The gap between the thresholds
will allow a smooth cutoff without changing the target-domain current.

For a given finite polynomial contribution, collect the original links
needed both to evaluate it and to prove its local gap bound. This
includes the links reconstructing its transport paths, endpoint bases,
and all possible one-step outputs. It also includes the intersecting
localization boxes, the radial paths and filling plaquettes used to
compare transports there, and the original neighborhoods needed to
reconstruct those data. Fill holes where the transport comparison uses
a surface: a plaquette bound only along the visible paths would not
control its interior.

Call this set of positive original links $\mathcal E$ and its required
plaquettes $\mathcal P$. Include all four edges of every plaquette
in $\mathcal E$, and define the local domain
\begin{equation}
 \mathcal V_{\mathcal E}
 =\{(U_e)_{e\in\mathcal E}:\|1-U(p)\|<\eps_{\rm out}
                                  \text{ for }p\in\mathcal P\}.
\end{equation}
The geometric formulas and their fixed-order derivatives are defined
throughout $\mathcal V_{\mathcal E}$. The supported-vector estimate
of section~\ref{app:s:supported} gives the gap when all output
components are retained. Both statements use only the specified links
in $\mathcal E$; neither requires a globally admissible completion.

\subsection{Polynomial expansion and color cancellation}

To control the spatial range, we expand the positive spectral projector
in polynomials of its Hermitian operator. Let $\mathcal K$ be an
original or sampled Hermitian Wilson operator $H(V)$, or the comparison
operator $H_{n,s}$. Choose $\Lambda>13$, above their common norm
bound. For a nonnegative integer $m$, define
\begin{equation}
 \begin{gathered}
 Z=1-\mathcal K^2/\Lambda^2,\qquad
 c_m=\frac{4^{-m}\binom{2m}{m}}{\Lambda},\qquad
 S_m=c_m\mathcal K Z^m,\\
 P_{\mathcal K}=\frac12\left(1+\sum_{m\ge0}S_m\right).
 \end{gathered}
 \label{app:a:series}
\end{equation}
This is the inverse-square-root series used in
\eqref{eq:local-polynomial-projector}. On a globally gapped field, $Z$ is nonnegative and its norm is
strictly less than one, so the series converges to
$P_{\mathcal K}=(1+\operatorname{sign}\mathcal K)/2$.
For a comparison this is $P_{n,s}$; for the original Wilson operator
it is the chiral projector $\widehat P_-$.
When only a local tuple in $\mathcal V_{\mathcal E}$ is given,
we use the individual finite polynomials $S_m$. The finite-matrix
estimate below will bound these polynomials directly.

The use of a polynomial series to prove overlap locality follows
ref.~\cite{HJL}; here we use the binomial rather than the Legendre
series. These are polynomials in $\mathcal K$, not yet polynomials
in the original links. The latter representation is constructed
separately in appendix~\ref{poly:section}.

Before assigning contributions to individual sites, substitute
\eqref{app:a:series} into the full curvature trace
$\ii\Tr\{P_{\mathcal K}[\delta_1P_{\mathcal K},
\delta_2P_{\mathcal K}]\}$, where $\delta_1,\delta_2$ are
variations of the operator. The constant $1/2$ in the first projector
multiplies a commutator with zero full trace. Remove that term first:
its diagonal trace at one site need not vanish. The remaining terms
contain three polynomial factors with prefactor $\ii/8$.

For a matrix $A$, write $\tau_y A=\tr\{A(y,y)\}$ for the trace
over the internal indices at $y$ specified above. At total order
$M=m_1+m_2+m_3$, define
\begin{equation}
 \mathcal F^R_{y,M}[\delta_1,\delta_2;\mathcal K]
  =\frac{\ii}{8}\sum_{m_1+m_2+m_3=M}
       \tau_y\{S_{m_1}[\delta_1S_{m_2},\delta_2S_{m_3}]\}.
 \label{app:a:diagonal}
\end{equation}
The superscript records the fermion representation. This decomposition
reproduces the full trace after summing over $y$ and $M$. Since the full
current and curvature are real, we may take the real part of each
finite contribution without changing their sums.

The two quantities needed for Proposition~\ref{prop:locality} are
\begin{equation}
 \begin{split}
 f^I_{y,M,n,s}
  &=\operatorname{Re}\mathcal F^R_{y,M}
                     [\partial_s,D_I;H_{n,s}],\\
 c^{IJ}_{y,M,h}
  &=\operatorname{Re}\mathcal F^R_{y,M}
                     [D_I,D_J;H(F_hU)].
 \end{split}
\end{equation}
In the first line, $\partial_s$ changes the comparison parameter at
fixed $U$, while $D_I$ varies the original links at fixed $s$.
In the second, both $D_I$ and $D_J$ vary the original links, inducing
the complete generator fields $\Xi_h^I$ and $\Xi_h^J$ on the
sampled lattice. Thus these are the contributions to the prescribed
current and pulled-back curvature, not to a response in which fine
links are varied independently. Their sums are
\begin{equation}
 \begin{split}
 j_{\mu,n}^a(x)
   &=\int_0^1\!\dd s\sum_{y,M}f^{(x,\mu,a)}_{y,M,n,s},\\
 B_{n,\eta}
   &=\sum_{x,\mu,a}\eta_\mu^a(x)j_{\mu,n}^a(x),\\
 (F_h^*\Omega)_{IJ}
   &=\sum_{y,M}c^{IJ}_{y,M,h}.
 \end{split}
 \label{app:a:current-components}
\end{equation}
Here $(F_h^*\Omega)_{IJ}$ denotes
$(F_h^*\Omega)(\eta^{[I]},\eta^{[J]})$. The sums over $M$ run
over nonnegative integers, and the sums over $y$ use the sampled sites
or coarse parent blocks appropriate to the operator.

The gauge-variation identity also needs a trace with a site generator
inserted, rather than a link variation. For a prescribed field
$\xi(y)\in\mathfrak{su}(5)$ on the sampled lattice, define
\begin{equation}
 m_{y,M,h}[\xi]
   =\operatorname{Re}\frac{\ii}{2}
          \tau_y\{S_M R(\xi(y))\}.
\end{equation}
Here $\mathcal K=H(F_hU)$. The matrix $R(\xi(y))$ acts on color
and is the identity on spin. To identify the sum, use
\begin{equation}
 P_{\mathcal K}-P_+
   =\frac12\left(\sum_{m\ge0}S_m-\gamma_5\right).
\end{equation}
The term containing $\gamma_5$ has zero spin trace against the site
insertion. Summing $m_{y,M,h}[\xi]$ therefore gives
$\ii\Tr\{(P_{\mathcal K}-P_+)R(\xi)\}$, the anomaly pairing
in \eqref{eq:ward}. The bounds below justify these sums and their
fixed-order original-link derivatives.

Exponential decay in polynomial order controls range. The decrease
with spacing uses \eqref{gen:ordered-traces}, proved from cubic anomaly
cancellation in section~\ref{sec:locality}: ordered color traces in
$R$ and $R^*$ agree whenever they contain at most four Lie-algebra
factors.

Spin conjugation gives an opposite relation for the real diagonal
contributions. It leaves the trace position fixed; the matrix $\mathcal B$ in
\eqref{eq:kernel-symmetry} acts only on spin, while complex conjugation
also conjugates the color matrices. Thus
\begin{equation}
 \mathcal F^{R^*}_{y,M}
       =-\overline{\mathcal F^R_{y,M}}.
\end{equation}
Consequently $\operatorname{Re}\mathcal F^R_{y,M}$ is half the
difference of the real contributions in $R$ and $R^*$. The anomaly
term $m_{y,M,h}[\xi]$ has the same property. A Taylor coefficient
with at most four Lie-algebra factors in its color trace must therefore
vanish: \eqref{gen:ordered-traces} equates the two traces, while the
contribution is their difference.

We use this fact by expanding the background transports about the
identity while keeping the prescribed generators in the trace. Each
prescribed generator supplies one color factor; each derivative of the
background transports supplies at most one more. If $r_{\rm ins}$
is the number of prescribed Lie-algebra insertions, the first background
order that can survive is
\begin{equation}
 p=5-r_{\rm ins}.
\end{equation}
For the current, the induced variation $D_I$ supplies one insertion.
The comparison derivative $\partial_s$ changes only scalar
interpolation coefficients and supplies no color factor. Thus
$r_{\rm ins}=1$ and $p=4$. For the curvature, the two variations
$D_I,D_J$ give $r_{\rm ins}=2$ and $p=3$. The anomaly trace has
one site insertion and again $p=4$.

The cancellation holds for arbitrary numerical Lie-algebra generators
$X\in\mathfrak{su}(5)$ inserted as $R(X)$; it is not a statement
about arbitrary matrices on the fermion representation space. We will
apply it to the actual generators induced by the original links. An
auxiliary parameter will vary the background while holding those
generators fixed. Their dependence on $U$ is still differentiated
when we take original-link derivatives. Because the cancellation is
valid for every generator value, it remains an identity after those
derivatives as well. The exact Taylor remainder of order $p$ will
then turn it into a bound on the full contribution.

\subsection{Finite operators with all output sites retained}

The color identity tells us which background orders vanish. To bound
what remains, we also need exponential decay in the polynomial order
$M$. We establish that bound using only the contribution's specified
finite set of original inputs. The local gap of
section~\ref{app:s:supported} applies to a spinor with finite support
provided that its entire image under the operator is retained. We will
use this image to construct a contracting finite matrix, then show how
its powers reproduce the chosen diagonal contribution.

First consider the response to a prescribed link variation. Let $V$
be a link field of spacing $h$, let $T_\nu$ be its covariant shift,
and let $\xi_\nu(y)$ be the Lie-algebra generator of the variation
on the positive link $(y,\nu)$. Denote this variation by $\delta_\xi$.
Acting on a spinor $\psi$, the varied shift is
\begin{equation}
 (\delta_\xi T_\nu\psi)(y)
   =R(\xi_\nu(y))R(V(y,\nu))\psi(y+h\hat\nu).
\end{equation}
The unvaried shift is unitary. Multiplication by $R(\xi_\nu(y))$
has norm bounded by a representation-dependent constant times the
largest generator norm. Adding the fixed number of terms in the Wilson
operator therefore gives
\begin{equation}
 \|\delta_\xi\mathcal K\|\le C\|\xi\|_\infty.
 \label{app:a:operator-derivatives}
\end{equation}
Here $\mathcal K$ is the Hermitian operator whose polynomial we are
estimating, and $\|\xi\|_\infty$ is the maximum over positive links.
For a comparison, take the larger maximum for the coarse and fine
fields. The inclusion $I_n$ is constant and has norm one, so embedding
the coarse variation preserves this bound. The same reasoning applies
after any fixed number of original-link derivatives: differentiate the
endpoint coefficients and the complete induced generators, whose
bounds were established above. One original link may change many
sampled links, but their combined operator response is controlled by
this maximum. A count of sampled sites is needed only when we later
sum the diagonal traces.

We next choose the finite matrices. A factor $S_m$ contains at most
$2m+1$ operator steps. Thus the contribution of total order $M$ in
\eqref{app:a:diagonal} contains at most $2M+3$ steps. Choose a finite
set of input sites $\mathcal S$ containing every allowed path of
these lengths that starts and ends at the chosen trace position $y$,
including the endpoints of each varied link along those paths. For a comparison,
include complete parent blocks and all 81 children of each parent.
The set is chosen from the possible operator steps, before evaluating
any link matrices. It therefore still contains a path if that path's
coefficient happens to vanish at a particular field. Let
$\mathcal S^+$ contain every site reachable in one operator step from
$\mathcal S$.

Let $\mathcal H_{\mathcal S}$ be the spinor space on these input
sites, including spin, color, and any child indices, and let
$E_{\mathcal S}$ be its coordinate projection. We use two restrictions
of the same operator:
\begin{equation}
 \begin{gathered}
 K_{\rm out}:\mathcal H_{\mathcal S}
                 \longrightarrow\mathcal H_{\mathcal S^+},\\
 (K_{\rm out}v)(y)
     =\sum_{z\in\mathcal S}\mathcal K(y,z)v(z),
             \qquad y\in\mathcal S^+,\\
 K_{\mathcal S}=E_{\mathcal S}\mathcal K E_{\mathcal S}.
 \end{gathered}
\end{equation}
The map $K_{\rm out}$ acts on input spinors in
$\mathcal H_{\mathcal S}$ and keeps all their output components in
$\mathcal H_{\mathcal S^+}$. The square matrix $K_{\mathcal S}$
keeps only entries with both endpoints in $\mathcal S$. For example,
a spinor near the edge of $\mathcal S$ can acquire a component just
outside it after one step. That component contributes to
$\|K_{\rm out}v\|$, but is absent from
$\|K_{\mathcal S}v\|$. The local lower bound is a statement about
the full output norm, and hence about $K_{\rm out}$.

On a torus, the sites in these matrices are periodic sites. First add
all hopping terms with the same periodic input and output, with their
fermion boundary signs, and then form
$K_{\rm out}^\dagger K_{\rm out}$. This counts each output once
and retains interference between hops reaching it. It also covers
periods one and two, for which distinct hops on the cover can have
the same periodic endpoints.

The original inputs in $\mathcal E$ include all the neighborhoods
needed for the local gap proof, so its lower bound holds on the entire
domain $\mathcal V_{\mathcal E}$. Together with the upper norm
bound, it gives a contraction on the input space:
\begin{equation}
 \begin{gathered}
 Z_{\mathcal S}
   =1_{\mathcal S}-\Lambda^{-2}K_{\rm out}^{\dagger} K_{\rm out},\\
 0\le Z_{\mathcal S}\le q_{\rm ad}1_{\mathcal S},
 \qquad q_{\rm ad}=1-g^2/\Lambda^2<1.
 \end{gathered}
 \label{app:a:compression}
\end{equation}
Here $1_{\mathcal S}$ is the identity on
$\mathcal H_{\mathcal S}$. We may use the same
$g=g_*/2\le1/2$ for the comparison and ordinary Wilson operators,
where $g_*$ is the free comparison gap. The bound
$\|K_{\rm out}\|\le13<\Lambda$ gives $Z_{\mathcal S}\ge0$;
the lower bound $\|K_{\rm out}v\|\ge g\|v\|$ gives its upper
bound $q_{\rm ad}<1$. No links outside $\mathcal E$ or globally
admissible completion are used.

This construction also preserves the two-step coefficients required by
the polynomial. Since $\mathcal K$ is Hermitian, the matrix entries
of $K_{\rm out}^\dagger K_{\rm out}$ between input sites are
those of $\mathcal K^2$: the product sums over every possible
intermediate output in $\mathcal S^+$. In particular, a path that
leaves $\mathcal S$ at its middle site and returns at its last site
is included. Squaring $K_{\mathcal S}$ would remove that path.
The contraction is therefore built from the full output map, while
$K_{\mathcal S}$ supplies the remaining single-operator factor.

With the same coefficient $c_m$ as in \eqref{app:a:series}, replace
each polynomial factor by
\begin{equation}
 S_{m,\mathcal S}=c_mK_{\mathcal S}Z_{\mathcal S}^m.
\end{equation}
These replacements leave the complete diagonal contribution unchanged.
Its allowed paths all lie in the chosen set, and each squared factor
still sums its full intermediate outputs. This is an identity for the
chosen finite contribution.
The sets are fixed as $U$, $s$, and the auxiliary background parameter
introduced below vary, so the identity also holds after differentiating.
In particular, derivatives of the squared factor act on both
$K_{\rm out}^\dagger$ and $K_{\rm out}$. The displayed product
order is retained; $K_{\mathcal S}$ need not commute with
$Z_{\mathcal S}$.

The contraction now controls derivatives of the polynomial. If $b$
ordered derivatives are distributed over the $m$ copies of
$Z_{\mathcal S}$, there are at most $(m+1)^b$ choices of their
locations. At most $b$ copies are differentiated. Each undifferentiated
copy still has norm at most $q_{\rm ad}$, while the differentiated
copies have the fixed-order bounds just established. Their product is
therefore bounded by $C_b(m+1)^bq_{\rm ad}^{m-b}$. In the full
contribution, the two prescribed variations in
\eqref{app:a:diagonal} and the $O((M+1)^2)$ choices of
$(m_1,m_2,m_3)$ add only fixed powers of $M+1$. At any fixed
derivative order, losing finitely many contraction factors changes
the constant, but leaves exponential decay in $M$.

We will need this decay along the auxiliary interpolation of the local
potential as well. The bound retaining all output components \eqref{app:s:weak-gap} and
its ordinary Wilson counterpart give a contraction factor
$q_{\rm loc}<1$ throughout that family. Choose one factor for both
the actual field and the family:
\begin{equation}
 q=\max(q_{\rm ad},q_{\rm loc})<1.
 \label{app:a:common-q}
\end{equation}
Write $f_{y,M}$ for any current, curvature, or anomaly contribution,
suppressing its fixed variation indices and other parameters. For every
ordered original-link derivative list $\boldsymbol K$ of length
$\ell$, the finite products then obey
\begin{equation}
 |D^{\boldsymbol K}f_{y,M}|
       \le C_\ell(M+1)^{a_\ell}q^M.
 \label{app:a:polynomial-bound}
\end{equation}
The nonnegative integer $a_\ell$ and constant $C_\ell$ may depend
on the derivative order. For link variations, this estimate uses the
complete induced generators. For a site insertion, its generator and
all original-link derivatives under consideration are assumed bounded.
The paths determine the finite input support, but the norm is taken on
the whole matrix product, so the cancellations that give $q^M$ remain
present. This estimate also applies to the first comparison $B_0$,
whose coarse endpoint is the original field $U$.

\subsection{Radial interpolation and Taylor estimates}

We now combine color cancellation with a small-spacing estimate. The
polynomial bound alone has no power of $h$. For a contribution confined
to a small spacetime patch, radial transport makes the local potential
proportional to distance from the patch center. Scaling this potential
and differentiating a short transport inserts the potential along its
path. Its bound then supplies two powers of $h$: one from the distance
to the center and one from the path length. The color identity will
require several such derivatives.
Throughout this operation we retain the original-link variation whose
current or curvature coefficient is being estimated.

For a comparison with $n\ge1$, both endpoints sample the same
connection; set $h=h_n$. For an ordinary sampled Wilson term, $h$
is its sampling spacing. Let $r_0>0$ be the fixed radius in
\eqref{app:s:weak-gap}, decreased if needed to remain within the
common chart and injectivity margins. Choose a constant
$C_{\rm supp}$ such that a box of radius
$hC_{\rm supp}(M+1)$ contains all order-$M$ paths, their one-step
outputs, the parent-to-child segments, the comparison localization
boxes, and the fixed Wilson plaquette halo. These constants do not depend on refinement or
derivative order. The complete original input set $\mathcal E$ also
contains the fixed neighborhoods needed to reconstruct these data and
the inherited frames. Those original inputs need not themselves lie
inside the radial patch.

First assume $hC_{\rm supp}(M+1)<r_0$. Choose the patch center
and radial paths independently of $U$. For a comparison, use a coarse
vertex $z$ as center and start with its actual inherited frame
$f_h(z,U)$. For an ordinary sampled Wilson term, use its trace site
$z=y$ and the actual frame there. Transport this frame radially through
the patch. The initial frame has uniformly bounded original-link
derivatives, and the geometric transport bounds give the same control
for the resulting radial frames.

In these frames write the connection along spacetime as
$\dd_y+A^{\rm rad}$, where $\dd_y$ differentiates spacetime
coordinates and
$A^{\rm rad}=\sum_\mu A_\mu^{\rm rad}(y,U)\,\dd y_\mu$.
Let $F^{\rm rad}_{\nu\mu}$ be its curvature components. At a
point $y=z+w$, with displacement $w$ from the radial center, the
radial transport identity and its original-link derivatives give
\begin{equation}
 \begin{split}
 A^{\rm rad}_\mu(z+w,U)
  &=\int_0^1t\sum_{\nu=1}^4 w^\nu
          F^{\rm rad}_{\nu\mu}(z+tw,U)\,\dd t,\\
 \|D^{\boldsymbol K}A^{\rm rad}_\mu(z+w,U)\|
  &\le C_\ell|w|.
 \end{split}
 \label{app:a:weak-jets}
\end{equation}
The parameter $t$ runs along the radius from $z$ to $z+w$.
The explicit displacement $w$ is what makes the potential small near
$z$. An original-link derivative acts on the curvature and on the
actual color frame, but the chosen radial path and its displacement
are fixed. Thus the distance factor survives every fixed derivative
order $\ell$. This argument requires bounded original-link derivatives
of the curvature.

Introduce a separate parameter $0\le u\le1$ and, at fixed $U$,
replace the connection on the patch by $\dd_y+uA^{\rm rad}$.
Recompute its transports along the fixed coarse edges and the three
segments representing each fine coefficient: from one parent to its
child, across the fine edge, and back to the other parent. At $u=1$
these are the actual endpoint coefficients in the copied radial frames
of section~\ref{app:s:local-comparison}; at $u=0$ all background
color transports are identities. The coarse change of frame is copied
to its children, keeping $I_n$ constant. This is a local family of
finite matrices evaluated from the original inputs already specified.
The original field is held fixed while $u$ varies; $s$ still labels
the interpolation between the coarse and fine operators.

It remains to specify the insertion on this auxiliary background.
We first evaluate the full induced generator field $\Xi^I(U)$ at the actual
original field, including the derivatives of all endpoint frames and
parent-to-child transports. For a fixed change of color frame $g$, its
positive-link generator becomes
$\xi^I_\nu(y)=\Ad_{g(y)}\Xi^I_\nu(y)$, where
$\Ad_gX=gXg^{-1}$. Transforming the background links and their
prescribed generators together leaves each diagonal trace unchanged.
Write $V$ for the endpoint field, or the pair of endpoint fields in a
comparison, and $V^g$ for these links in the changed frames. With the
other indices suppressed, the current and curvature contributions obey
\begin{equation}
 \begin{split}
 f(V;\Xi^I)&=f(V^g;\Ad_g\Xi^I),\\
 c(V;\Xi^I,\Xi^J)
     &=c(V^g;\Ad_g\Xi^I,\Ad_g\Xi^J).
 \end{split}
\end{equation}
These equalities first hold for any fixed $g$. We then substitute the
actual radial frame change $g(U)$, copied from parents to children for
a comparison. The result is an identity as a function of $U$.
Consequently, subsequent original-link derivatives act on both $g(U)$
and the complete induced generator fields, including the original and inherited
frames on which they depend.

For each fixed $U$, keep the numerical positive-link generator
$\xi^I(U)=\Ad_{g(U)}\Xi^I(U)$ constant as $u$ varies. For a
curvature contribution keep both $\xi^I(U)$ and $\xi^J(U)$
constant in this way. These matrices retain their dependence on $U$,
so an original-link derivative $D^{\boldsymbol K}$ differentiates
them. Holding an insertion fixed in $u$ therefore has no implication
that it is fixed under the original-link derivatives. At $u=1$ this
prescription gives exactly the transformed variation entering the
original contribution.

For clarity, consider one positive link $V(u)$ in the auxiliary
background, with prescribed generator $\xi$. Its insertion and that
on the inverse link are
\begin{equation}
 \delta_\xi V(u)=\xi V(u),\qquad
 \delta_\xi V(u)^{-1}=-V(u)^{-1}\xi.
\end{equation}
The generator on the reverse edge is consequently
$-\Ad_{V(u)^{-1}}\xi$. Although the positive generator $\xi$
is independent of $u$, this reverse generator varies through $V(u)$.
Its background derivatives must therefore be included. They insert the
same potential factors as derivatives of the positive transport, and
keep the two edge orientations consistent.

This prescription is also why the insertion is not recomputed as
$[D_IV(u)]V(u)^{-1}$. That expression would differentiate a new
sampling map with its potential scaled by $u$, instead of carrying
the actual induced variation through the estimate. Nor is a separate
$g^{-1}\dd g$ insertion added to the one-form: the change of frame
was made in the pointwise covariance identity, whose $U$ dependence
is already differentiated in full. Here $\dd$ denotes differentiation
in the original-link variables.

We can now bound the derivatives with respect to $u$. Let $T_u$ be
transport along any one coefficient path. Its length is $O(h)$ and
its points lie at distance $O(h(M+1))$ from $z$. Differentiating
with respect to $u$ inserts one $A^{\rm rad}$ integrated along
this path, with unitary transports before and after the insertion.
The distance bound in \eqref{app:a:weak-jets} and the path length
therefore give a factor $h^2(M+1)$ per background derivative. For
fixed nonnegative derivative orders $\ell,b$ and
$|\boldsymbol K|=\ell$, the estimate is
\begin{equation}
 \|D^{\boldsymbol K}\partial_u^b T_u\|
       \le C_{\ell,b}[h^2(M+1)]^b.
\end{equation}
An original-link derivative can act on an inserted potential or on a
transport between insertions. In the first case
\eqref{app:a:weak-jets} keeps its distance factor; in the second,
the $b$ potentials already inserted by $\partial_u^b$ are still
present. Thus all $b$ powers of $h^2$ survive. This is the feature
that will make the refinement estimate hold together with every fixed
order of original-link derivatives.

Each operator row contains a fixed number of coefficient transports, so the
same bound controls the finite operator and its response to a prescribed
generator field. For $\xi=\xi^I$ or $\xi^J$, fixed $b\ge0$,
$|\boldsymbol K|=\ell$, and $k=0,1$, we obtain
\begin{equation}
 \bigl\|D^{\boldsymbol K}\partial_u^b\partial_s^k
                  K_{\rm out}\bigr\|
 +\bigl\|D^{\boldsymbol K}\partial_u^b\partial_s^k
                  \delta_\xi K_{\rm out}\bigr\|
       \le C_{\ell,b,k}[h^2(M+1)]^b.
\end{equation}
The generator field is held fixed in both $u$ and $s$, while its full
original-link dependence is retained. The constants at these derivative
orders are uniform in $u,s,h,M$ and volume. For a comparison,
$\partial_s$ differentiates only the scalar interpolation coefficients;
for an ordinary sampled Wilson operator, omit it. Compressing the input
and output indices and applying the product rule gives the same bound
for $K_{\mathcal S}$, $Z_{\mathcal S}$, and their first
variations. A site insertion is treated using the maximum norms of its
generator and the original-link derivatives under consideration.

The auxiliary family now has the three properties needed for the
estimate: at $u=1$ it gives the required contribution, at $u=0$
its background transports are identities, and throughout the interval
its finite matrices have a common contraction bound. Let $f(u,U)$
denote a current, curvature, or anomaly contribution in this family,
with its positive-link generators or site generator held fixed as $u$
varies. As above, $p=5-r_{\rm ins}$, where $r_{\rm ins}$ counts
the prescribed Lie-algebra insertions.

At $u=0$, the ordered color-trace identity gives
$\partial_u^bf(0,U)=0$ for every $b<p$. The identity holds for
every admitted original field and arbitrary numerical values of the
insertions. It therefore also holds when those values are the actual
induced generator fields depending on $U$. Differentiating in the original-link
variables still gives zero, including the terms in which derivatives
act on an induced generator or a frame. This is how the color cancellation remains
available for the derivative estimates.

With the lower Taylor coefficients absent, the entire value at $u=1$
is given by the integral remainder:
\begin{equation}
 f(1,U)=\frac1{(p-1)!}\int_0^1(1-u)^{p-1}
                       \partial_u^p f(u,U)\,\dd u.
 \label{app:a:remainder}
\end{equation}
No approximation of the connection dependence is made. Differentiating
this equality with respect to the original links gives the corresponding
integral for each $D^{\boldsymbol K}f(1,U)$. We can use the
finite-matrix contraction uniformly at every $u$ in this integral.

To bound the differentiated integrand, keep the polynomial products
intact and distribute the derivatives among their factors. At total
order $M$, the $\ell$ additional original-link derivatives, the
$p$ background derivatives, and the two prescribed variations in
\eqref{app:a:diagonal} can differentiate at most $\ell+p+2$
of the $Z_{\mathcal S}$ factors. Every other such factor still
contracts by $q$. The $p$ background derivatives supply
$[h^2(M+1)]^p$ in total, also when they act on a varied operator
$\delta_\xi K_{\rm out}$. Choosing their locations and summing
the triples of polynomial orders costs only a fixed power of $M+1$.
The resulting bound on the differentiated integrand is
\begin{equation}
 C_\ell h^{2p}(M+1)^{b_\ell}q^{M-(\ell+p+2)},
\end{equation}
with $b_\ell$ a fixed nonnegative integer. The anomaly site trace
contains fewer prescribed variations and obeys this bound as well.
Absorb the fixed factor $q^{-(\ell+p+2)}$ into $C_\ell$ and
integrate over $u$. Restoring the trace position, polynomial order,
and spacing in the notation gives
\begin{equation}
 |D^{\boldsymbol K}f_{y,M,h}|
       \le C_\ell h^{2p}(M+1)^{b_\ell}q^M.
 \label{app:a:diagonal-bound}
\end{equation}
The spacing $h$ is the ordinary sampling spacing or the coarse spacing
of a comparison with $n\ge1$; in the latter case the bound is uniform
in $s$. Raising the original-link derivative order changes $b_\ell$
and $C_\ell$, but leaves the factor $h^{2p}$.

Finally consider a contribution whose required paths do not fit in the
radial patch, including a winding contribution. It has
$hC_{\rm supp}(M+1)\ge r_0$. Estimate it at the actual field
using \eqref{app:a:polynomial-bound}, and multiply the right-hand
side by $[hC_{\rm supp}(M+1)/r_0]^{2p}\ge1$. This produces
the same $h^{2p}$ factor at the cost of increasing the polynomial
power of $M+1$. Thus \eqref{app:a:diagonal-bound} holds for both
short and long contributions. It gives $p=4$ for the current and
site trace and $p=3$ for the curvature. For the first comparison
$B_0$, use \eqref{app:a:polynomial-bound} directly: its spacing is
fixed at $h_0=1$, and no refinement gain is required.

The estimate also applies when the prescribed generators are supplied
independently of an original-link variation. At derivative order zero,
the proof uses only their maximum norms, the radial bound on the spatial
potential, and the operator gaps at the actual field and throughout the
local background interpolation. It uses no spatial derivatives of the
generators. Thus the same finite-polynomial estimate holds for samples
of any fixed smooth family of spatial connections once these bounds
have been checked; its constants then depend on that family. A site
generator may also depend on the sampling spacing. The operator and
Taylor estimates are linear in its maximum norm, so this dependence
introduces no further condition.

\subsection{Smooth extensions and convergence in the locality norms}

The bound \eqref{app:a:diagonal-bound} applies at one trace position
$y$, with the original variation indices fixed. To obtain the locality
norm, we must sum the positions that can contribute to a given original
output, as well as the original links on which further derivatives can
act. We must also define each coefficient for unrestricted values of
its finite set of input matrices. We first count these inputs, then
construct the smooth extension without losing the spacing or polynomial
order bounds.

Fix the first original variation index $I=(x,\mu,a)$. For total
polynomial order $M$ and spacing $h$, choose
$\rho_{h,M}=R_*+Ch(M+1)$ large enough to contain the complete
original-input neighborhood about the trace position $y$. The constant
$R_*$ includes the fixed enlargements needed for reconstruction and
the local gap, while the second term bounds the distance reached by the
polynomial paths. A contribution to this component vanishes unless
that neighborhood contains the original link in $I$. Its possible trace
positions therefore lie within a region of original four-volume at
most $C(1+\rho_{h,M})^4$. Since the trace positions have density
$h^{-4}$, their number is at most
\begin{equation}
 C h^{-4}(1+\rho_{h,M})^4.
 \label{app:a:density}
\end{equation}
For a comparison, these positions are coarse parent blocks; the 81 child
indices are already included in each block trace. The sampled links
changed by $D_I$ do not require another density factor. Their combined
effect entered the operator estimate
\eqref{app:a:operator-derivatives} through the maximum induced
generator norm, before the trace was taken.

All subsequent index sums concern the original lattice. The same
neighborhood contains at most $C(1+\rho_{h,M})^4$ original link
and generator indices, where $C$ includes the 24 generator directions.
Each of the $\ell$ derivatives in an ordered list can act only at one
of these indices. For curvature, the second variation index $J$ must
also lie there. Thus these sums add fixed powers of the number of
original inputs. On a torus, several lifted links can represent the same
periodic link. Differentiating that periodic link sums the derivatives
with respect to all its occurrences; the same count bounds these
sums, with the derivative order retained throughout.

The estimates so far hold on $\mathcal V_{\mathcal E}$, where the
required local plaquettes obey the outer admissibility bound. The
supremum in \eqref{eq:locality}, however, ranges over all independent
matrices on the finite input set $\mathcal E$. To supply a function
on that larger domain, we multiply each coefficient by a smooth factor
that equals one near every target admissible tuple and vanishes before
the local formulas lose their bounds. We can then set the product to
zero outside $\mathcal V_{\mathcal E}$. This changes neither the
coefficient nor its derivatives on the target domain.

The required plaquette set $\mathcal P$ was fixed with the input
neighborhood. For each $p\in\mathcal P$, define
$z_p=\tr_{\mathbf5}[(1-U(p))^\dagger(1-U(p))]$.
This is a smooth nonnegative function of the four input links of the
plaquette, and $\|1-U(p)\|^2\le z_p$. We use it to turn off
the coefficient while the plaquette is still inside the outer domain.
Define the smooth scalar function $\rho$ and the two thresholds by
\begin{equation}
 \rho(t)=
 \begin{cases}e^{-1/t},&t>0,\\0,&t\le0,\end{cases}
 \qquad u_-=\eps_{\rm out}^2/16,
 \qquad u_+=\eps_{\rm out}^2/8,
\end{equation}
and set
\begin{equation}
 \chi(z)=\frac{\rho(u_+-z)}
                   {\rho(u_+-z)+\rho(z-u_-)}.
 \label{poly:cutoff}
\end{equation}
The denominator is positive for every real $z$, because
$u_-<u_+$. Hence $\chi$ is smooth, equals one for $z\le u_-$,
and equals zero for $z\ge u_+$. Between these thresholds it provides
a smooth transition with $0\le\chi\le1$.

Both thresholds lie in the required places. On a target admissible
tuple, $\eps=\eps_{\rm out}/10$ gives
\begin{equation}
 z_p\le5\eps^2=\eps_{\rm out}^2/20<u_-.
\end{equation}
Every factor is therefore identically one in a neighborhood of that
tuple. Wherever a factor has nonzero support, including the closure of
that support, its plaquette satisfies
$\|1-U(p)\|^2\le z_p\le u_+<\eps_{\rm out}^2$.
Thus the transition ends strictly before the boundary of the local
domain on which the coefficient is defined.

Multiply the coefficient by
$\prod_{p\in\mathcal P}\chi(z_p)$ and extend the product by zero.
Its support closure lies inside $\mathcal V_{\mathcal E}$, so the
result is smooth on $\prod_{e\in\mathcal E}\SU(5)$. The operation
acts on a coefficient as a function of its original-link inputs. It
does not change the spatial connection used to compute that coefficient.
Some inputs may appear only in the cutoff plaquettes. They enter through
those plaquette matrices and require no further reconstruction or
enlargement of the neighborhood.

The extension retains the estimate needed for the sums. There are at
most $C(1+\rho_{h,M})^4$ cutoff factors. Their derivatives of each
fixed order are uniformly bounded, since every factor uses four link
matrices and the thresholds are fixed. Undifferentiated factors are at
most one. Distributing $\ell$ derivatives over the product therefore
introduces only a fixed power of the number of factors. In every
product-rule term, the original coefficient or one of its derivatives
remains, carrying $h^{2p}q^M$ from
\eqref{app:a:diagonal-bound}. Summing the derivative indices adds
only further powers of the number of original inputs, including links
that occur only in the cutoffs.

We can now combine the small-spacing factor, the input counts and the
smooth extensions in the norm \eqref{eq:locality}. That norm measures
the support from the first original variation link, whereas
$\rho_{h,M}$ was measured from the trace position. Because this
original link lies in the same neighborhood, centering the support on
it gives radius at most $2\rho_{h,M}+O(1)$. The fixed part of this
radius contributes only a constant to the exponential weight. The part
that grows with $M$ contributes
$\exp\{C_{\rm rad}\lambda h(M+1)\}$, for a fixed
$C_{\rm rad}>0$.

At each order $M$, the trace-position sum converts $h^{2p}$ to
$h^{2p-4}$. The remaining counts, including all derivative indices and
cutoff derivatives, are bounded by powers of $M+1$, since $h\le1$.
The complete contribution to the norm is consequently bounded by
\begin{equation}
 C_\ell h^{2p-4}(M+1)^{A_\ell}q^M
       \exp\{C_{\rm rad}\lambda h(M+1)\}.
 \label{app:a:summed-bound}
\end{equation}
Here the nonnegative integer $A_\ell$ collects those powers and can
grow with the fixed derivative order $\ell$. The contraction factor
$q<1$ does not change with $\ell$.

The support weight must leave enough decay in $M$ to sum these powers.
Choose one positive rate satisfying
$0<\lambda<-\log q/(2C_{\rm rad})$. Then, uniformly for $h\le1$,
$q^M\exp\{C_{\rm rad}\lambda h(M+1)\}\le Cq^{M/2}$.
The remaining exponential $q^{M/2}$ sums every fixed power of $M+1$.
Thus increasing the derivative order changes the bounding constant, but
does not require a smaller admissibility domain or a new locality rate.

Sum over the polynomial orders and trace positions, and integrate over
$0\le s\le1$ for the current. The current has $p=4$ and hence
$h^{2p-4}=h^4$; the curvature has $p=3$ and hence $h^2$. This gives
\begin{equation}
 \|B_0\|_{\lambda,\ell}\le C_\ell,\qquad
 \|B_n\|_{\lambda,\ell}\le C_\ell h_n^4\quad(n\ge1),\qquad
 \|F_h^*\Omega\|_{\lambda,\ell}\le C_\ell h^2.
 \label{app:a:scale-bounds}
\end{equation}
The sampled Wilson estimate applies at every $h=3^{-n}\le1$, including
$h=1$. For the first comparison $B_0$, use
\eqref{app:a:polynomial-bound} at its fixed spacing $h_0=1$ and
perform the same counts and extension. It needs only a finite bound.
All these estimates use domains and a positive locality rate common to
every fixed $\ell$.

The convergent bounds also give the finite-support decomposition required
in section~\ref{sec:locality-definition}. At each original output,
assign every extended contribution to the integer ceiling $r$ of its
chosen support radius and add the contributions with that value of $r$.
The supports were chosen independently of the field, so this grouping
is unchanged when the inputs vary or are differentiated. Even though
infinitely many refinement levels can enter one group, their
$\sum_{n\ge1}h_n^4$ bound is summable. Together with the remaining
$q^{M/2}$ decay, this gives uniform convergence with every fixed
number of derivatives on the unrestricted finite product of links in
the radius-$r$ neighborhood. Including the finite $B_0$ contribution,
write the resulting smooth coefficient as $j_{\mu,r}^a(x)$ and set
\begin{equation}
 (\mathfrak L_r)_\eta
    =\sum_{x,\mu,a}\eta_\mu^a(x)j_{\mu,r}^a(x),
 \qquad \mathfrak L_{\rm ext}=\sum_{r\ge0}\mathfrak L_r.
\end{equation}
Each $j_{\mu,r}^a(x)$ uses only finitely many original links near its
own output, and these coefficients satisfy \eqref{eq:current-locality}. The form
$\mathfrak L_r$ then sums these coefficients against the variation
at every output. On the target domain all the cutoffs are one, so
$\mathfrak L_{\rm ext}$ agrees with the measure term obtained from
\eqref{app:a:current-components}. This proves
Proposition~\ref{prop:locality}.

\subsection{Locality of the exterior derivative}

The smooth extension of the measure term gives a closed extension of
the overlap curvature by taking its exterior derivative. To use this
extension in the locality norm, we must check what differentiation does
to the center of a coefficient's support. One of the two derivative
terms retains the old center; the other makes a formerly differentiated
input link the first variation index of the two-form.

Write $j_{I,r}=j_{\mu,r}^a(x)$ and $j_I=\sum_rj_{I,r}$ for the
extended coefficients, where $I=(x,\mu,a)$. Consider $D_Ij_{J,r}$
at fixed $I$. The coefficient before differentiation was centered at
$J$. It can depend on the link in $I$ only if that link lies within
its radius-$r$ input neighborhood. Thus only $J$ within distance $r$
of $I$ can contribute, giving at most $C(1+r)^4$ possible indices
$J$. Each input was within $r$ of $J$, so the full support lies
within $2r+1$ of $I$. The extra constant allows for the endpoints
of the links in the graph-distance convention.

The estimate for $j_{J,r}$ with $\ell+1$ derivatives controls
$D^{\boldsymbol K}D_Ij_{J,r}$, where $\boldsymbol K$ is an ordered
list of length $\ell$. To compensate for the larger radius and the
sum over $J$, choose a reduced rate $0<\lambda_d<\lambda/2$.
Then
\begin{equation}
 \begin{split}
 &\sup_I\sum_{r\ge0}e^{\lambda_d(2r+1)}
       \sum_{J,\boldsymbol K}\sup_U
                 |D^{\boldsymbol K}D_Ij_{J,r}|\\
 &\quad\le C e^{\lambda_d}
       \|\mathfrak L_{\rm ext}\|_{\lambda,\ell+1}
       \sum_{r\ge0}(1+r)^4e^{-(\lambda-2\lambda_d)r}<\infty.
 \end{split}
 \label{eq:exterior-locality-bound}
\end{equation}
The leftover decay $e^{-(\lambda-2\lambda_d)r}$ sums the count of
possible centers. The other derivative term, $D_Jj_{I,r}$, is already
centered at $I$ and is bounded directly by the original locality norm.

An exterior derivative also contains the bracket of the two variation
vector fields. With structure constants defined by
$[T^a,T^b]=\sum_c f^{ab}{}_cT^c$, and with
$I=(x,\mu,a)$ and $J=(y,\nu,b)$, the full expression is
\begin{equation}
 (\dd\mathfrak L_{\rm ext})_{IJ}
  =D_Ij_J-D_Jj_I
       +\delta_{xy}\delta_{\mu\nu}\sum_c f^{ab}{}_c j_{(x,\mu,c)}.
\end{equation}
The sign follows from \eqref{eq:link-derivative}: the vector fields
generated by left multiplication of a link have the opposite bracket
sign to the matrices $T^a$. The exterior derivative subtracts their
vector-field bracket, giving the positive structure-constant term in
this formula, as in \eqref{eq:curvature}. That term couples indices
on the same original link and has fixed coefficients, so it satisfies
the same locality estimate without enlarging the support.

For example, $\lambda_d=\lambda/3$ works at every fixed derivative
order. The differentiated series therefore defines a smooth local
two-form on unrestricted links. It is closed because $\dd^2=0$.
On the target domain $\mathfrak L_{\rm ext}$ and all its derivatives
agree with the measure term, whose exterior derivative is $\Omega$
by differential integrability. Hence $\dd\mathfrak L_{\rm ext}$
is the required closed local extension of the curvature.

\subsection{Finite volume}

Locality also bounds how much a finite torus changes the current. Fix
a periodic original field, an original output link, and one lift of
that link to the covering space. Evaluate the current on the torus and
compare it with the same convergent finite-support expansion evaluated
on the periodic lift of the field. In the covering-space expression,
trace paths must close in the covering space. In the torus expression,
they may also close after a translation by a period. Both expressions
are functions of the original periodic field, so a derivative with
respect to one torus link acts on every lifted copy of that link in
the covering-space expression.

A contribution whose enlarged input neighborhood projects injectively
to the torus has the same local inputs and paths in both expressions,
and the two contributions agree. A difference requires either a path
that closes by winding around a period or distinct lifted inputs that
become identified on the torus. Winding paths retain their actual
parallel transports and fermion boundary signs. Thus the difference
between the currents comes entirely from terms large enough to reach
one of these identifications.

In terms of the support bound already used, such a term must obey
\begin{equation}
 Ch(M+1)\ge cL_{\min}-C'R_*-C'',
 \qquad L_{\min}=\min_\mu L_\mu,
\end{equation}
Here $L_\mu$ is the period in direction $\mu$, in original units,
and $c,C',C''$ are fixed positive constants. The terms involving
$R_*$ and the fixed offsets account for the enlargements of the
polynomial paths. For sufficiently large $L_{\min}$, they can be
absorbed into the period term, giving
$M+1\ge c'L_{\min}/h$ with fixed $c'>0$. The factor $1/h$
has a direct geometric meaning: traversing a fixed original length
requires more operator steps as the sampled lattice is refined.

Reserve a further fixed fraction of the exponential decay in $M$ for
this lower bound on the order, leaving the rest to sum the polynomial
factors from derivatives and input counts. For current comparisons with
$n\ge1$, the resulting tail bound is
\begin{equation}
 C_\ell\sum_{n\ge1}h_n^4
         e^{-\lambda_{\rm vol}L_{\min}/h_n}
       \le C'_\ell e^{-\lambda_{\rm vol}L_{\min}}.
 \label{app:a:finite-volume}
\end{equation}
The positive rate $\lambda_{\rm vol}$ can be chosen independently
of derivative order, because every fixed polynomial factor can still
be summed by the decay left in $M$. The first comparison $B_0$ has
the same exponential volume bound at its fixed spacing. For the sampled
curvature $F_h^*\Omega$ at one fixed spacing $h$, the corresponding
large-period bound is
$C_\ell h^2e^{-\lambda_{\rm vol}L_{\min}/h}$; no sum of
curvatures is needed. For bounded $L_{\min}$, the uniform total norm gives the same final
exponential volume bound for the summed current after increasing the
constant. For the sampled curvature in this range, we use only the
uniform $C_\ell h^2$ bound. The additional factor
$e^{-\lambda_{\rm vol}L_{\min}/h}$ is asserted only for sufficiently
large $L_{\min}$.

The derivative-index sums include all occurrences of an identified
periodic link, as required when differentiating the lifted expression.
These finite-volume bounds compare values and derivatives on the target
admissible domain, where every cutoff is identically one. Identifying
lifted input matrices outside that domain need not make the separately
chosen smooth extensions agree; no such agreement is needed for this
comparison.

\section{Trace bounds on a fixed connection cylinder}
\label{app:d:cylinder-traces}

The phase corrections in section~\ref{sec:global} contain the same
two-link and site-generator traces estimated in
appendix~\ref{app:a:locality}. We apply that finite-polynomial estimate
to the fixed connection cylinder and prove
Lemma~\ref{lem:fixed-family-traces}. Only the sampling spacing tends to
zero. Constants $C_\gamma$ may depend on the chosen cylinder, loop,
original spacing and finite volume, but are uniform in $h,r,t$.
No original-link derivatives are needed here.

\begin{proof}
Write $H=H[U_h(r,t)]$ and $P=\widehat P_-[U_h(r,t)]$.
The ordinary curvature $\kappa_h$ has the two prescribed positive-link
generators $\Xi_r,\Xi_t$, since
$\partial_\tau P=\delta_{\Xi_\tau}P$. At the auxiliary boundary,
$m_h^1$ has the site generator $A_t(1,t)$.
Section~\ref{sec:global} bounds their maximum norms uniformly in
$h,r,t$. We use only these bounds, with no spatial derivatives or
positive powers of $h$ assigned to the generators.

Expand the traces as in \eqref{app:a:diagonal} and its site-insertion
counterpart. The full trace of the constant-projector commutator is
removed before the site decomposition. The reference subtraction gives
no additional site term because $P_+-\tfrac12=\gamma_5/2$ has zero
spin trace against a color generator. To use the estimate of
appendix~\ref{app:a:locality}, we check its actual-field gap, radial
bound and gap throughout the local background interpolation.

The cylinder has bounded spatial curvature. Its sampled plaquettes
therefore differ from the identity by $O_\gamma(h^2)$, and the ordinary
Wilson gap holds everywhere on the cylinder for sufficiently small $h$.
Choose a fixed radial-patch radius $r_0$ within the common spatial chart
and injectivity margins. At a trace position $y_0$, spatial radial
transport gives
\begin{equation}
 \|A^{\rm rad}_\mu(y_0+w)\|\le C_\gamma|w|.
\end{equation}
An order-$M$ contribution uses at most $2M+3$ operator steps.
For $hC(M+1)<r_0$, choose the finite input and full-output sets of
appendix~\ref{app:a:locality} inside this patch. Their choice depends
only on the allowed steps, not on the link values.

Transform the links and their already specified insertion fields into
these radial bases. In particular, $A_t(1,t)$ is conjugated as the insertion
in $m_h^1$; it is not recomputed as a temporal connection coefficient
in the new bases. Introduce the same local family
$\dd_y+uA^{\rm rad}$, $0\le u\le1$, holding the numerical generators
on positive links and at sites fixed in $u$. Reverse-link variations
remain $\delta_\Xi V(u)^{-1}=-V(u)^{-1}\Xi$.
The family has identity background transports at $u=0$ and reproduces
the required contribution at $u=1$.

It remains to check its gap. Let $K_{\rm out}(u)$ retain every output
of spinors supported in the chosen input set. At $u=0$, extend an input
spinor by zero and apply the free Wilson gap to obtain
$\|K_{\rm out}(0)\psi\|\ge\|\psi\|$. The patch has no winding
paths, so any fermion boundary signs within it can be removed by fixed
signs on its components. Each edge has length $h$, while the radial
potential is bounded by $C_\gamma r_0$. Thus
\begin{equation}
 \|K_{\rm out}(u)-K_{\rm out}(0)\|
 \le C_\gamma h r_0\le\tfrac12
\end{equation}
for all sufficiently small $h$, uniformly in $u,r,t,M$. The full-output
gap is at least $1/2$. Together with the fixed upper operator bound,
this supplies a common contraction $q_\gamma<1$ for the finite matrices
used in appendix~\ref{app:a:locality}.

All inputs to its Taylor estimate have now been checked. There are two
prescribed color factors in the projector-variation trace, and one in
a site trace. Ordered color cancellation therefore requires, respectively,
three or four derivatives of the background transports. Applying
\eqref{app:a:diagonal-bound} at original-link derivative order zero gives
the following bounds. Here $\Phi^{\rm links}_{y_0,M}$ and
$\Phi^{\rm site}_{y_0,M}[\xi]$ denote the real finite diagonal
contributions at trace position $y_0$ and total polynomial order $M$,
with two prescribed link generators or one site generator $\xi$,
respectively. The fixed nonnegative powers $b,b'$ account for
the finite products and their derivatives:
\begin{equation}
 \begin{aligned}
 |\Phi^{\rm links}_{y_0,M}|&\le
 C_\gamma h^6(M+1)^bq_\gamma^M,\\
 |\Phi^{\rm site}_{y_0,M}[\xi]|&\le
 C_\gamma h^8(M+1)^{b'}q_\gamma^M\|\xi\|_\infty .
 \end{aligned}
\end{equation}
For the remaining orders, $hC(M+1)\ge r_0$, the actual-field gap
bounds the finite contribution. The long-order argument following
\eqref{app:a:diagonal-bound} gives the same estimates after increasing
$b,b'$. It includes winding contributions with their actual holonomies
and fermion boundary signs.

Sum the polynomial orders and then the $h^{-4}\prod_\mu L_\mu$
fine sites. This is the only fine-site count: the link variations were
already bounded by maximum generator norms. We obtain
\begin{equation}
 \begin{aligned}
 \left|\ii\Tr\{P[\partial_rP,\partial_tP]\}\right|&\le C_\gamma h^2,\\
 \left|\ii\Tr\{(P-P_+)R(\xi)\}\right|
       &\le C_\gamma h^4\|\xi\|_\infty .
 \end{aligned}
\end{equation}
The first bound is $|\kappa_h|\le C_\gamma h^2$. Taking
$\xi=A_t(1,t)$ in the second gives $|m_h^1|\le C_\gamma h^4$.

The site estimate requires only a maximum norm, so it also applies to
an $h$-dependent generator field. For the connection reconstructed from
a fixed original field $U$, take the copied Ward parameter
$\xi_{h_{n_{\max}}}=\omega_{n_{\max}}$. Copying preserves
$\|\omega_{n_{\max}}\|_\infty=\|\omega\|_\infty$, even across
changes of original ancestor. Hence the remaining anomaly in
\eqref{eq:ward-telescope} is $O(h_{n_{\max}}^4)$ and tends to zero.
\end{proof}

\section{A polynomial representation of the measure current}
\label{poly:section}

The measure current is already defined, but evaluating it still involves
reconstructed connections, transport equations, and infinitely many
comparisons. We now express this same current as a convergent series of
polynomials in the original link matrices and their adjoints. The
geometric construction enters the numerical coefficients. Each polynomial
then has explicit dependence on the field being evaluated.
This gives an exact representation on the admissible domain and bounds
for finite approximations, without expanding in the gauge coupling or
around the trivial field.

The construction has two steps. First, take a contribution depending on
finitely many links and average it with a polynomial weight. The result
is a polynomial in those links. Sharpening the weight makes the average
approach the original contribution. Second, write this limit as an
initial average plus successive corrections, and sum over the spatial
ranges of the contributions. The locality bounds control that last sum.
We preserve the gauge and lattice transformation laws throughout. The
full sum equals the previously constructed current, so its integrability
conditions are unchanged; a finite approximation need not satisfy them
exactly. Appendix~\ref{phys:section} uses these polynomials to study
the classical continuum limit and finite approximations to fermion
amplitudes.

\subsection{Local contributions and their symmetries}

We start with the local decomposition from
appendix~\ref{app:a:locality}. Fix an original positive edge
$e=(x,x+\hat\mu)$ and write $U_e=U(x,\mu)$. The integer $r\ge0$
now denotes an input radius in original lattice units. The contribution
at this radius depends on a finite set $S_{e,r}$ of links within radius
$r$ of $e$, including $e$ itself. On a finite torus, $S_{e,r}$
contains the actual independent periodic links: different copies of one
link on the covering lattice are identified before any averaging.
The contribution to the measure term is
\begin{equation}
 j_\mu^a(x;U)=\sum_{r\geq0}j_{\mu,r}^a(x;U),\qquad
 \mathfrak L_{e,r}(\eta)
   =\sum_{a=1}^{24}\eta_\mu^a(x)j_{\mu,r}^a(x;U).
\end{equation}
The label $e$ has two roles that must be distinguished. It is one of the
links on which the coefficient can depend, but it also specifies the
variation that this part of the measure term multiplies. Thus
$\mathfrak L_{e,r}$ depends only on input links in $S_{e,r}$, while
its tangent argument is only $\eta_\mu(x)$. Summing over all output
edges and radii recovers $\mathfrak L_\eta$.

Averaging will vary the input matrices independently over the whole group,
including tuples that are not admissible. This is why we need the smooth
extensions constructed in appendix~\ref{app:a:locality}. Its plaquette
cutoff \eqref{poly:cutoff} equals one near the target admissible domain
and turns off the local coefficient before its geometric formulas or
operator gap can fail. Extending the product by zero gives a smooth
function on $\SU(5)^{S_{e,r}}$, the space of independent input matrices.
This extends the coefficient, without extending the geometric construction
to arbitrary link fields.

We need to approximate both the coefficient and its derivatives. For a
scalar function $f$ of links in a finite set $S$ and a nonnegative integer
$\ell$, define
\begin{equation}
 \|f\|_{S,\leq\ell}
 =\sum_{d=0}^{\ell}\ \sum_{I_1,\ldots,I_d}
       \sup_{U\in\SU(5)^S}|D_{I_1}\cdots D_{I_d}f(U)|.
\end{equation}
Here $I=(y,\nu,b)$ specifies an input link and Lie-algebra generator,
and $D_I$ is the original-link derivative of
\eqref{eq:link-derivative}. The inner sum includes all ordered lists of
$d$ such indices in $S$; at $d=0$ it is just $\sup|f|$. The
supremum is over independent $\SU(5)$ matrices, with no admissibility
restriction. We omit $S$ when it is understood. Summing derivative
indices is useful because the averaging changes all the inputs: the
response to those changes is already included in this norm.

The locality estimates give, for every fixed $\ell$,
\begin{equation}
 \sup_{x,\mu,a}\sum_{r\geq0}e^{\lambda r}
       \|j_{\mu,r}^a(x;\cdot)\|_{\leq\ell}\leq C_\ell,
 \qquad |S_{e,r}|\leq C(1+r)^4 .
 \label{poly:input-bound}
\end{equation}
The number $|S_{e,r}|$ counts the input links. One positive decay rate
$\lambda$ works at every fixed derivative order, independently of
lattice volume and admissible component; $C_\ell$ may depend on that
order. Thus contributions with large input radius are exponentially
small even after all derivative locations have been summed.

Before approximating, we arrange that the contribution at each radius
has the transformation laws of the full current. For gauge covariance,
average each $\mathfrak L_{e,r}$ over gauge transformations at the
vertices of its input links, using normalized Haar measure. Transform
the input links and the tangent argument together. Acting with another
gauge transformation then only changes the integration variables, so the
averaged one-form is gauge invariant, or equivalently its current
coefficients are gauge covariant.

Translations use the same translated rule at every output site. For
coordinate permutations, reflections, and complex conjugation, transform
all output coefficients together, moving their output edges and input
sets along with them. Average over this finite symmetry group, including
the axial signs and reversed-link transports of
\eqref{eq:axial}--\eqref{eq:axial-components}. This is an operation
on the current at all output links. Give the resulting contributions
a common radius containing the transformed inputs. The support
enlargement and averages preserve \eqref{poly:input-bound}.

The full current already obeys these laws on the admissible domain, so
averaging its decomposition does not change its sum there. We use this
symmetrized decomposition from now on and write
\begin{equation}
 \mathfrak L_r=\sum_e\mathfrak L_{e,r}.
\end{equation}
Each $\mathfrak L_{e,r}$ has finitely many inputs, whereas
$\mathfrak L_r$ sums the output edges of the whole lattice. Its
coefficients at all sites and directions transform as a current. A
coordinate permutation that changes the periods or fermion boundary
signs compares the corresponding tori with transformed boundary
conditions, as in section~\ref{sec:symmetries}.

\subsection{Polynomial averaging of the measure term}

Consider first a smooth scalar function $f$ of one link matrix $U$.
We will average its values at other matrices $V$, using a weight that
is polynomial in both $U$ and $V$. After integration, the factors
in $V$ become numerical coefficients and the factors in $U$ remain.
This produces a polynomial even when $f$ itself is not polynomial.

Let $G=\SU(5)$ and normalize Haar measure by $\int_G\dd g=1$.
For a positive integer $m$, define
\begin{equation}
 b(g)=\frac{1+\operatorname{Re}\tr(g)/5}{2},\qquad
 Z_m=\int_Gb(g)^m\,\dd g,\qquad k_m(g)=Z_m^{-1}b(g)^m .
 \label{poly:kernel}
\end{equation}
Here the trace is in the fundamental representation. The function
$b$ lies between zero and one, and $Z_m$ makes the nonnegative weight
$k_m$ integrate to one. The weight is polynomial in the entries of
$g$ and $g^\dagger$. It is unchanged by conjugation
$g\mapsto hgh^{-1}$, inversion, or complex conjugation. A function
unchanged by conjugation is called a class function.

Hold $U$ fixed and average $f(gU)$. Changing the integration variable
to $V=gU$ gives
\begin{equation}
 \int_G k_m(g)f(gU)\,\dd g
       =\int_G k_m(VU^{-1})f(V)\,\dd V .
 \label{poly:one-link-average}
\end{equation}
The right-hand side makes the polynomial dependence visible: $f(V)$
does not depend on $U$, and all the $U$ dependence is in the weight.
Expand that weight into finitely many terms, with index set $A_m$:
\begin{equation}
 \begin{split}
 k_m(VU^{-1})
    &=\sum_{\alpha\in A_m}
        p_{m,\alpha}(U,U^\dagger)q_{m,\alpha}(V,V^\dagger),\\
 \int_G k_m(VU^{-1})f(V)\,\dd V
    &=\sum_{\alpha\in A_m}p_{m,\alpha}(U,U^\dagger)
       \int_G q_{m,\alpha}(V,V^\dagger)f(V)\,\dd V .
 \end{split}
 \label{poly:one-link-polynomial}
\end{equation}
The functions $p_{m,\alpha}$ and $q_{m,\alpha}$ are polynomials
collecting the $U$ and $V$ factors in each term. Every integral in
the last line is a number independent of $U$. This step did not expand
$f$ about the identity or restrict $U$ to be close to it.

For these polynomial averages to recover $f(U)$, the weight must
concentrate on $V$ close to $U$ as $m$ grows. Equivalently, $g$ must
concentrate near the identity. Using the Frobenius matrix norm
$\|M\|_F^2=\tr(M^\dagger M)$, we have
\begin{equation}
 1-b(g)=\|1-g\|_F^2/20 .
\end{equation}
Thus $b$ reaches its unique maximum at $g=1$. Let $d_G$ denote
distance measured by the shortest path on the group, using the metric
from the inner product $\langle\cdot,\cdot\rangle$ of
section~\ref{sec:measure}. This metric is invariant under left and right
multiplication. Near the identity, $1-b(g)$ is bounded above and below
by positive constants times $d_G(g,1)^2$. Consequently $b(g)^m$
is concentrated within distance $O(m^{-1/2})$. Away from the identity,
$b$ is uniformly less than one, so this part of the integral decreases
exponentially in $m$. Integration gives
\begin{equation}
 Z_m\geq c m^{-12},\qquad
 \int_Gd_G(g,1)k_m(g)\,\dd g\leq C m^{-1/2}.
 \label{poly:concentration}
\end{equation}
The constants $c,C>0$ are independent of $m$. The power $m^{-12}$
is the volume of a region of width $m^{-1/2}$ in the 24-dimensional
group. After normalization, the second bound states that the average
link displacement is $O(m^{-1/2})$. Multiplying this displacement
by a derivative bound for $f$ controls the approximation error.

The measure term is linear in a link variation, so averaging only its
scalar coefficients would miss part of the transformation. At the
output link, write $\delta U=\eta U$ as before. Holding $g$ fixed,
left multiplication $U\mapsto gU$ changes this variation to
$\delta(gU)=(g\eta g^{-1})(gU)$. Right multiplication
$U\mapsto Ug$ instead gives $\delta(Ug)=\eta(Ug)$.
Thus left averaging rotates the generator, whereas right averaging
does not. The pullback of a one-form means evaluating it on both the
transformed links and this transformed variation.

For an input set $S$ and a link $f\in S$, let $L_{f,g}$ replace
$U_f$ by $gU_f$ and $R_{f,g}$ replace it by $U_fg$, leaving the
other links fixed. Link reversal exchanges the two operations, since
$(gU_f)^{-1}=U_f^{-1}g^{-1}$. We therefore average them with equal
weight at each link:
\begin{equation}
 \mathcal H_m^S
 =\prod_{f\in S}\left\{\frac12\int_G k_m(g)
                 (L_{f,g}^*+R_{f,g}^*)\,\dd g\right\}.
 \label{poly:pullbacks}
\end{equation}
The stars denote the pullbacks just described. The product means that
we perform this averaging once for every input link; the operations
at different links commute. An added link on which the coefficient does
not depend and at which the form has no tangent component contributes
an identity operation. We may therefore enlarge input sets when needed
and omit the superscript $S$ when it is understood.

To see what this does to a current coefficient, take a real one-form
$\beta$ on $G^S$ and write
\begin{equation}
 \beta(\eta)=\sum_{\substack{x,\mu\\(x,x+\hat\mu)\in S}}
       \sum_{a=1}^{24}\eta_\mu^a(x)b_\mu^a(x;U),\qquad
 b_\mu(x;U)=\sum_a b_\mu^a(x;U)T^a .
\end{equation}
The matrices $T^a$ are the same orthonormal Lie-algebra basis used for
the current. The components $b_\mu^a$ here are real scalar functions;
$b_\mu$ is their Lie-algebra-valued matrix. Define
$\Theta_\mu(x)=\dd U(x,\mu)U(x,\mu)^{-1}$. This matrix-valued
one-form returns the generator when evaluated on a variation:
$\Theta_\mu(x)(\eta)=\eta_\mu(x)$. Hence
\begin{equation}
 \beta=\sum_{(x,x+\hat\mu)\in S}
                 \langle b_\mu(x),\Theta_\mu(x)\rangle .
\end{equation}
Left multiplication gives $L_g^*\Theta=\Ad_g\Theta$, where
$\Ad_g X=gXg^{-1}$. Invariance of the inner product moves this
rotation onto the coefficient as $\Ad_{g^{-1}}b_\mu$. Right
multiplication leaves $\Theta$ unchanged. This accounts for the
extra matrix factor in the averaging formula below.

Fix the output edge $e=(x,x+\hat\mu)$ and change the integration
variables to the translated input matrices $V=(V_f)_{f\in S}$.
Writing $\dd V=\prod_{f\in S}\dd V_f$, the matrix coefficient of
the averaged form is
\begin{equation}
 \begin{split}
 (\mathcal H_m\beta)_\mu(x;U)=\frac12\int_{G^S}
 &\prod_{f\in S}k_m(U_f^{-1}V_f)\\[-2pt]
 &\times\left[b_\mu(x;V)
       +U_eV_e^{-1}b_\mu(x;V)V_eU_e^{-1}\right]\,\dd V .
 \end{split}
 \label{poly:coefficient-integral}
\end{equation}
The notation $(\mathcal H_m\beta)_\mu(x;U)$ means the matrix
coefficient multiplying $\Theta_\mu(x)$, in the same way that
$b_\mu$ is the coefficient of $\beta$. The field $U$ is held fixed;
only $V$ is integrated. Right multiplication at the output gives
$b_\mu(x;V)$. Left multiplication gives
$\Ad_{g^{-1}}b_\mu(x;V)$ with $g=V_eU_e^{-1}$, which is the
second bracketed term. At other input links there is no tangent
component for this output. The class-function property equates their
left and right scalar averages and gives the common kernel product.

Expanding the kernels now gives polynomials in $U_f,U_f^\dagger$,
with numerical integrals as coefficients. Each kernel has degree at most
$m$ in its link argument, and the output rotation adds at most two
factors of $U_e$ or $U_e^\dagger$. The total degree is therefore at
most $m|S|+2$. The kernel symmetries ensure that $\mathcal H_m$
commutes with gauge transformations, link reversal, the lattice
transformations, and complex conjugation. It preserves reality and
introduces no dependence on links outside $S$. As an average of
pullbacks, it also commutes with exterior differentiation.

We must control approximation errors in derivatives as well as values.
Use $\|\beta\|_{\le\ell}=\sup_{x,\mu,a}
\|b_\mu^a(x;\cdot)\|_{S,\le\ell}$, with the derivative norm
defined above. Under a product of link translations, each derivative
generator is either unchanged or rotated in the adjoint representation.
At order $\ell$ there are at most $\ell$ such indices, together
with the one current index. The rotations therefore cost a constant
$C_\ell$ depending on the fixed 24-dimensional Lie algebra and the
derivative order, not on the number of input links.

To compare a translated coefficient with the original one, change the
links successively. For each multiplying matrix $g$, follow a shortest
group path from $1$ to $g$ and integrate the derivative along it.
The resulting difference is bounded by path length times a link
derivative. If a coefficient or derivative index is also rotated, its
rotation changes only when that link is the output or one of the
$\ell$ differentiated links. There are at most $\ell+1$ such
occurrences. The derivative norm already sums over every changed link,
so this estimate brings no additional factor of $|S|$.
Averaging the path lengths with \eqref{poly:concentration} gives
\begin{equation}
 \|\mathcal H_m\beta\|_{\leq\ell}
       \leq C_\ell\|\beta\|_{\leq\ell},
 \qquad
 \|\mathcal H_m\beta-\beta\|_{\leq\ell}
       \leq C_\ell m^{-1/2}\|\beta\|_{\leq\ell+1}.
 \label{poly:approximation}
\end{equation}
The first estimate keeps the averaging uniformly bounded; the second
makes it approach the original one-form as $m\to\infty$. At every
fixed derivative order, one additional derivative bounds the error.

\subsection{The polynomial series and finite approximations}

We now apply these averages to $\mathfrak L_{e,r}$ and write the
limit as a series. The parameter $r$ still controls the spatial input
radius. A new index $k$ controls how sharply we concentrate the
averaging weight; it is not the refinement level $n$ used to construct
the original current. Choose $m=4^k$, for which the approximation
error is $O(2^{-k})$, and define
\begin{equation}
 \begin{split}
 \mathfrak P_{e,r,-1}&=\mathcal H_1\mathfrak L_{e,r},\\
 \mathfrak P_{e,r,k}&=(\mathcal H_{4^{k+1}}-\mathcal H_{4^k})
                         \mathfrak L_{e,r}\quad(k\geq0),\\
 \mathfrak P_{r,k}&=\sum_e\mathfrak P_{e,r,k}.
 \end{split}
 \label{poly:blocks}
\end{equation}
The index $k=-1$ denotes the initial average $\mathcal H_1$;
$k\ge0$ labels the change between two successive averages. At output
$e=(x,x+\hat\mu)$, denote their scalar and matrix current
coefficients by
\begin{equation}
 \mathfrak P_{e,r,k}(\eta)
   =\sum_a\eta_\mu^a(x)j_\mu^{a;(r,k)}(x;U),\qquad
 j_\mu^{(r,k)}(x;U)=\sum_a j_\mu^{a;(r,k)}(x;U)T^a .
\end{equation}
Each term at fixed $e,r,k$ depends on a finite set of input links.
Summing over all $e$ gives a polynomial contribution to the current
on the whole lattice, with the same gauge, axial, and
complex-conjugation laws as the exact current. It is this collection
of output coefficients, transformed together, that carries the lattice
symmetries.

The spatial bound \eqref{poly:input-bound} controls the sum over $r$,
and the approximation bound controls the sum over $k$. With the fixed
positive rate $\lambda_{\rm pol}=\lambda/2$, they give
\begin{equation}
 \sup_{x,\mu,a}\sum_r e^{\lambda_{\rm pol}r}
    \|j_\mu^{a;(r,k)}(x;\cdot)\|_{\leq\ell}
 \leq C_\ell'2^{-k}\quad(k\geq0).
 \label{poly:block-bound}
\end{equation}
The initial average has the analogous finite bound without the
$2^{-k}$ factor. Both sums therefore converge absolutely, including
every fixed order of original-link derivatives and their summed
positions, at one positive exponential locality rate.

The cancellation between successive averages identifies their sum.
For example, the initial term plus the first correction is
$\mathcal H_1\mathfrak L_{e,r}
 +(\mathcal H_4-\mathcal H_1)\mathfrak L_{e,r}
 =\mathcal H_4\mathfrak L_{e,r}$.
In general, summing through $k=k_{\max}\ge0$ gives exactly
$\mathcal H_{4^{k_{\max}+1}}\mathfrak L_{e,r}$, whose limit is
$\mathfrak L_{e,r}$. Summing the radii and outputs gives
\begin{equation}
 \mathfrak L
   =\sum_{r\geq0}\left(\mathfrak P_{r,-1}
                     +\sum_{k\geq0}\mathfrak P_{r,k}\right).
 \label{poly:identity}
\end{equation}
This identity holds throughout the admissible domain, in every fixed
original-link derivative norm at the stated positive locality rate.
It establishes equality with the comparison current. Its curvature, gauge variation,
and integrals along paths are consequently the same. In particular,
the Wilson line $W[\gamma]$ is unchanged for every admissible loop.
The averaging and the sum of corrections are illustrated in
Fig.~\ref{fig:polynomial-group-averaging}.

\begin{figure}[htbp]
  \centering
  \includegraphics[width=\linewidth]{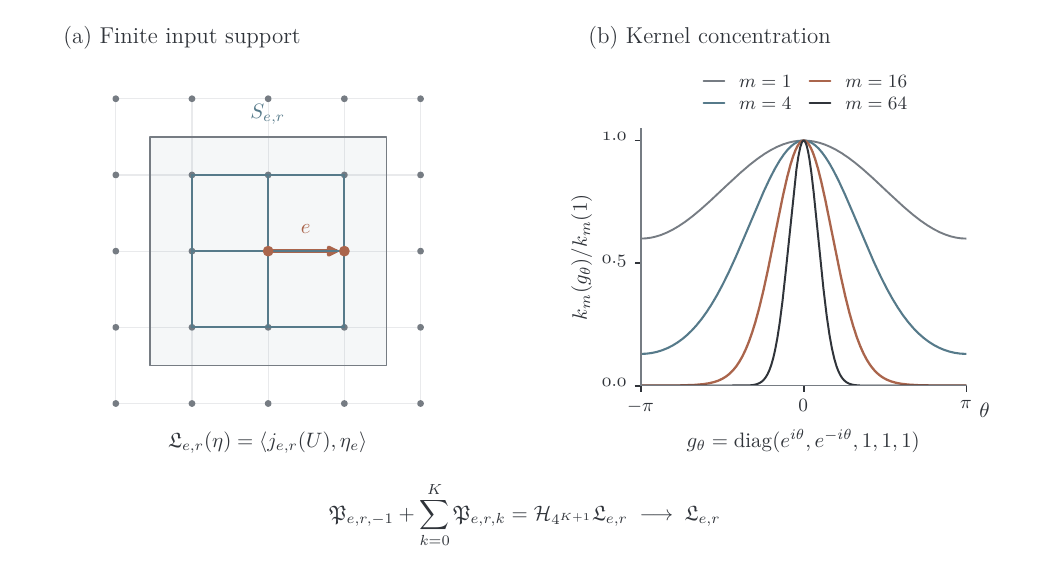}
  \caption{(a) The contribution $\mathfrak L_{e,r}$ depends on links in $S_{e,r}$ within radius $r$ of output edge $e$. Here $j_{e,r}=\sum_a j_{\mu,r}^a(x)T^a$ and $\eta_e=\eta_\mu(x)$ for $e=(x,x+\hat\mu)$. (b) The relative weights $k_m(g_\theta)/k_m(1)$ sharpen as $m$ increases; they are not a one-dimensional Haar density. The operator $\mathcal H_m$ averages the whole measure term, including its variation argument. Its polynomial weight leaves polynomial dependence on the original links and their adjoints, with numerical coefficients supplied by integration. Below, the initial polynomial $\mathfrak P_{e,r,-1}$ and corrections through $k=K$ give a sharper average, approaching $\mathfrak L_{e,r}$ as $K$ increases.}
  \label{fig:polynomial-group-averaging}
\end{figure}

For a finite approximation, choose nonnegative integers
$r_{\max},k_{\max}$ and retain
\begin{equation}
 \mathfrak L^{[r_{\max},k_{\max}]}
   =\sum_{r=0}^{r_{\max}}\left(\mathfrak P_{r,-1}
                       +\sum_{k=0}^{k_{\max}}\mathfrak P_{r,k}\right).
\end{equation}
The corresponding components are
$j_\mu^{a;[r_{\max},k_{\max}]}(x;U)$. Square brackets denote a
partial sum; parentheses $(r,k)$ label one polynomial contribution.
These cutoffs independently limit spatial range and the accuracy of
the averaging.

Use the locality norm \eqref{eq:locality}, summed over derivative
orders through $\ell$:
$\|\beta\|_{\lambda',\le\ell}
 =\sum_{d=0}^\ell\|\beta\|_{\lambda',d}$.
As before, it refers to the specified finite-range decomposition at
each output. The omitted radii have an exponentially small tail, and
the omitted accuracy corrections have a geometric tail. For each fixed
$\ell$ and $0\le\lambda'<\lambda_{\rm pol}$,
\begin{equation}
 \|\mathfrak L-\mathfrak L^{[r_{\max},k_{\max}]}\|_{\lambda',\leq\ell}
 \leq C_{\ell,\lambda'}
       \left(e^{-(\lambda_{\rm pol}-\lambda')r_{\max}}+2^{-k_{\max}}\right).
 \label{poly:prefix-error}
\end{equation}
At $\lambda'=0$ this bounds each current component, and the higher
orders bound its summed link derivatives. Finite sums retain the exact
transformation laws when their polynomial coefficients are defined by
the exact integrals. They need not satisfy the covariant divergence,
curvature, or full Wilson-line identities exactly. Those identities
hold for the infinite sum; \eqref{poly:prefix-error} controls the
current error when only finitely many terms are retained.

\subsection{Determining the numerical coefficients}

The coefficients in \eqref{poly:coefficient-integral} are numbers
independent of the field $U$ to be evaluated. Their defining integrals
still contain the geometric construction of the current, however, so
we must explain how to determine those numbers with controlled errors.
There are two tasks: approximate the integrand uniformly by finitely
many comparison calculations, then evaluate the remaining finite group
integral. This concerns the computability of the prescription, not its
practical efficiency.

All choices are fixed in original units of spacing one: the geometric
weights, the allowed operator paths, the symmetry averages, and the
plaquette cutoff. We first obtain explicit margins that keep the
geometric and spectral formulas away from their singularities. The
free comparison matrices depend continuously on momentum and the
comparison parameter in a compact domain. Subdivide that domain and
bound their entries on each piece by interval arithmetic, meaning
upper and lower enclosures for each numerical operation. The continuity
estimates then give a lower gap bound throughout the piece. The proved
positive gap ensures that a sufficiently fine finite subdivision will
certify a positive rational lower bound everywhere. This bound fixes
a localization length and a rational outer plaquette threshold within
the geometric and spectral margins, once for all refinement levels
and derivative orders.

These margins also control the derivatives needed for an error bound.
For each requested fixed order, bound the derivatives of the geometric
weights in appendix~\ref{app:g:weights} and the cutoff
\eqref{poly:cutoff}. Their explicit exponential and integral formulas
give rational upper bounds: use the endpoint estimates near the ends
of their supports and interval bounds on the remaining compact regions.
For matrix inverses, roots, and logarithms in the reconstruction, the
fixed margins allow the resolvent identities and product rule to bound
the derivatives. Differentiating the transport equations then carries
these bounds along the inherited parent-to-child paths. Their total
length is at most two in original units, independently of refinement,
so this step does not accumulate an increasing transport length.

Next truncate the two sums used to evaluate a local comparison
coefficient. The first is over refinement levels $n$. The second is
over the total order $M$ in the sign-function polynomial expansion of
appendix~\ref{app:a:locality}. This order $M$ is distinct from the
averaging power $m$ and accuracy index $k$ used in this appendix.
Choose a rational number strictly between the contraction factor in
\eqref{app:a:common-q} and one. At any fixed derivative order, counting
derivative positions in the finite matrix products and original links
in their inputs gives a computable constant and an integer power of
$M+1$. After the trace-site sum, the bound at each $n\ge1$ contains
$3^{-4n}$ times this power of $M+1$ and the chosen rational number
to the power $M$.

The first comparison $n=0$ has its own bound with the same exponential
decay in $M$. Cutoff products and symmetry averages change only the
computable constants and polynomial powers. Thus discarded refinement
levels are bounded by a geometric series in $3^{-4n}$, while discarded
polynomial orders are bounded by a polynomially weighted geometric
series in $M$. Given a requested uniform error, choose both truncations
to fit within it. The remaining calculation involves finitely many
matrix operations, transport equations, and parameter integrals, each
with bounded error. The argument works at any fixed derivative order
and gives a uniform approximation to the coefficient integrand.

The uniform integrand error must be converted into an error on each
numerical coefficient. For example, in
\eqref{poly:one-link-polynomial}, replacing $f$ by an approximation
$\widetilde f$ with $\sup|\widetilde f-f|\le\delta$ changes the
coefficient of $p_{m,\alpha}$ by at most
$\delta\int_G|q_{m,\alpha}(V,V^\dagger)|\,\dd V$.
The multiplier is bounded using the known polynomial
$q_{m,\alpha}$ on the compact group. The same argument applies to
each scalar component in \eqref{poly:coefficient-integral}.
Choose the integrand accuracy so that this contribution uses only part
of the assigned coefficient error.

It remains to integrate over the finite compact product $G^S$.
The continuity bounds control the numerical integration error, which
must fit within the remaining error allowance. Each retained polynomial has finitely many monomials of
known degree, so coefficient errors give bounds on its value and any
requested fixed order of link derivatives. Add this numerical error
to the truncation bound \eqref{poly:prefix-error}; the latter alone
assumes exact polynomial coefficients.

The approximate measure term must remain real if its line integral is
to define a phase of modulus one. This can be enforced by rounding
conjugate-paired monomials together, or by taking the real part of each
scalar current component in the fixed Lie-algebra basis. Reality by
itself does not preserve the exact gauge and lattice laws. Those laws
hold for the polynomials with exactly defined coefficients; errors from
numerically evaluating the coefficients are controlled separately.

There is one further choice when applying the rule on finite tori.
Choose the input neighborhoods, cutoff plaquettes, and assigned radii
on the covering lattice before imposing periodic identifications. The
input sets must contain all reconstruction and local gap neighborhoods.
If this enlarged support has no periodic identifications, the same
coefficients apply on every torus containing it. For a set that wraps
around the torus, integrate over the actual independent periodic link
variables. Its finite graph and fermion boundary signs may then label
the numerical coefficients. Impose those identifications before
averaging: two lifted copies of one periodic link are one integration
variable.

The exponential support bounds remain uniform in either case, and the
rule contains no continuous physical spacing. It specifies the
coefficients to arbitrary accuracy, with bounds on both discarded terms
and numerical evaluation. The practical computational cost remains to
be established.

For the $\SU(N_{\rm c})$ theories covered by
section~\ref{sec:other-representations}, replace the averaging weight
by $b(g)=(1+\operatorname{Re}\tr(g)/N_{\rm c})/2$ and carry out
the same construction. The concentration width is still $m^{-1/2}$,
but the group dimension is $N_{\rm c}^2-1$. Hence
$Z_m\ge c\,m^{-(N_{\rm c}^2-1)/2}$, while the approximation error
still decreases as $m^{-1/2}$. The constants may depend on the fixed
group and representation; the convergence and finite-error arguments
are otherwise unchanged.

\section{Classical continuum estimates and finite-cutoff amplitudes}
\label{phys:section}

Appendix~\ref{poly:section} expresses the measure current as an exact
series of polynomials in the original links. We now use that representation
to answer two questions. On a fixed smooth physical gauge background,
does the measure term survive as the original lattice spacing tends to
zero? On a fixed lattice, how accurately does a finite polynomial sum
reproduce the measure and fermion amplitudes? These are different limits.
Neither changes the role of auxiliary refinement, which has already
been used to construct the exact current.

For the classical continuum limit, restore the original physical spacing
$a$. The site $x$ has physical position $ax$, the periods
$\ell_\mu=aL_\mu>0$ are held fixed, and the auxiliary lattice at level
$n$ has physical spacing $ah_n$. Let $A$ range over a compact smooth
family of $\SU(5)$ connections on a fixed bundle over this four-torus.
The original field $U_a(A)$ consists of the actual edge holonomies of
$A$. For sufficiently small $a$, all these fields lie in the fixed
admissible domain. We evaluate the previously constructed current on
$U_a(A)$; we do not replace its reconstructed connection by $A$.

Write $\boldsymbol\vartheta=(\vartheta^1,\ldots,\vartheta^d)$ for the
finitely many commuting parameters of the background family. The
continuum estimates use derivatives with respect to these parameters,
rather than the original-link derivatives $D_I$. Each derivative order is fixed before taking $a\to0$. The
constants below may depend on the smooth family and these orders, but
are uniform within the family. For a Standard Model background, use
$\phi_q$ to form its associated $\SU(5)$ connection and pull the current
back along the corresponding link map. The same estimates then apply.

The current $j_\mu^b(x;U)$ defined in the main text multiplies a
dimensionless link variation. A physical current density instead enters
a spacetime sum with volume element $a^4$. To relate the two, use
L\"uscher's convention $\delta U=a\eta_{\rm L}U$. Thus the
dimensionless generator used in this paper is
\begin{equation}
 \eta_{{\rm edge},\mu}(x)
   =\delta U(x,\mu)U(x,\mu)^{-1}
   =a\,\eta_{{\rm L},\mu}(ax).
\end{equation}
Here $\eta_{\rm L}$ is the same variation expressed in physical units;
this equality defines it at the lattice sites. We use $b$ for the
Lie-algebra index in this appendix to distinguish it from the spacing
$a$. Matching the two expressions for the measure term gives
\begin{equation}
 j_{{\rm L},\mu}^b(ax;A)
     =a^{-3}j_\mu^b(x;U_a(A)),\qquad
 \mathfrak L_{\eta_{\rm edge}}
     =\sum_{x,\mu,b}\eta_{{\rm edge},\mu}^b(x)j_\mu^b(x)
     =a^4\sum_{x,\mu,b}\eta_{{\rm L},\mu}^b(ax)
                                      j_{{\rm L},\mu}^b(ax).
\end{equation}
One power of $a$ comes from the link variation and three from the
current. A dimensionless current of order $a^p$ therefore gives a
physical current of order $a^{p-3}$. We will show that the physical
current is $O(a^2)$ for the complete polynomial contributions of
appendix~\ref{poly:section}, and $O(a^6)$ for their exact sum. The
stronger result uses the additional cancellation established next.

\subsection{Vanishing derivatives at the vacuum}
\label{app:a:vacuum-jets}

At the vacuum $U(x,\mu)=1$, the exact current and its first three
original-link derivatives vanish. Appendix~\ref{app:a:locality}
proved the underlying color cancellation while treating the endpoint
links as independent variables. We must check that it still holds when
these links and their variations are computed from the original field.

Consider one finite contribution $f^I_{y,M,n,s}$ to the current in
\eqref{app:a:current-components}. The index $I$ specifies the original
link variation, $y$ the trace position, $M$ the total sign-series order,
$n$ the comparison level, and $s$ the comparison parameter. At the
vacuum, every coarse and fine endpoint link is the identity. Temporarily
use their independent logarithms as background variables and keep the
induced variation fixed. Each term is linear in one prescribed
Lie-algebra insertion. The color-trace cancellation from
\eqref{gen:ordered-traces} sets to zero every background derivative of
orders zero through three, for every value of that insertion: at most
four Lie-algebra factors occur in its color trace.

Now substitute the actual endpoint links and their induced generators
$\Xi^{I,q}_{n,\nu}(y)$ from \eqref{app:a:receiving-source}, where
$q={\rm c},{\rm f}$ labels the coarse and fine endpoints. Apply an
original-link derivative. It can differentiate an endpoint background,
or it can differentiate the induced generator, including its dependence
on the transported bases. In the second case there is still just one
Lie-algebra insertion, with a different numerical value. After at most
three such derivatives, every chain-rule term therefore contains at
most three background derivatives and one insertion. The same color
cancellation sets each term to zero. This includes the first comparison
$B_0$, since both $U$ and $F_{1/3}U$ are identities at the vacuum.

The local cutoffs equal one near the vacuum. The bounds
\eqref{app:a:scale-bounds} allow these derivatives to pass through the
sums over polynomial orders, trace positions, and refinement levels.
If $\boldsymbol K=(K_1,\ldots,K_\ell)$ is any ordered list of
original-link derivative indices, with
$D^{\boldsymbol K}=D_{K_1}\cdots D_{K_\ell}$ and
$|\boldsymbol K|=\ell$, the result is
\begin{equation}
 \left.D^{\boldsymbol K}j_\mu^b(x;U)\right|_{U=1}=0
 \qquad(0\le|\boldsymbol K|\le3).
\end{equation}
This is a statement about finitely many derivatives of the exact
current. It uses smoothness, not analyticity in the original links.

\subsection{Continuum bounds for polynomial contributions}

Take one polynomial contribution $j_\mu^{(r,k)}(x;U)$ from
appendix~\ref{poly:section}. Here $r$ is its input radius in original
lattice units and $k$ labels the initial average or a correction between
successive averages. We keep its coefficients at all output sites and
directions together, with the numerical coefficients defined by their
exact integrals. This complete collection obeys the gauge and axial
transformation laws. An isolated monomial need not do so.

We will expand this contribution on the smooth field $U_a(A)$ near an
output link and use its symmetries to eliminate the lowest powers of
$a$. First express the current at the midpoint of that link. The
current matrix $j_\mu(x)$ transforms in the color space at $ax$,
whereas reversing the link exchanges its endpoints. Transporting both
orientations to the midpoint lets us compare them in the same color
frame. For $e=(x,x+\hat\mu)$, define
$X_{x,\mu}=a(x+\tfrac12\hat\mu)$ and
$t_\mu(x;A)=\operatorname{PT}_A(X_{x,\mu}\leftarrow ax)$, the
parallel transport of the physical connection along the half edge. Set
\begin{equation}
 \widehat j_\mu^{(r,k)}(x;A)
   =t_\mu(x;A)\,j_\mu^{(r,k)}(x;U_a(A))\,t_\mu(x;A)^{-1}.
\end{equation}
Its scalar components are
$\widehat j_\mu^{b;(r,k)}=-\tr_{\mathbf5}(T^b\widehat
j_\mu^{(r,k)})$. The corresponding physical current is
$a^{-3}\widehat j_\mu^{(r,k)}$. The midpoint transport only changes
the color frame; it does not change the current or its norm.

Choose radial gauge for $A$ about this midpoint: identify nearby color
frames by parallel transport along straight radial segments, keeping the
midpoint frame fixed as the background parameters vary. Let $F$ denote
the curvature of $A$ and $\nabla^A$ its spacetime covariant derivative.
In a local frame,
$\nabla_\nu^A F_{\rho\sigma}
=\partial_\nu F_{\rho\sigma}+[A_\nu,F_{\rho\sigma}]$.
These are derivatives of physical fields, distinct from $D_I$ and from
the overlap operator $D$.

Radial gauge expresses the potential near the midpoint as curvature
integrated along radial segments. The potential vanishes at the
midpoint and grows linearly with distance there. At a fixed lattice
displacement it is consequently $O(a)$, and integration along an edge
of length $a$ gives a link logarithm beginning at $O(a^2)$.
For an input edge $f$, with fixed orientation and displacement in
original lattice units, this gives
\begin{equation}
 \log U_f=a^2L_f^{(2)}(F)
              +a^3L_f^{(3)}(\nabla^A F)+\cdots .
 \label{phys:true-input}
\end{equation}
The maps $L_f^{(2)}$ and $L_f^{(3)}$ are obtained by integrating the
radial expansion along that edge. They are linear in $F$ and
$\nabla^A F$, respectively, evaluated at the output midpoint. Work
inside a fixed physical ball where this radial gauge is defined. At
any fixed expansion order, the coefficients, the remainder, and their
fixed background-parameter derivatives grow at most polynomially with
the lattice distance of $f$ from the output.

The smooth coefficient extensions allow a Taylor expansion in the
independent input-link logarithms. Substituting \eqref{phys:true-input}
then expresses each power of $a$ in terms of curvature and its
covariant derivatives at the midpoint. Exponential locality controls
the sums over input links, including the powers of their relative
coordinates. For now retain only contributions whose enlarged input
supports lie inside the physical ball, with no periodic identifications.
Appendix~\ref{poly:section} assigns the same numerical coefficients
whenever the full enlarged support embeds in the torus. We may thus use
all axis reflections and permutations of the covering-lattice rule,
even when they do not preserve the periods of the chosen torus. We
bound the omitted supports below.

Consider the dimensionless current through order $a^4$. Its constant
term vanishes by gauge covariance: an adjoint-valued current at the
vacuum would have to be unchanged by every global $\SU(5)$ rotation,
and no nonzero such vector exists. There is no term of order $a$
because each input logarithm starts at order $a^2$.

Full inversion of the four coordinates about the midpoint excludes
every even power of $a$. To see this, assign weight two to a curvature
factor and weight one to each covariant derivative. A coefficient of
$a^p$ in the radial expansion has total weight $p$. Curvature is even
under inversion, while each derivative changes sign, so this
coefficient has parity $(-1)^p$. The axial current is odd under full
inversion: this transformation has determinant $+1$ but reverses its
vector index. Thus even coefficients cannot occur. In particular,
orders two and four vanish. The transports along both halves of the
output edge are identities in radial gauge, so reversing that edge
adds no color rotation to this argument.

Only order three remains. It is linear in
$\nabla_\nu^A F_{\rho\sigma}$, with the four spacetime indices
$\mu,\nu,\rho,\sigma$, where $\mu$ is the current direction. Its
coefficient must transform as a pseudotensor. Reflection of any one
axis requires that axis to occur an odd number of times among these
four indices. Every axis must therefore occur exactly once. Axis
permutations then fix the coefficient to be proportional to
$\epsilon_{\mu\nu\rho\sigma}$, with $\epsilon_{1234}=1$.

The color part is equally constrained. Gauge covariance requires its
linear map on $\mathfrak{su}(5)$ to commute with the adjoint action.
Such a map is a real scalar multiple of the identity. One way to see
this is to extend it complex-linearly to $\mathfrak{sl}_5(\mathbb C)$:
the adjoint representation is irreducible because this Lie algebra is
simple, and Schur's lemma makes the extension scalar. Preservation of
$\mathfrak{su}(5)$ makes the scalar real. The only possible
order-three term is therefore
\begin{equation}
 c\,\epsilon_{\mu\nu\rho\sigma}\nabla_\nu^A F_{\rho\sigma}=0,
\end{equation}
where $c$ is a real coefficient and the equality follows from the
covariant Bianchi identity. The absence of a gauge-covariant axial current of dimension three in the classical continuum analysis was already noted by L\"uscher~\cite[sec.~8]{LuscherNonabelian}. All powers through $a^4$ vanish. Hence
$\widehat j_\mu^{(r,k)}=O(a^5)$, or $O(a^2)$ in physical
normalization. This bound does not require the order-five coefficient
to be nonzero.

We also need the continuum bound for spacetime variations of the
current. Subtracting matrices at neighboring sites is not gauge
covariant unless we first bring them to a common color frame. For any
midpoint current $J_\mu(x)$, define the normalized physical difference
\begin{equation}
 (\Delta_\nu^{\rm cov}J)_\mu(x)
  =\frac{\Ad_{\tau_{\mu\nu}(x;A)}
                     J_\mu(x+\hat\nu)-J_\mu(x)}{a},
 \qquad
 \tau_{\mu\nu}(x;A)
   =\operatorname{PT}_A(X_{x,\mu}\leftarrow X_{x+\hat\nu,\mu}).
\end{equation}
Here $\Ad_\tau J=\tau J\tau^{-1}$, and the transport follows the
short straight segment between the midpoints. The numerator compares
the two currents in the color space at $X_{x,\mu}$; the denominator
is their physical separation. Write
$(\Delta_{\rm phys}^{\rm cov})^{\ell_x}$ for any ordered sequence
of $\ell_x$ such differences and
$\partial_{\boldsymbol\vartheta}^{\ell_{\rm src}}$ for any ordered
sequence of $\ell_{\rm src}$ background-parameter derivatives.
Both orders are fixed nonnegative integers.

For the explicit terms of the local expansion, a normalized difference
is bounded by a physical derivative of the smooth coefficient. It
therefore preserves the power of $a$. We need not assume this for the
remainder. Bound each difference of the remainder by a factor $a^{-1}$
and expand through a sufficiently high power $a^N$: with
$N\ge\ell_x+4$, an $O(a^{N+1})$ remainder becomes
$O(a^{N-\ell_x-2})=O(a^2)$ after the differences and the factor
$a^{-3}$ for physical normalization. Smooth-family bounds control
any fixed number of parameter derivatives in the same expansion.

For one physical difference, for example, write the dimensionless current
as $a^5 C_5+O(a^6)$, where $C_5$ is a smooth coefficient field.
The difference quotient of $C_5$ remains bounded, while that of the
remainder is at worst $O(a^5)$. Both contributions are therefore
$O(a^2)$ after physical normalization.

The polynomial factors in relative lattice coordinates are absorbed
by reducing the exponential locality rate once, to
$\lambda_{\rm phys}=\lambda_{\rm pol}/2>0$. The resulting estimate is
\begin{equation}
 \sup_{x,\mu}
 \left\|(\Delta_{\rm phys}^{\rm cov})^{\ell_x}
       \partial_{\boldsymbol\vartheta}^{\ell_{\rm src}}
       a^{-3}\widehat j_\mu^{(r,k)}(x;A)\right\|
 \leq C_{\ell_x,\ell_{\rm src}}a^2
           e^{-\lambda_{\rm phys}r}d_k,
 \quad d_{-1}=1,\quad d_k=2^{-k}\ (k\geq0).
 \label{phys:block-bound}
\end{equation}
The norm is the fixed matrix operator norm. The factor $d_k$ records
the initial average at $k=-1$ and the geometrically decreasing
corrections at $k\ge0$. Constants are uniform in the chosen smooth
family; the radius $r$ is still in original lattice units.

The inputs omitted from the physical ball obey the same estimate. Their
radius satisfies $r\ge c/a$ for a fixed $c>0$, so their exponential
locality bound is smaller than every fixed power of $a$. A fixed
portion of that decay absorbs both the difference quotients and the
physical normalization. Since the physical periods remain positive,
this also controls supports that wind around the torus and their
dependence on nontrivial flat holonomies. No global radial gauge or
expansion about a globally trivial connection is required.

Let $\widehat j_\mu(x;A)$ be the exact current in the midpoint frame,
and let $\widehat j_\mu^{[r_{\max},k_{\max}]}(x;A)$ denote the
finite sum defined in appendix~\ref{poly:section}, transported to the
same frame. Summing \eqref{phys:block-bound} over the omitted radii
and averaging corrections gives
\begin{equation}
 \sup_{x,\mu}
 \left\|(\Delta_{\rm phys}^{\rm cov})^{\ell_x}
       \partial_{\boldsymbol\vartheta}^{\ell_{\rm src}}
       a^{-3}\bigl(\widehat j_\mu-\widehat j_\mu^{[r_{\max},k_{\max}]}\bigr)
                              (x;A)\right\|
 \leq C_{\ell_x,\ell_{\rm src}}a^2
       \left(e^{-\lambda_{\rm phys}r_{\max}}+2^{-k_{\max}}\right).
 \label{phys:prefix-bound}
\end{equation}
The rates $\lambda_{\rm pol}$ and $\lambda_{\rm phys}$ are
independent of the fixed derivative orders. Thus polynomial truncation
errors retain the same $a^2$ suppression on smooth physical
backgrounds. The estimate concerns complete polynomial contributions
with their exact numerical coefficients, so that their gauge and
lattice transformation laws hold.

\subsection{The exact current and the measure response}

The exact current has more cancellation than an individual polynomial
contribution. Its first three original-link derivatives vanish at the
vacuum, as shown in subsection~\ref{app:a:vacuum-jets}. Its Taylor
expansion in link logarithms therefore begins at degree four. To use
this fact in the local radial expansion, we must account for the
supports outside the physical ball: the cancellation holds for the
full current, not necessarily for each term retained inside that ball.

For each Taylor coefficient of degree below four, the sum of retained
and omitted contributions is zero. The retained coefficient is
therefore minus the omitted tail. Every omitted support has radius
$r\ge c/a$, and the differentiated locality bounds make this tail
exponentially small in $1/a$. They also control the coordinate factors
from the radial expansion and the required fixed derivatives. Thus the
local expansion inherits the exact current's Taylor zeros up to a
remainder smaller than every fixed power of $a$.

Each link logarithm begins at $a^2$, so degree four first permits
$a^8$ in the dimensionless current. This coefficient is quartic in
$F$. Full inversion excludes it by the axial transformation law, just
as for the even powers above. The exact current is consequently
$\widehat j_\mu=O(a^9)$, giving the stronger physical bound
$a^{-3}\widehat j_\mu=O(a^6)$.

For $\ell_x$ physical differences, expand through power $a^N$
with $N\ge\ell_x+8$. The remainder then stays $O(a^6)$ after the
possible $a^{-\ell_x}$ loss and physical normalization. Fixed
background-parameter derivatives are controlled by the same
smooth-family bounds. By contrast, the averaging in
appendix~\ref{poly:section} samples a local coefficient at other input
matrices. A fixed $(r,k)$ contribution need not retain the full
current's vanishing derivatives at the vacuum. Its bound remains
$O(a^2)$ as in \eqref{phys:block-bound}.

To obtain the measure term, pair the current with the link variation
and sum over links. A smooth physical variation gives
$\eta_{{\rm edge},\mu}=O(a)$, while fixed physical volume contains
$O(a^{-4})$ links. Together with the dimensionless current bounds,
this gives an $O(a^2)$ response for finite complete polynomial sums
and an $O(a^6)$ response for the exact current. Where the fermion
operator is invertible, the exact response enters
$\delta_\eta(-\log Z_F)$ as $+\ii\mathfrak L_\eta$ in
\eqref{eq:determinant-variation}. There is thus no surviving measure
contribution to this variation on the stated smooth families as
$a\to0$. This is a classical background-field estimate.

The bounds survive fixed background-parameter derivatives and
integration along a fixed parameter interval of a smooth path of
connections, provided the variations and required derivatives remain
uniformly bounded. A changing choice of gauge frame adds a gauge
variation $X_\omega$. Its link generator can be spatially rough, so
it need not satisfy the $O(a)$ bound used above. Instead, sum the
pairing by parts to put the backward difference on the current:
\begin{equation}
 \sum_\mu(\nabla_\mu^*j_\mu)(x)
 =\sum_\mu\left[j_\mu(x)
    -\Ad_{U(x-\hat\mu,\mu)^{-1}}j_\mu(x-\hat\mu)\right].
\end{equation}
This is the dimensionless backward difference of
\eqref{eq:backward-difference}. Insert the half-edge transports that
relate $j_\mu$ to its midpoint current. Along direction $\mu$, these
transports lie on the same straight line and compose exactly: the
expression in brackets is $a$ times the midpoint covariant difference
at $x-\hat\mu$, transported to the color frame at $ax$. Thus the
unscaled backward difference supplies one additional power of $a$
relative to the normalized physical difference. Its size is $O(a^6)$
for a complete polynomial contribution and $O(a^{10})$ for the exact current. Pairing with
$\omega$ and summing over the $O(a^{-4})$ sites gives
$Ca^2\|\omega\|_\infty$ and $Ca^6\|\omega\|_\infty$,
respectively. Here $\|\omega\|_\infty$ is the maximum matrix norm
over sites. No spatial derivative of $\omega$ is needed.

Now allow the gauge frame to depend on both the path
parameter $t$ and the background parameters
$\boldsymbol\vartheta$. At each site write the change of frame as
$g(t,\boldsymbol\vartheta)$. Its generator along the path is
$\omega=(\partial_tg)g^{-1}$; along a parameter
$\vartheta^\alpha$ it is
$(\partial_{\vartheta^\alpha}g)g^{-1}$. Differentiating a covariant
current in this moving frame adds a commutator with the latter
generator. To retain a bound after $\ell_{\rm src}$ parameter
derivatives, it suffices to bound these generators and their parameter
derivatives through that order, uniformly over sites, the path, and the
background family as $a\to0$. Spatial derivatives of $g$ are not
required.

\subsection{Polynomial approximation of fixed-cutoff amplitudes}

We next keep the original lattice fixed and increase the polynomial
cutoffs $r_{\max},k_{\max}$. The physical spacing, the number of
sites, and the fermion operator $D[U]$ are unchanged. We ask first
how accurately the finite current determines the measure at a given
admissible field $U$, and then how this error affects fermion
amplitudes. Throughout the comparison, use the same base field and
normalized measure $s_0$ in each connected component as in
\eqref{eq:section}.

The infinite polynomial series already gives the exact measure.
Equation~\eqref{poly:identity} identifies its sum with
$\mathfrak L$, so inserting that sum in \eqref{eq:section}
leaves $s(U)$ unchanged. In particular, the phase obtained by
integrating the current still cancels the projector phase around
every closed loop. This equality holds along any admissible piecewise
smooth path, including paths through spatially rough fields. It
requires no path-length bound uniform as $a\to0$.

A finite sum need not cancel the projector phase exactly. Two paths
from the base field to the same endpoint can therefore give different
approximate measures. To define one approximation, choose a measurable
assignment $U\mapsto\gamma_U$ in each component. Each path joins
the fixed base field to $U$, remains admissible, is piecewise smooth,
and has finite length. Keep this assignment fixed as
$r_{\max},k_{\max}$ increase. Transport $s_0$ along the chosen
path using the original chiral projectors, and replace only the
scalar measure correction by the finite polynomial sum. For
$\gamma=\gamma_U$, the result is
\begin{equation}
\begin{aligned}
 s_\gamma^{[r_{\max},k_{\max}]}(U)
 &=\exp\left(-\ii\int_\gamma
                   \mathfrak L^{[r_{\max},k_{\max}]}\right)
       \operatorname{PT}_{\nabla^{\det}}(\gamma)s_0\\
 &=\exp\left(\ii\int_\gamma
         [\mathfrak L-\mathfrak L^{[r_{\max},k_{\max}]}]\right)
       s(U).
\end{aligned}
\label{phys:approximate-phase}
\end{equation}
The second line follows by using the same path to express the exact
measure. It shows that the exact and approximate measures belong to
the same determinant line at $U$ and differ only by a phase.
The phase has modulus one because both currents are real. A finite
approximation is therefore normalized, but it need not depend smoothly
on $U$ or satisfy the gauge-invariance condition $s(U^g)=g\,s(U)$.

The current error accumulates along the path and over the lattice
links. On a lattice with $V=\prod_\mu L_\mu$ sites, measure
the length of $\gamma$ by
\begin{equation}
 \ell_\infty(\gamma)
   =\int_0^1\max_{x,\mu}
        \|\partial_tU_t(x,\mu)U_t(x,\mu)^{-1}\|\,\dd t .
\end{equation}
The integrand is the largest dimensionless link generator along the
path. Equation~\eqref{poly:prefix-error}, with
$\lambda'=0,\ell=0$, bounds the error in each current
component uniformly on the admissible domain. Pairing with the
path generator and summing the $4V$ links therefore bounds the
phase error by the current error times $CV\ell_\infty(\gamma)$.
For a real phase difference $\theta$,
$\lvert e^{\ii\theta}-1\rvert\leq|\theta|$, so
\begin{equation}
 \|s_\gamma^{[r_{\max},k_{\max}]}(U)-s(U)\|
 \leq CV\ell_\infty(\gamma)
          \left(e^{-\lambda_{\rm pol}r_{\max}}+2^{-k_{\max}}\right).
 \label{phys:section-error}
\end{equation}
At fixed lattice and endpoint, the chosen path has finite length,
and the right-hand side tends to zero as both polynomial cutoffs
increase.

If the polynomial coefficients are evaluated numerically, their
errors must first be combined into a uniform error bound for each
retained current component, as described in
appendix~\ref{poly:section}. This additional current error adds
inside the parentheses in \eqref{phys:section-error}. Keep the
current real so that the approximate measure retains unit norm.
In a joint limit with $a\to0$, the section error tends to zero
if the entire right-hand side does, including
the growth in $V$ and the path length. If the paths themselves
vary with the cutoffs, use their actual lengths in this bound.
Figure~\ref{fig:distinct-limits} distinguishes this approximation
from the auxiliary refinement and classical continuum limits.

\begin{figure}[htbp]
  \centering
  \includegraphics[width=\linewidth]{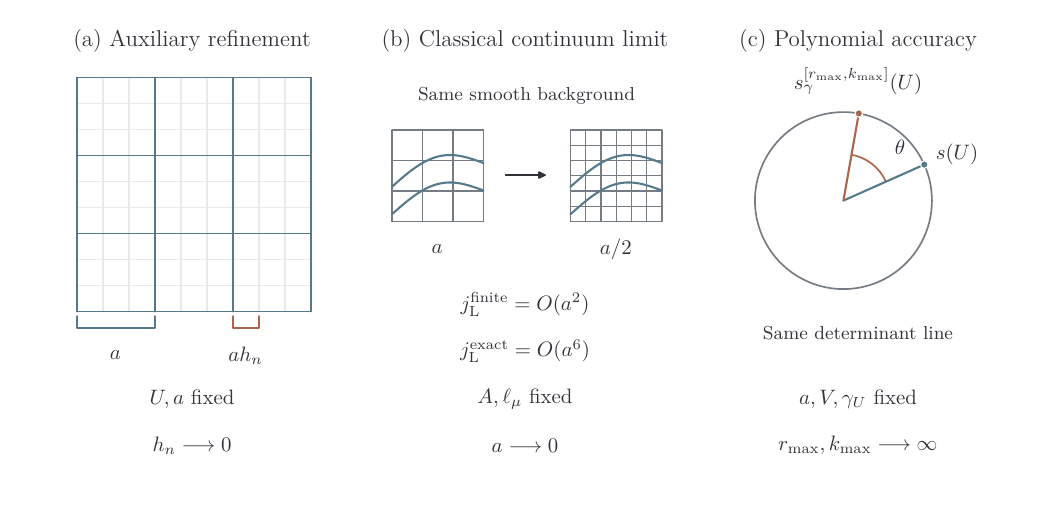}
  \caption{Three distinct limits. (a) Auxiliary refinement decreases $ah_n$, with $h_n=3^{-n}$, at fixed original links $U$ and spacing $a$. (b) The classical continuum limit decreases $a$ at fixed smooth connection $A$ and physical periods $\ell_\mu$. The displayed powers bound the physical current $j_{\rm L}$: finite complete polynomial sums use exact coefficients, and ``exact'' denotes the full current. (c) Polynomial approximation increases $r_{\max},k_{\max}$ at fixed lattice size $V$ and prescribed path $\gamma_U$. The circle shows the relative phase $\theta$ between the exact measure $s(U)$ and its finite approximation $s_\gamma^{[r_{\max},k_{\max}]}(U)$. All drawings are schematic.}
  \label{fig:distinct-limits}
\end{figure}

To obtain fermion amplitudes, couple independent external Grassmann
sources linearly to the fermion fields. These sources anticommute;
the background parameters used earlier in this appendix commute.
In the reference bases of section~\ref{sec:fermion-measure}, the
action is $\bar c M c$, and the source terms are linear in the
finite set of integration coordinates $c_k,\bar c_l$. Expand
the action and source exponentials before integrating. The expansions
terminate because every Grassmann coordinate squares to zero, and
Berezin integration selects the terms containing every integration
coordinate exactly once. Each coefficient of an external-source
monomial is consequently a finite sum of products of entries of
$M$ and reference basis vectors, multiplied by the measure phase.

For example, select source factors that leave $p$ coordinates
of each kind to be supplied by the action exponential. The remaining
integral is a $p\times p$ minor of $M$, with sign fixed by
the Grassmann ordering. A minor remains defined when $M$ is
singular or rectangular. When the fermion and antifermion ranks
differ, the source factors must supply the unmatched coordinates;
they can also supply the coordinates of fermion zero modes. This
finite polynomial description therefore includes zero-mode
amplitudes without introducing an inverse Dirac operator. For instance,
with two coordinates of each kind and $M=\operatorname{diag}(0,m)$,
where $m\ne0$, sources supplying $c_1,\bar c_1$ leave the second
pair to the action. Their coefficient is $m$, up to the fixed
Grassmann-ordering sign and measure phase, even though $\det M=0$.

Changing the measure while keeping the operator and insertions fixed
multiplies every source coefficient by the same phase in
\eqref{phys:approximate-phase}. The section bound thus gives
pointwise convergence of every coefficient at fixed lattice cutoff.
To integrate this convergence over gauge fields, we also need a
bound independent of $r_{\max},k_{\max}$. The Dirac operator
is uniformly bounded on the fixed lattice, and the reference bases
are orthonormal. The finite sums of their entries just described
therefore bound each source coefficient uniformly over admissible
fields. Replacing the measure changes only its unit-modulus phase,
so the same bound applies to every approximation.

With a fixed absolutely integrable gauge weight, the same component
normalizations in both measures, and bounded insertions, this common
bound permits dominated convergence for each unnormalized
gauge-integrated coefficient. Finite length of each selected path
is enough for pointwise convergence. The dominating bound comes
from unit norm, so the path lengths need not be bounded uniformly
over endpoints. If the exact gauge-integrated partition function
$Z$ on this fixed lattice is nonzero, the approximating partition
functions also become nonzero and division by them gives convergence
of normalized observables.

This amplitude comparison keeps the fermion operator, observable,
gauge weight, and component phases fixed. Equality of the comparison
and polynomial representations supplies no additional ultraviolet
bound for observables built from the auxiliary reconstructed
connection, and does not compare measures constructed with different
reconstruction weights.

\subsection{Other representations}

The same arguments apply to the $\SU(N_{\rm c})$ theories
covered by section~\ref{sec:other-representations}. Their adjoint
representation has no invariant vector, eliminating a constant
current, and its invariant linear maps are scalar. Together with
the same axial transformation law, this leaves
$\epsilon_{\mu\nu\rho\sigma}\nabla_\nu^A F_{\rho\sigma}$
as the only possible term through order $a^4$ in the
dimensionless polynomial current. It vanishes by the Bianchi
identity, giving the $O(a^2)$ physical current and measure
response for complete polynomial contributions.

For the exact current, the same identity
\eqref{gen:ordered-traces} applies. With one prescribed
Lie-algebra insertion, it gives the same vanishing original-link
derivatives through third order at the vacuum. The same radial expansion
and inversion argument then give the stronger $O(a^6)$
physical current and measure response. These estimates retain the
smooth-family and fixed-derivative conditions stated above. The
fixed-cutoff amplitude argument uses the finite Grassmann integral
and the current error bound. It applies under the same conditions on
prescribed paths, gauge weight, insertions, and component normalizations.

\bibliographystyle{JHEP}
\bibliography{lcf}
\end{document}